\documentclass[12pt]{article}

\usepackage{amsfonts,amssymb,amsmath,exscale,relsize,enumitem,amsthm,mathtools, physics}
\usepackage{graphicx}
\usepackage{float}
\usepackage[thinlines]{easytable}
\usepackage{csquotes}
\usepackage{setspace}
\usepackage{booktabs}
\usepackage{longtable}
\usepackage{pdflscape}
\usepackage{needspace}

\usepackage[authoryear,round]{natbib}

\usepackage[
colorlinks=true,
citecolor=blue,
linkcolor=blue,
urlcolor=blue,
pagebackref=true
]{hyperref}

\usepackage{varioref}
\usepackage{subcaption}
\usepackage{xcolor}
\usepackage{setspace}
\usepackage{appendix}
\usepackage[margin=1.35in]{geometry}
\usepackage{setspace} 

\DeclareMathOperator*{\argmax}{argmax}

\newtheorem{remark}{Remark}

\mathchardef\mhyphen="2D

\newtheorem{thm}{Theorem}[section]
\newtheorem{lem}{Lemma}[section]

\renewcommand{\vec}[1]{\boldsymbol{#1}}
\numberwithin{equation}{section}

\title{Bayesian Model Pursuit and Near-Oracle Sparse Signal Discovery Under Dependence}

\author{%
	Prasenjit Ghosh\thanks{Department of Statistics, Texas A\&M University, College Station, TX 77843, USA. Email: \texttt{prasenjit@stat.tamu.edu}}%
	\hspace{1em} 
	Arijit Chakrabarti\thanks{Applied Statistics Unit, Indian Statistical Institute, Kolkata - 700108, India. Email: \texttt{arc@isical.ac.in}}%
}

\date{} 

\begin{document}
	
	\maketitle
	
\begin{abstract}
	
Sparse signal discovery is a fundamental problem in large-scale inference, where the goal is to identify a relatively small number of active signals hidden among a large collection of null effects. Despite the ubiquity of dependence in modern applications, comparatively little is known about the extent to which dependence information can be exploited for efficient sparse signal recovery from a Bayes-risk perspective. In this paper, we develop a Bayesian Step-Down (BSD) procedure for sparse signal discovery under arbitrary known covariance dependence. BSD adopts a Bayesian model-pursuit strategy that sequentially accumulates evidence regarding competing sparse signal configurations while explicitly incorporating the covariance structure of the data.
	
To assess its effectiveness, we introduce a Bayes Oracle under a class of sparse one-factor dependence structures and compare BSD with the Oracle, the recently proposed MRD--GBS procedure of \citet{ghosh2026covariance}, the MRD procedure of \citet{COHEN_SACK_XU_2009}, and the Benjamini--Hochberg method. Our simulation studies reveal a striking phenomenon: across a wide variety of dimensions, sparsity levels, and one-factor dependence structures, BSD exhibits near-oracle behavior and often becomes virtually indistinguishable from the Bayes Oracle in terms of Bayes risk and support recovery performance. Moreover, the remarkably close empirical agreement between BSD and MRD--GBS persists across both Oracle-benchmark and general covariance settings, despite the fundamentally different Bayesian and frequentist principles underlying the two procedures. These findings provide new insight into the attainable Bayes-risk frontier for sparse signal discovery under dependence and suggest that BSD may serve as a practically useful benchmark for evaluating competing multiple testing procedures in dependent settings where exact Oracle calculations are unavailable. Finally, we show that the proposed BSD method admits a residual representation leading to admissibility under arbitrary covariance dependence together with substantial computational simplifications.
	
\end{abstract}

\medskip

\noindent \textbf{Keywords and phrases.}
Bayesian multiple testing; sparse signal discovery; Bayes Oracle; Bayesian Step-Down procedure; dependence; covariance structure; model pursuit; Bayes risk; large-scale inference; multiple testing under dependence.

\medskip

\noindent \textbf{MSC 2020 subject classifications.}
Primary 62C10, 62F03, 62F15; Secondary 62C25, 62J15, 62H15.
	
\section{Introduction}
\label{sec:introduction}

Sparse signal discovery is a fundamental problem in modern statistical inference, where the objective is to identify a relatively small number of active signals hidden among a large collection of null effects. Such problems arise naturally in genomics, bioinformatics, neuroimaging, astronomy, finance, economics, medicine, environmental science, and numerous other scientific disciplines. The resulting inferential task is inherently high-dimensional and has led to extensive developments in large-scale multiple testing methodology. A central challenge is the control of erroneous discoveries under multiplicity, motivating a substantial literature on procedures controlling global error measures such as the family-wise error rate (FWER) and the false discovery rate (FDR); see, for example, \citet{BH1995}, \citet{BL1999}, \citet{BY2001}, \citet{STO_2002}, \citet{STS_2004}, \citet{BKY2006}, \citet{BLROQ2009}, \citet{SUN_CAI_2009}, and the references therein. While FWER-controlling procedures often become overly conservative in large-scale settings, FDR-based methodologies typically provide a more attractive balance between false discoveries and false non-discoveries and have consequently become central to modern multiple testing theory and practice.

Alongside developments in error-rate control, a substantial body of work has investigated optimality properties of multiple testing procedures under sparsity from both Bayesian and frequentist perspectives. Important themes include asymptotic Bayes optimality under sparsity (ABOS), Oracle risk approximations, local false discovery rate methods, and minimaxity considerations; see, for example, \citet{ABDJ2006}, \citet{SUN_CAI_2007}, \citet{BGT2008}, \citet{BCFG2011}, \citet{NR2012}, \citet{DG2013}, \citet{GTGC2015}, \citet{GC2016}, \citet{PaulChakrabarti2025AISM}, and references therein. More recently, attention has shifted toward understanding the fundamental limits of sparse multiple testing through sharp asymptotic risk characterizations. Notable developments include the sharp multiple testing boundary established by \citet{ACR2024} and the sharp asymptotic minimaxity results for testing procedures induced by one-group global--local shrinkage priors obtained by \citet{PGC2025}. Collectively, these investigations have significantly advanced our understanding of Bayes-optimal sparse signal recovery under independence and clarified the extent to which computationally feasible procedures can approximate ideal Oracle performance.

A substantial portion of the theoretical development in multiple testing has relied on the assumption that the underlying test statistics are independent. In many scientific applications, however, substantial dependence arises naturally through spatial, temporal, biological, financial, or network-driven interactions. Ignoring such dependence may substantially distort the behavior of multiple testing procedures, and several authors have demonstrated that traditional marginal (p)-value based methods can exhibit highly unstable behavior under strong dependence; see, for example, \citet{QBKY2005}, \citet{QKY2005}, \citet{GGQY2007}, \citet{KY2007}, and \citet{QXGY2007}. Moreover, \citet{Efron_2007} argued that failure to incorporate dependence may lead to substantially misleading inferential conclusions in highly correlated settings. At the same time, dependence may contain valuable information regarding the underlying signal configuration and can therefore be exploited to improve inferential efficiency; see \citet{BH2000}, \citet{GRW2006}, and \citet{HJ2010}. Consequently, dependence may be viewed not merely as a nuisance requiring adjustment but also as a potential source of information for sparse signal recovery.

These observations have motivated a substantial literature on multiple testing procedures that explicitly account for dependence among the test statistics; see, for example, \citet{BY2001}, \citet{ROM_SHK_WOLF}, \citet{SKS2008}, \citet{LEEK_STO_2008}, \citet{BLROQ2009}, \citet{FRI_KLO_CAU_2009}, and \citet{SUN_CAI_2009}. Nevertheless, despite considerable progress on controlling global error criteria under dependence, comparatively little is known about the attainable limits of sparse signal discovery from a Bayes-risk perspective. Most existing developments focus primarily on validity properties such as FWER or FDR control, whereas much less attention has been devoted to understanding how closely a dependence-aware multiple testing procedure can approach Bayes-optimal sparse signal recovery.

A natural benchmark is provided by the Bayes Oracle. Under a specified probabilistic model, the Bayes Oracle minimizes the corresponding Bayes risk and therefore represents the ideal performance achievable by a procedure possessing complete knowledge of the underlying data-generating mechanism. Although such procedures are generally unavailable in practice, they provide a useful benchmark for assessing the effectiveness of implementable methodologies and for understanding the attainable limits of sparse signal recovery. Closely related questions concerning optimal sparse signal recovery under independence have been investigated through asymptotic Bayes optimality under sparsity, oracle-risk approximations, and sharp asymptotic minimaxity analyses; see, for example, \citet{ABDJ2006}, \citet{BGT2008}, \citet{BCFG2011}, \citet{DG2013}, \citet{GTGC2015}, \citet{GC2016}, \citet{ACR2024}, \citet{PaulChakrabarti2025AISM}, \citet{PGC2025}. By contrast, analogous questions become considerably more challenging in the presence of dependence. Once the test statistics become correlated, the likelihood no longer factorizes, posterior inclusion probabilities depend on the full observation vector, and both the structure of the Bayes Oracle and the associated risk calculations become substantially more complicated. Consequently, a comparable understanding of Bayes-optimal sparse signal recovery under dependence remains largely undeveloped. This gap provides one of the principal motivations for the present investigation.

Motivated by these considerations, we study sparse signal discovery under dependence through a Bayesian framework and revisit the Bayesian Step-Down (BSD) procedure originally proposed by \citet{GhoshChakrabarti2015}. The methodology is formulated within a two-group Bayesian model in which each hypothesis is associated with a latent indicator specifying whether the corresponding signal is active or inactive. Rather than relying solely on marginal evidence, BSD explicitly exploits the joint dependence structure of the data and proceeds sequentially through a step-down mechanism that adaptively updates the available evidence at each stage of the testing process. The resulting procedure is applicable under arbitrary known covariance dependence within the Gaussian two-groups framework and provides a natural mechanism for incorporating dependence information into the signal-discovery process.

A central premise underlying BSD is that dependence is not merely a nuisance requiring adjustment. Rather, the covariance structure itself may contain valuable information regarding the underlying active signal configuration, and sequential incorporation of this information can substantially improve sparse signal recovery. Viewed from this perspective, dependence becomes a potential source of inferential strength rather than simply an obstacle to valid inference. More broadly, sparse signal discovery may be viewed as a model search problem over a vast collection of candidate active configurations. Such a perspective is closely related to the Bayesian model selection and stochastic search literature; see, for example, \citet{GM1993,GM1997,OS2009}. In a problem involving \(n\) hypotheses, there are \(2^n\) possible configurations of active and inactive signals, rendering exhaustive exploration infeasible even for moderately large dimensions. BSD may therefore be interpreted as a posterior-guided model-pursuit procedure that sequentially accumulates evidence regarding competing sparse configurations while explicitly accounting for the dependence structure among the observations. Under this perspective, covariance information influences not only the evidence associated with individual hypotheses but also the manner in which information propagates across coordinates as the testing process evolves. Consequently, BSD may be viewed as a dependence-aware mechanism for navigating a high-dimensional configuration space. Although our theoretical and empirical investigations are conducted within a Gaussian framework, the underlying principles of model pursuit and sequential evidence accumulation are considerably more general and suggest potential extensions to a broader class of dependence-aware multiple testing and signal recovery problems.


A striking empirical phenomenon emerges from our numerical investigations. Across a broad collection of sparse one-factor dependence structures, BSD exhibits performance remarkably close to that of the corresponding Bayes Oracle. In many such settings, the two procedures become nearly indistinguishable in terms of Bayes risk, signal recovery accuracy, and overall classification performance. Remarkably, this phenomenon is already evident in relatively small and moderate dimensions, suggesting that the near-Oracle behavior of BSD is not merely an asymptotic artifact but may persist in practically relevant finite-sample regimes. This observation is particularly noteworthy because the Bayes Oracle represents the optimal decision rule under the assumed probabilistic model, whereas BSD is a computationally tractable sequential procedure motivated by Bayesian model pursuit. The empirical proximity between the two suggests that highly efficient sparse signal recovery under dependence may be achievable through careful exploitation of covariance geometry and motivates a deeper investigation into the relationship between dependence-aware model-pursuit strategies and Bayes-optimal decision rules.

Interestingly, our numerical investigations also reveal that the recently proposed MRD--GBS procedure developed by \citet{ghosh2026covariance} often exhibits performance remarkably similar to that of BSD and the Bayes Oracle in sparse regimes. This behavior is substantially less pronounced for the original MRD procedure of \citet{COHEN_SACK_XU_2009}. Although MRD--CSX often attains power comparable to, and sometimes exceeding, that of BSD and MRD--GBS, these gains are frequently accompanied by substantially inflated false discovery rates, particularly in sparse regimes, resulting in a less favorable overall risk profile. This suggests that near-oracle performance cannot be attributed solely to residual-based dependence adjustment, but also depends critically on an appropriate calibration of the sequential residual search. Since BSD and MRD--GBS arise from fundamentally different inferential paradigms, their close empirical agreement suggests that successful sparse signal recovery under dependence may be governed less by the Bayesian--frequentist distinction itself than by the ability of a procedure to effectively exploit dependence information. These findings provide additional insight into the role of covariance geometry in large-scale inference and raise the broader question of what structural features enable near-oracle performance under dependence.

Taken together, these observations motivate a substantially broader perspective on BSD than the one originally envisioned in \citet{GhoshChakrabarti2015}. Rather than viewing BSD solely as a dependence-aware Bayesian multiple testing procedure arising from the classical two-groups formulation, we investigate its role as a computationally tractable surrogate for the Bayes Oracle under dependence and as a benchmark procedure for studying sparse signal recovery in correlated settings. In particular, we establish explicit connections between BSD and covariance-adaptive residual-based testing procedures, provide a rigorous decision-theoretic justification through recent admissibility results of \citet{GC_ADMISSIBILITY_2026} under arbitrary covariance dependence, and study empirically the extent to which BSD can approximate Bayes-optimal sparse signal recovery across a broad collection of dependence structures. Viewed from this perspective, BSD serves not merely as a Bayesian testing algorithm, but as a statistically principled framework for understanding information accumulation, covariance adaptation, and near-oracle behavior in large-scale multiple testing problems under dependence.

The contributions of this paper are multifaceted.

\begin{enumerate}
	
\item We revisit, reformulate, and reinterpret the Bayesian Step-Down (BSD) procedure for sparse signal discovery under arbitrary covariance dependence. A central contribution of the paper is the development of a posterior-guided model-pursuit perspective for BSD. Rather than viewing BSD merely as a Bayesian multiple testing algorithm, we show that it may be understood as a dependence-aware sequential search mechanism that explores candidate sparse configurations by repeatedly extending the currently identified active configuration using local posterior evidence.
	
\item We introduce a Bayes Oracle under a class of sparse one-factor dependence structures and use it to investigate the attainable limits of sparse signal recovery under dependence. To the best of our knowledge, systematic Oracle-based investigations of sparse signal discovery under dependence remain relatively scarce in the multiple testing literature.
	
\item We establish an explicit residual representation of the BSD statistics and connect BSD to covariance-adaptive residual-based step-down procedures such as MRD of \citet{COHEN_SACK_XU_2009}. This representation reveals that the Bayesian posterior-odds statistics driving BSD are completely characterized by an underlying covariance-adjusted residual system, thereby providing a geometric interpretation of posterior evidence accumulation under dependence.
	
\item Combining the residual representation with recent admissibility results of \citet{GC_ADMISSIBILITY_2026} for locally monotone residual-based procedures, we provide a rigorous decision-theoretic justification for BSD under arbitrary covariance dependence. In particular, we show that BSD is admissible with respect to the vector loss function associated with the multiple testing problem. This establishes one of the strongest forms of decision-theoretic optimality available in classical statistical decision theory and demonstrates that the admissibility of BSD derives from the covariance-adaptive residual geometry underlying the procedure rather than from its Bayesian formulation alone.
	

\item We show that the residual representation also yields substantial computational simplifications. While the Bayes Oracle requires posterior evaluation over an exponentially large configuration space, BSD replaces exhaustive model exploration with a covariance-adaptive sequential search strategy whose complexity grows only quadratically with the number of hypotheses. Consequently, BSD provides a computationally tractable surrogate for Bayesian model exploration in large-scale dependence-aware multiple testing problems.
	
\item We conduct extensive simulation studies comparing BSD with the Bayes Oracle, the recently proposed MRD--GBS \citep{ghosh2026covariance} and the MRD--CSX \citep{COHEN_SACK_XU_2009} procedures, and the Benjamini--Hochberg method. Our numerical investigations reveal a striking empirical phenomenon: across a broad collection of sparsity levels, dimensions, and one-factor dependence structures, BSD consistently exhibits near-oracle behavior and often becomes virtually indistinguishable from the Bayes Oracle in terms of Bayes risk and signal recovery performance.
	
\item We demonstrate that BSD provides a useful empirical benchmark for evaluating multiple testing procedures under dependence. While exact Oracle calculations rapidly become infeasible under general covariance structures, the Bayes risk of BSD can often be estimated accurately through simulation, thereby providing valuable information regarding the attainable performance frontier in dependence-aware sparse signal discovery problems. 

\item Our numerical investigations reveal an additional and potentially surprising phenomenon. Although MRD–GBS was developed from a fundamentally different frequentist perspective and was not derived through Bayes-risk considerations or posterior model-pursuit arguments, its performance is often remarkably close to that of BSD and the Bayes Oracle across a broad collection of sparse one-factor dependence settings. This observation suggests that successful sparse signal recovery under dependence may be governed less by the Bayesian–frequentist distinction itself than by the ability of a procedure to effectively exploit covariance structure and accumulate information throughout the sequential testing process.
	
\end{enumerate}

The remainder of the paper is organized as follows. Section~\ref{MODEL_PROBLEM_FORMULATION} introduces the multiple-testing framework, the underlying Bayesian formulation, and the Bayes Oracle benchmark. Section~\ref{sec:BSD_METHOD} develops the Bayesian Step-Down procedure and discusses its interpretation from a model-pursuit perspective. Section~\ref{sec:THEORETICAL_PROPERTIES} establishes several theoretical properties of BSD, including its connection with the MRD framework of \citet{COHEN_SACK_XU_2009}, the admissibility of the resulting testing procedure, and computational complexity considerations. Section~\ref{sec:SIMULATION} presents an extensive simulation study under a broad collection of dependence structures and compares the performance of BSD with the Bayes Oracle and several competing procedures. Section~\ref{sec:DISCUSSION} concludes with a discussion of the findings and several directions for future research. Additional technical details and supplementary simulation results are deferred to the Appendix.

\section{Problem Formulation and the Bayes Oracle}
\label{MODEL_PROBLEM_FORMULATION}

We consider the problem of simultaneously testing the means of a collection of jointly normal random variables. Let $\boldsymbol{X} = (X_1,\ldots,X_n)$ denote an observed random vector satisfying
\begin{equation}\label{CONDITIONAL_DIST_X}
\boldsymbol{X}\mid\boldsymbol{\theta}\sim N_n(\boldsymbol{\theta}, \boldsymbol{\Sigma}),
\end{equation}

where $\boldsymbol{\theta} = (\theta_1,\ldots,\theta_n)\in\mathbb{R}^{n}$ is an $n\times 1$ unknown mean vector and $\boldsymbol{\Sigma} = ((\sigma_{ij}))
$ is an $n\times n$ known positive definite covariance matrix with an arbitrary dependence structure. Our objective is to simultaneously test

\begin{equation}
	\label{eq:TESTING_PROBLEM_Bayesian}
	H_{0i}:\theta_i=0
	\qquad
	\mbox{against}
	\qquad
	H_{Ai}:\theta_i\neq 0,
	\qquad
	i=1,\ldots,n.
\end{equation}


Throughout the paper, we view (\ref{eq:TESTING_PROBLEM_Bayesian}) not merely as a collection of individual testing problems, but as a structured signal recovery problem in which the primary objective is to identify the unknown configuration of true and false null hypotheses while accounting for the dependence structure among the observations. This perspective naturally motivates the use of latent indicator variables and Bayesian model-selection ideas, which we introduce next.


%
%
%
%
%
%
%
%
%
%
%

\subsection{A Bayesian Formulation}

To study the attainable limits of sparse signal recovery under dependence, we adopt a standard two-groups formulation for the unknown mean vector $\boldsymbol{\theta}$. For each coordinate, let $\nu_i$ denote an unobserved indicator variable taking the value $1$ if the corresponding signal is active and $0$ otherwise. The latent indicators are assumed to satisfy

$$\nu_i
\stackrel{\mathrm{i.i.d.}}{\sim}
\mathrm{Bernoulli}(p),
\qquad
i=1,\ldots,n,$$

where $p\in(0,1)$ denotes the prior probability that a signal is present which is also interpreted as the theoretical proportion of active signals.

Conditional on $\nu_i=0$, we assume that $\theta_i=0$, whereas conditional on $\nu_i=1$,

$$\theta_i \sim N(0,\psi^2),$$

where $\psi^2>0$ controls the signal strength that is typically assumed to be large. Equivalently, the individuals $\theta_i$'s are modeled as independent and identically distributed observations from the following two-groups prior distribution:

\begin{equation}
	\label{eq:two_group_prior}
	\theta_i
	\stackrel{\mathrm{i.i.d.}}{\sim}
	(1-p)\delta_0
	+
	pN(0,\psi^2),
	\qquad
	i=1,\ldots,n.
\end{equation}

The above two-groups formulation provides a probabilistic description of sparse signal configurations and serves as the foundation for the Bayes Oracle considered throughout this paper. Closely related hierarchical two-group formulations have played a central role in Bayesian multiple testing, multiplicity adjustment, and sparse model selection; see, for example, \citet{SB2006,SB2010}.


Let $\boldsymbol{\nu} = (\nu_1,\ldots,\nu_n)$ denote the vector of latent indicators and define $\boldsymbol{B}_{\boldsymbol{\nu}}
= \mathrm{diag}(\nu_1,\ldots,\nu_n)$.

Under \eqref{CONDITIONAL_DIST_X} and (\ref{eq:two_group_prior}), the conditional distribution of $\boldsymbol{X}$ given $\boldsymbol{\nu} = (\nu_1,\ldots,\nu_n)$ is given by
\begin{equation}
	\label{eq:conditional_x}
	\boldsymbol{X}
	\mid
	\boldsymbol{\nu}
	\sim
	N_n
	\left(
	\boldsymbol{0},
	\,
	\boldsymbol{\Sigma}
	+
	\psi^2 \boldsymbol{B}_{\boldsymbol{\nu}}
	\right).
\end{equation}

The corresponding prior probability of a configuration
$\boldsymbol{\nu}\in\{0,1\}^{n}$
is

\begin{equation}\label{eq:prior_latent_configuration}
\pi(\boldsymbol{\nu})
=
\prod_{i=1}^{n}
p^{\nu_i}
(1-p)^{1-\nu_i}.
\end{equation}

Consequently, by averaging over all possible configurations $\boldsymbol{\nu}\in\{0,1\}^{n}$, we obtain the marginal distributions of $	\boldsymbol{X}$ as

\begin{equation}
	\label{eq:marginal_x}
	\boldsymbol{X}
	\sim
	\sum_{\boldsymbol{\nu}\in\{0,1\}^{n}}
	\pi(\boldsymbol{\nu})
	N_n
	\left(
	\boldsymbol{0},
	\,
	\boldsymbol{\Sigma}
	+
	\psi^2 \boldsymbol{B}_{\boldsymbol{\nu}}
	\right).
\end{equation}

The testing problem (\ref{eq:TESTING_PROBLEM_Bayesian}) is therefore equivalent to simultaneously testing

\begin{equation}
	\label{eq:testing_nu}
	H_{0i}:\nu_i=0
	\qquad
	\mbox{against}
	\qquad
	H_{Ai}:\nu_i=1,
	\qquad
	i=1,\ldots,n.
\end{equation}

\subsection{A Decision-Theoretic Formulation}

The two-groups formulation naturally induces a multiple decision problem concerning the latent signal configuration $$
\boldsymbol{\nu} = (\nu_1,\ldots,\nu_n)$$. Let
$$
\boldsymbol{\phi}(\boldsymbol{X})
=
\left(
\phi_1(\boldsymbol{X}),
\ldots,
\phi_n(\boldsymbol{X})
\right)
$$

denote a multiple testing procedure, where

\[
\phi_i(\boldsymbol{X})
=
\begin{cases}
	1, & \text{if } H_{0i} \text{ is rejected},\\
	0, & \text{if } H_{0i} \text{ is not rejected}.
\end{cases}
\]

Following the standard compound decision framework, we consider an additive loss function obtained by summing the losses incurred by the individual testing decisions. Specifically, for each coordinate, define

\begin{equation}
	\label{eq:individual_loss}
	L_i\!\left(\nu_i,\phi_i(\boldsymbol{X})\right)
	=
	\delta_0(1-\nu_i)\phi_i(\boldsymbol{X})
	+
	\delta_A\nu_i\left\{1-\phi_i(\boldsymbol{X})\right\},
\end{equation}

where $\delta_0>0$ and $\delta_A>0$ denote the loss associated with a false discovery (Type I error), and a missed discovery (Type II error), respectively.


The overall loss of a multiple testing procedure is then defined as

\begin{equation}
	\label{eq:additive_loss}
	L\!\left(
	\boldsymbol{\nu},
	\boldsymbol{\phi}
	\right)
	=
	\sum_{i=1}^{n}
	L_i\!\left(
	\nu_i,
	\phi_i(\boldsymbol{X})
	\right).
\end{equation}

The corresponding Bayes risk of a multiple testing procedure
$\boldsymbol{\phi}$
is given by

\begin{equation}
	\label{eq:bayes_risk}
	R(\boldsymbol{\phi})
	=
	E\!\left[
	L\!\left(
	\boldsymbol{\nu},
	\boldsymbol{\phi}
	\right)
	\right],
\end{equation}

where the expectation is taken with respect to the joint distribution induced by the two-groups prior and the sampling model. When $\delta_0=\delta_A=1$, the additive loss reduces to the total number of classification errors, and the corresponding Bayes risk equals the expected sum of false discoveries and false non-discoveries. This special case will play a central role in the numerical investigations reported later in the paper. The Bayes Oracle considered in the next subsection is the multiple testing procedure that minimizes the risk \eqref{eq:bayes_risk}.

\subsection{The Bayes Oracle}

%
%
%
%
%
%
%
%
%

Under the additive loss formulation (\ref{eq:additive_loss}), the Bayes Oracle is the multiple testing procedure that minimizes the Bayes risk (\ref{eq:bayes_risk}) under the assumed two-groups model. Standard Bayesian decision theory implies that the Oracle rejects the $i$th null hypothesis whenever the posterior expected loss of rejection is smaller than that of acceptance. Consequently, the Bayes Oracle rejects the $i$-th null hypothesis $H_{0i}$ whenever the corresponding posterior inclusion probability

\begin{equation}
	\label{eq:oracle_rule}
	P(\nu_i=1 \mid \boldsymbol{x})
	>
	\frac{\delta_0}{\delta_0+\delta_A},
	\qquad i=1,\ldots,n.
\end{equation}

Thus, the Bayes Oracle is completely characterized by the posterior inclusion probabilities defined as
\begin{equation}
	\label{eq:posterior_prob}
	P(\nu_i=1|\boldsymbol{x})
	=
	\frac{
		\sum\limits\limits_{\boldsymbol{\nu}\in\{0,1\}^{n}:\nu_i=1}
		\pi(\boldsymbol{\nu})
		f(\boldsymbol{x}\mid\boldsymbol{\nu})
	}
	{
		\sum\limits_{\boldsymbol{\nu}\in\{0,1\}^{n}}
		\pi(\boldsymbol{\nu})
		f(\boldsymbol{x}\mid\boldsymbol{\nu})
	}, \qquad i=1,\ldots,n.
\end{equation}
where $f(\boldsymbol{x}\mid\boldsymbol{\nu})$ denotes the conditional density of $\boldsymbol{X}$ given $\boldsymbol{\nu}$.




%
%
%


The Bayes risk of the Oracle is given by

\begin{equation}
	\label{eq:oracle_risk}
	R_{\mathrm{Oracle}}
	=
	\sum\limits_{\boldsymbol{\nu}\in\{0,1\}^{n}}
	\pi(\boldsymbol{\nu})
	\int_{\mathbb{R}^{n}}
	L\!\left(
	\boldsymbol{\nu},
	\boldsymbol{\phi}^{\mathrm{Oracle}}(\boldsymbol{x})
	\right)
	f(\boldsymbol{x}\mid\boldsymbol{\nu})
	\,d\boldsymbol{x},
\end{equation}

where
\(
\pi(\boldsymbol{\nu})
\)
denotes the prior probability distribution \eqref{eq:prior_latent_configuration} of the latent configuration
\(
\boldsymbol{\nu}
\)
and
\(
f(\boldsymbol{x}\mid\boldsymbol{\nu})
\)
denotes the corresponding conditional density \eqref{eq:conditional_x}. 

Since the Oracle minimizes the Bayes risk among all multiple testing procedures under the additive loss (\ref{eq:additive_loss}), it represents the ideal benchmark for assessing the attainable limits of sparse signal recovery under the assumed probabilistic model. The corresponding Bayes risk is given by (\ref{eq:oracle_risk}). In the next subsection, we discuss the computational challenges associated with implementation of the Oracle and evaluation of its risk under general covariance dependence.

%

\subsection{The Bayes Oracle Under General Dependence: Computational Challenges}

The Bayes Oracle provides the ideal benchmark for sparse signal recovery under the additive loss formulation introduced in Section~\ref{MODEL_PROBLEM_FORMULATION}.2. As shown in Section~\ref{MODEL_PROBLEM_FORMULATION}.3, the Oracle rejects the $i$th null hypothesis whenever the posterior probability of the corresponding alternative exceeds a threshold determined by the relative costs of false discoveries and missed discoveries. Consequently, the Oracle is completely characterized by the collection of posterior inclusion probabilities
\[
P(\nu_i=1\mid\boldsymbol{X}),
\qquad
i=1,\ldots,n.
\]

Under the dependent two-groups model, however, evaluation of these posterior inclusion probabilities becomes computationally challenging. Indeed, for each coordinate $i$,
\[
P(\nu_i=1\mid\boldsymbol{X})
=
\frac{
	\sum\limits_{\boldsymbol{\nu}\in\{0,1\}^{n}:\nu_i=1}
	\pi(\boldsymbol{\nu})
	f(\boldsymbol{X}\mid\boldsymbol{\nu})
}{
	\sum\limits_{\boldsymbol{\nu}\in\{0,1\}^{n}}
	\pi(\boldsymbol{\nu})
	f(\boldsymbol{X}\mid\boldsymbol{\nu})
},
\]
where the summations extend over all possible latent signal configurations. Under a general covariance structure $\boldsymbol{\Sigma}$, the likelihood does not admit a factorized representation and the posterior probability associated with any particular coordinate depends on the entire latent configuration vector. Consequently, computation of a single posterior inclusion probability requires summation over a model space of cardinality $2^n$. As pointed out by \citet{XCML_2011}, exact implementation of the full Bayes Oracle therefore entails computational complexity of order $O(n2^n)$.

The computational burden associated with the Bayes Oracle extends well beyond the evaluation of a single Oracle decision rule. The primary objective of the present study is not merely to compute Oracle decisions but to estimate the corresponding Bayes risks accurately enough to permit meaningful comparisons among competing procedures. Achieving such accuracy requires averaging the Oracle loss over a large number of independently generated datasets. In the present study, this entails $G=5000$ Monte Carlo replications at each sparsity level and for each dependence structure considered. Since exact Oracle implementation for a single dataset already requires computation on the order of $O(n2^n)$, the overall computational cost of Oracle risk estimation grows exponentially with dimension and must be incurred repeatedly throughout the simulation experiment. Consequently, the practical limitation is not merely the evaluation of the Oracle itself, but the challenge of obtaining sufficiently accurate Oracle risk estimates to serve as a reliable benchmark. Thus, despite its conceptual importance as the optimal benchmark under the assumed model, the Bayes Oracle rapidly becomes computationally inaccessible under general covariance dependence, thereby restricting exact Oracle comparisons to relatively low-dimensional problems.

An important exception arises under certain structured dependence models, most notably one-factor models. In this setting, the latent factor representation induces conditional independence among the observations given the common factor, leading to substantial simplifications in both posterior calculations and risk evaluation. Consequently, one-factor models provide a rare setting in which accurate numerical approximations to Oracle risk become feasible and therefore offer a natural laboratory for studying the attainable limits of sparse signal recovery under dependence.

\subsection{A Tractable Bayes Oracle Under One-Factor Dependence}

An important exception to the computational difficulties described above arises under one-factor dependence structures. Suppose that the observations admit the latent factor representation

\[
X_i=\theta_i+\lambda_i Z+\varepsilon_i,
\qquad
i=1,\ldots,n,
\]

where \(Z\sim N(0,1)\) is a common latent factor, \(\lambda_i\in\mathbb{R}\) denotes the loading of the \(i\)th coordinate on the common factor, and

\[
\varepsilon_i\sim N(0,\sigma_i^2),
\qquad
i=1,\ldots,n,
\]

are mutually independent idiosyncratic error terms independent of \(Z\). Let

\[
\boldsymbol{\lambda}
=
(\lambda_1,\ldots,\lambda_n)^T
\]

denote the corresponding vector of factor loadings. 

Consequently, the covariance matrix of
\(\boldsymbol{X}\)
may be written as

\[
\boldsymbol{\Sigma}
=
\boldsymbol{D}+\boldsymbol{\lambda}\boldsymbol{\lambda}^T,
\]

where

\[
\boldsymbol{D}
=
\operatorname{diag}
(\sigma_1^2,\ldots,\sigma_n^2).
\]

%
%
%
%
%
%

Conditional on $Z=z$, the observations \(X_1,\ldots,X_n\) become independent and the corresponding conditional posterior inclusion probabilities factorize across coordinates as follows.
For \(d\in\{0,1\}\), define
\[
f_{d,i}(x_i\mid z)
\]
to be the conditional density of \(X_i\) given \(Z=z\) and \(\nu_i=d\). Under the present normal two-groups model,
\[
f_{0,i}(x_i\mid z)
=
\phi_{\sigma_i^2}(x_i-\lambda_i z),
\qquad
f_{1,i}(x_i\mid z)
=
\phi_{\sigma_i^2+\psi^2}(x_i-\lambda_i z),
\]
where \(\phi_{\tau^2}\) denotes the \(N(0,\tau^2)\) density. Let
\[
m_i(x_i\mid z)
=
(1-p)f_{0,i}(x_i\mid z)
+
p f_{1,i}(x_i\mid z).
\]
Then the posterior inclusion probability appearing in the Oracle rule can be written as
\[
P(\nu_i=1\mid \boldsymbol{x})
=
\frac{
	\int_{\mathbb R}
	\phi(z)\,
	p f_{1,i}(x_i\mid z)
	\prod_{k\neq i} m_k(x_k\mid z)
	\,dz
}{
	\int_{\mathbb R}
	\phi(z)\,
	\prod_{k=1}^{n} m_k(x_k\mid z)
	\,dz
}.
\]
Thus, under the one-factor model, computation of the Oracle rule no longer requires summation over all \(2^n\) latent configurations. Instead, each posterior inclusion probability can be evaluated by one-dimensional numerical integration with respect to the common factor \(Z\).

Consequently, the Bayes risk of the Oracle may be accurately estimated by Monte Carlo simulation from the one-factor model, with the Oracle decisions evaluated through the above one-dimensional quadrature. This avoids direct enumeration of the full configuration space and makes Oracle-based benchmarking feasible in dimensions where brute-force calculation would be computationally prohibitive.

Nevertheless, the computational intractability of the Oracle under a general covariance structure $\boldsymbol{\Sigma}$ motivates the search for practically implementable procedures that can effectively exploit dependence information while retaining much of the sparse recovery performance of the optimal rule. The Bayesian Step-Down procedure developed in the next section is motivated precisely from this perspective.

\section{The Bayesian Step-Down Procedure}
\label{sec:BSD_METHOD}

The Bayes Oracle developed in Section~\ref{MODEL_PROBLEM_FORMULATION} provides the ideal benchmark for sparse signal recovery under the assumed two-groups model. However, under general covariance dependence, exact implementation of the Oracle requires exploration
of a model space containing $2^n$ competing signal configurations and rapidly becomes computationally infeasible. Consequently, rather than attempting to directly approximate the full Bayes rule, it is natural to seek computationally
tractable procedures that can effectively exploit dependence information while retaining desirable statistical and decision-theoretic properties.

The Bayesian Step-Down (BSD) procedure proposed by
\citet{GhoshChakrabarti2015} was originally developed from this perspective. Although formulated within the classical two-groups framework described in Section~\ref{MODEL_PROBLEM_FORMULATION}, BSD may be viewed more broadly as a dependence-aware Bayesian model-pursuit procedure for sparse signal discovery. The central premise underlying BSD is that dependence is not merely a nuisance requiring adjustment. Rather, the covariance structure itself contains information regarding the underlying active configuration. BSD seeks to exploit this information through a sequential evidence-accumulation mechanism that adaptively updates the inferential landscape after every rejection. Viewed from this perspective, BSD is not simply a multiple testing
procedure, but a covariance-guided search strategy through a vast collection of competing sparse signal configurations.


Recall that a frequentist step-down procedure begins by examining the most significant test statistic and asking whether at least one null hypothesis can be rejected. Equivalently, one may ask whether the global null hypothesis can plausibly explain the observed data. A natural Bayesian analogue of this question may be formulated within the two-groups model.

Rather than comparing the global null hypothesis against the entire model space $\{0,1\}^{n}$, BSD restricts attention to the collection of models containing exactly one active signal,
\[
\mathcal{M}_1
=
\left\{
\boldsymbol{\nu}\in\{0,1\}^{n}
:
\sum_{i=1}^{n}\nu_i=1
\right\}.
\]

These models constitute the simplest departures from the global null and therefore represent the most plausible alternatives at the initial stage of the search process. For each coordinate $j$, one compares the posterior probability of the model in which the $j$-th hypothesis is the unique signal with the posterior probability of the global null model. If the largest of these posterior odds exceeds a predetermined threshold, the corresponding null hypothesis is rejected and the associated coordinate is incorporated into the current active configuration. The corresponding observation $X_j$ is then set aside and removed from future competition. Nevertheless, its effect continues to propagate through the covariance structure and influences all subsequent posterior comparisons. The procedure then proceeds to the next stage with the reduced set of candidate coordinates. Otherwise, the procedure terminates and accepts all remaining null hypotheses.

The same principle is then applied recursively. After a signal has been identified, attention is restricted to the remaining active set of hypotheses, and the procedure again compares the local null model against the collection of models containing exactly one additional signal among the remaining coordinates. In this way, BSD performs a sequential exploration of increasingly complex signal configurations, updating the inferential geometry after every rejection.

An important feature of this mechanism is that, once a coordinate is declared active, it is removed from further consideration in subsequent stages of the search. Thus, BSD does not repeatedly revisit previously identified signals. Instead, the procedure progressively enlarges the current active configuration while simultaneously reducing the set of remaining candidate coordinates. Consequently, the search space contracts adaptively as evidence accumulates. In this sense, BSD effectively performs a sequential exploration of the configuration space by adding one signal at a time and removing it from future competition. Through the dependence structure encoded by the covariance matrix, each newly identified signal also alters the inferential landscape for the remaining coordinates, allowing information to propagate throughout the sequential search process.

Viewed from this perspective, BSD may be interpreted as a dependence-aware posterior model-pursuit procedure. Rather than attempting to approximate the full Bayes rule over the entire model space, BSD constructs a sequence of local posterior comparisons that guide the search for the underlying active configuration. Because these comparisons are based on the joint likelihood, the resulting procedure incorporates the covariance structure at every stage. Consequently, each rejection modifies the evidence available to the remaining hypotheses through the dependence encoded by $\boldsymbol{\Sigma}$, allowing information to accumulate adaptively throughout the search process.

We emphasize that our goal is not to approximate the full Bayes rule directly, but rather to develop a computationally tractable dependence-aware procedure that performs well relative to natural Bayesian benchmarks while possessing attractive decision-theoretic properties.

\subsection{The Bayesian Step-Down Procedure}

Viewed as a model-pursuit procedure, BSD constructs the active configuration sequentially. Suppose that, before stage $t$, the coordinates
\[
j_1,\ldots,j_{t-1}
\]
have already been identified as active. These coordinates are retained as part of the current active configuration, but the corresponding observations
\[
X_{j_1},\ldots,X_{j_{t-1}}
\]
are set aside and removed from future competition, although their influence continues to propagate through the accumulated active configuration and the covariance structure. Thus, the posterior comparisons at stage $t$ are carried out using the reduced data vector
\[
\boldsymbol{X}^{(j_1,\ldots,j_{t-1})},
\]
consisting of the observations whose indices remain under consideration. The procedure then searches among these remaining coordinates for the most plausible candidate signal to be added next. Operationally, this search is implemented through a collection of stagewise posterior-odds statistics that compare the local null configuration on the reduced active set with configurations obtained by adding exactly one additional signal among the remaining hypotheses.


 Let us define
\begin{eqnarray}\label{eq:BSD_STAT_DEFN}
	S_{tj}^{(j_1,\dots,j_{t-1})}(\vec{X}) &=& \frac{\pi\left(\nu_j=1,\vec{\nu}^{(j_1,\dots,j_{t-1}, j)}=\vec{0}\mid\vec{X}^{(j_1,\dots,j_{t-1})}\right)}{\pi\left(\nu_j=0,\vec{\nu}^{(j_1,\dots,j_{t-1},j)}=\vec{0}\mid \vec{X}^{(j_1,\dots,j_{t-1})}\right)}
\end{eqnarray}
for $t,j=1,\dots,n,$ $1\leqslant j_{1}\neq\dots\neq j_{t-1}\leqslant n$ and $j_{\ell} \neq j$ for all $\ell = 1,\dots,t-1$.

\subsubsection{Interpretation of the BSD Statistics}

The BSD statistics admit a simple interpretation within the model-pursuit framework described above. At stage $t$, the coordinates
\[
j_1,\ldots,j_{t-1}
\]
have already been incorporated into the current active configuration. For each remaining coordinate $j$, the statistic
\[
S_{tj}^{(j_1,\ldots,j_{t-1})}(\boldsymbol{X})
\]
compares the posterior probability of the configuration obtained by augmenting the current active configuration with signal $j$ to the posterior probability of the current active configuration itself.

Equivalently, for each fixed stage $t$, the numerator of
\[
S_{tj}^{(j_1,\ldots,j_{t-1})}(\boldsymbol{X})
\]
may be viewed as the posterior probability of a plausible alternative obtained by adding exactly one additional signal to the current active configuration, while the denominator corresponds to the posterior probability of the associated local null configuration. Thus,
$S_{tj}^{(j_1,\ldots,j_{t-1})}(\boldsymbol{X})$ represents a stagewise posterior odds in favour of extending the current active configuration through coordinate $j$.

Consequently, the BSD statistic quantifies the posterior evidence in favour of including coordinate $j$ among the active signals. Larger values of
\[
S_{tj}^{(j_1,\ldots,j_{t-1})}(\boldsymbol{X})
\]
indicate stronger support for adding coordinate $j$ to the currently identified active set.

Unlike the Bayes Oracle, which attempts to evaluate posterior
probabilities over the entire configuration space simultaneously, BSD performs a sequence of local comparisons between neighboring configurations. At each stage, only those configurations obtained by adding a single signal to the current active configuration are considered. Consequently, BSD explores the model space incrementally, constructing increasingly complex signal configurations while avoiding
the combinatorial burden associated with exhaustive enumeration of all $2^n$ possibilities. Thus, BSD may be viewed as a greedy posterior model-pursuit strategy that sequentially constructs increasingly complex active configurations through a sequence of locally optimal posterior extensions.


\subsubsection{The Proposed Procedure}

For $t=1,\dots,n$, let us define the indices $j_t(\vec{X})$ as,
\begin{align}\label{eq:BSD_INDEX}
	j_t(\vec{X})&=\argmax_{j \in \{1,\dots,n\}\setminus\{j_1(\vec{X}),\dots,j_{t-1}(\vec{X})\}} S^{(j_1(\vec{X}),\dots,j_{t-1}(\vec{X}))}_{tj}(\vec{X}).
\end{align}

Given a predetermined threshold $\delta > 0$, the proposed Bayesian Step Down (BSD) procedure is now described below.

\begin{enumerate}
	\item At stage 1, consider the statistics $S_{1j}(\vec{X})$, where $j\in\{1,\dots,n\}$. If $S_{1j_1(\vec{X})}(\vec{X})\leqslant \delta$, stop and accept all the $H_{0i}$'s. Otherwise, reject $H_{0j_1(\vec{X})}$ and continue to stage 2.
	
	\item At stage 2, consider the statistics $S^{(j_1(\vec{X}))}_{2j}(\vec{X})$, where $j\in \{1,\dots,n\}\setminus\{j_1(\vec{X}) \}$. If $S^{(j_1(\vec{X}))}_{2j_2(\vec{X})}(\vec{X})\leqslant\delta$, stop and accept all the remaining $H_{0i}$'s. Otherwise, reject $H_{0j_2(\vec{X})}$ and continue to stage 3.
	
	\item In general, at stage $t$, consider the $(n-t+1)$ many statistics $S^{(j_1(\vec{X}),\dots,j_{t-1}(\vec{X}))}_{tj}(\vec{X})$, where $j \in \{1,\dots,n\}\setminus\{j_1(\vec{X}),\dots,j_{t-1}(\vec{X})\}$. If $S^{(j_1(\vec{X}),\dots,j_{t-1}(\vec{X}))}_{tj_t(\vec{X})}(\vec{X})\leqslant\delta$, stop and accept all the remaining $H_{0i}$'s. Otherwise, reject $H_{0j_t(\vec{X})}$ and move to stage $(t+1)$.
	
	\item We continue in this way until an acceptance occurs or we are exhausted with all the null hypotheses (that is $t=n$), in which case we must stop.
\end{enumerate}

The threshold $\delta$ is chosen to reflect the relative loss assigned to false discoveries and missed discoveries. In particular, motivated by the Bayes Oracle rule in (\ref{eq:oracle_rule}), we take
\[
\delta=\frac{\delta_0}{\delta_A}.
\]
Indeed, the Bayes Oracle rejects whenever
\[
\frac{P(\nu_i=1\mid\boldsymbol{x})}
{P(\nu_i=0\mid\boldsymbol{x})}
>
\frac{\delta_0}{\delta_A},
\]
which is equivalent to (\ref{eq:oracle_rule}). Thus, when the BSD statistic is interpreted as a stagewise posterior odds in favour of augmenting the current active configuration by including coordinate $j$, the choice $\delta=\delta_0/\delta_A$ provides a natural decision-theoretic calibration consistent with the Bayes Oracle. In the symmetric loss case $\delta_0=\delta_A=1$, this reduces to the threshold $\delta=1$.


Under independence, the BSD statistics coincide exactly with the posterior-odds quantities driving the Bayes Oracle rule, and the resulting BSD procedure reduces to the Bayes-optimal multiple testing procedure. Thus, BSD recovers the Oracle solution whenever the observations are independent.

A key distinction from the independent setting is that, under dependence, information regarding the active configuration is no longer contained solely within individual coordinates. Rather, evidence propagates through the covariance structure and must be incorporated into the sequential search process.

An important feature of the BSD statistics is that they explicitly incorporate the dependence structure among the observations through the covariance matrix $\boldsymbol{\Sigma}$ and its appropriate submatrices appearing in the stagewise posterior probabilities. Although the observations corresponding to previously selected coordinates are removed from future competition, they are not removed from future influence. Indeed, the statistics
\[
S_{tj}^{(j_1,\ldots,j_{t-1})}(\boldsymbol{X})
\]
computed at stage $t$ depend implicitly on the earlier observations through the active configuration accumulated during the preceding stages of the procedure. Consequently, decisions made at one stage influence the posterior comparisons performed at subsequent stages, allowing information to propagate through the covariance structure as the search evolves.


Although the Bayes Oracle is formally defined through posterior probabilities over the entire configuration space, exhaustive exploration of all $2^n$ configurations may be unnecessary in sparse settings. Under sparsity, the posterior distribution tends to concentrate on a comparatively small collection of configurations containing only a few active signals, while the vast majority of candidate models receive negligible posterior support. BSD exploits this concentration phenomenon by sequentially pursuing the most plausible posterior extensions of the current active configuration. Consequently, rather than attempting a global search over the entire model space, BSD concentrates its computational effort on regions that accumulate substantial posterior mass and are therefore most likely to contain the true signal configuration. Viewed from this perspective, BSD functions as a posterior-guided model-pursuit strategy that seeks to capture the dominant posterior structure of the sparse recovery problem without incurring the combinatorial complexity of exhaustive Bayesian model averaging.

From the perspective adopted in this paper, BSD may be viewed as a dependence-aware Bayesian model-pursuit strategy operating on a high-dimensional sparse configuration space. Rather than attempting to identify the underlying signal configuration through exhaustive exploration of all $2^n$ possibilities, BSD navigates the configuration space through a sequence of locally most plausible posterior extensions. At each stage, the procedure augments the current active configuration by incorporating the coordinate receiving the strongest posterior support. The covariance structure enters directly into these stagewise posterior comparisons, allowing information accumulated from earlier discoveries to propagate throughout the search process. In this way, BSD performs a covariance-guided, posterior-driven exploration of the sparse region of the configuration space, sequentially refining its estimate of the underlying signal configuration while maintaining computational tractability.


The stagewise posterior-odds statistics introduced above provide the computational engine underlying BSD. An important and somewhat surprising feature of these statistics is that they admit an equivalent residual-based representation, which establishes a close connection between BSD and the MRD methodology developed by \citet{COHEN_SACK_XU_2009}. We explore this connection in the next section.

%
%
%
%


\section{Connection between BSD and MRD}
\label{sec:THEORETICAL_PROPERTIES}

The development in Section~\ref{sec:BSD_METHOD} motivated BSD as a dependence-aware Bayesian model-pursuit strategy for sparse signal recovery. The procedure was constructed through a sequence of stagewise posterior comparisons that guide a search over increasingly complex sparse configurations while explicitly incorporating the covariance structure of the observations.

At first sight, the resulting methodology appears fundamentally Bayesian in nature. Indeed, the BSD statistics are defined through posterior probabilities associated with competing sparse configurations, and the procedure itself is naturally interpreted as a posterior-guided exploration of the configuration space. Consequently, one might expect the computational machinery underlying BSD to be fundamentally different from that of residual-based multiple testing procedures.

Although BSD was introduced in Section~\ref{sec:BSD_METHOD} through a sequence of Bayesian posterior comparisons defining a dependence-aware model-pursuit strategy, the procedure possesses a remarkable hidden structure. In this section, we show that the BSD statistics admit an equivalent representation in terms of covariance-adaptive residual quantities. This representation reveals a deep structural connection between the BSD procedure developed in this article and the Maximum Residual Down (MRD) procedure of \citet{COHEN_SACK_XU_2009}, thereby exposing the residual geometry underlying the BSD search process and providing important insight into its computational and decision-theoretic properties.

More specifically, we show that, at every stage of the procedure, the BSD statistics can be expressed as locally monotone transformations of the corresponding MRD residuals, with the transformation determined by the stage, the candidate coordinate, and the accumulated rejection history. Consequently, the Bayesian posterior evidence driving BSD is completely characterized by the same covariance-adaptive residual geometry that underlies MRD. This observation not only yields a substantial computational simplification, but also places BSD within a broader class of residual-based step-down procedures possessing attractive admissibility properties.

We begin by briefly reviewing the MRD residual quantities introduced by \citet{COHEN_SACK_XU_2009}, which serve as the canonical residual-based step-down procedure under arbitrary covariance dependence. For that, we adopt here similar convention of notations used in \citet{COHEN_SACK_XU_2009}. 

Let $\boldsymbol{X}^{(j_{1},\dots,j_{t})}$ be an $(n-t)\times1$ vector consisting of those components of $\boldsymbol{X}=(X_1,\dots,X_n)$ with $X_{j_{1}},\dots,X_{j_{t}}$ left out. Suppose $\boldsymbol{\Sigma}_{(j_{1},\dots,j_{t})} $ is the $(n-t)\times(n-t)$ sub-matrix obtained after eliminating the $j_{1},\dots,j_{t}$-th rows and the corresponding columns of $\boldsymbol{\Sigma}$. Let $\boldsymbol{\sigma}_{(j)}^{(j_{1},\dots,j_{t})} $ be the $(n-t-1)\times1$ vector obtained by eliminating the $j_{1},\dots,j_{t}$-th and $j$-th elements of the $j$-th column vector of $\boldsymbol{\Sigma}$. Define
\begin{eqnarray}
	\sigma_{j\cdot(j_{1},\dots,j_{t})} &=& \sigma_{jj} - {\boldsymbol{\sigma}_{(j)}^{(j_{1},\dots,j_{t})}}^{T}{\boldsymbol{\Sigma}^{-1}_{(j_{1},\dots,j_{t},j)}}{\boldsymbol{\sigma}_{(j)}^{(j_{1},\dots,j_{t})}}. \nonumber
\end{eqnarray}

The MRD procedure is based on a collection of adaptively formed residual
statistics defined as
\begin{eqnarray}\label{MRD_STATISTICS}
	U_{tj}^{(j_{1},\dots,j_{t-1})}(\boldsymbol{X}) &=& \frac{X_j - {\boldsymbol{\sigma}_{(j)}^{(j_{1},\dots,j_{t-1})}}^{T}{\boldsymbol{\Sigma}^{-1}_{(j_{1},\dots,j_{t-1},j)}}{\boldsymbol{X}^{(j_{1},\dots,j_{t-1},j)}}}{\sigma^{\frac{1}{2}}_{j\cdot(j_{1},\dots,j_{t-1})}},\nonumber
\end{eqnarray}
for $t,j=1,\dots,n$, $1\leqslant j_{1}\neq\dots\neq j_{t-1}\leqslant n$ and $j_{\ell}\neq j$ for all $\ell=1,\dots,t-1$.

For $1 \leqslant t \leqslant n$, we define the index $\widetilde{j}_t(\boldsymbol{X})$ as
\begin{align}\label{eq:MRD_INDEX}
	\widetilde{j}_t(\boldsymbol{X})&=\argmax_{j \in \{1,\dots,n\}\setminus\{\widetilde{j}_1(\boldsymbol{X}),\dots,\widetilde{j}_{t-1}(\boldsymbol{X})\}} |U^{(\widetilde{j}_1(\boldsymbol{X}),\dots,\widetilde{j}_{t-1}(\boldsymbol{X}))}_{tj}(\boldsymbol{X})|.
\end{align}

Given a sequence of positive constants
$C_1 \geqslant C_2 \geqslant \dots \geqslant C_n$,
the MRD procedure operates in a step-down manner as follows. At stage $t$, the null hypothesis corresponding to the largest residual statistic is rejected if
\[
\left|
U^{(\widetilde{j}_1(\boldsymbol{X}),\ldots,
	\widetilde{j}_{t-1}(\boldsymbol{X}))}_{t\widetilde{j}_t(\boldsymbol{X})}
(\boldsymbol{X})
\right| > C_t.
\]
Otherwise, the procedure stops and accepts all remaining null hypotheses. The process continues until an acceptance occurs or there are no more null hypotheses left to be rejected.

%
%
%

The MRD procedure requires specification of the stagewise critical constants \(C_1,\ldots,C_n\). In the numerical comparisons reported in this paper, we consider two calibrations of this residual-based step-down rule. The original MRD procedure of \citet{COHEN_SACK_XU_2009}, denoted here by
MRD--CSX, uses
\[
C^{\mathrm{CSX}}_1
=
\Phi^{-1}\left(1-\frac{\alpha}{2n}\right),
\]
and, for \(i=2,\ldots,n\),
\[
C^{\mathrm{CSX}}_i
=
0.71\,\Phi^{-1}\left(1-\frac{\alpha}{2(n-i+1)}\right).
\]
The MRD--GBS procedure of \citet{ghosh2026covariance} instead uses the \citet{GBS2009} stagewise calibration. Specifically, for \(i=1,\ldots,n\), define
\[
\alpha_i
=
\frac{i\alpha}{n+1-i(1-\alpha)}.
\]
The corresponding two-sided residual critical constants are
\[
C^{\mathrm{GBS}}_i
=
\Phi^{-1}\left(1-\frac{\alpha_i}{2}\right),
\qquad i=1,\ldots,n.
\]
Thus, MRD--CSX and MRD--GBS employ the same covariance-adaptive residualization mechanism but differ in their stagewise calibration. Throughout the paper, MRD--GBS serves as the primary frequentist comparator because it combines the MRD residual search framework with the generalized step-down calibration of \citet{GBS2009}.

\begin{remark}
Note that the indices $j_t(\boldsymbol{X})$ and $\widetilde{j}_{t}(\boldsymbol{X})$ defined in \eqref{eq:BSD_INDEX} and \eqref{eq:MRD_INDEX}, respectively,
need not coincide in general. Consequently, the BSD and MRD procedures 	may generate substantially different stagewise rejection sequences even when applied to the same dataset.
	
This distinction is important because both procedures are residual-adaptive. The choice made at any stage determines the collection of coordinates that are removed from future competition and therefore influences all subsequent stagewise comparisons. As a result, a discrepancy between the selected indices at one stage may propagate throughout the remainder of the search process, leading to different conditioning histories and different future rejection paths. Thus, BSD should not be viewed merely as a numerical reparameterization of MRD, but rather as a distinct sequential procedure that may explore the active configuration space in a fundamentally different manner.
\end{remark}

  Theorem \ref{THM:BSD_MRD_CONNECTION} presented below characterizes the relationship between the proposed BSD method and the MRD method due to \citet{COHEN_SACK_XU_2009}.\newline

\begin{thm}\label{THM:BSD_MRD_CONNECTION}
	Under the present set-up, the BSD statistics and the MRD statistics are associated through the following functional relationship:
	\begin{eqnarray}\label{eq:BSD_MRD_RELATION}
		S_{tj}^{(j_1,\dots,j_{t-1})}(\vec{X}) 
		&=& \frac{p}{1-p} \times \sqrt{\frac{\sigma_{j \cdot (j_1,\dots,j_{t-1})}}{\psi^2+\sigma_{j \cdot (j_1,\dots,j_{t-1})}}}\nonumber\\
		&& \times\exp\bigg\{\frac{\psi^2}{2(\psi^2+\sigma_{j \cdot (j_1,\dots,j_{t-1})})}\{U_{tj}^{(j_1,\dots,j_{t-1})}(\vec{X})\}^2\bigg\}.
	\end{eqnarray}
\end{thm} 

\begin{proof}
	See Appendix.
\end{proof}

Theorem~\ref{THM:BSD_MRD_CONNECTION} reveals that the Bayesian posterior-odds statistics driving BSD are completely characterized by the corresponding covariance-adaptive
residual quantities underlying the MRD procedure. Thus, although BSD was developed from a Bayesian model-pursuit perspective and MRD was motivated through residual-based multiple testing considerations, both procedures are ultimately driven by the same residual geometry induced by the covariance structure.

The representation \ref{eq:BSD_MRD_RELATION} shows that, for each fixed stage and conditioning history, the BSD statistic is a strictly increasing function of the square of the corresponding MRD residual. Consequently, larger residual evidence necessarily corresponds to larger posterior support for extending the current active configuration. In this sense, the covariance-adaptive residual quantities provide the geometric mechanism through which posterior evidence accumulates during the BSD search process.

It is important to note, however, that the transformation appearing in \eqref{eq:BSD_MRD_RELATION} depends explicitly on the stage of the procedure and the previously identified active configuration through the residual variance $\sigma_{j\cdot(j_1,\ldots,j_{t-1})}$. Consequently, BSD should not be viewed as a global monotone transformation of MRD. Rather, BSD belongs to a broader class of locally monotone residual-based procedures whose scoring rules adaptively evolve as the search progresses. Moreover, the BSD and MRD procedures employ distinct stagewise calibration mechanisms through their respective thresholding schemes. As a result, even though both procedures are driven by the same covariance-adaptive residual geometry, they may generate substantially different sequential search paths through the configuration space.

\subsection{Admissibility of Locally Monotone Residual-Based Step-Down Procedures}


The development in Sections 3 and 4 has shown that BSD may be viewed as a dependence-aware Bayesian model-pursuit strategy whose stagewise posterior evidence is completely characterized by covariance-adaptive residual quantities. We now turn to the decision-theoretic properties of such procedures.

Unlike the Bayes-risk formulation considered in Section~\ref{MODEL_PROBLEM_FORMULATION}, the
admissibility theory developed in this section is formulated within the classical vector-loss framework of Cohen and Sackrowitz (2005, 2007, 2008, 2009). This frequentist formulation treats the multiple testing problem as a simultaneous decision problem and provides a natural
criterion for evaluating the admissibility of dependence-aware
multiple testing procedures.

Towards that, we adopt the framework of \citet{COHEN_SACK_XU_2009} in this context. Let us assume that we observe a random vector $\boldsymbol{X}=(X_1,\dots,X_n)$ (obtained through some suitable transformation, if necessary) such that $\boldsymbol{X} \sim N_{n}(\boldsymbol{\theta},\boldsymbol{\Sigma})$ where  $\boldsymbol{\theta}=(\theta_1,\dots,\theta_n)\in\mathbb{R}^{n}$ is an $n\times 1$ unknown, but fixed mean vector, and $\boldsymbol{\Sigma}=((\sigma_{ij}))$ is an $n \times n$ known positive definite matrix with an arbitrary covariance structure. We are interested in testing simultaneously 
\begin{equation}\label{eq:TESTING_PROBLEM}
	H_{0i}:\theta_{i}=0 \mbox{ against } H_{Ai}:\theta_{i}\neq 0, \mbox{ for } i=1,\dots,n.
\end{equation}
 Recall that, in the frequentist paradigm, any multiple testing procedure
$\Phi(\boldsymbol{x})=(\phi_1(\boldsymbol{x}),\dots,\phi_n(\boldsymbol{x}))$
induces an individual test function $\phi_j(\boldsymbol{x})$ for testing
$H_{0j}$ against $H_{Aj}$, where $\phi_j(\boldsymbol{x})$ denotes the probability
of rejecting the $j$-th null hypothesis when the observation
$\boldsymbol{X}=\boldsymbol{x}$ is realized. We consider the standard $0-1$ loss function
corresponding to $\phi_j$, given by
\begin{equation}\label{INDIVIDUAL_LOSS}
	L_j\big(\phi_j(\boldsymbol{X}),\boldsymbol{\theta}\big)
	=
	I\{\theta_j=0\}\phi_j(\boldsymbol{X})
	+
	I\{\theta_j\neq0\}
	\big(1-\phi_j(\boldsymbol{X})\big),
\end{equation}
while the corresponding risk function is given by
\begin{equation}\label{INDIVIDUAL_RISK}
	R_j\big(\phi_j,\boldsymbol{\theta}\big)
	=
	I\{\theta_j=0\}
	E_{\boldsymbol{\theta}:\theta_j=0}\big(\phi_j(\boldsymbol{X})\big)
	+
	I\{\theta_j\neq0\}
	E_{\boldsymbol{\theta}:\theta_j\neq0}
	\big(1-\phi_j(\boldsymbol{X})\big).
	\nonumber
\end{equation}

We consider the overall loss function for the multiple testing procedure
$\Phi(\boldsymbol{X})$ to be the vector loss
\begin{equation}\label{eq:VECTOR_LOSS}
	L\big(\Phi(\boldsymbol{X}),\boldsymbol{\theta}\big)
	=
	\big(
	L_1\big(\phi_1(\boldsymbol{X}),\boldsymbol{\theta}\big),
	\dots,
	L_n\big(\phi_n(\boldsymbol{X}),\boldsymbol{\theta}\big)
	\big),
\end{equation}
with corresponding vector risk function
\begin{equation}\label{VECTOR_RISK}
	R\big(\Phi,\boldsymbol{\theta}\big)
	=
	\big(
	R_1\big(\phi_1,\boldsymbol{\theta}\big),
	\dots,
	R_n\big(\phi_n,\boldsymbol{\theta}\big)
	\big).
	\nonumber
\end{equation}

The vector-loss formulation \eqref{eq:VECTOR_LOSS} considered above should not be confused with the additive Bayes-loss formulation introduced in Section~\ref{MODEL_PROBLEM_FORMULATION}. The former
is used to study admissibility in the classical frequentist sense, whereas the latter is used to define Bayes-optimal sparse signal recovery and the corresponding Bayes Oracle.

A multiple testing procedure $\Phi(\boldsymbol{X})$ is said to be inadmissible with respect to the vector loss function (\ref{eq:VECTOR_LOSS}) if there exists another multiple testing procedure $\Phi^{*}(\boldsymbol{X})$ such that
$R_j\big(\phi^{*}_j,\boldsymbol{\theta}\big)
\leq
R_j\big(\phi_j,\boldsymbol{\theta}\big)$
for all $j=1,\dots,n$ and all $\boldsymbol{\theta}\in\mathbb{R}^n$,
with strict inequality holding for at least one $j$ and some
$\boldsymbol{\theta}\in\mathbb{R}^n$. A multiple testing procedure is said to be
admissible if it is not inadmissible in the aforesaid sense. It is natural
that a multiple testing procedure which is inadmissible with respect to the
vector loss function (\ref{eq:VECTOR_LOSS}) also becomes inadmissible whenever
the loss is a non-decreasing function of the numbers of type I and type II
errors. In this context, it is worth recalling that
\citet{COHEN_KOL_SACK_2007},
\citet{COHEN_SACK_2005_B},
\citet{COHEN_SACK_2007} and
\citet{COHEN_SACK_2008}
showed that in many common applications involving dependent test statistics,
typical $p$-value based stepwise testing procedures, including the celebrated
BH method, are inadmissible with respect to the vector loss function
(\ref{eq:VECTOR_LOSS}). Consequently, such procedures also become inadmissible
whenever the risk is a non-decreasing function of the expected numbers of
type I and type II errors. This reveals an undesirable feature of many
traditional stepwise multiple testing procedures under dependence.

The BSD--MRD representation established in Theorem~\ref{THM:BSD_MRD_CONNECTION} reveals that the Bayesian posterior-odds statistics underlying BSD are completely characterized by covariance-adaptive residual quantities. This connection provides more than a computational interpretation of BSD. It also places BSD within a broader class of residual-based step-down procedures whose statistical properties can be studied through the geometry of the underlying residual system.

A natural question is whether the locally monotone residual structure identified in the previous section has any decision-theoretic implications. In particular, one may ask whether BSD possesses the admissibility properties long regarded as fundamental desiderata of dependence-aware multiple testing procedures. To address this question, we appeal to a recent general admissibility result of \citet{GC_ADMISSIBILITY_2026} for locally monotone residual-based step-down procedures. The result is stated below for completeness.

\begin{lem}\label{lem:LEM_ADMISSIBILITY_MONOTONE_TRANSFORM}
	
Suppose $\boldsymbol{X}\sim N_{n}(\boldsymbol{\theta},\boldsymbol{\Sigma})$, where $\boldsymbol{\theta}\in\mathbb{R}^{n}$ is unknown, but fixed and $\boldsymbol{\Sigma}$ is an $n \times n$ arbitrary but known positive definite covariance matrix. Then, for the two sided multiple testing problem (\ref{eq:TESTING_PROBLEM}), any step-down multiple testing procedure based on statistics $S_{tj}$'s, where each $S_{tj}$ is obtained through a locally adaptive strictly increasing transformation of the absolute value of the corresponding MRD statistic $U_{tj}$, is admissible with respect to the vector loss function \eqref{eq:VECTOR_LOSS}.
\end{lem}
\begin{proof}
See \citet{GC_ADMISSIBILITY_2026}.
\end{proof}

Lemma~\ref{lem:LEM_ADMISSIBILITY_MONOTONE_TRANSFORM} identifies a broad structural mechanism underlying admissibility. In particular, it shows that admissibility is fundamentally determined by the local monotone ordering structure induced by the MRD residual statistics, rather than by the specific form of the residual scoring rule used to rank competing hypotheses. As a consequence, admissibility extends beyond the MRD procedure itself to a substantially larger class of locally monotone residual-based step-down procedures.

The admissibility of the BSD procedure now follows immediately by combining Theorem~\ref{THM:BSD_MRD_CONNECTION} with Lemma~\ref{lem:LEM_ADMISSIBILITY_MONOTONE_TRANSFORM} as presented by the following theorem.

\begin{thm}\label{THM:BSD_Admissibility}
Suppose $\boldsymbol{X}\sim N_{n}(\boldsymbol{\theta},\boldsymbol{\Sigma})$, where $\boldsymbol{\theta}\in\mathbb{R}^{n}$ is unknown, but fixed and $\boldsymbol{\Sigma}$ is an $n \times n$ arbitrary but known positive definite covariance matrix. Then, for the two sided multiple testing problem (\ref{eq:TESTING_PROBLEM}) the BSD procedure based on the statistics $S_{tj}$'s defined in \eqref{eq:BSD_STAT_DEFN} is admissible with respect to the vector loss function \eqref{eq:VECTOR_LOSS}.
\end{thm}

\begin{proof}
See Appendix.
\end{proof}

Theorem~\ref{THM:BSD_Admissibility} provides a rigorous decision-theoretic justification for the BSD procedure under arbitrary covariance dependence. Although BSD was originally motivated from a Bayesian model-pursuit perspective, the theorem shows that its statistical validity extends beyond the Bayesian framework. In particular, BSD cannot be uniformly improved upon with respect to the vector loss function \eqref{eq:VECTOR_LOSS}, and therefore possesses one of the strongest forms of optimality available in classical decision theory.

The result is particularly noteworthy because admissibility is established without requiring independence among the observations. Indeed, admissibility continues to hold under an arbitrary known positive definite covariance structure. This stands in sharp contrast to many traditional multiple testing procedures whose admissibility properties may fail under dependence. Consequently, Theorem~\ref{THM:BSD_Admissibility} provides a rigorous theoretical justification for using BSD as a dependence-aware multiple testing procedure.

More broadly, Theorems~\ref{THM:BSD_MRD_CONNECTION} and \ref{THM:BSD_Admissibility} reveal an interesting structural phenomenon. Although BSD was developed through Bayesian posterior comparisons and motivated as a model-pursuit strategy for sparse signal discovery, its admissibility ultimately derives from the covariance-adaptive residual geometry underlying the procedure. Thus, the decision-theoretic properties of BSD are not tied to the specific Bayesian formulation itself, but rather to the locally monotone residual structure through which posterior evidence is accumulated during the sequential search process.

\subsection{Computational Complexity of BSD}

The Bayes Oracle introduced in Section~\ref{MODEL_PROBLEM_FORMULATION} requires posterior evaluation over the entire configuration space
\[
\{0,1\}^n,
\]
whose cardinality is $2^n$. Consequently, exact implementation of the Oracle rule requires posterior calculations over an exponentially growing collection of signal configurations and becomes computationally prohibitive even for moderate values of $n$. This combinatorial explosion constitutes one of the principal obstacles to Bayesian multiple testing under dependence.

BSD circumvents this difficulty by replacing exhaustive posterior exploration with a sequential model-pursuit strategy. Rather than evaluating posterior probabilities over all possible signal configurations simultaneously, BSD incrementally constructs an active configuration through a sequence of local posterior comparisons. At each stage, only those configurations obtained by adding a single signal to the current active configuration are considered.

More specifically, at stage $t$, BSD evaluates posterior odds for only
\[
n-t+1
\]
candidate extensions of the current active configuration. Consequently, in the worst-case scenario, where the procedure continues until the $n$-th stage, the total number of stagewise posterior comparisons performed by BSD is
\[
\sum_{t=1}^{n}(n-t+1)
=
n+(n-1)+\cdots+1
=
\frac{n(n+1)}{2},
\]
which grows quadratically with the number of hypotheses $n$.

Thus BSD replaces the exponential configuration-space exploration required by the Bayes Oracle with a quadratic-time sequential search procedure. While the Oracle requires consideration of all $2^n$ possible model configurations, BSD restricts attention to a carefully selected collection of neighboring configurations whose cardinality grows only quadratically with $n$.

Furthermore, Theorem~\ref{THM:BSD_MRD_CONNECTION} establishes that the BSD statistics admit an equivalent representation in terms of the MRD residual statistics. Consequently, BSD inherits the computational simplifications developed for the MRD procedure by \citet{ghosh2026covariance}. Under the original formulation of \citet{COHEN_SACK_XU_2009}, computation of the residual statistic associated with each active coordinate requires inversion of a corresponding leave-one-out covariance submatrix. Thus, if $k$ hypotheses remain active at a given stage, direct implementation requires $k$ separate matrix inversions. In the worst-case scenario, the total number of matrix inversions required throughout the procedure is
\[
\sum_{k=1}^{n} k
=
\frac{n(n+1)}{2}.
\]

In contrast, the alternative precision-matrix representation developed by \citet{ghosh2026covariance} eliminates this redundancy by expressing all active residual statistics at a given stage through a single inverse covariance matrix. As a result, at most one matrix inversion is required per stage, reducing the total number of matrix inversions from $n(n+1)/2$ to at most $n$. Since the BSD statistics are locally monotone transformations of the corresponding MRD residual statistics, the same computational advantages become available for BSD through Theorem~\ref{THM:BSD_MRD_CONNECTION}.

Thus BSD achieves computational scalability through two complementary mechanisms: a sequential model-pursuit strategy that replaces exhaustive exploration of the $2^n$ configuration space by only $n(n+1)/2$ local posterior comparisons, and a precision-matrix representation that substantially reduces the computational burden associated with evaluating the stagewise statistics.

\section{Simulation Study}
\label{sec:SIMULATION}

The primary objective of this simulation study is to investigate the finite-sample operating characteristics of the proposed BSD procedure under a variety of dependence structures and to assess how closely its performance approaches that of the Bayes Oracle whenever the latter is available. While the BSD procedure was originally motivated as a computationally feasible Bayesian multiple-testing method under dependence, its exact relationship with the corresponding Bayes-optimal decision rule remains difficult to quantify theoretically. Consequently, simulation studies provide a natural framework for evaluating the extent to which BSD succeeds in approximating the Bayes Oracle.

A particular advantage of one-factor dependence structures is that they admit an explicit latent-variable representation that allows direct computation of the Bayes Oracle. This enables a detailed comparison between the Oracle and BSD procedures across a broad range of sparsity regimes and dependence configurations. In addition to BSD and the Oracle, we also compare several competing multiple-testing procedures in order to place the performance of BSD within a broader methodological context.

Throughout the simulation study, special attention is paid not only to overall misclassification risk but also to the individual components of testing performance, including false discovery rates, false non-discovery rates, power, and support recovery characteristics. Collectively, these experiments provide insight into the extent to which the BSD procedure is able to exploit dependence information for accurate signal detection in large-scale multiple-testing problems.

The mean vector is generated according to the standard two-group formulation

\[
\theta_i
\sim
(1-p)\delta_0
+
pN(0,\psi^2),
\qquad
i=1,\ldots,n,
\]

independently across coordinates, where \(p\) denotes the proportion of non-null signals and \(\delta_0\) is a point mass at zero. 

%

Following the calibration used in \citet{BCFG2011}, we set
\[
\psi^2 = 2\log n.
\]
This choice corresponds to the classical universal threshold \(\sqrt{2\log n}\) introduced by \citet{DJ1994}, which has played a fundamental role in sparse high-dimensional estimation, signal recovery, and multiple testing. Signal strengths of this order frequently arise as critical detection boundaries in sparse high-dimensional inference, separating detectable signals from those that are asymptotically indistinguishable from noise; see, for example, \citet{ABDJ2006}, \citet{JS2004}, \citet{BCFG2011}, and \citet{HJ2010}. Consequently, the resulting simulation setting is neither trivially easy nor excessively difficult, thereby providing a meaningful and scientifically relevant benchmark for evaluating the ability of competing procedures to recover sparse signals under dependence.

The sparsity parameter \(p\) is varied over the grid

\[
0.01,\,
0.03,\,
0.05,\,
0.10,\,
0.15,\,
0.20,\,
0.25,\,
0.30,\,
0.35,\,
0.40,\,
0.45,\,
0.50.
\]

For each configuration, results are based on \(G=5000\) independently generated datasets.

The one-factor representation considered in Section~\ref{MODEL_PROBLEM_FORMULATION}.5 permits exact computation of the Bayes Oracle and therefore provides an ideal setting for benchmarking the BSD procedure. To assess the robustness of the resulting conclusions across a variety of dependence patterns, we consider six distinct specifications of the loading vector
\(\boldsymbol{\lambda}=(\lambda_1,\ldots,\lambda_n)^\top\). These loading configurations induce a broad spectrum of covariance structures ranging from independence and homogeneous positive dependence to heterogeneous, clustered, and mixed-sign factor models. Specifically, we consider:

\begin{enumerate}
	\item \textbf{Independence:}
	\(\lambda_i = 0\) for all \(i\).
	
	\item \textbf{Positive Factor Dependence:}
	\(\lambda_i = \sqrt{0.7}\) for all \(i\).
	
	\item \textbf{Alternating Factor Dependence:}
	\(\lambda_i = (-1)^i\sqrt{0.7}\).
	
	\item \textbf{Linear Loadings:}
	the loadings increase linearly from \(0.45\) to \(0.85\).
	
	\item \textbf{Sparse Factor Dependence:}
	\(\lambda_i=0.85\) for the first quarter of the coordinates and
	\(\lambda_i=0.20\) for the remaining coordinates.
	
	\item \textbf{Two-Block Factor Dependence:}
	\(\lambda_i=0.85\) for the first half of the coordinates and
	\(\lambda_i=0.35\) for the remaining coordinates.
\end{enumerate}

Collectively, these six models encompass a diverse collection of dependence structures generated within a common one-factor framework. The resulting experimental design allows us to investigate the robustness of the BSD procedure across markedly different covariance geometries while simultaneously assessing the extent to which its operating characteristics continue to approximate those of the Bayes Oracle under increasingly complex forms of dependence.

\subsection{Competing Procedures and Performance Measures}
The primary goal of the simulation study is to assess how closely the proposed BSD procedure approximates the Bayes Oracle under dependence. To place the performance of BSD in a broader methodological context, we compare it with the following competing procedures.

\begin{enumerate}
	\item \textbf{Bayes Oracle.}
	The Bayes-optimal decision rule corresponding to the underlying one-factor model and the symmetric misclassification loss considered in this paper.
	
	\item \textbf{BSD.}
	The Bayesian Step-Down procedure proposed in this paper.
	
	\item \textbf{MRD--GBS.}
	The GBS-calibrated MRD procedure recently proposed by \citet{ghosh2026covariance}.
	
	\item \textbf{MRD--CSX.}
	The original MRD procedure of \citet{COHEN_SACK_XU_2009}.
	
	\item \textbf{BH.}
	The Benjamini--Hochberg procedure \citep{BH1995}.
\end{enumerate}

For each procedure, we report the Bayes misclassification risk together with several commonly used measures of multiple-testing performance, including the false discovery rate (FDR), false non-discovery rate (FNR), power, average number of rejections, average number of false discoveries, and average number of false nondiscoveries. Since the Bayes misclassification risk corresponds directly to the underlying multiple-decision loss function, it serves as our primary criterion for comparing the overall performance of the competing procedures. Additional insight into the mechanisms driving the observed risk differences is obtained through separate analyses of FDR, FNR, power, and signal-recovery behavior.

\subsection{Bayes Oracle Benchmarking Under One-Factor Dependence}

Several findings emerging from the simulation study proved to be particularly striking. Most notably, across all six one-factor dependence structures considered in this paper, the BSD procedure exhibits operating characteristics that are nearly indistinguishable from those of the Bayes Oracle, even in relatively small dimensions such as \(n=20\) and \(n=50\). The agreement becomes even more pronounced as the dimension increases, suggesting that BSD is capable of reproducing Oracle behavior with remarkable accuracy across a broad range of sparsity regimes and dependence configurations.

Equally noteworthy is the performance of the MRD--GBS procedure. Although derived from a fundamentally different methodological framework, MRD--GBS consistently attains Bayes misclassification risks that remain remarkably close to those of both BSD and the Bayes Oracle across a wide variety of dependence structures and sparsity regimes considered in this paper. While BSD provides an almost indistinguishable approximation to the Oracle throughout the experiments, MRD--GBS frequently operates very near the same Bayes-risk frontier. The emergence of similar operating characteristics from two conceptually distinct methodologies is striking and suggests that effective use of dependence information plays a central role in sparse-signal recovery under correlation. This observation is particularly intriguing because BSD arises from a Bayesian model-pursuit framework, whereas MRD--GBS is a step-down frequentist testing procedure. Taken together, the numerical results suggest that both BSD and MRD--GBS are highly effective at exploiting dependence information for sparse signal recovery, with BSD providing the closest overall approximation to the Bayes-optimal decision rule among the competing methods considered here.

\subsubsection{Bayes Misclassification Risks}

A central question motivating this study is whether the BSD procedure is capable of reproducing the Bayes Oracle behavior observed under independence even in the presence of substantial dependence. The Bayes Oracle represents the optimal benchmark under the assumed model and therefore provides a natural reference point for evaluating the efficiency of competing multiple testing procedures. The results presented below demonstrate that BSD remains remarkably close to the Oracle across a wide range of sparsity levels, dimensions, and covariance structures.
\begin{figure}[!htbp]
	\centering
	
	\begin{subfigure}{0.48\textwidth}
		\centering
		\includegraphics[width=\linewidth]{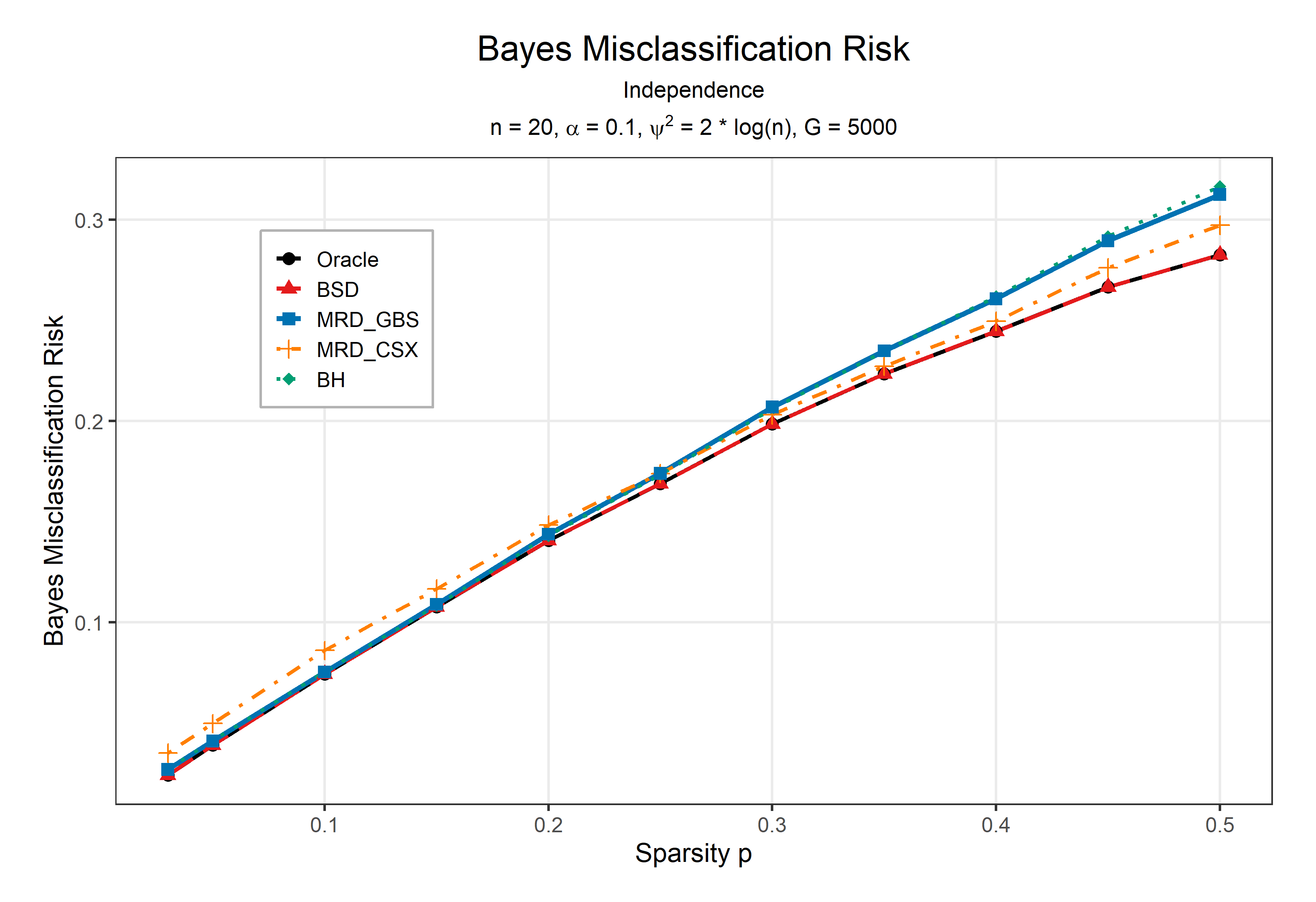}
		\caption{Independence}
	\end{subfigure}
	\begin{subfigure}{0.48\textwidth}
		\centering
		\includegraphics[width=\linewidth]{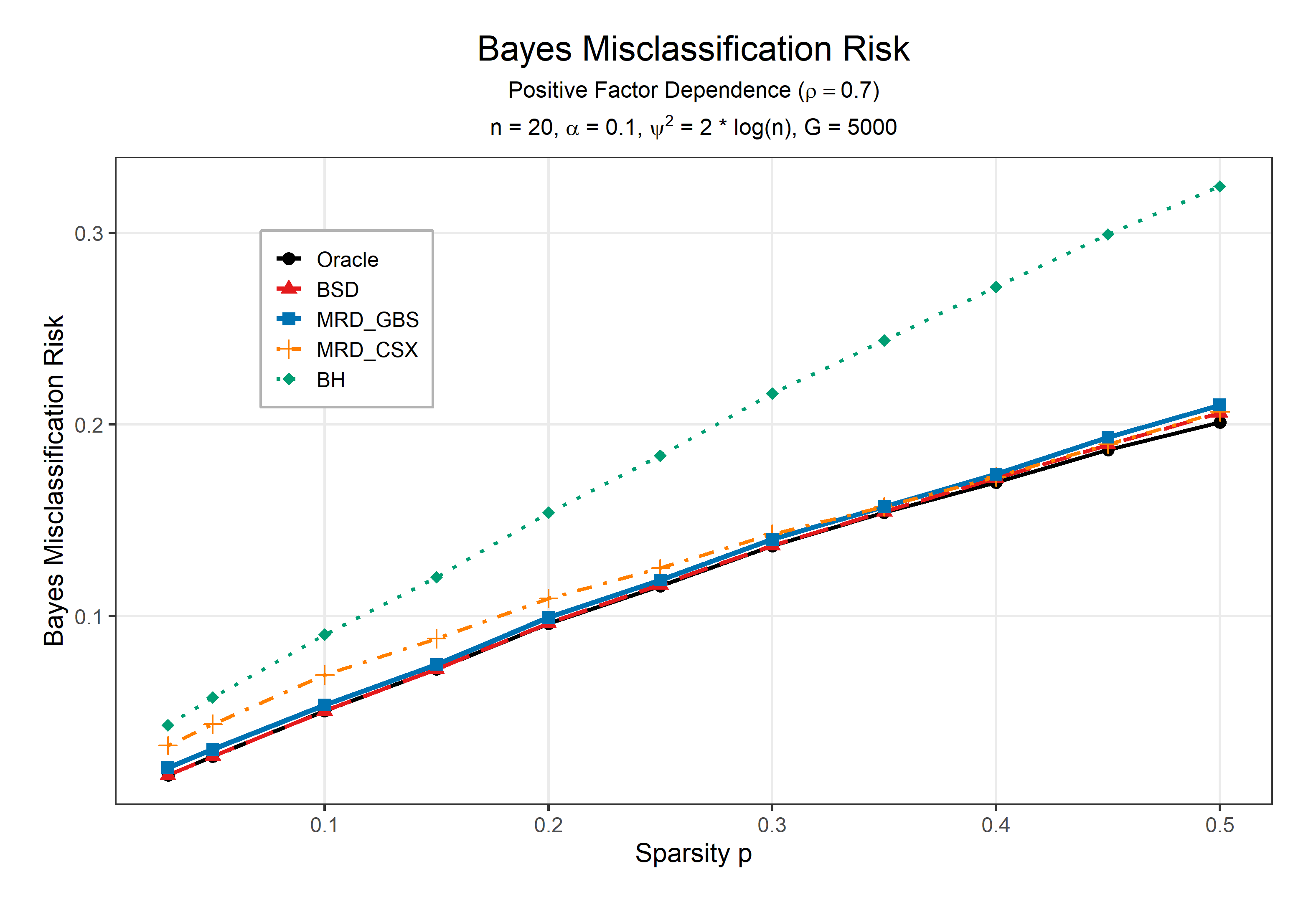}
		\caption{Positive factor}
	\end{subfigure}
	\vspace{0.25cm}
	\begin{subfigure}{0.48\textwidth}
		\centering
		\includegraphics[width=\linewidth]{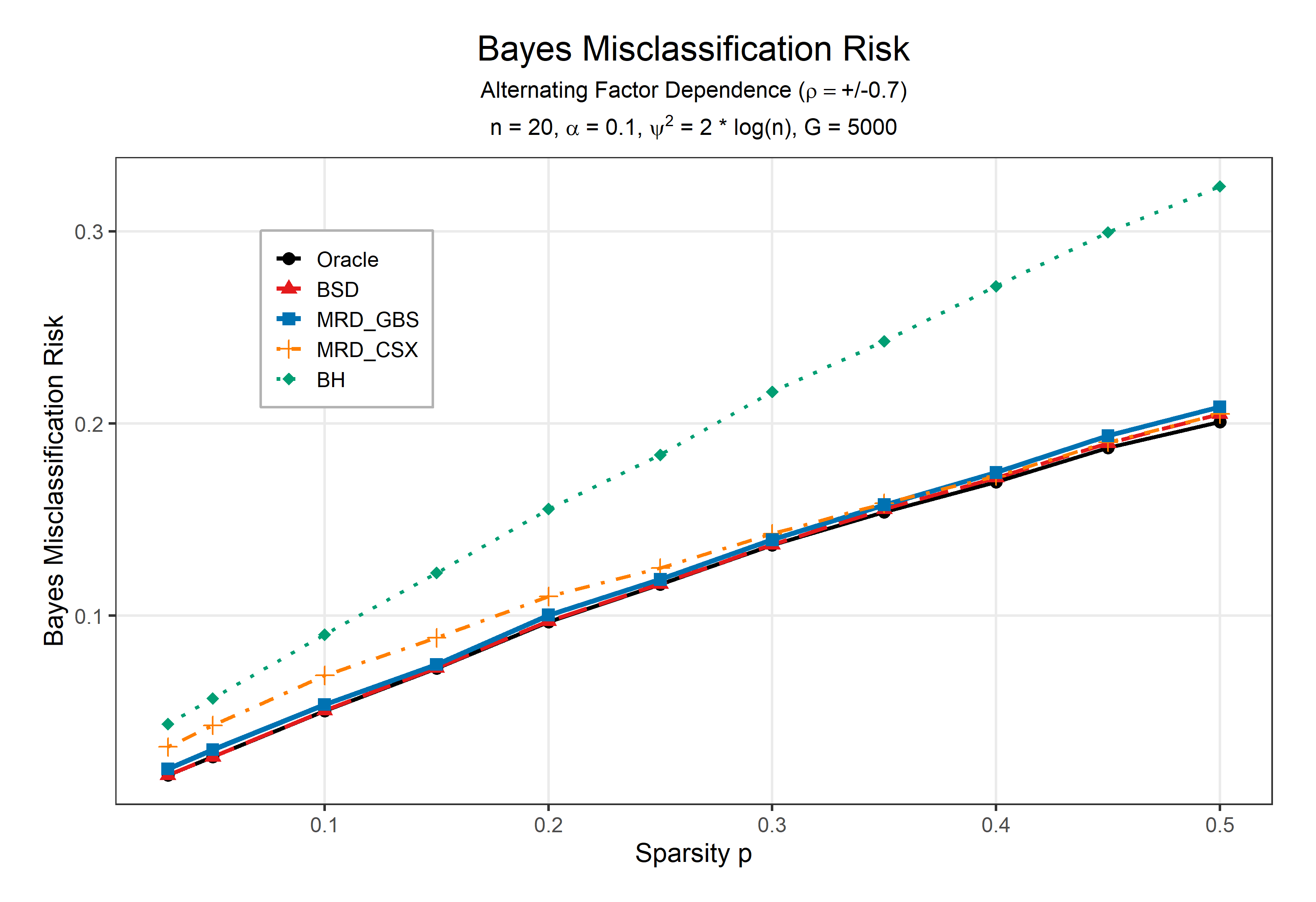}
		\caption{Alternating factor}
	\end{subfigure}
	\begin{subfigure}{0.48\textwidth}
		\centering
		\includegraphics[width=\linewidth]{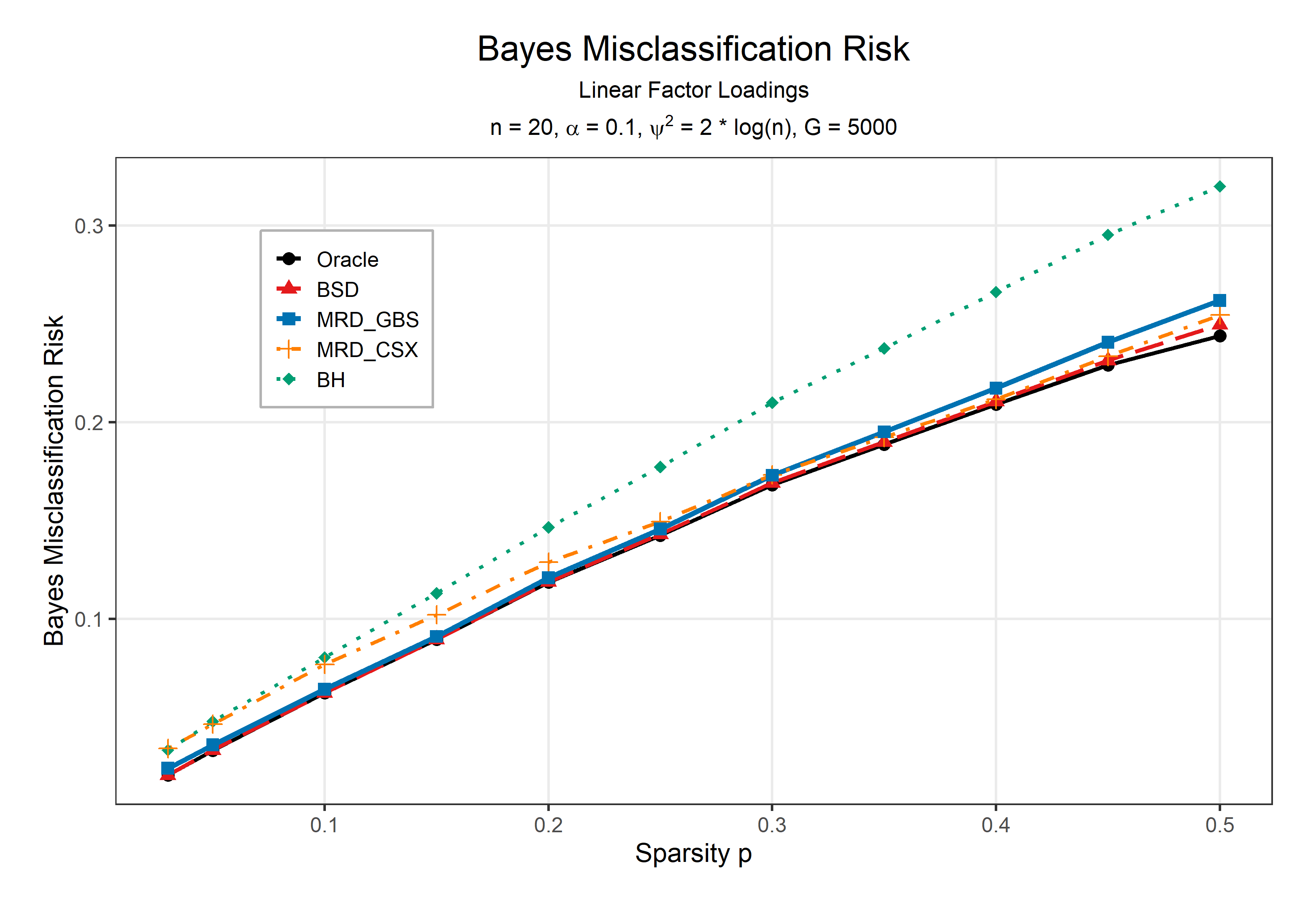}
		\caption{Linear loadings}
	\end{subfigure}
		\vspace{0.25cm}
	\begin{subfigure}{0.48\textwidth}
		\centering
		\includegraphics[width=\linewidth]{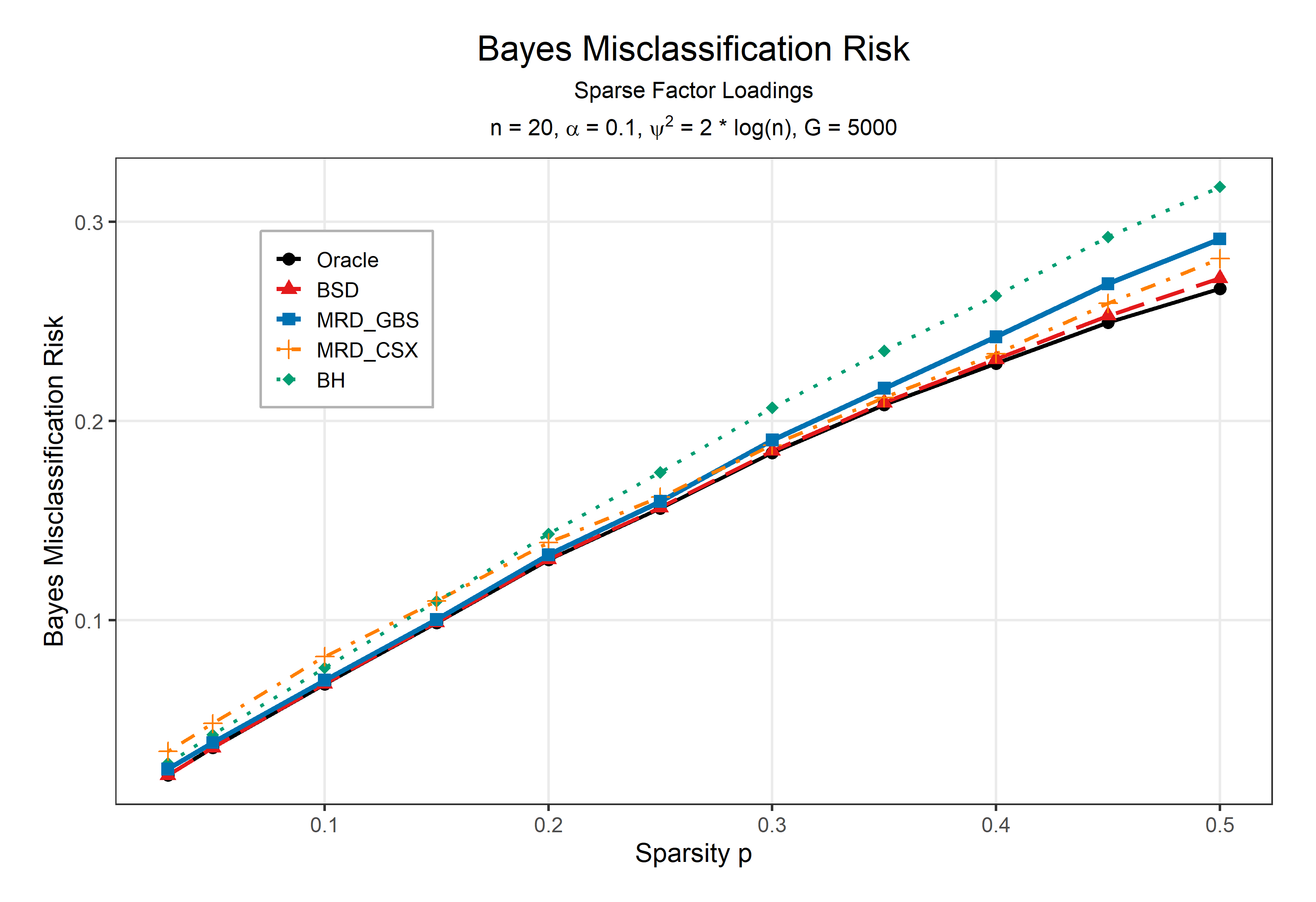}
		\caption{Sparse factor}
	\end{subfigure}
	\begin{subfigure}{0.48\textwidth}
		\centering
		\includegraphics[width=\linewidth]{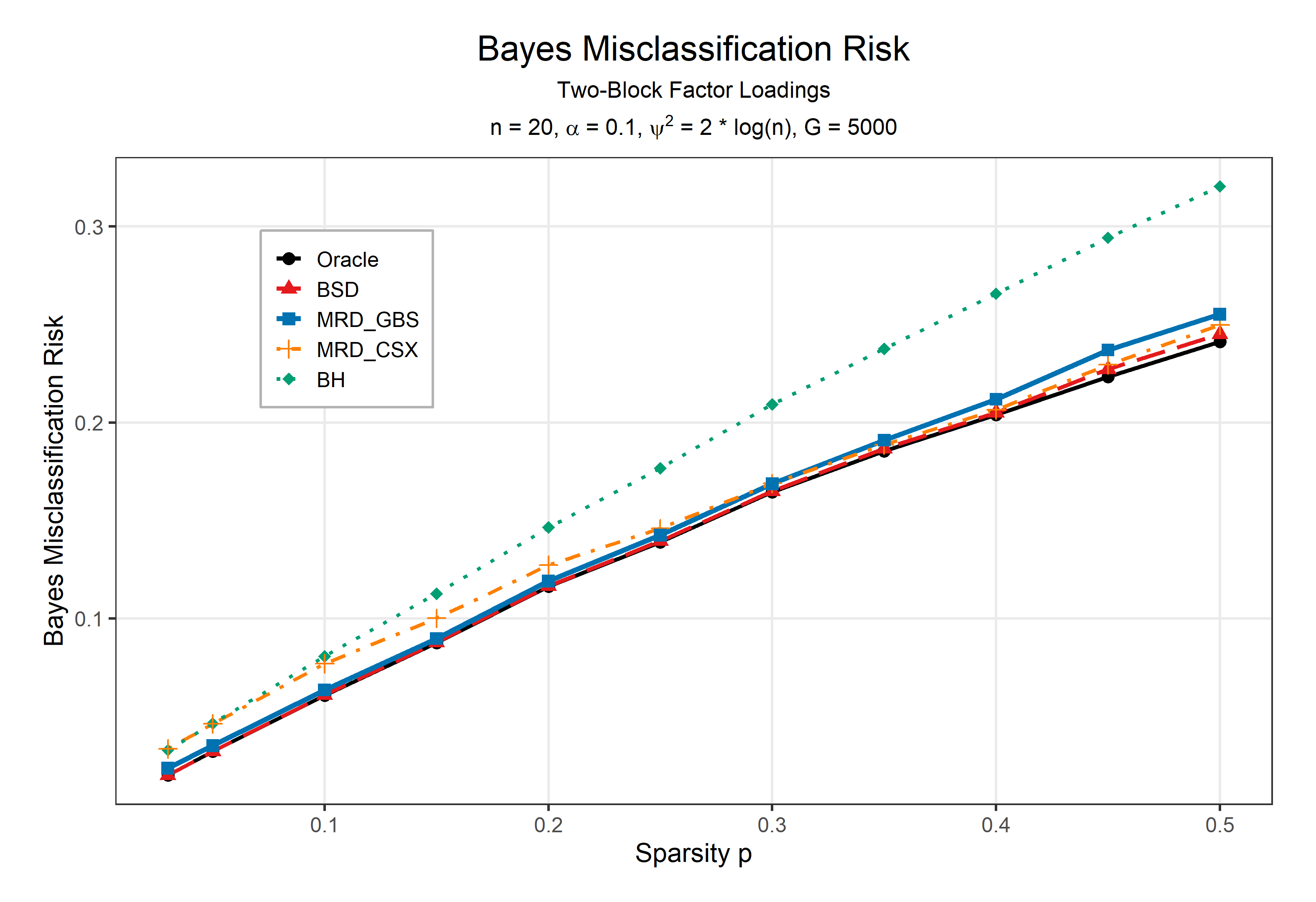}
		\caption{Two-block factor}
	\end{subfigure}
	
	\caption{Bayes misclassification risks of the competing procedures under the six one-factor dependence structures when \(n=20\), \(\alpha=0.1\), \(\psi^2=2\log n\), and \(G=5000\) Monte Carlo replications.}
	\label{fig:onefactor-risk-n20}
\end{figure}

Figure~\ref{fig:onefactor-risk-n20} displays the Bayes misclassification risks of the competing procedures under the six one-factor dependence structures when $n=20$. Several noteworthy patterns emerge. Most strikingly, BSD tracks the Bayes Oracle remarkably closely across all sparsity levels and covariance structures considered. Even in this relatively small-dimensional setting, the BSD risk curves are often visually indistinguishable from their Oracle counterparts, suggesting that the procedure successfully captures much of the information available to the optimal decision rule. A second noteworthy observation concerns the performance of MRD--GBS. Although derived from a fundamentally different methodological framework, its risks remain surprisingly close to those of BSD and the Oracle throughout the experiments. In contrast, MRD--CSX and BH typically incur noticeably larger risks, particularly in the moderate- to high-sparsity regimes.

\begin{figure}[!htbp]
	\centering
	
	\begin{subfigure}{0.48\textwidth}
		\centering
		\includegraphics[width=\linewidth]{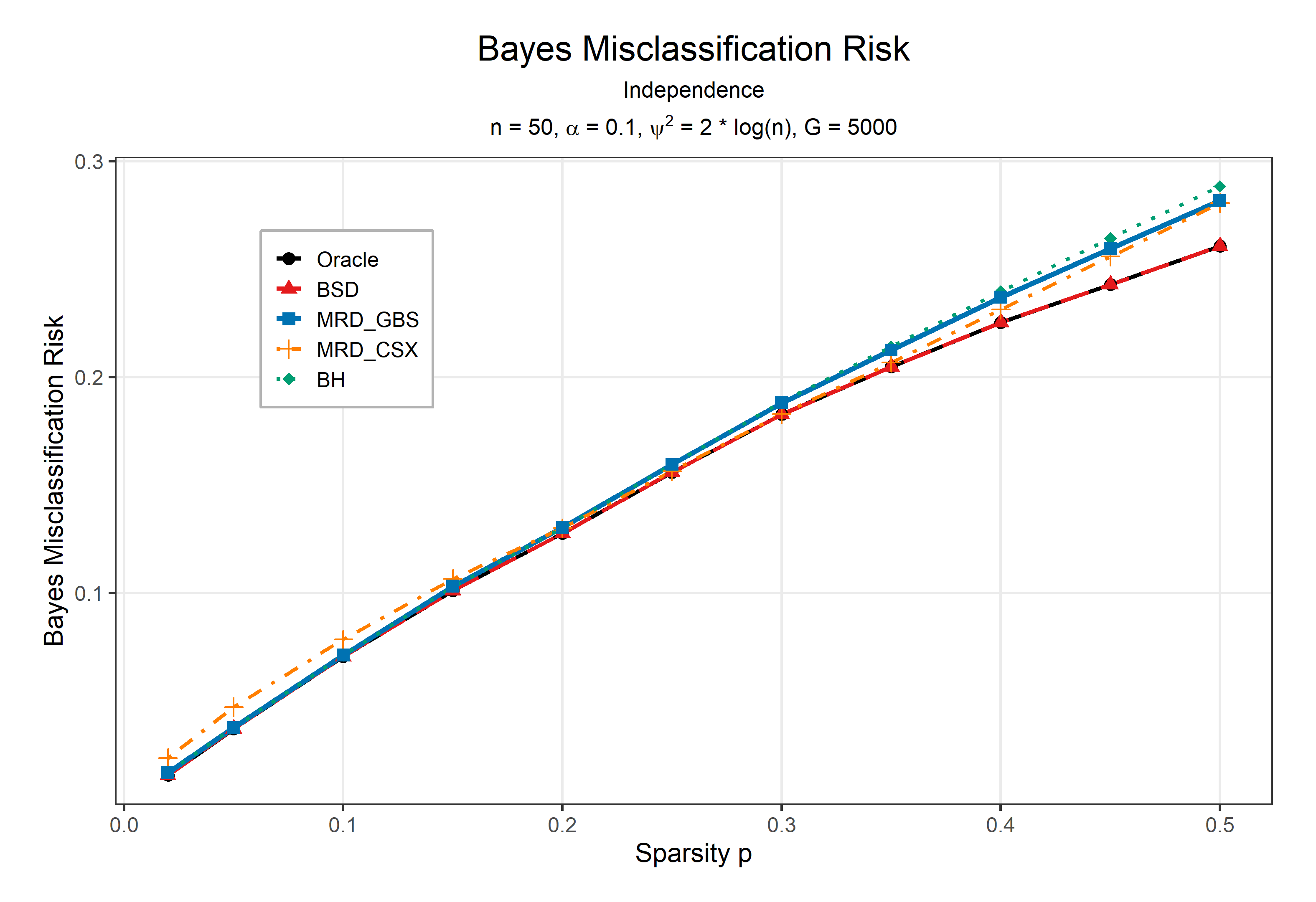}
		\caption{Independence}
	\end{subfigure}
	\begin{subfigure}{0.48\textwidth}
		\centering
		\includegraphics[width=\linewidth]{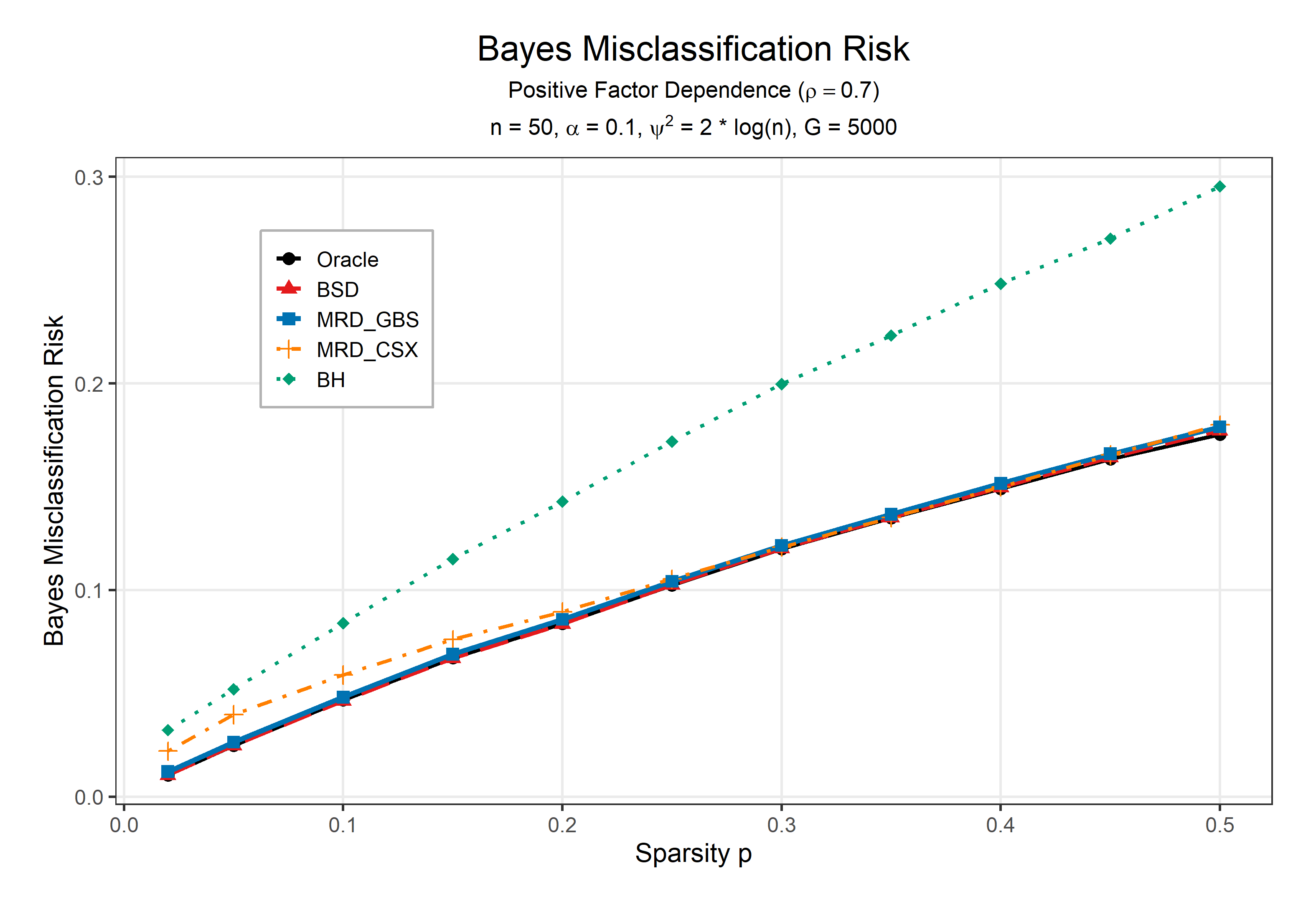}
		\caption{Positive factor}
	\end{subfigure}
	\vspace{0.25cm}
	\begin{subfigure}{0.48\textwidth}
		\centering
		\includegraphics[width=\linewidth]{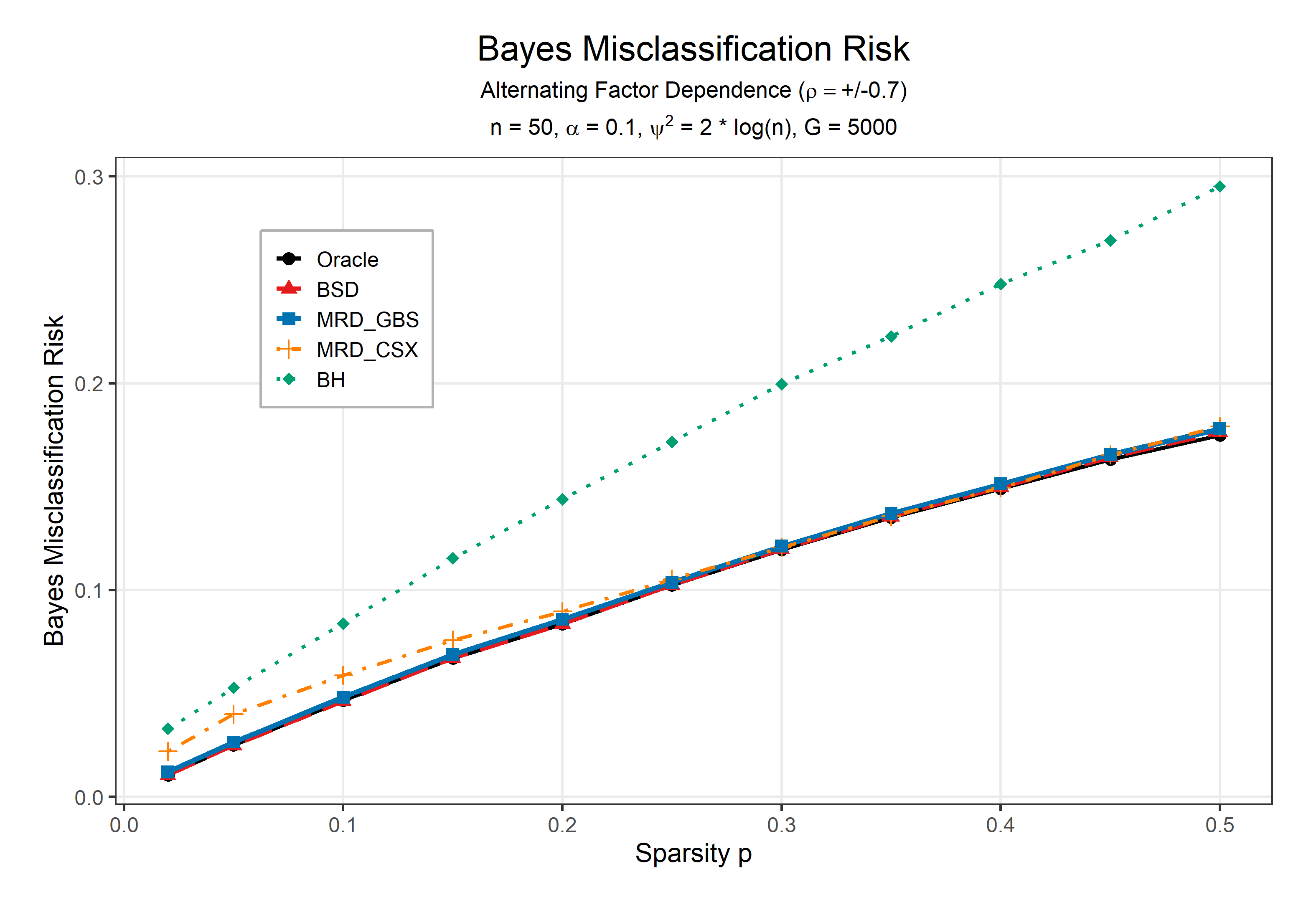}
		\caption{Alternating factor}
	\end{subfigure}
    \begin{subfigure}{0.48\textwidth}
		\centering
		\includegraphics[width=\linewidth]{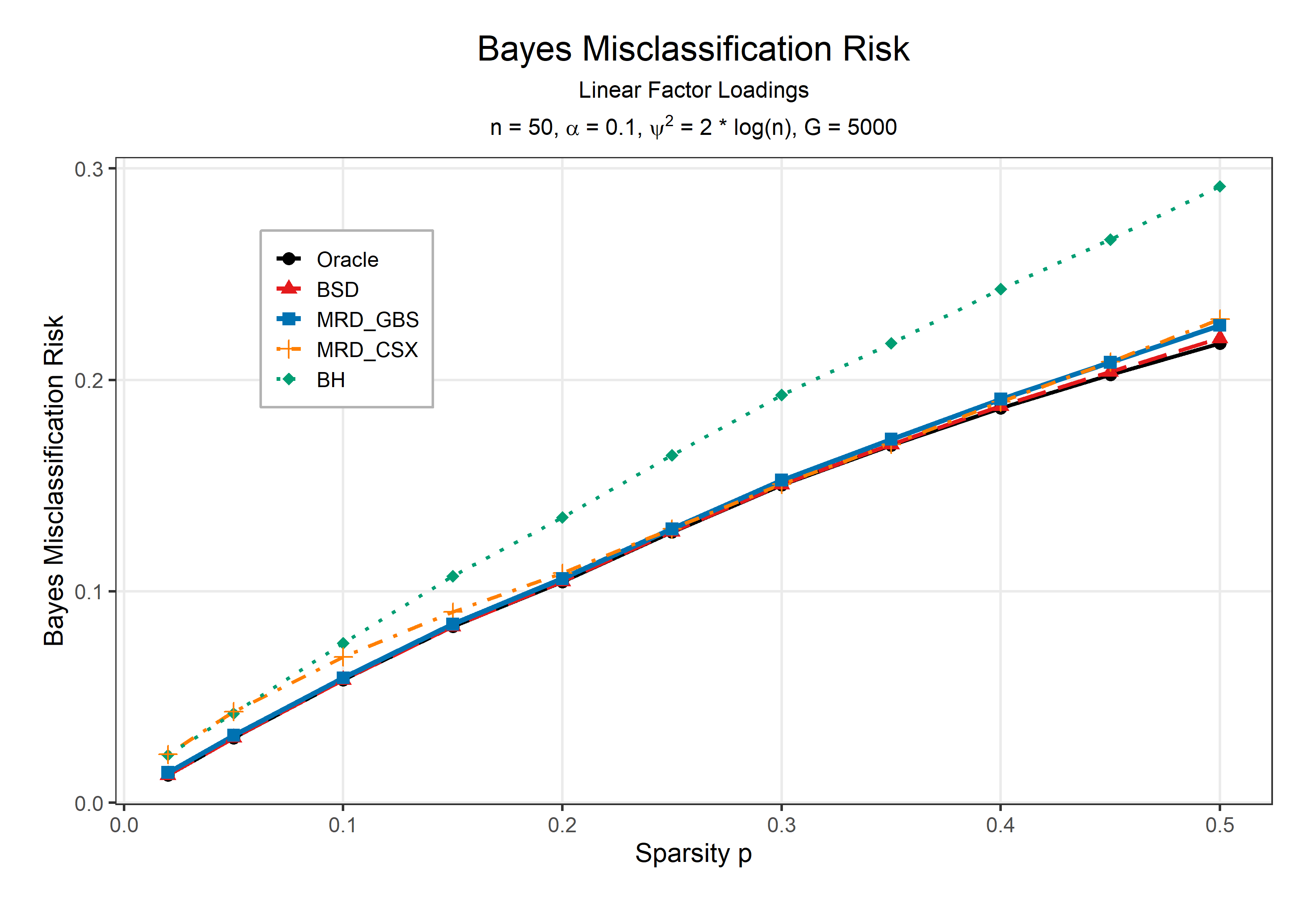}
		\caption{Linear loadings}
	\end{subfigure}
		\vspace{0.25cm}
	\begin{subfigure}{0.48\textwidth}
		\centering
		\includegraphics[width=\linewidth]{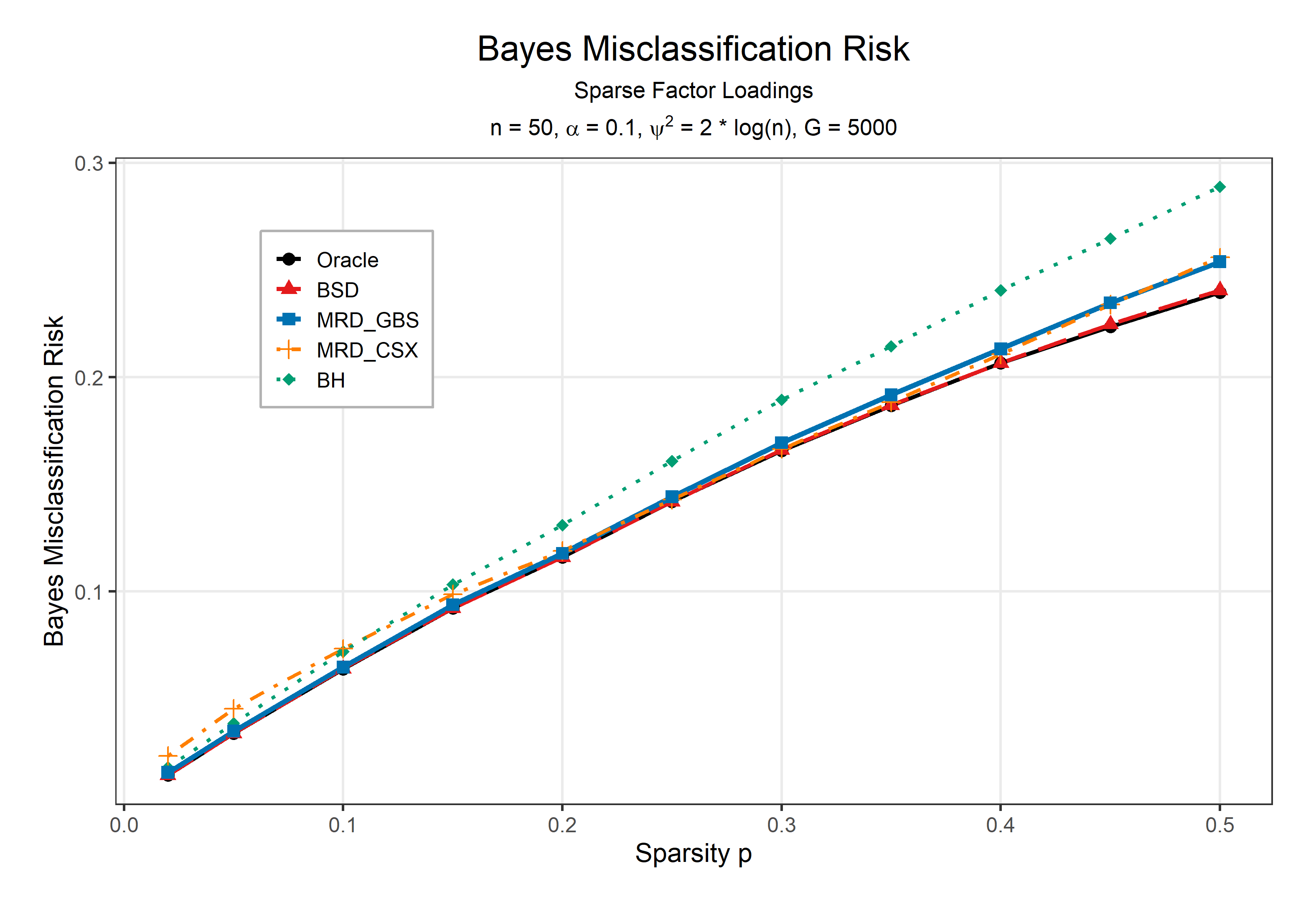}
		\caption{Sparse factor}
	\end{subfigure}
	\begin{subfigure}{0.48\textwidth}
		\centering
		\includegraphics[width=\linewidth]{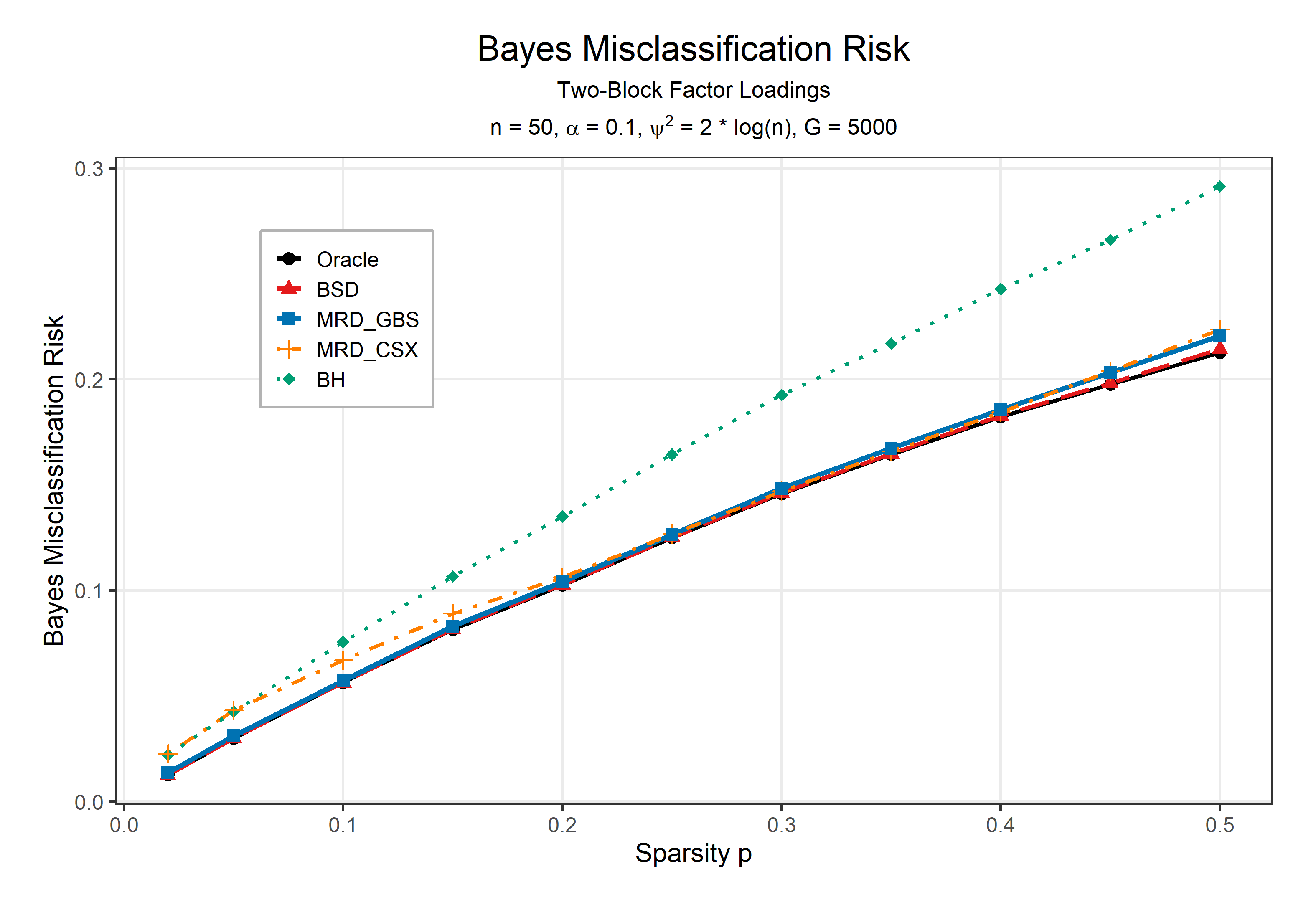}
		\caption{Two-block factor}
	\end{subfigure}
	
	\caption{Bayes misclassification risks of the competing procedures under the six one-factor dependence structures when \(n=50\), \(\alpha=0.1\), \(\psi^2=2\log n\), and \(G=5000\) Monte Carlo replications.}
	\label{fig:onefactor-risk-n50}
\end{figure}

The corresponding results for $n=50$ are presented in Figure~\ref{fig:onefactor-risk-n50}. The overall qualitative conclusions remain unchanged, but the agreement between BSD and the Bayes Oracle becomes even more pronounced. Across all six dependence structures, the BSD risk curves continue to lie extremely close to the Bayes Oracle over the entire range of sparsity levels considered. At the same time, MRD--GBS remains the strongest frequentist competitor and consistently outperforms both MRD--CSX and BH. The persistence of these patterns across dimensions provides strong evidence that the close agreement between BSD and the Oracle is not merely a small-sample phenomenon.

\begin{figure}[!htbp]
	\centering
	
	\begin{subfigure}{0.48\textwidth}
		\centering
		\includegraphics[width=\linewidth]{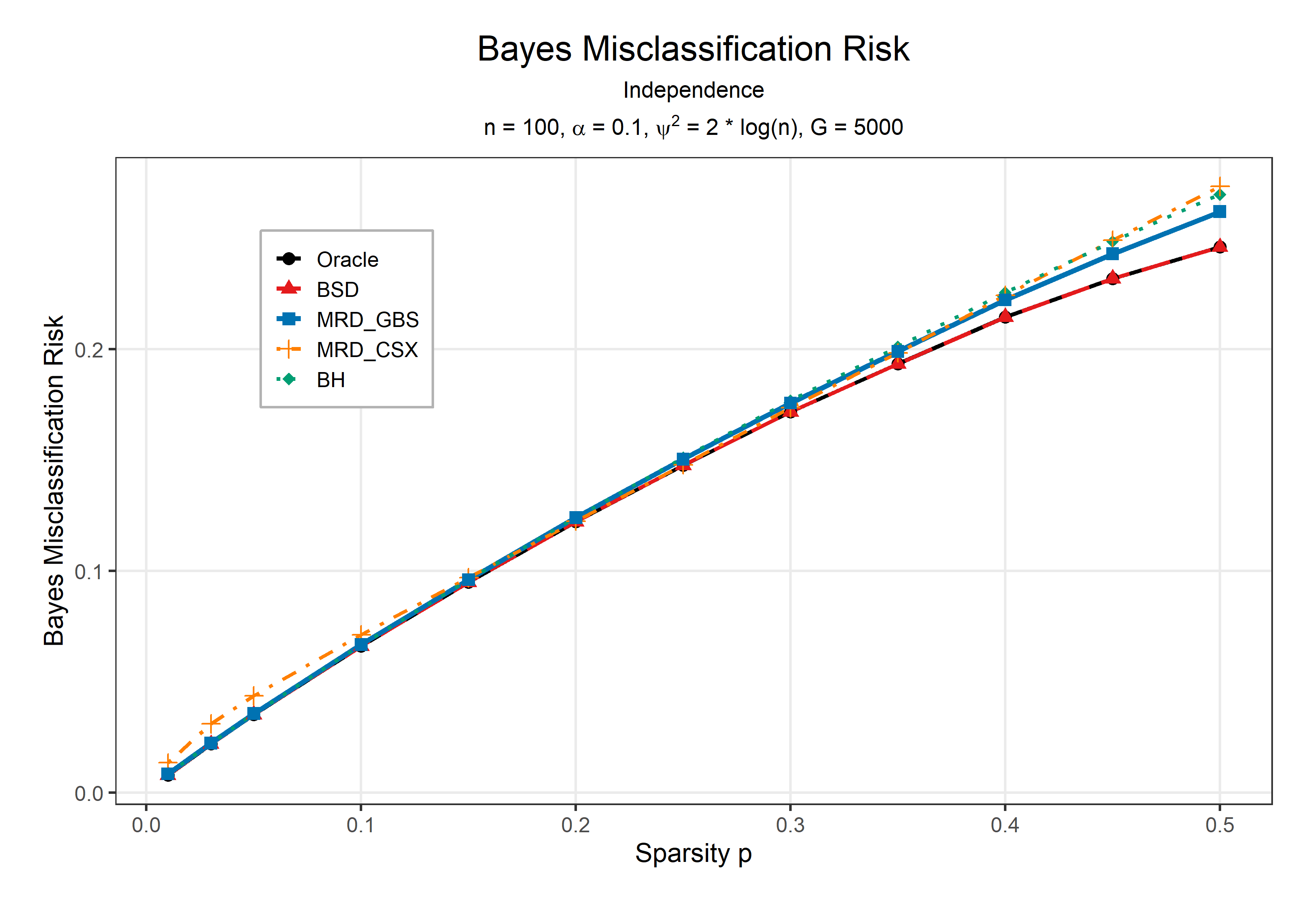}
		\caption{Independence}
	\end{subfigure}
	\begin{subfigure}{0.48\textwidth}
		\centering
		\includegraphics[width=\linewidth]{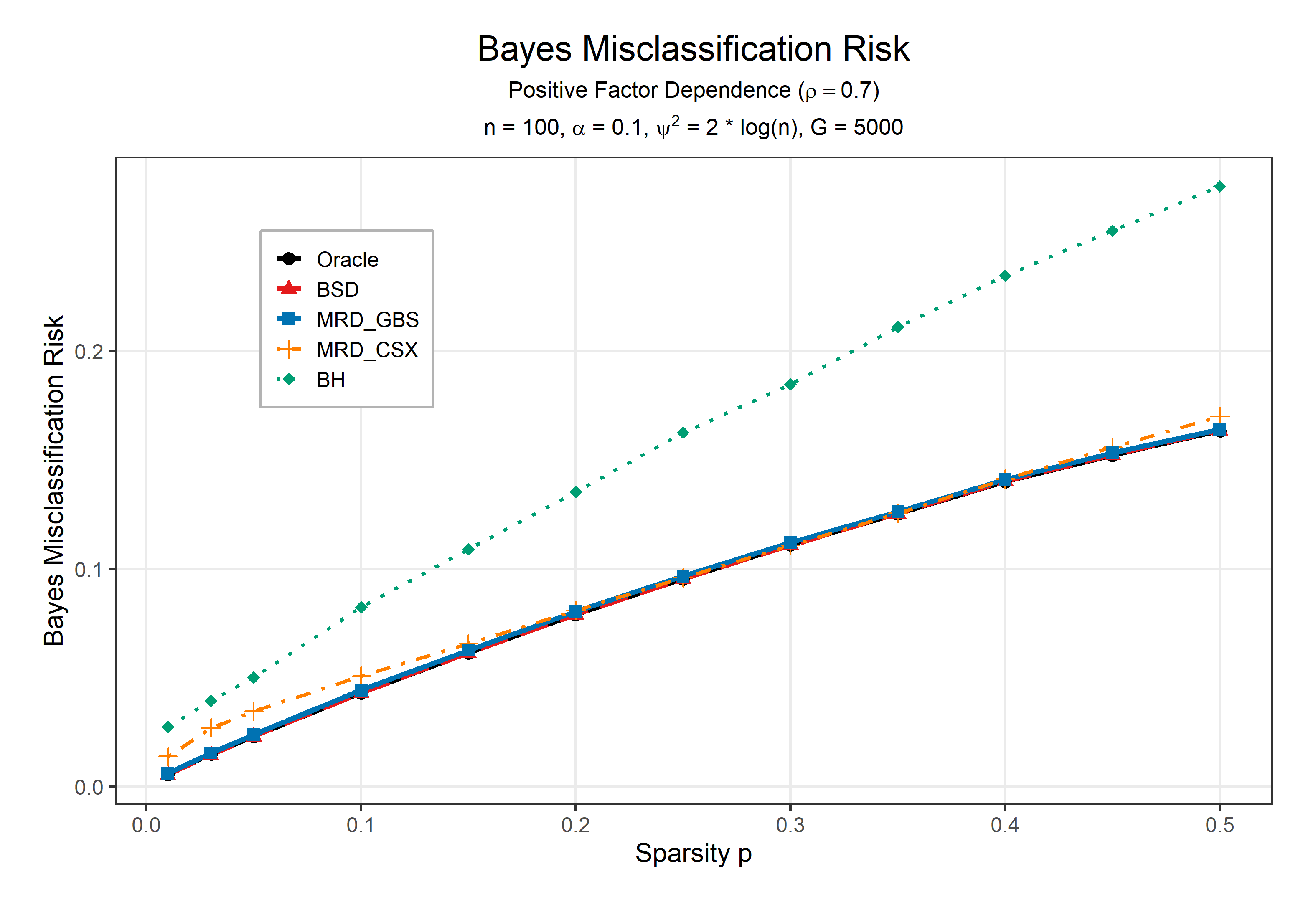}
		\caption{Positive factor}
	\end{subfigure}
		\vspace{0.25cm}
	\begin{subfigure}{0.48\textwidth}
		\centering
		\includegraphics[width=\linewidth]{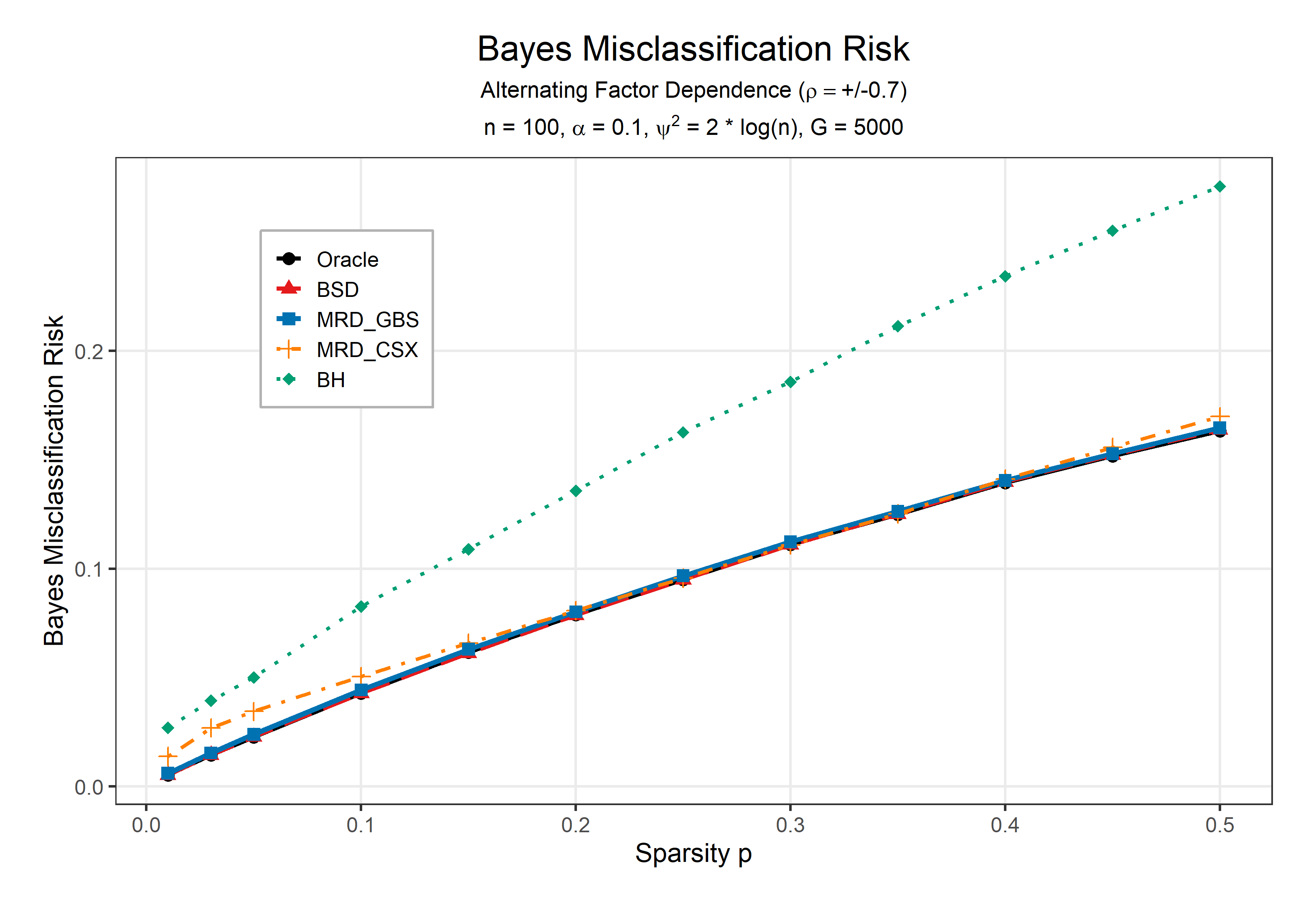}
		\caption{Alternating factor}
	\end{subfigure}	
	\begin{subfigure}{0.48\textwidth}
		\centering
		\includegraphics[width=\linewidth]{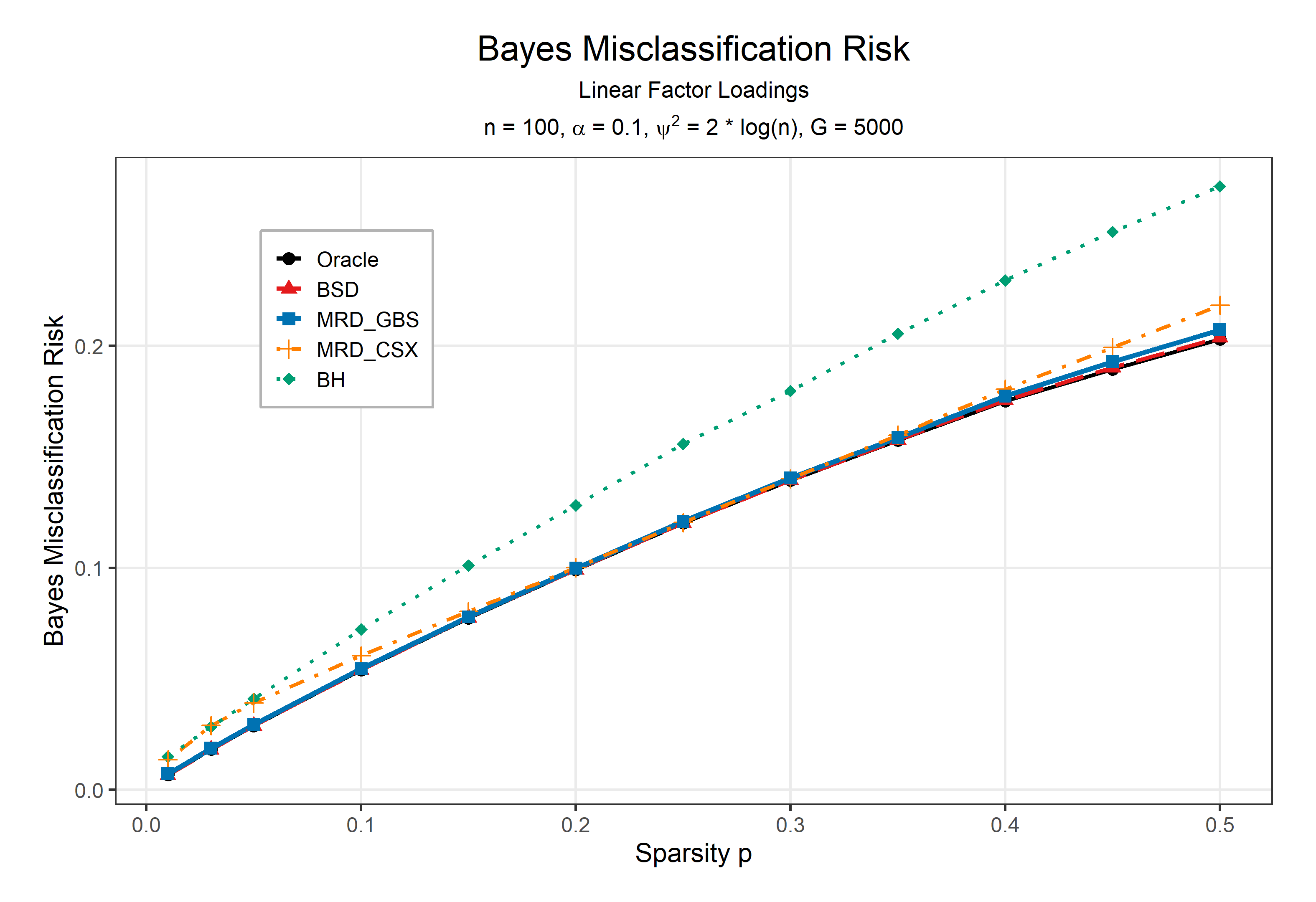}
		\caption{Linear loadings}
	\end{subfigure}
		\vspace{0.25cm}
	\begin{subfigure}{0.48\textwidth}
		\centering
		\includegraphics[width=\linewidth]{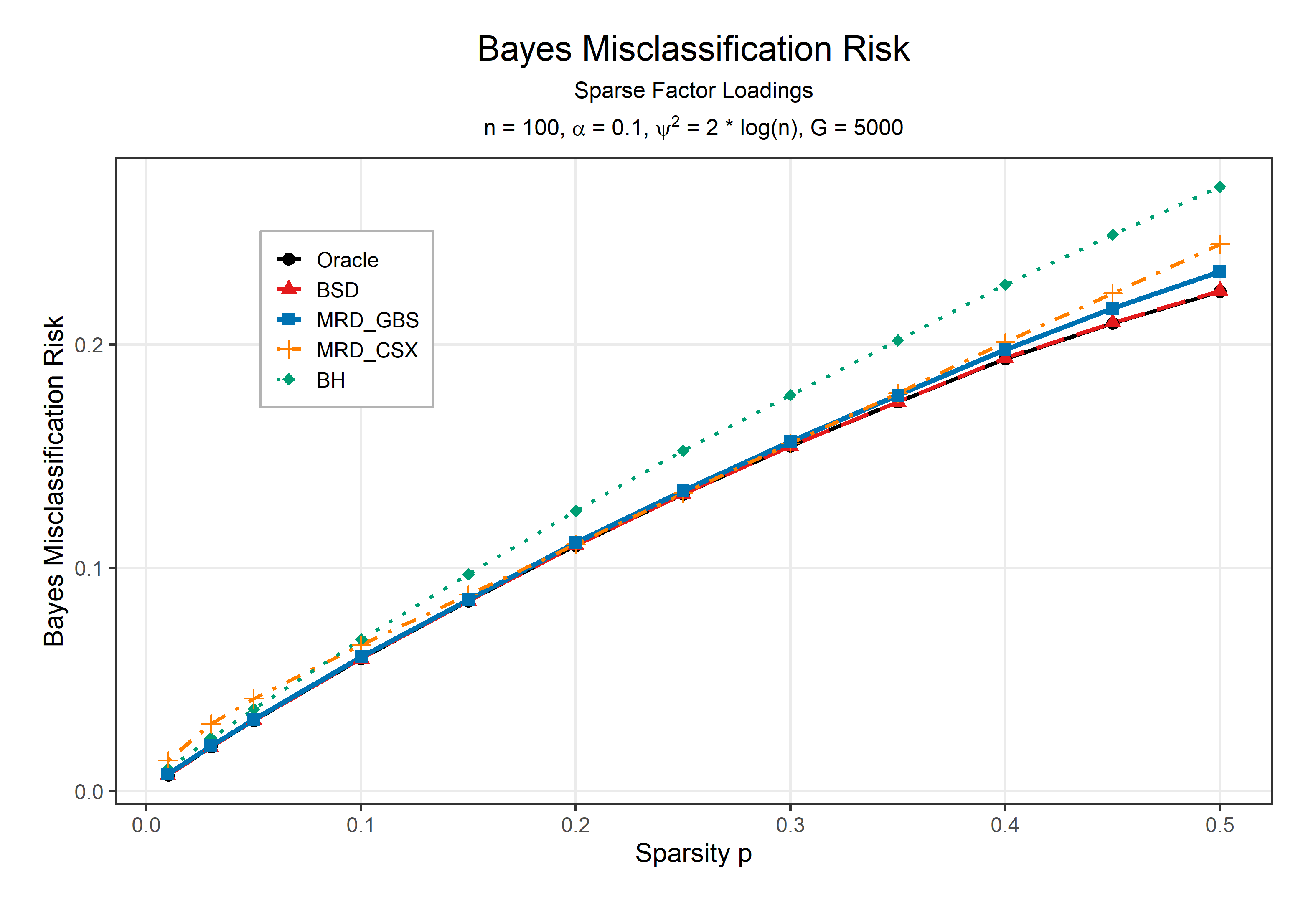}
		\caption{Sparse factor}
	\end{subfigure}
	\begin{subfigure}{0.48\textwidth}
		\centering
		\includegraphics[width=\linewidth]{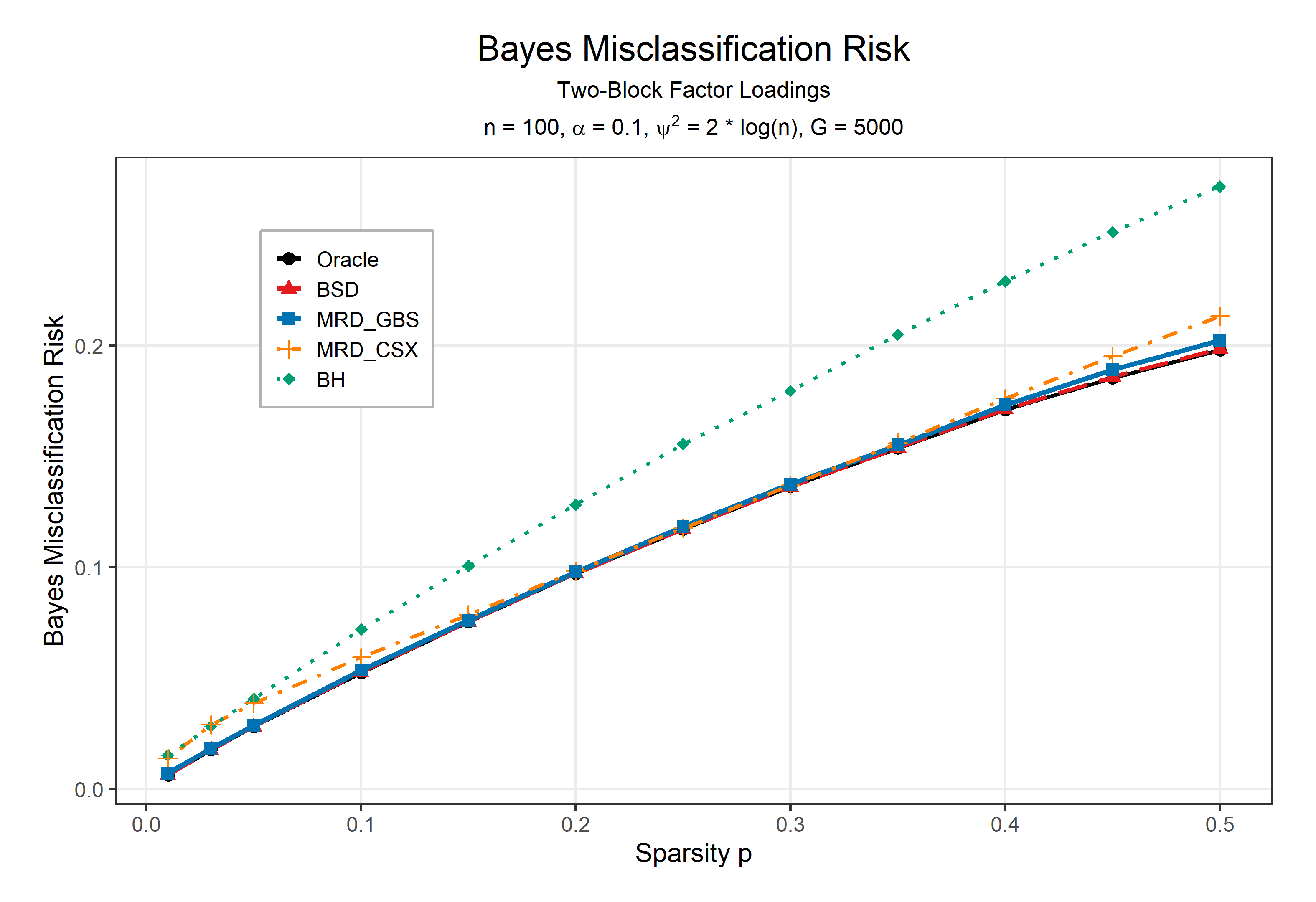}
		\caption{Two-block factor}
	\end{subfigure}
	
	\caption{Bayes misclassification risks of the competing procedures under the six one-factor dependence structures when \(n=100\), \(\alpha=0.1\), \(\psi^2=2\log n\), and \(G=5000\) Monte Carlo replications.}
	\label{fig:onefactor-risk-n100}
\end{figure}
Figure~\ref{fig:onefactor-risk-n100} further reinforces the preceding observations. Across all six covariance structures, BSD remains virtually indistinguishable from the Bayes Oracle, with excess risks that are negligible relative to the overall magnitude of the Bayes risk despite the increasing complexity of the underlying dependence patterns. Equally noteworthy is the behavior of MRD--GBS, whose risks frequently remain very close to the Bayes Oracle and are substantially smaller than those of MRD--CSX and BH. Taken together, the results in Figures~\ref{fig:onefactor-risk-n20}--\ref{fig:onefactor-risk-n100} provide compelling empirical evidence that BSD is capable of reproducing Oracle behavior with remarkable accuracy under a broad range of dependence structures. They also suggest that MRD--GBS captures many of the same signal-recovery mechanisms as the Oracle and BSD procedures despite arising from a fundamentally different methodological framework.

\begin{table}[ht]
	\centering
\caption{Average Bayes misclassification risks across all sparsity levels under the six one-factor dependence structures when $n=100$.}
	\label{tab:avg-risk-n100}
	\begin{tabular}{lccccc}
		\hline
		Dependence Structure & Oracle & BSD & MRD--GBS & MRD--CSX & BH \\
		\hline
		Independence       & 0.1295 & 0.1295 & 0.1339 & 0.1372 & 0.1356 \\
		Positive Factor    & 0.0844 & 0.0844 & 0.0855 & 0.0893 & 0.1474 \\
		Alternating Factor & 0.0842 & 0.0843 & 0.0855 & 0.0893 & 0.1474 \\
		Linear Loadings    & 0.1058 & 0.1059 & 0.1071 & 0.1118 & 0.1400 \\
		Sparse Factor      & 0.1168 & 0.1169 & 0.1193 & 0.1237 & 0.1366 \\
		Two-Block Factor   & 0.1032 & 0.1033 & 0.1047 & 0.1094 & 0.1398 \\
		\hline
	\end{tabular}
\end{table}
Table~\ref{tab:avg-risk-n100} further quantifies the close agreement between BSD and the Bayes Oracle. Across all six dependence structures, the average Bayes misclassification risks of BSD are nearly identical to those of the Oracle, with discrepancies appearing only beyond the third decimal place in most cases. MRD--GBS also exhibits remarkably competitive performance and consistently remains much closer to the Bayes Oracle than either MRD--CSX or BH. Particularly noteworthy is the behavior under the positive-factor, alternating-factor, linear-loading, sparse-factor, and two-block-factor models, where BSD effectively reproduces the Oracle risk. These numerical summaries reinforce the visual impression conveyed by the risk profiles and provide additional evidence that BSD successfully captures the essential dependence structure underlying Oracle behavior.

While Bayes misclassification risk provides a comprehensive summary of overall
decision quality, it does not reveal how the competing procedures balance false
discoveries against missed signals. To better understand the mechanisms driving
the observed risk differences, we next examine the false discovery rates of the
competing procedures under the same collection of one-factor dependence structures.

\subsubsection{False Discovery Rates}

We next examine the false discovery rate (FDR) behavior of the competing procedures. Since both BSD and the Bayes Oracle are derived from a Bayesian classification framework and were not designed to explicitly control the FDR, it is of particular interest to investigate the false discovery behavior induced by their Bayes-risk optimality considerations. In contrast, the BH and MRD--GBS procedures are calibrated through mechanisms that directly target false discovery control. Figures~\ref{fig:fdr-n20}--\ref{fig:fdr-n100} display the empirical FDRs across the six dependence structures and the three dimensions considered in our study.

\begin{figure}[p]
	\centering
	
	\begin{subfigure}{0.48\textwidth}
		\centering
		\includegraphics[width=\linewidth]{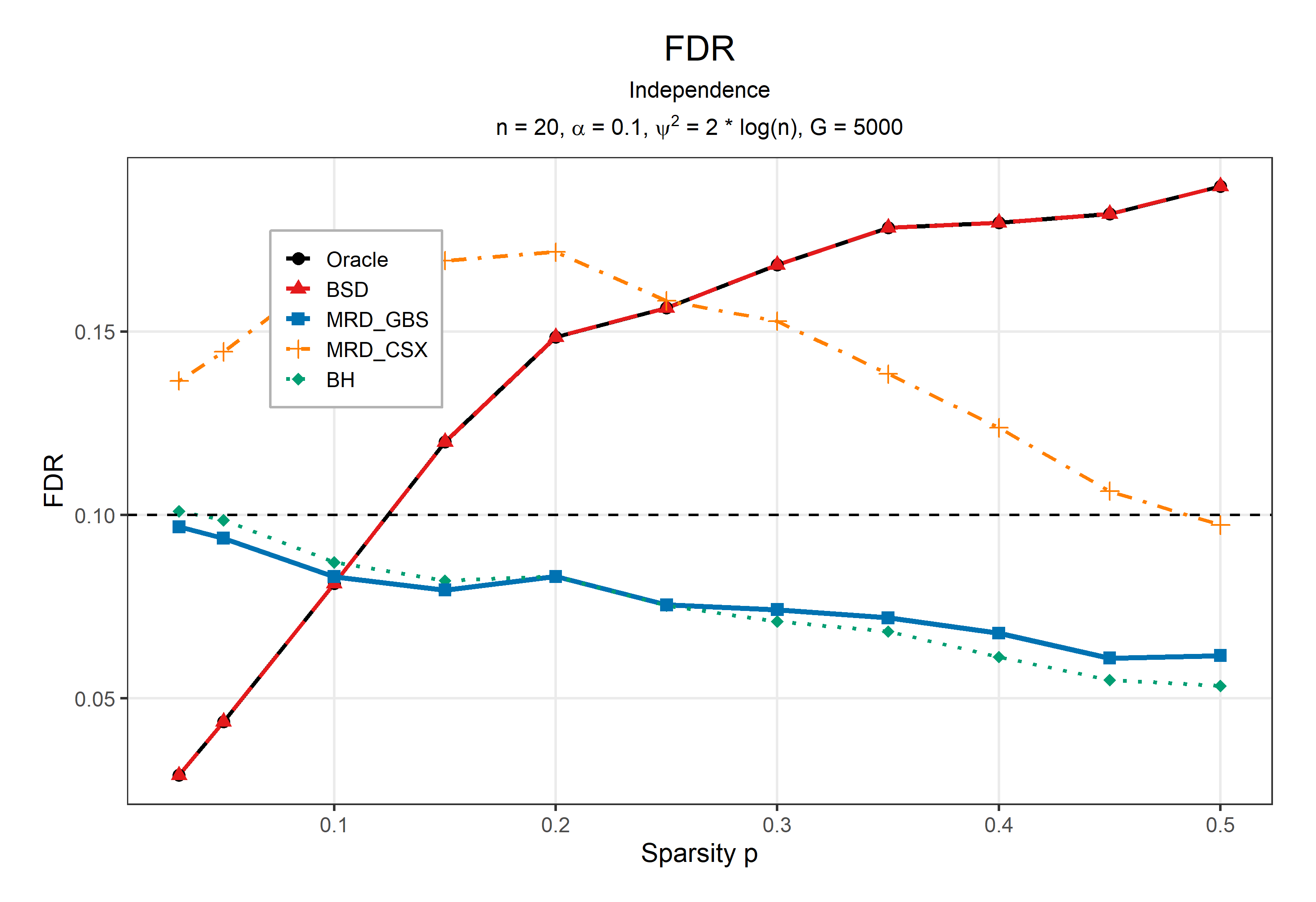}
		\caption{Independence}
	\end{subfigure}
	\hfill
	\begin{subfigure}{0.48\textwidth}
		\centering
		\includegraphics[width=\linewidth]{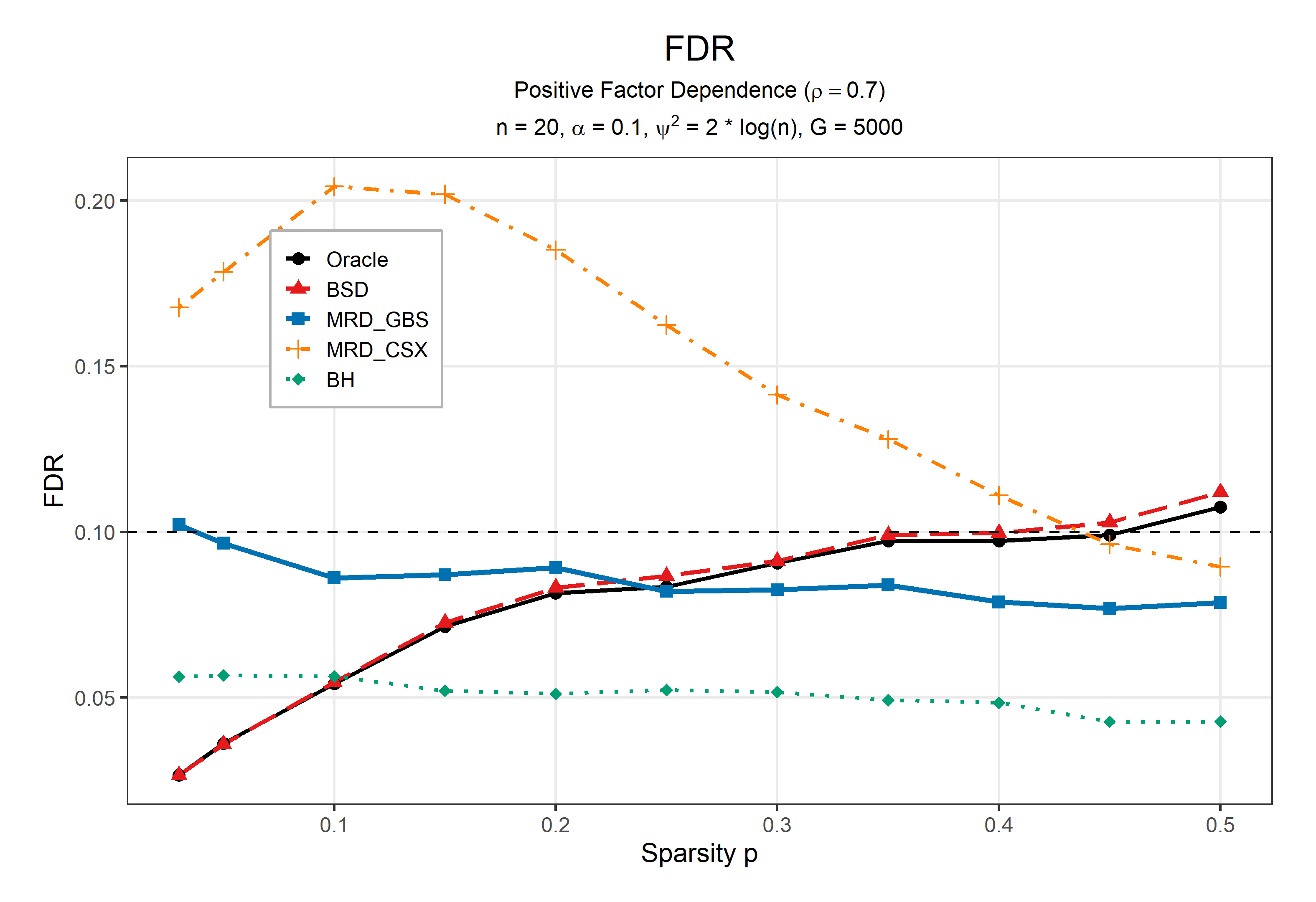}
		\caption{Positive Factor Dependence}
	\end{subfigure}
	
	\vspace{0.3cm}
	
	\begin{subfigure}{0.48\textwidth}
		\centering
		\includegraphics[width=\linewidth]{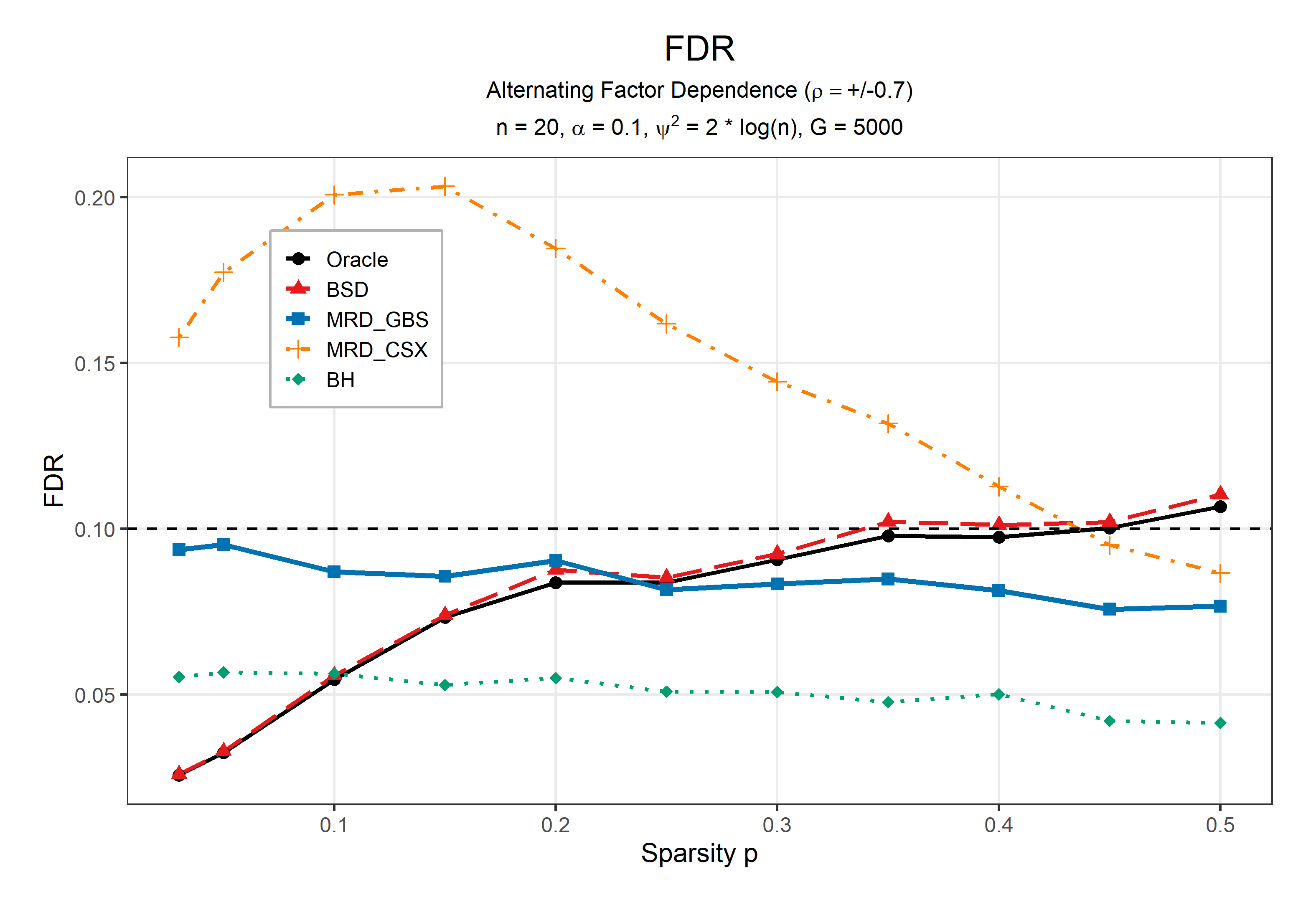}
		\caption{Alternating Factor Dependence}
	\end{subfigure}
	\hfill
	\begin{subfigure}{0.48\textwidth}
		\centering
		\includegraphics[width=\linewidth]{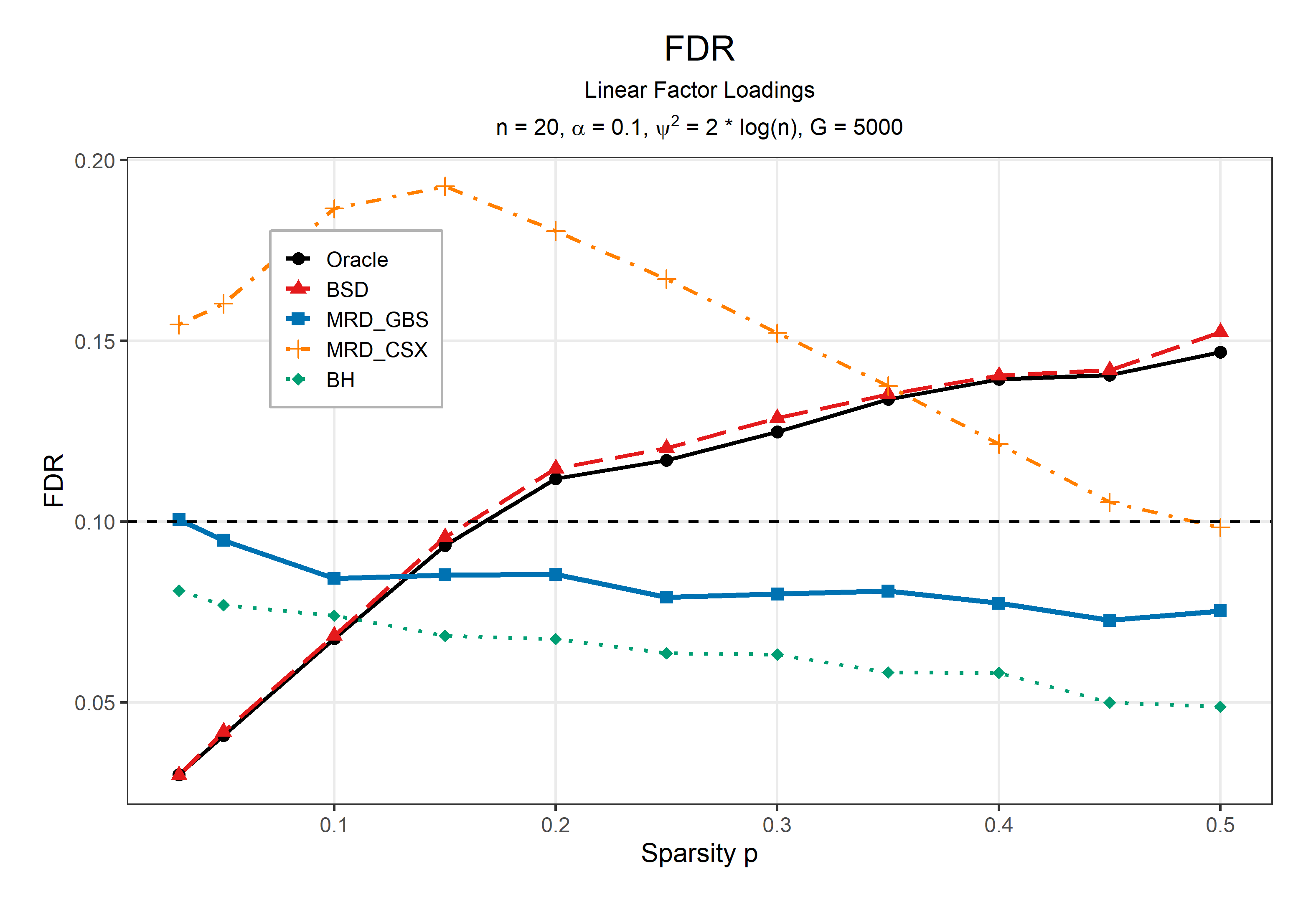}
		\caption{Linear Loadings}
	\end{subfigure}
	
	\vspace{0.3cm}
	
	\begin{subfigure}{0.48\textwidth}
		\centering
		\includegraphics[width=\linewidth]{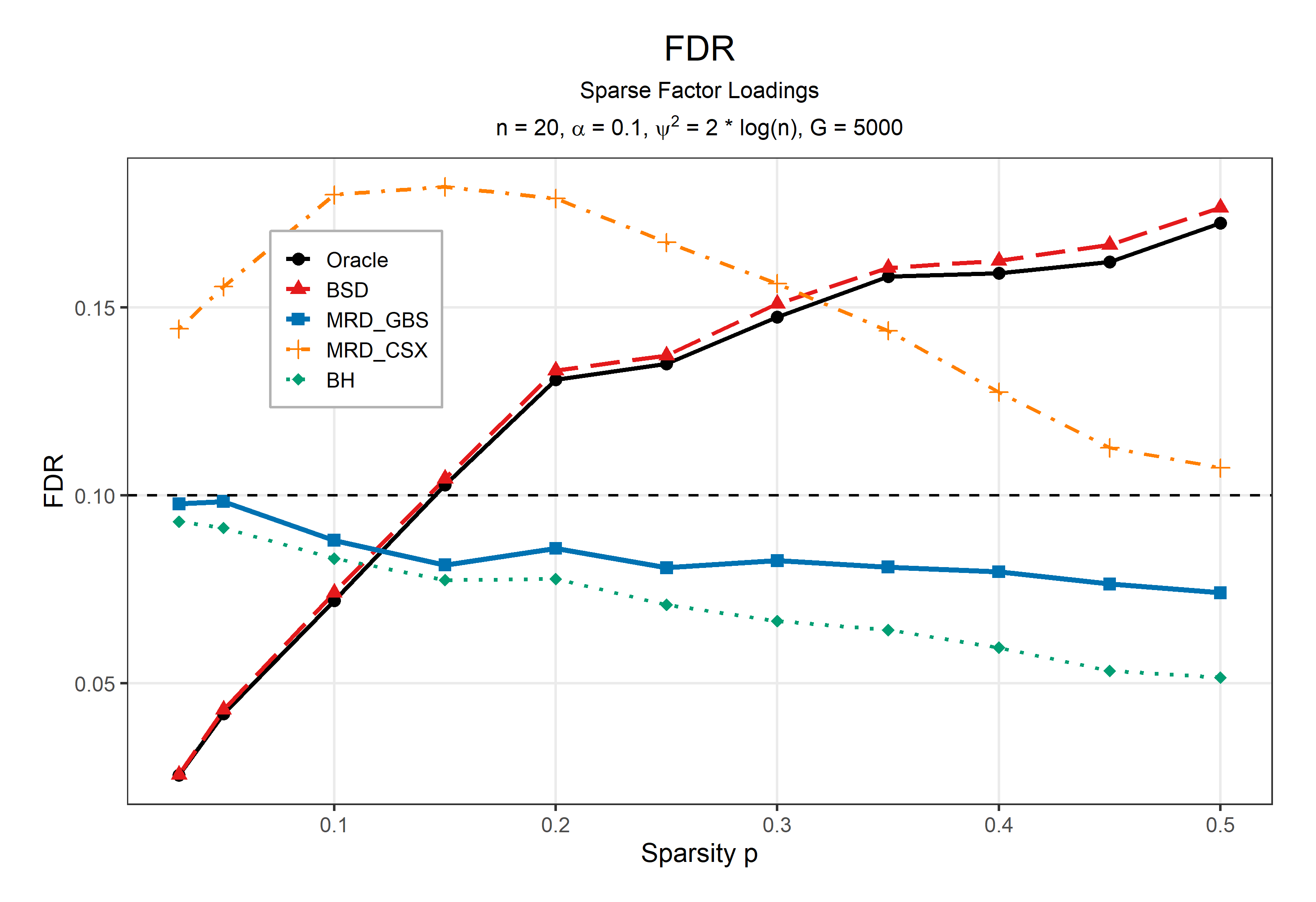}
		\caption{Sparse Factor Dependence}
	\end{subfigure}
	\hfill
	\begin{subfigure}{0.48\textwidth}
		\centering
		\includegraphics[width=\linewidth]{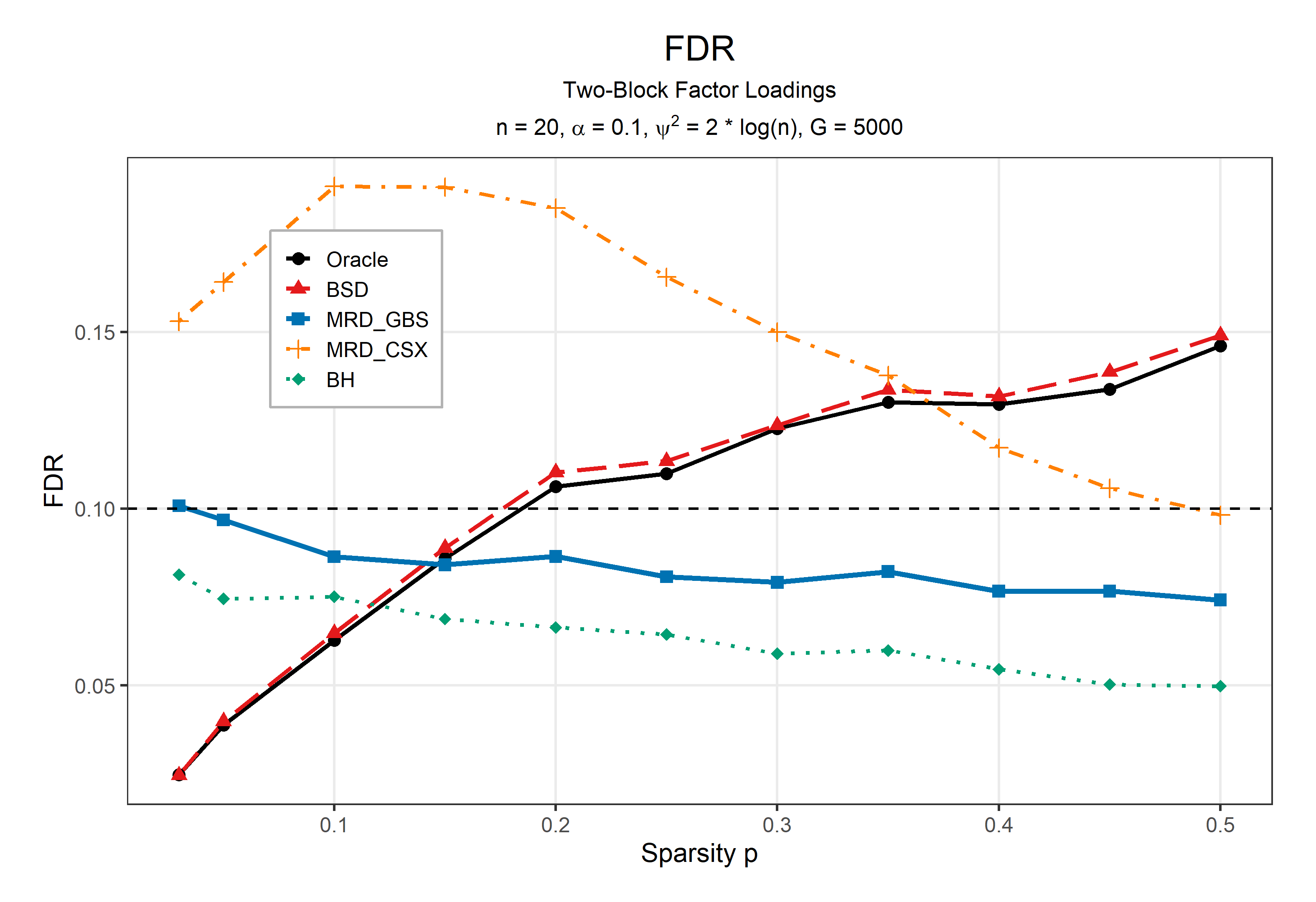}
		\caption{Two-Block Factor Dependence}
	\end{subfigure}
	
	\caption{Empirical false discovery rates under the six one-factor dependence structures when $n=20$, \(\alpha=0.1\), \(\psi^2=2\log n\), and \(G=5000\) Monte Carlo replications.}
	\label{fig:fdr-n20}
\end{figure}


\begin{figure}[p]
	\centering
	
	\begin{subfigure}{0.48\textwidth}
		\centering
		\includegraphics[width=\linewidth]{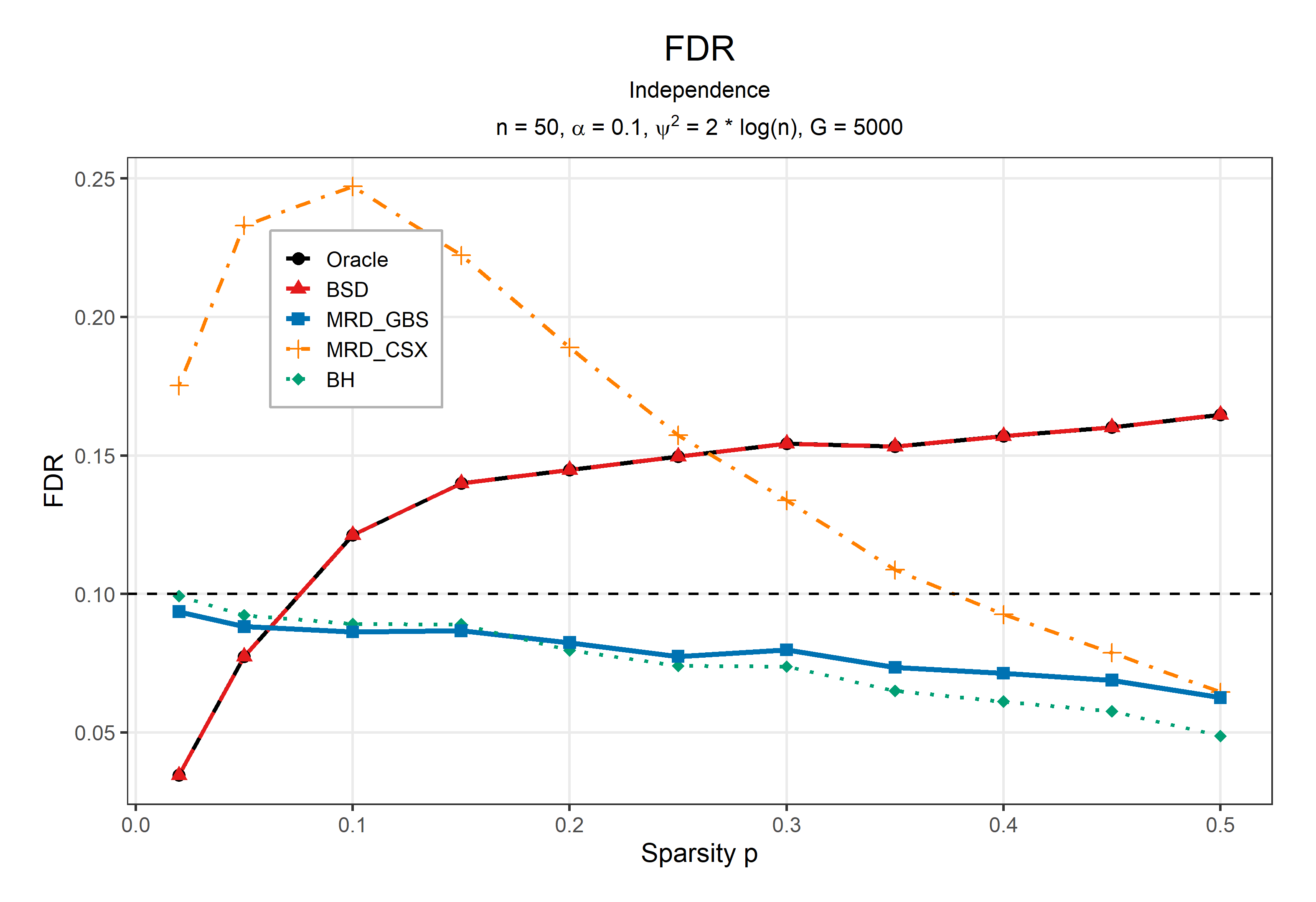}
		\caption{Independence}
	\end{subfigure}
	\hfill
	\begin{subfigure}{0.48\textwidth}
		\centering
		\includegraphics[width=\linewidth]{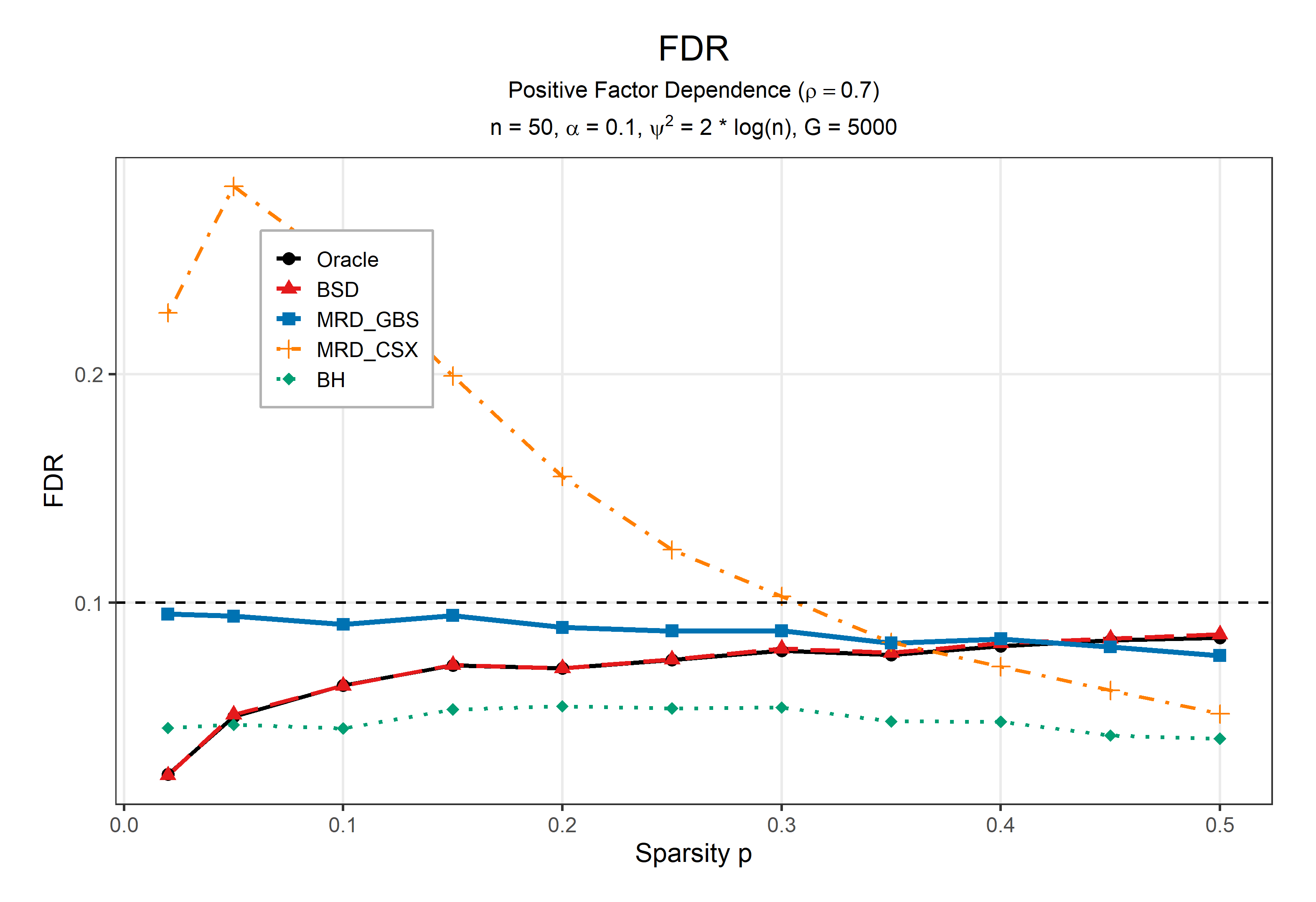}
		\caption{Positive Factor Dependence}
	\end{subfigure}
	
	\vspace{0.3cm}
	
	\begin{subfigure}{0.48\textwidth}
		\centering
		\includegraphics[width=\linewidth]{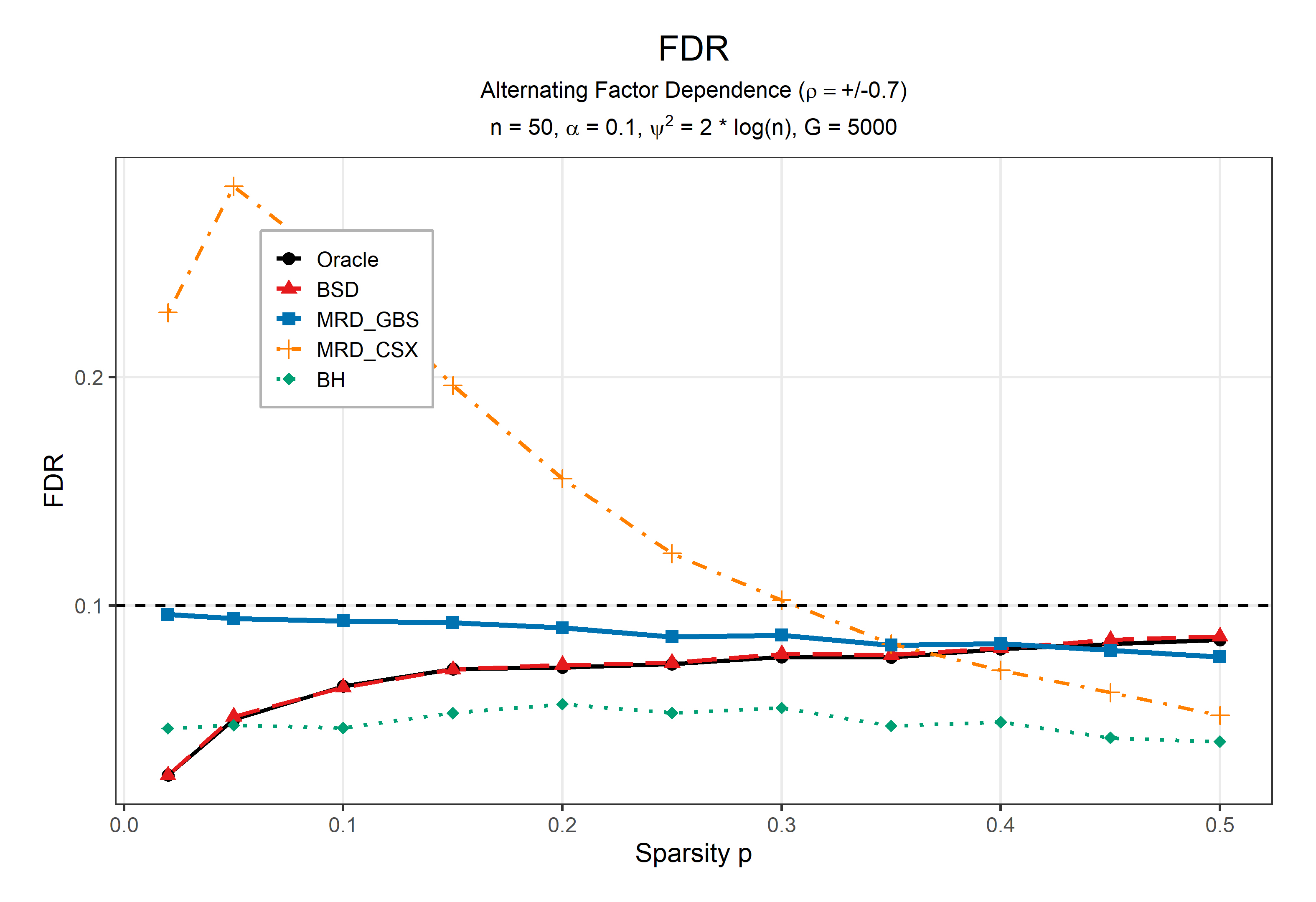}
		\caption{Alternating Factor Dependence}
	\end{subfigure}
	\hfill
	\begin{subfigure}{0.48\textwidth}
		\centering
		\includegraphics[width=\linewidth]{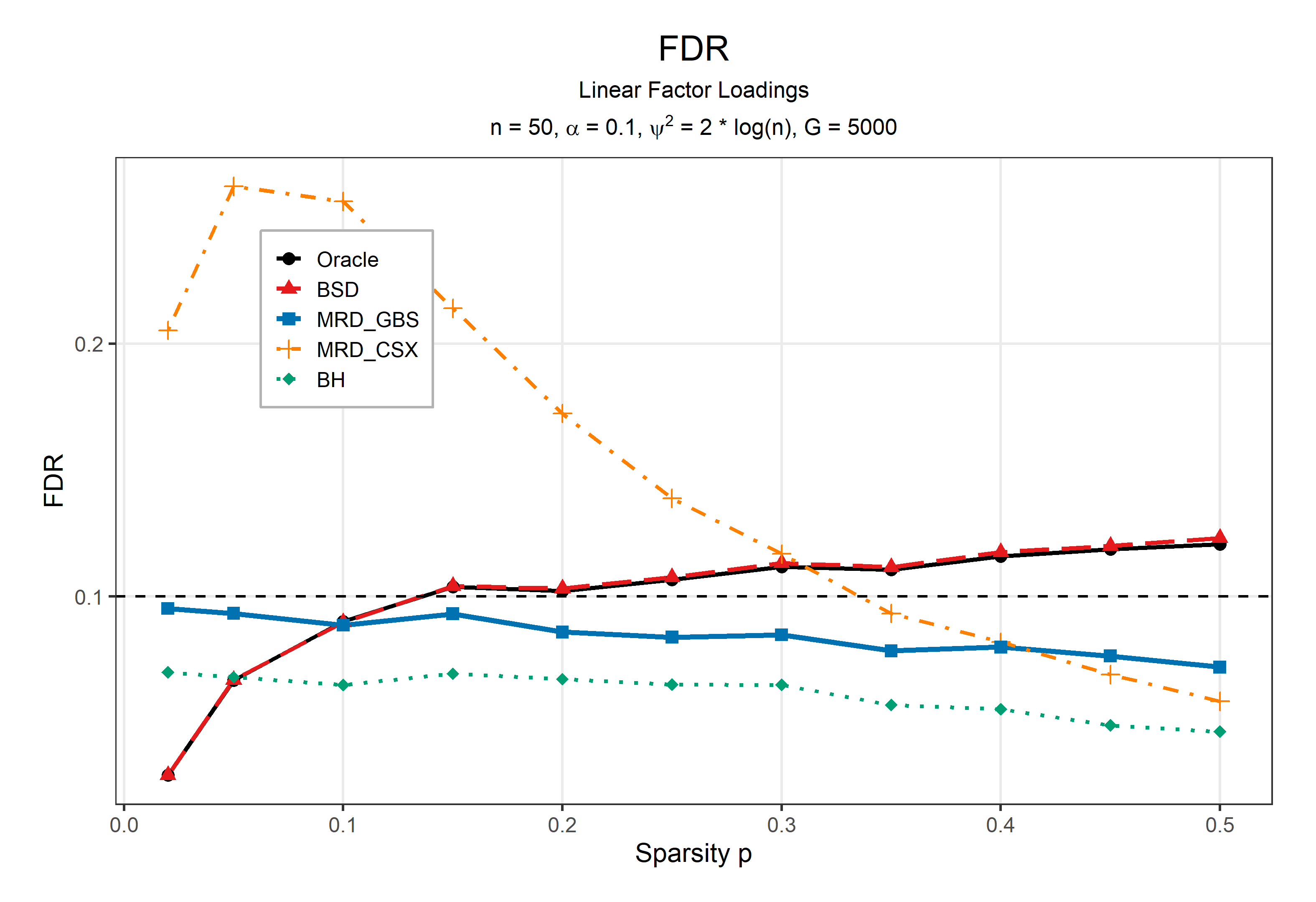}
		\caption{Linear Loadings}
	\end{subfigure}
	
	\vspace{0.3cm}
	
	\begin{subfigure}{0.48\textwidth}
		\centering
		\includegraphics[width=\linewidth]{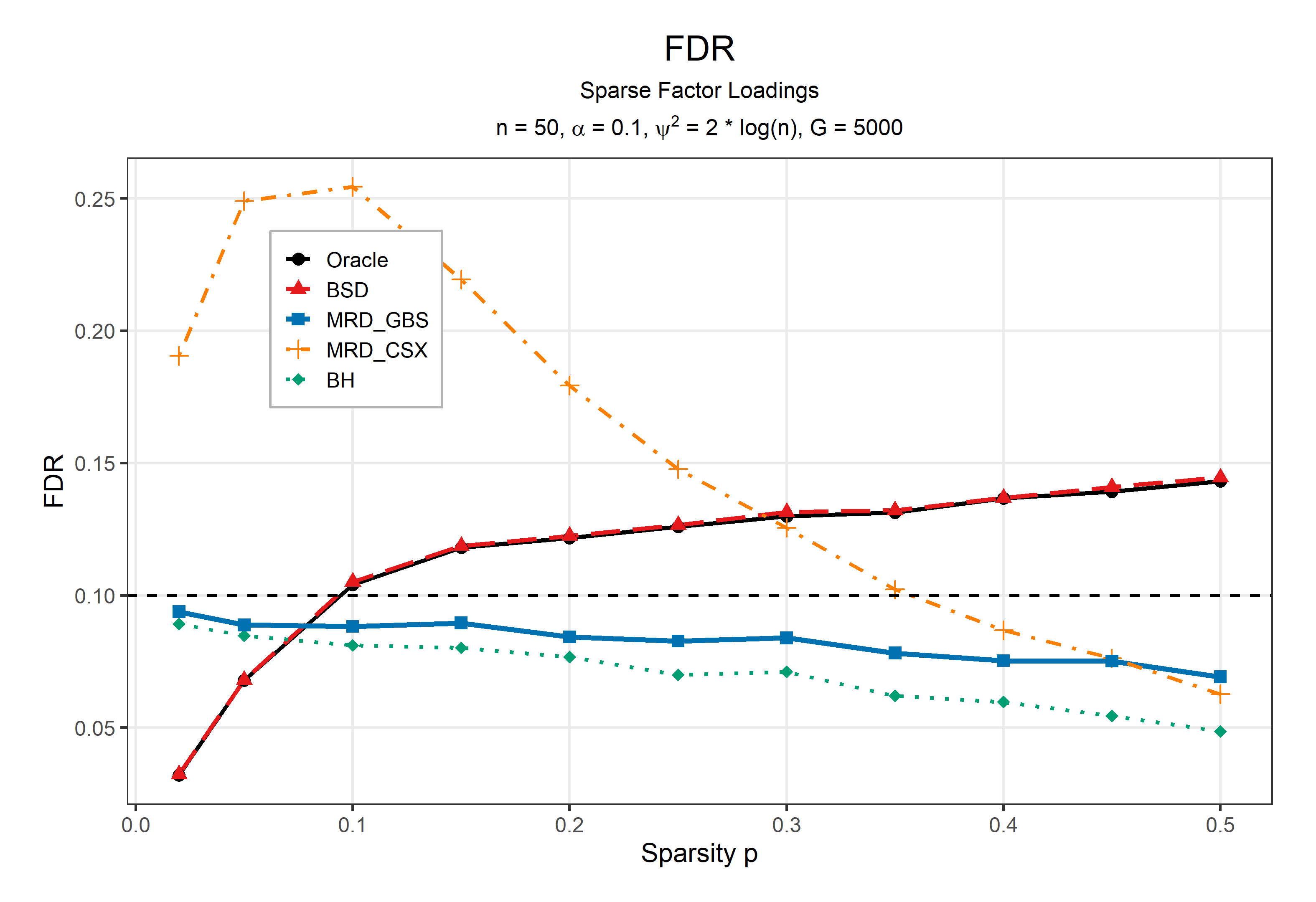}
		\caption{Sparse Factor Dependence}
	\end{subfigure}
	\hfill
	\begin{subfigure}{0.48\textwidth}
		\centering
		\includegraphics[width=\linewidth]{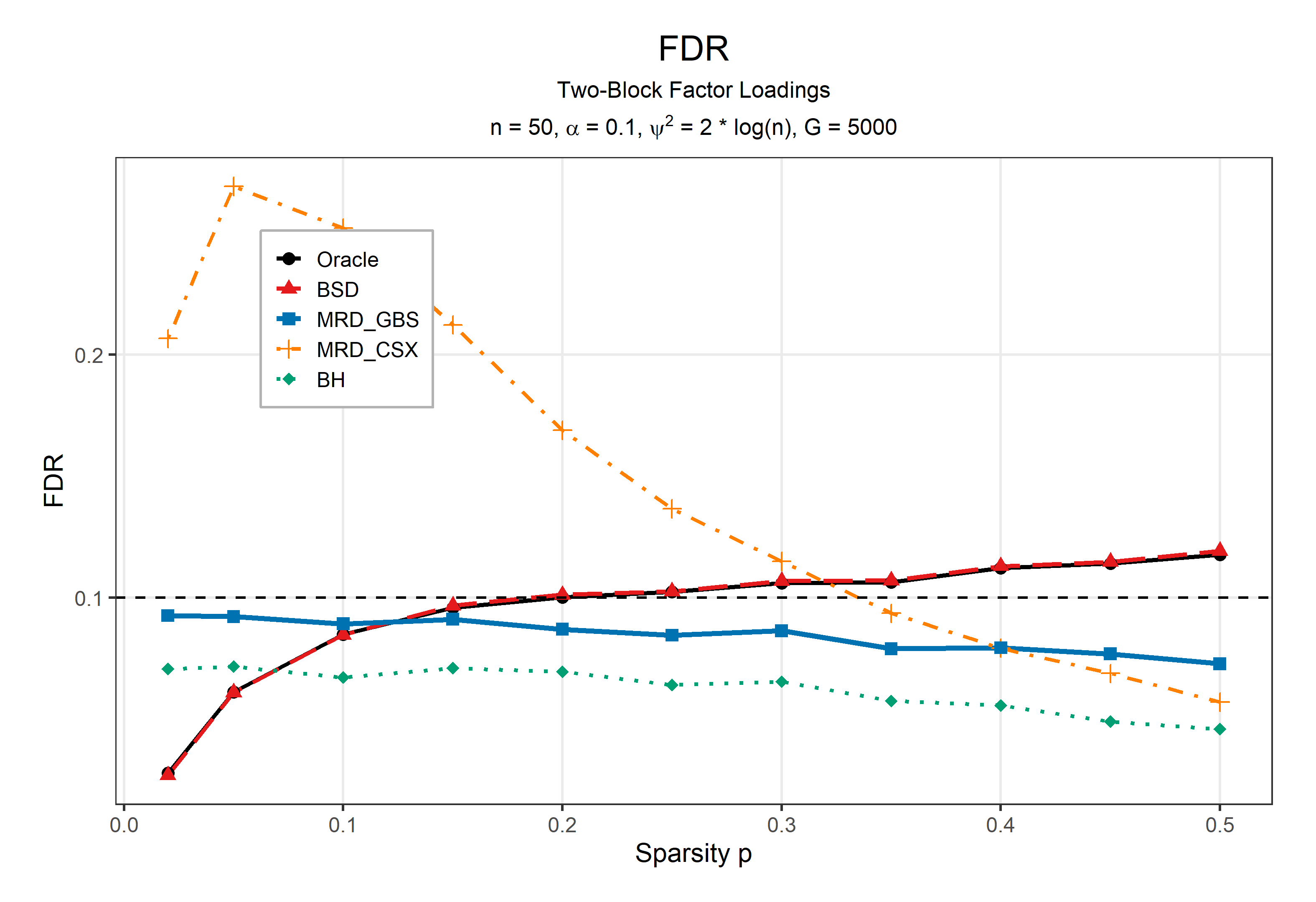}
		\caption{Two-Block Factor Dependence}
	\end{subfigure}
	
	\caption{Empirical false discovery rates under the six one-factor dependence structures when $n=50$, \(\alpha=0.1\), \(\psi^2=2\log n\), and \(G=5000\) Monte Carlo replications.}
	\label{fig:fdr-n50}
\end{figure}


\begin{figure}[p]
	\centering
	
	\begin{subfigure}{0.48\textwidth}
		\centering
		\includegraphics[width=\linewidth]{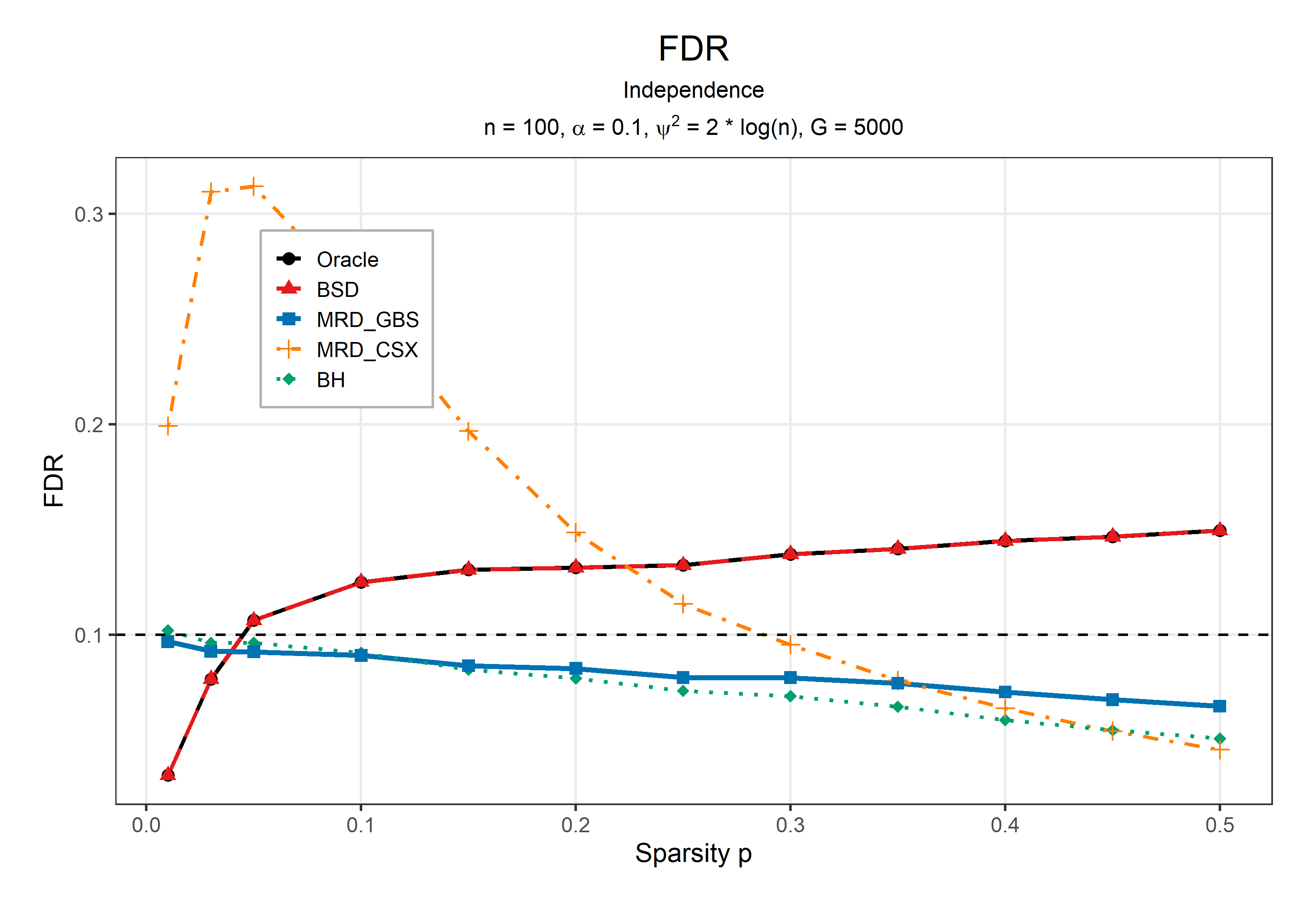}
		\caption{Independence}
	\end{subfigure}
	\hfill
	\begin{subfigure}{0.48\textwidth}
		\centering
		\includegraphics[width=\linewidth]{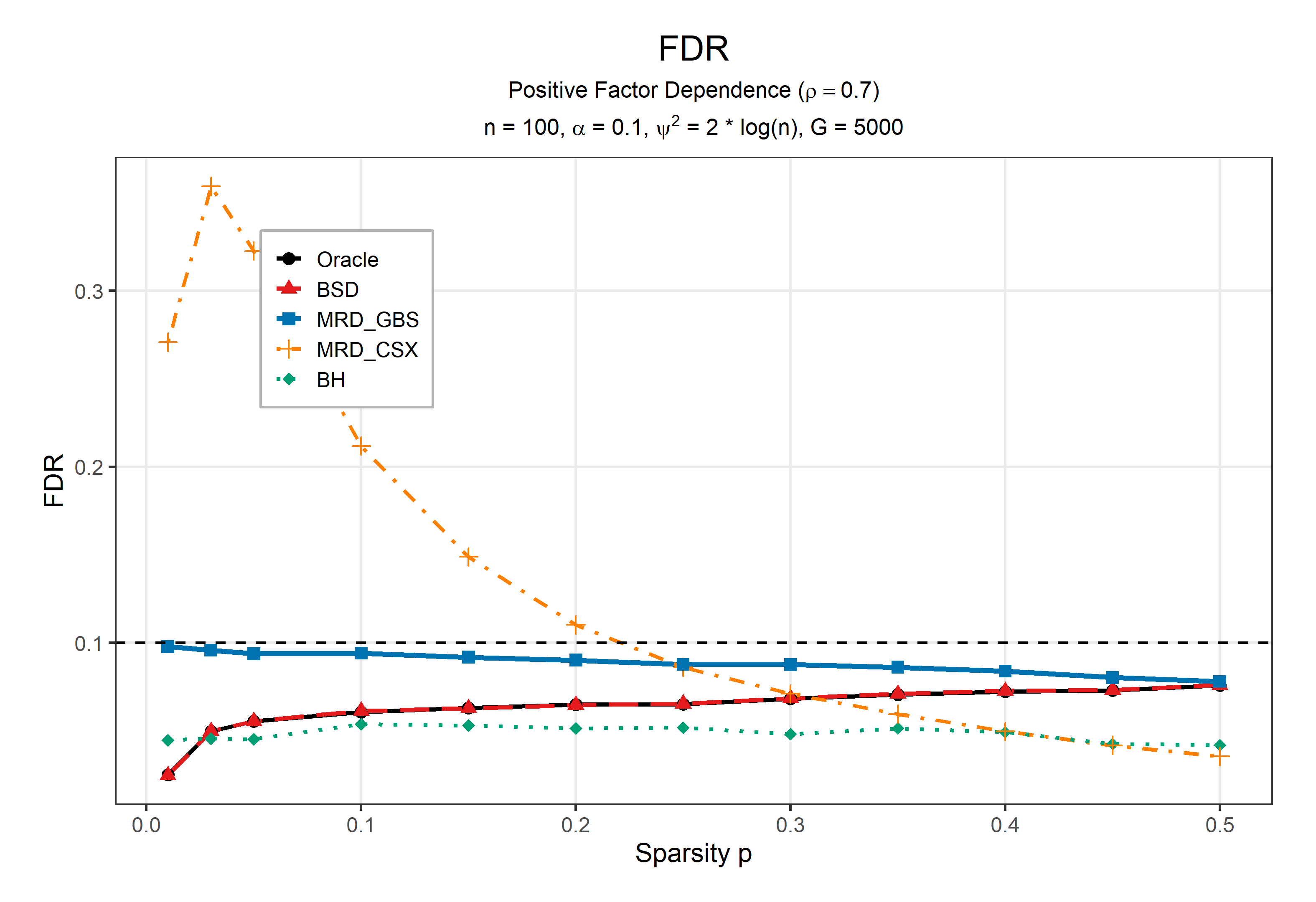}
		\caption{Positive Factor Dependence}
	\end{subfigure}
	
	\vspace{0.3cm}
	
	\begin{subfigure}{0.48\textwidth}
		\centering
		\includegraphics[width=\linewidth]{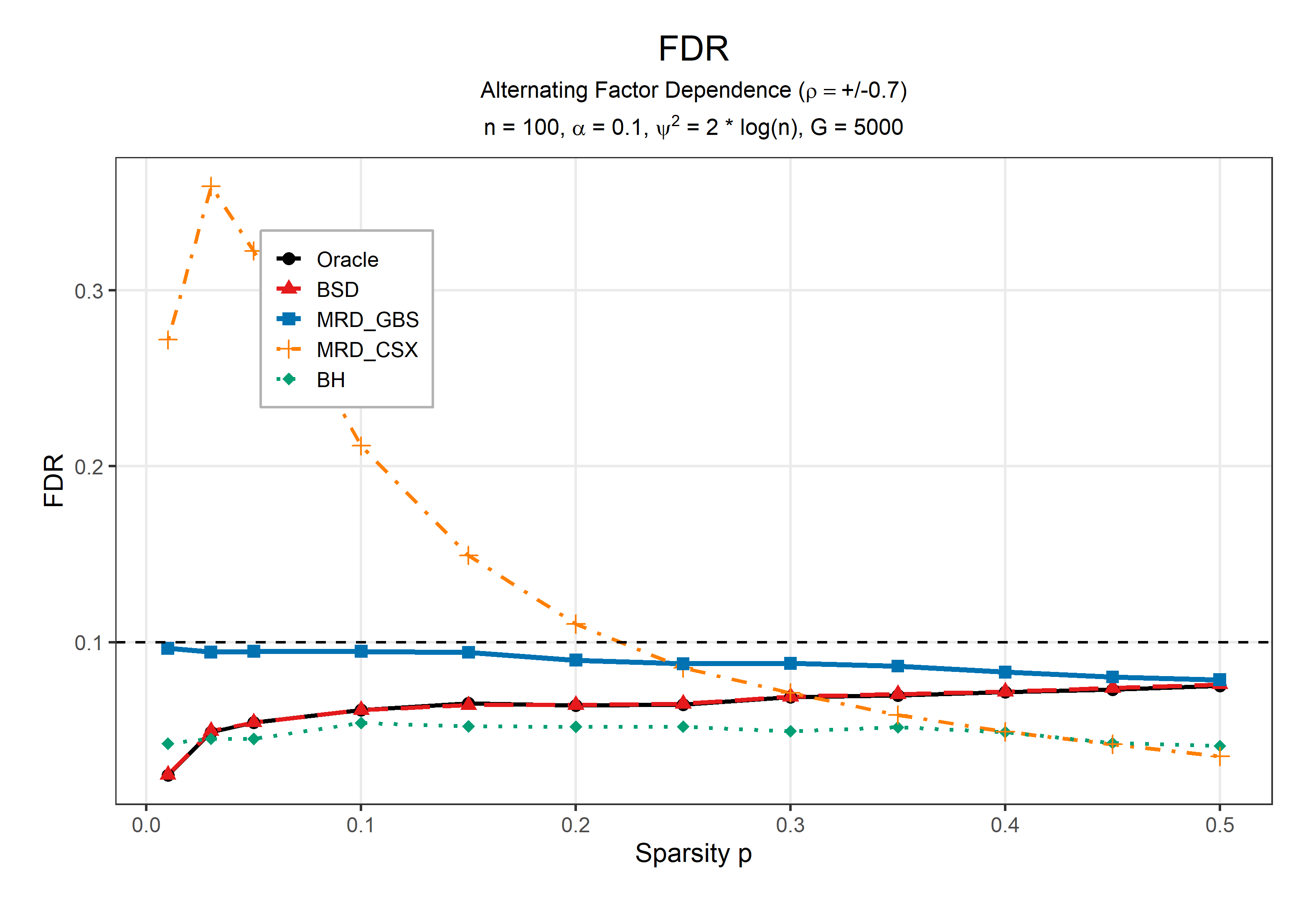}
		\caption{Alternating Factor Dependence}
	\end{subfigure}
	\hfill
	\begin{subfigure}{0.48\textwidth}
		\centering
		\includegraphics[width=\linewidth]{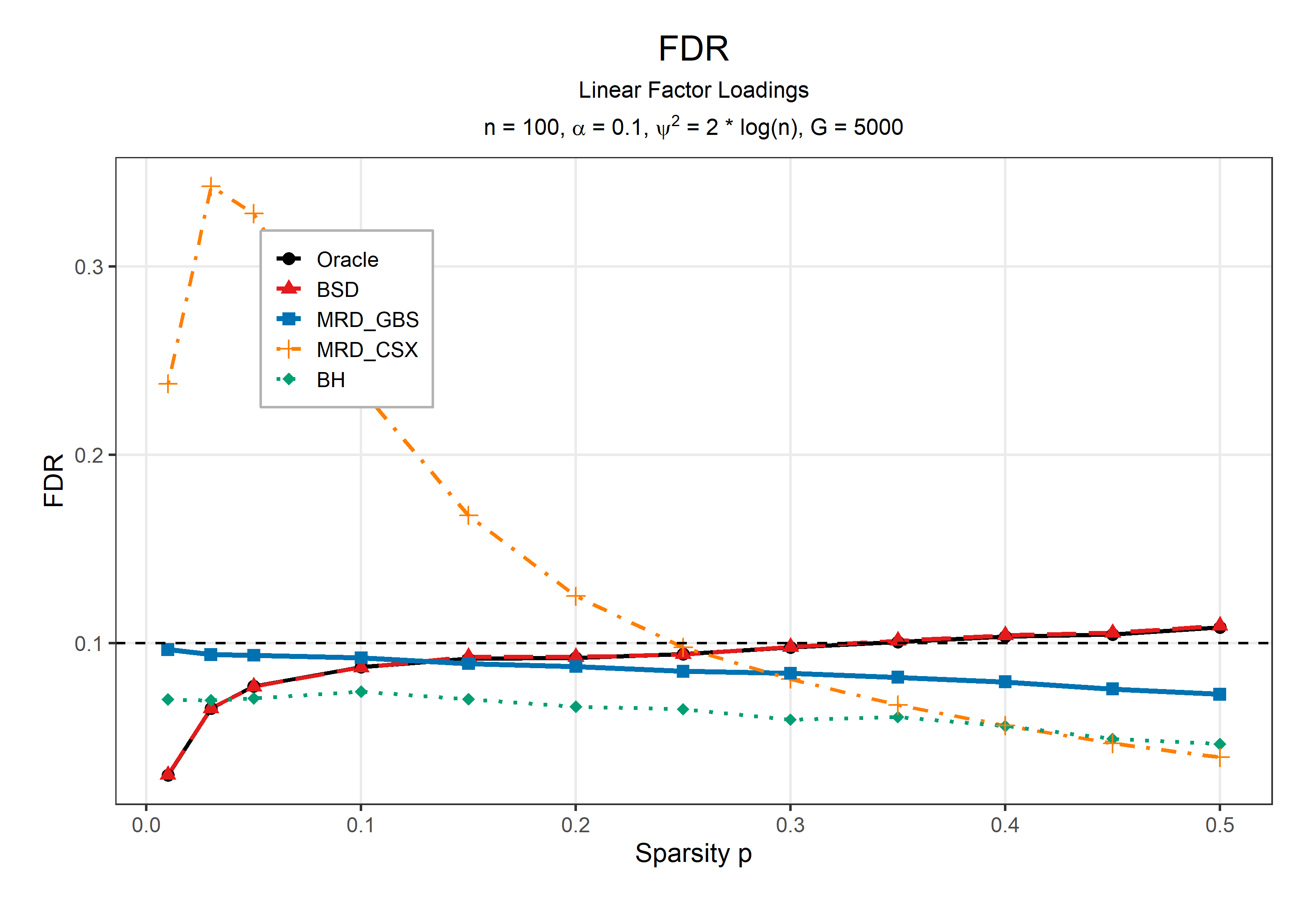}
		\caption{Linear Loadings}
	\end{subfigure}
	
	\vspace{0.3cm}
	
	\begin{subfigure}{0.48\textwidth}
		\centering
		\includegraphics[width=\linewidth]{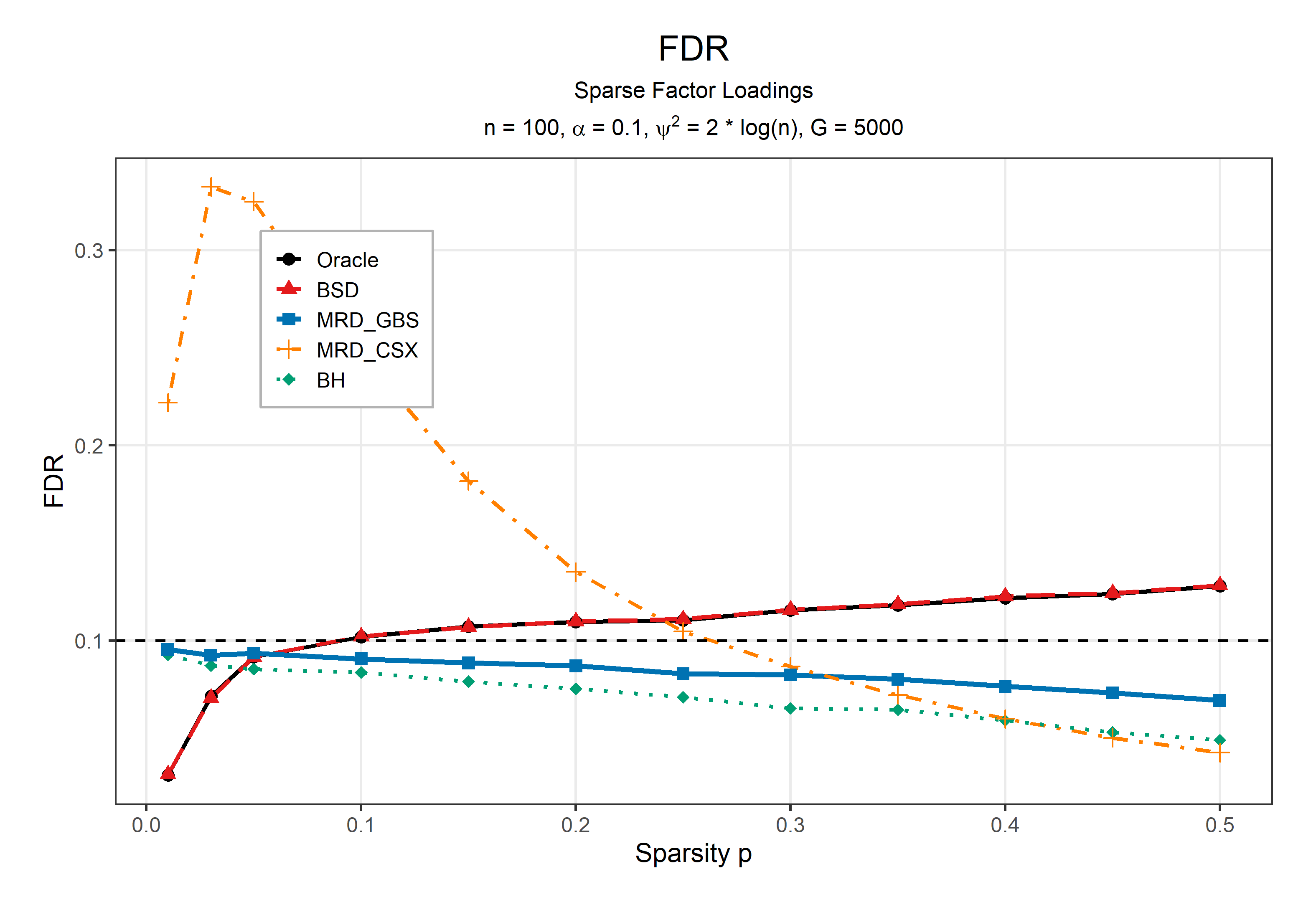}
		\caption{Sparse Factor Dependence}
	\end{subfigure}
	\hfill
	\begin{subfigure}{0.48\textwidth}
		\centering
		\includegraphics[width=\linewidth]{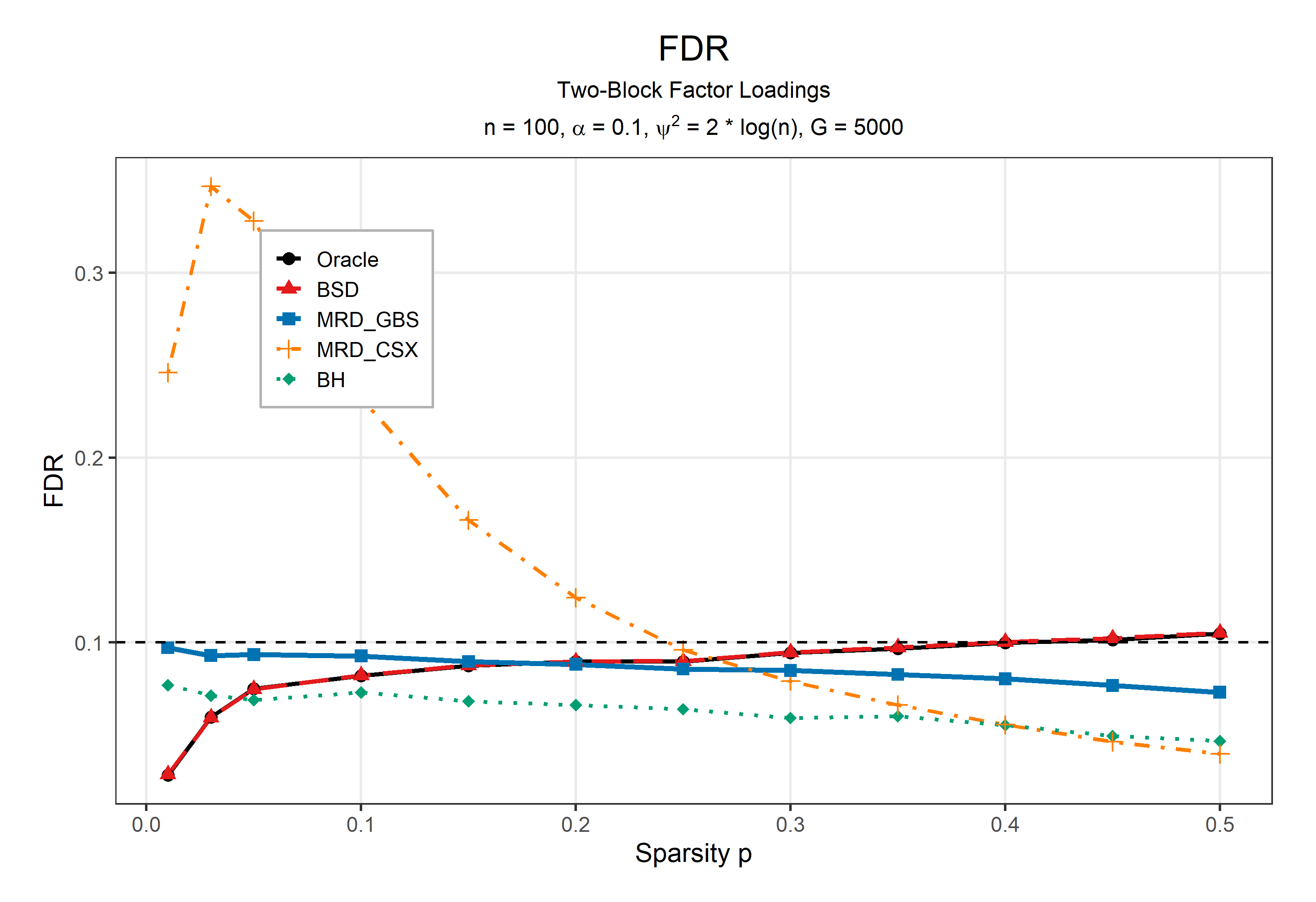}
		\caption{Two-Block Factor Dependence}
	\end{subfigure}
	
	\caption{Empirical false discovery rates under the six one-factor dependence structures when $n=100$, \(\alpha=0.1\), \(\psi^2=2\log n\), and \(G=5000\) Monte Carlo replications.}
	\label{fig:fdr-n100}
\end{figure}


A striking feature of the results is the near-perfect agreement between BSD and the Bayes Oracle across all dependence structures and dimensions considered. The two procedures exhibit almost identical FDR curves throughout the entire sparsity range considered, providing further evidence that BSD successfully reproduces the operating characteristics of the Oracle rule under a broad collection of covariance structures. It is worth emphasizing that neither BSD nor the Bayes Oracle is calibrated to attain any preassigned FDR level. Both procedures arise from a Bayes classification framework whose objective is to minimize overall misclassification loss rather than to maximize discoveries subject to an FDR constraint. Consequently, there is no theoretical reason for their FDRs to coincide with the nominal level $\alpha=0.10$. Indeed, the Oracle frequently operates at substantially smaller FDR levels, particularly in moderate and dense sparsity regimes. This behavior is entirely consistent with its underlying objective, since the Oracle is designed to optimize the overall trade-off between false discoveries and missed discoveries rather than to maximize the number of discoveries subject to an FDR constraint. The close agreement between BSD and the Oracle therefore suggests that BSD successfully reproduces the error-allocation strategy induced by Bayes-risk minimization and continues to recover nearly the same decision boundary selected by the Bayes Oracle even under substantial covariance dependence.

The MRD--GBS procedure displays a markedly different behavior. Across virtually all configurations, its empirical FDR remains close to or below the nominal level $\alpha=0.1$, often exhibiting a conservative tendency, particularly as the dimension increases. This observation is consistent with the strong support-recovery performance observed earlier and helps explain why MRD--GBS achieves Bayes risks that remain surprisingly close to the Bayes Oracle despite arising from a fundamentally different methodological framework.

In contrast, MRD--CSX exhibits substantial variability in its FDR behavior. For smaller sparsity levels, its FDR frequently exceeds the nominal level by a considerable margin, particularly when $n=50$ and $n=100$, before gradually decreasing as the signal proportion increases. The BH procedure, on the other hand, tends to be noticeably conservative under the dependence structures considered here, often producing empirical FDR values substantially below the nominal level.

Overall, the FDR results reinforce the conclusions drawn from the Bayes risk analysis. BSD closely tracks the Bayes Oracle not only in overall misclassification risk but also in its allocation of Type I errors. At the same time, the consistently stable FDR behavior of MRD--GBS provides additional evidence of its strong finite-sample performance under both independence and dependence.


\subsubsection{False Non-Discovery Rates}

The empirical false non-discovery rates (FNRs) are displayed in Figures~\ref{fig:fnr-n20}--\ref{fig:fnr-n100}. Since FNR measures the proportion of missed signals among all non-discoveries, it provides a direct assessment of the signal-recovery ability of the competing procedures. Consequently, the FNR results offer important insight into the mechanisms underlying the Bayes risk comparisons presented in the previous subsection.

\begin{figure}[p]
	\centering
	
	\begin{subfigure}{0.48\textwidth}
		\centering
		\includegraphics[width=\linewidth]{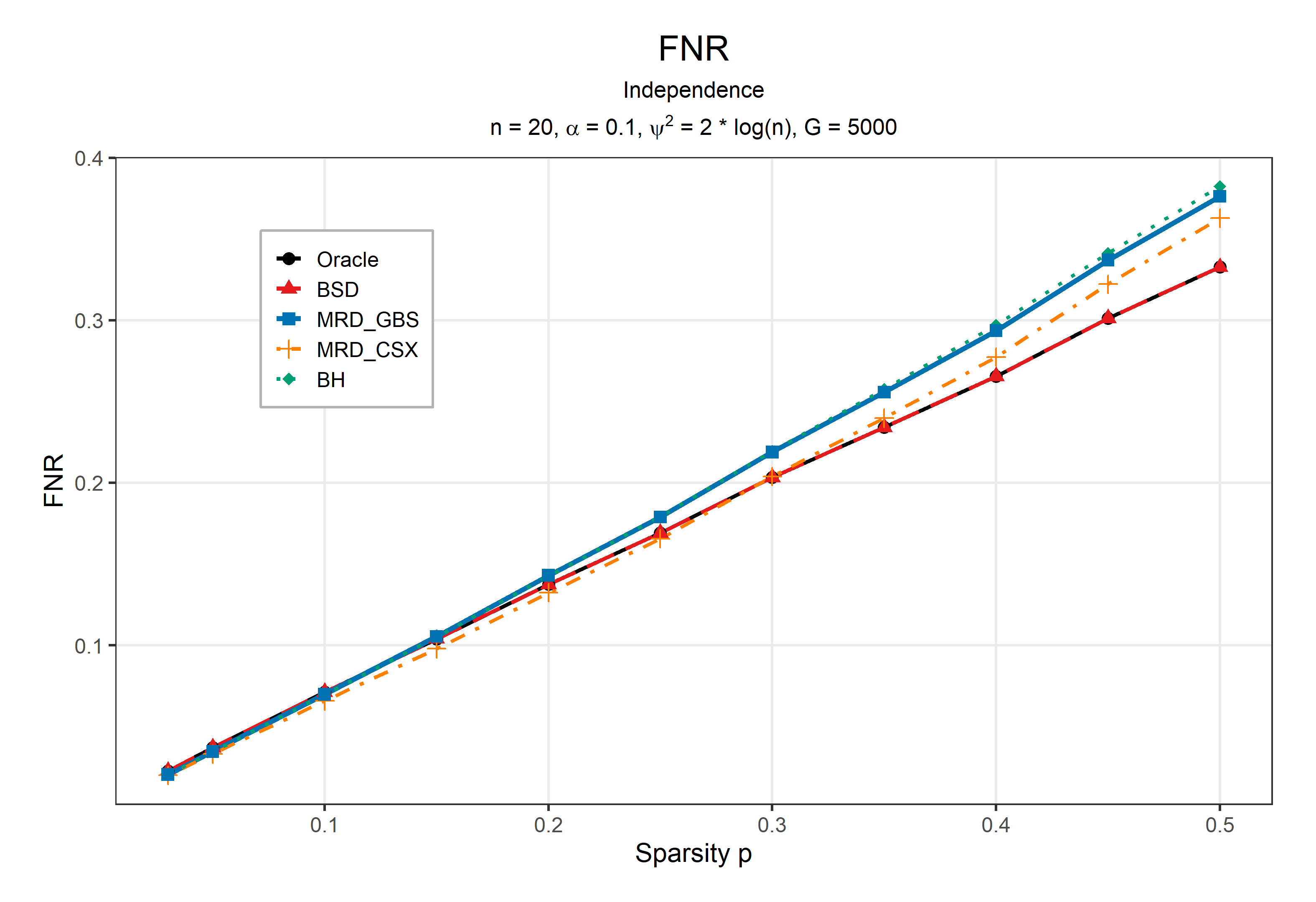}
		\caption{Independence}
	\end{subfigure}
	\hfill
	\begin{subfigure}{0.48\textwidth}
		\centering
		\includegraphics[width=\linewidth]{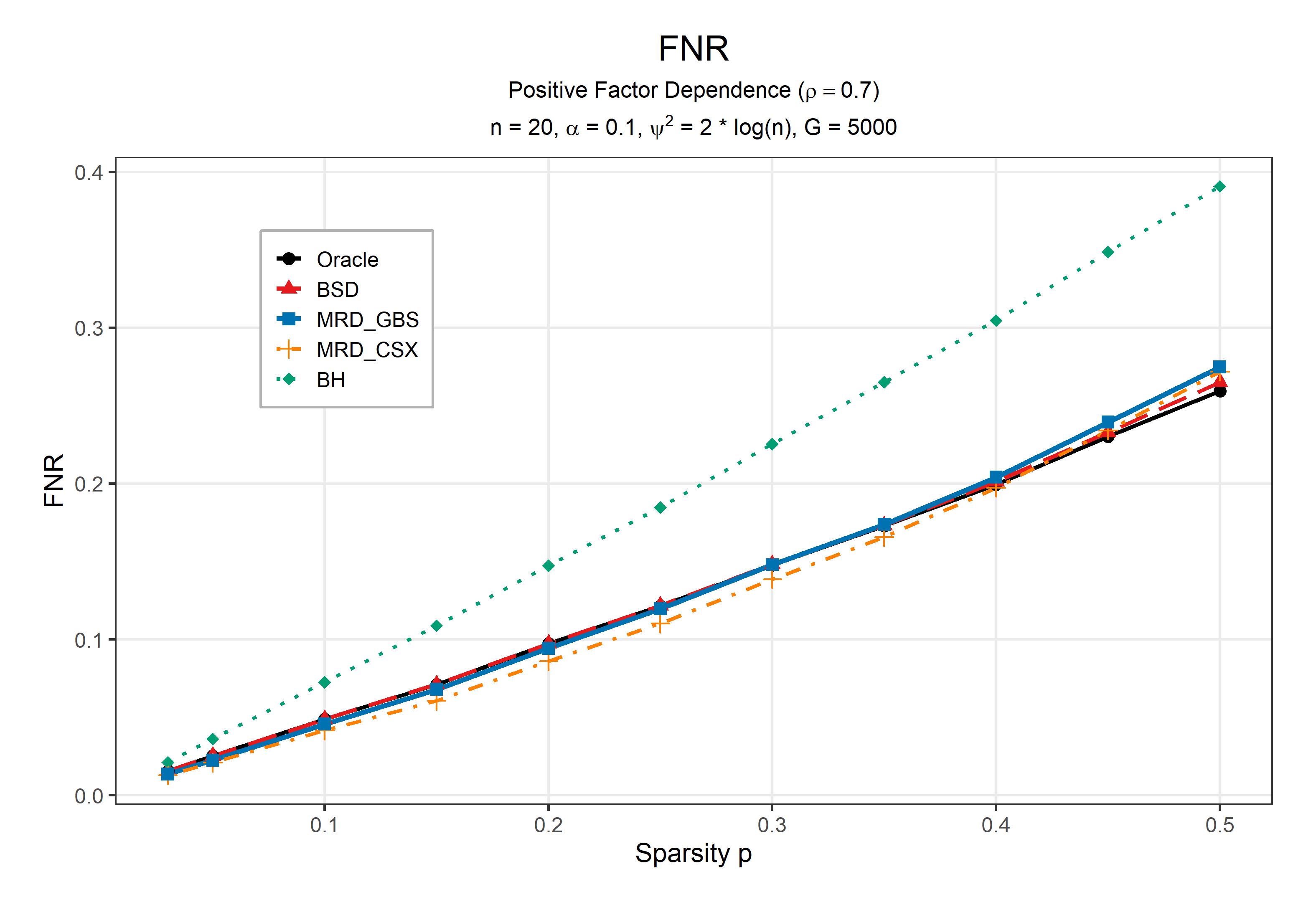}
		\caption{Positive Factor Dependence}
	\end{subfigure}
	
	\vspace{0.3cm}
	
	\begin{subfigure}{0.48\textwidth}
		\centering
		\includegraphics[width=\linewidth]{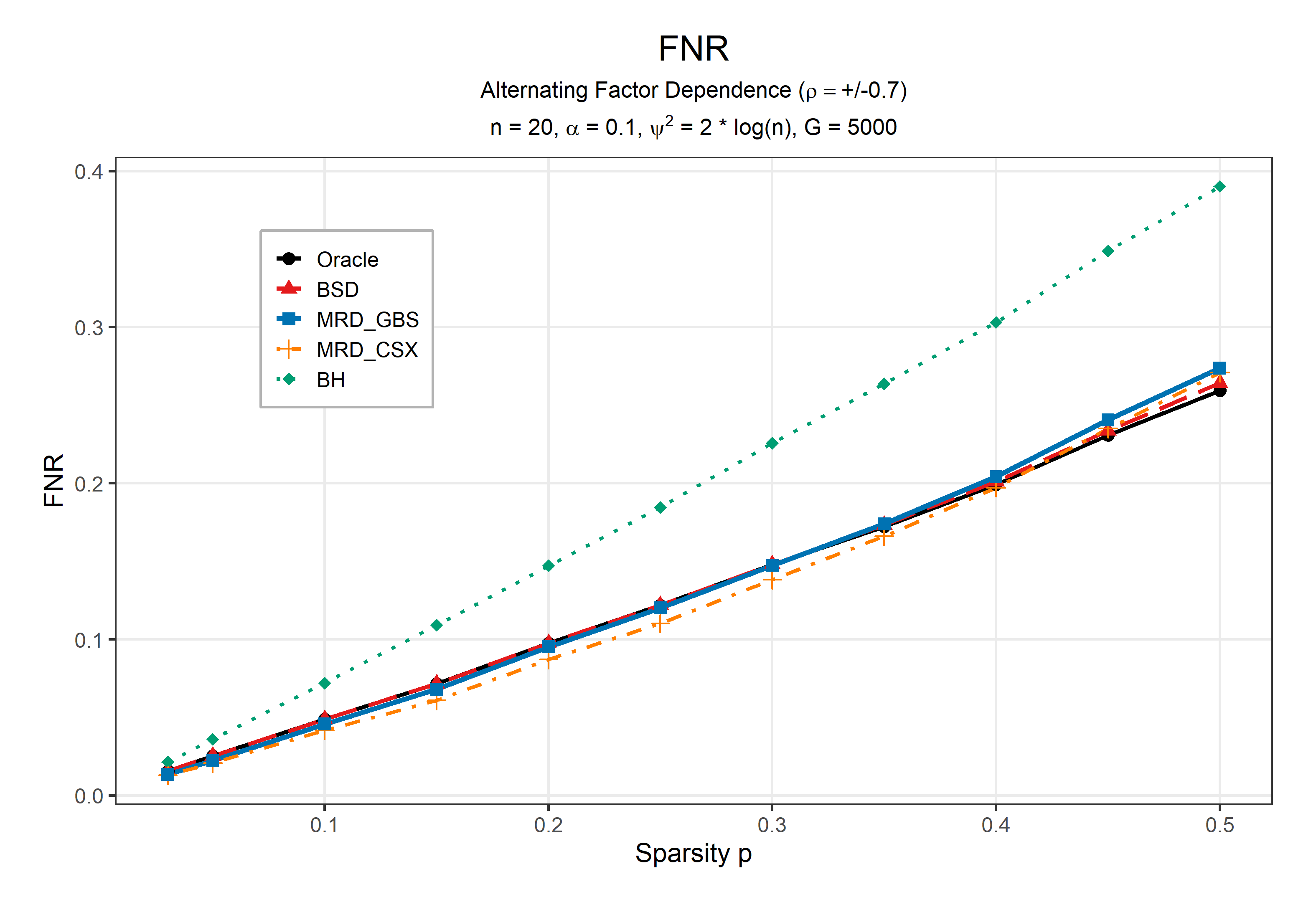}
		\caption{Alternating Factor Dependence}
	\end{subfigure}
	\hfill
	\begin{subfigure}{0.48\textwidth}
		\centering
		\includegraphics[width=\linewidth]{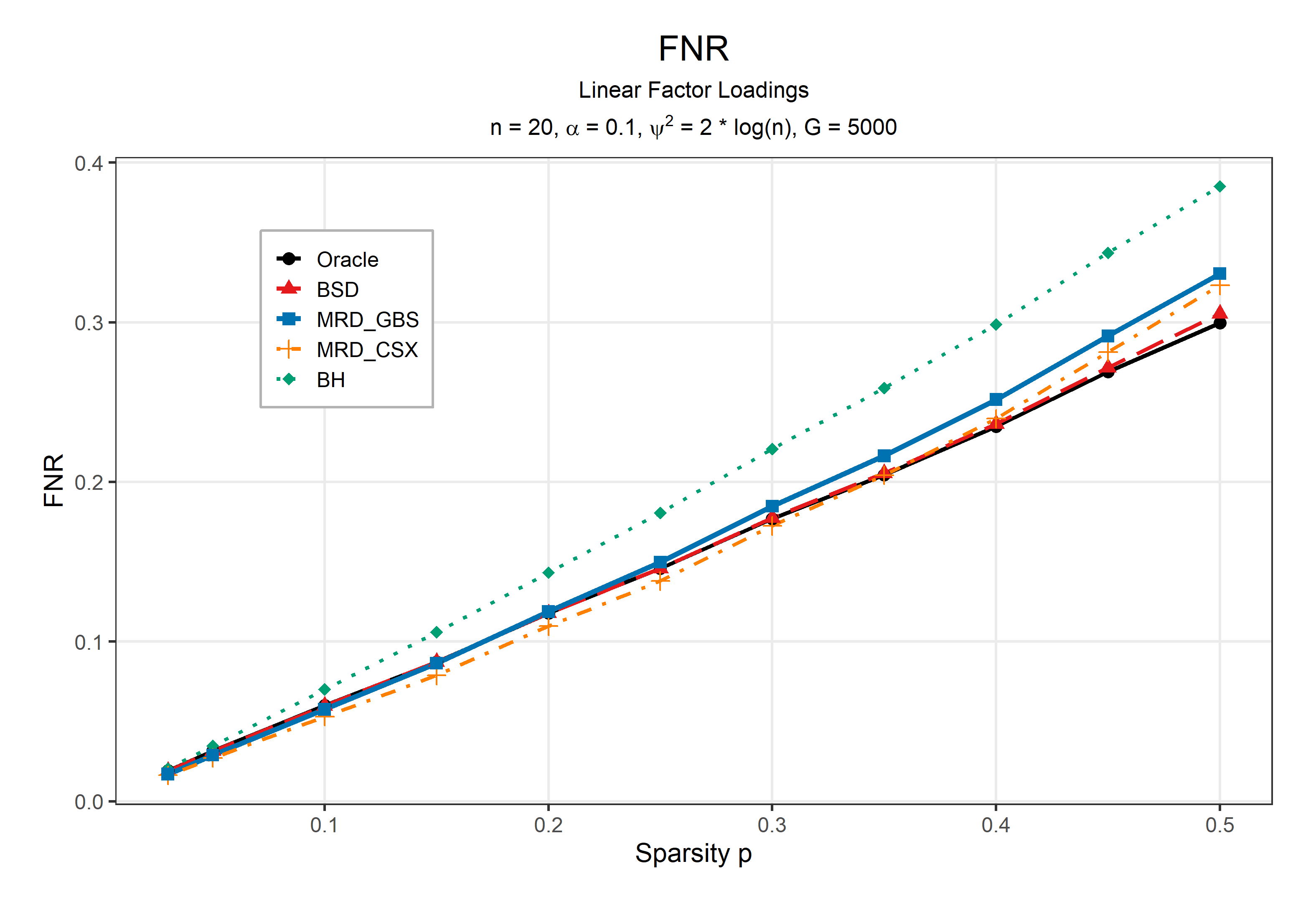}
		\caption{Linear Loadings}
	\end{subfigure}
	
	\vspace{0.3cm}
	
	\begin{subfigure}{0.48\textwidth}
		\centering
		\includegraphics[width=\linewidth]{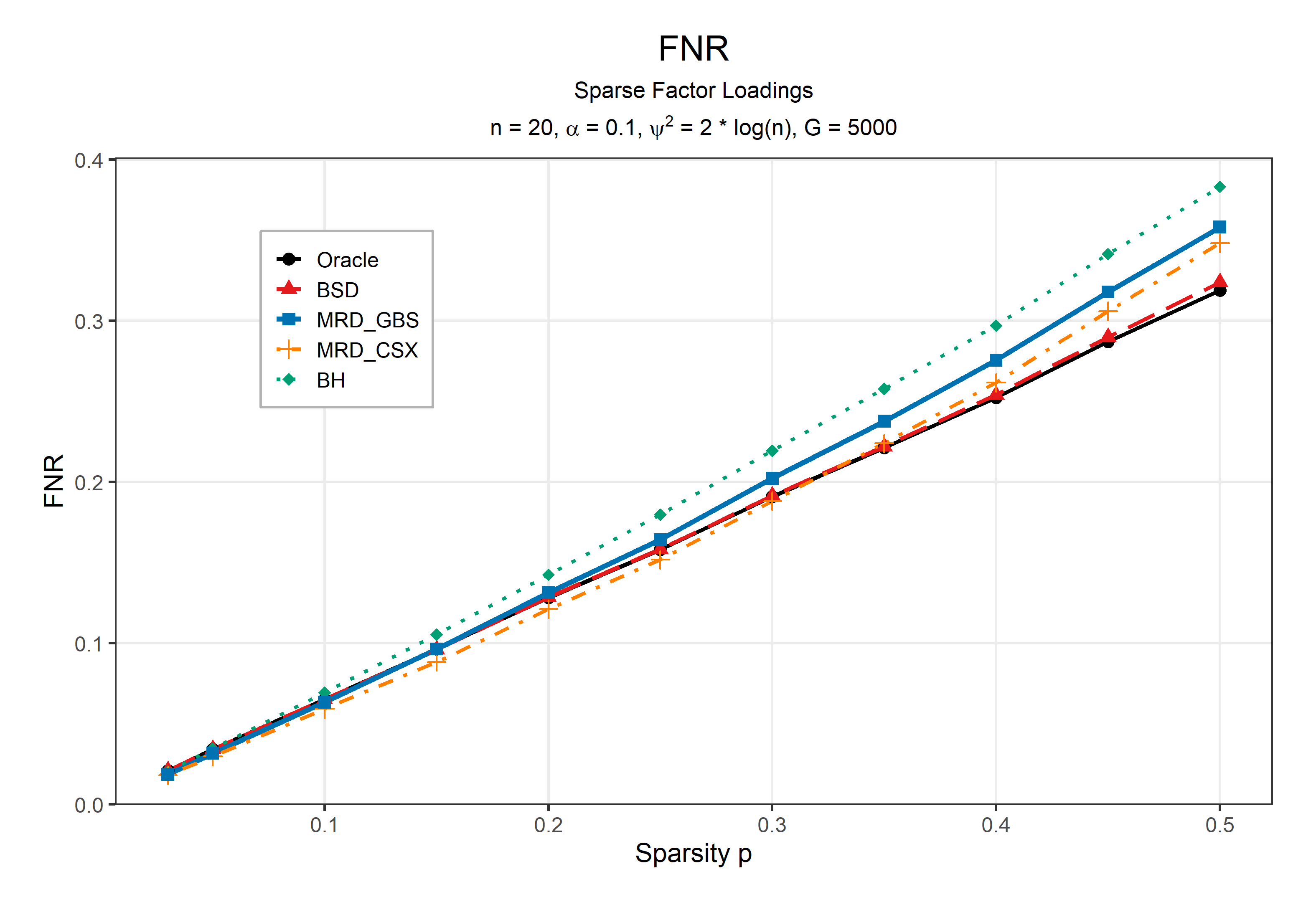}
		\caption{Sparse Factor Dependence}
	\end{subfigure}
	\hfill
	\begin{subfigure}{0.48\textwidth}
		\centering
		\includegraphics[width=\linewidth]{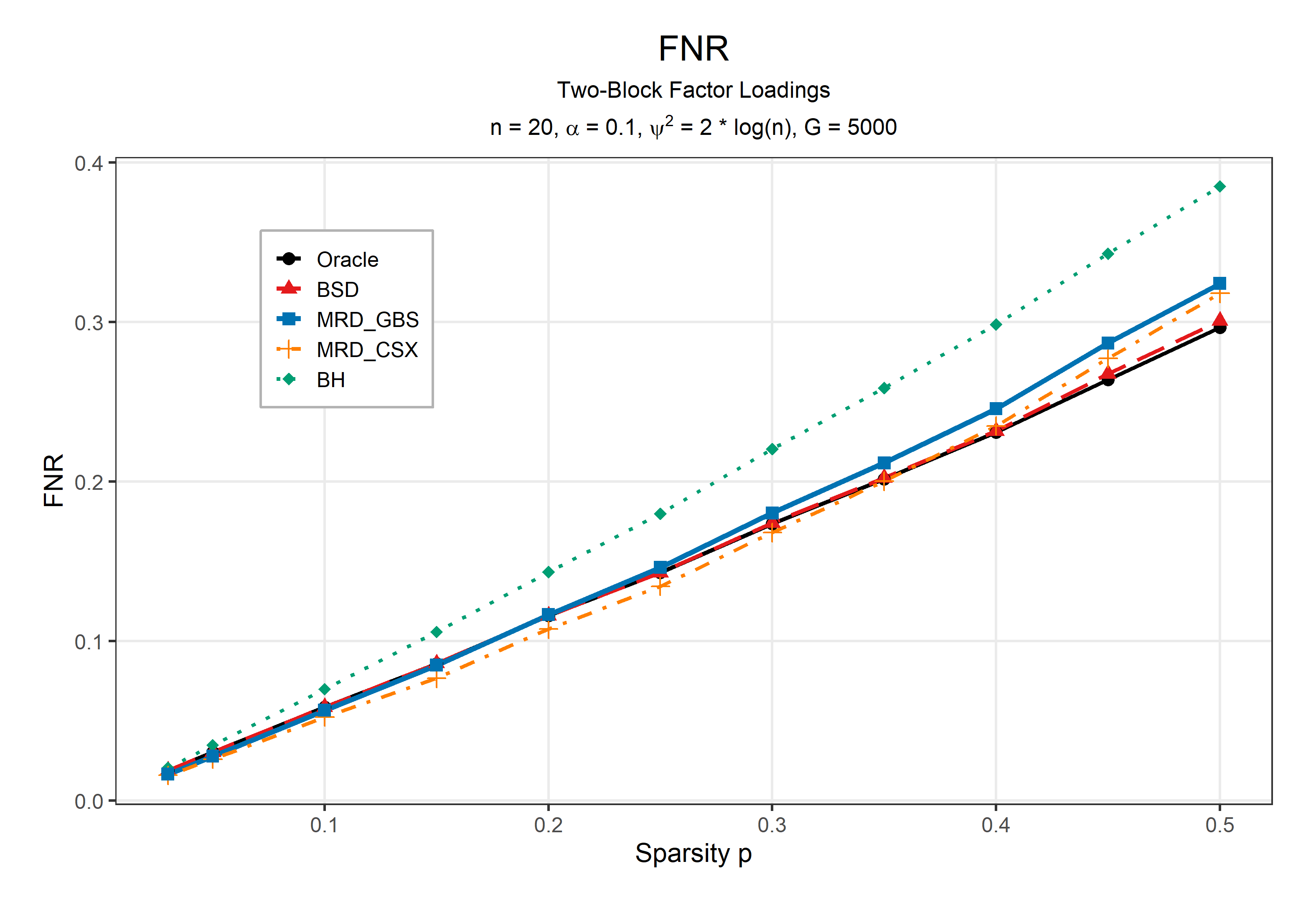}
		\caption{Two-Block Factor Dependence}
	\end{subfigure}
	
	\caption{Empirical false non-discovery rates under the six one-factor dependence structures when $n=20$, \(\alpha=0.1\), \(\psi^2=2\log n\), and \(G=5000\) Monte Carlo replications.}
	\label{fig:fnr-n20}
\end{figure}

\begin{figure}[p]
	\centering
	
	\begin{subfigure}{0.48\textwidth}
		\centering
		\includegraphics[width=\linewidth]{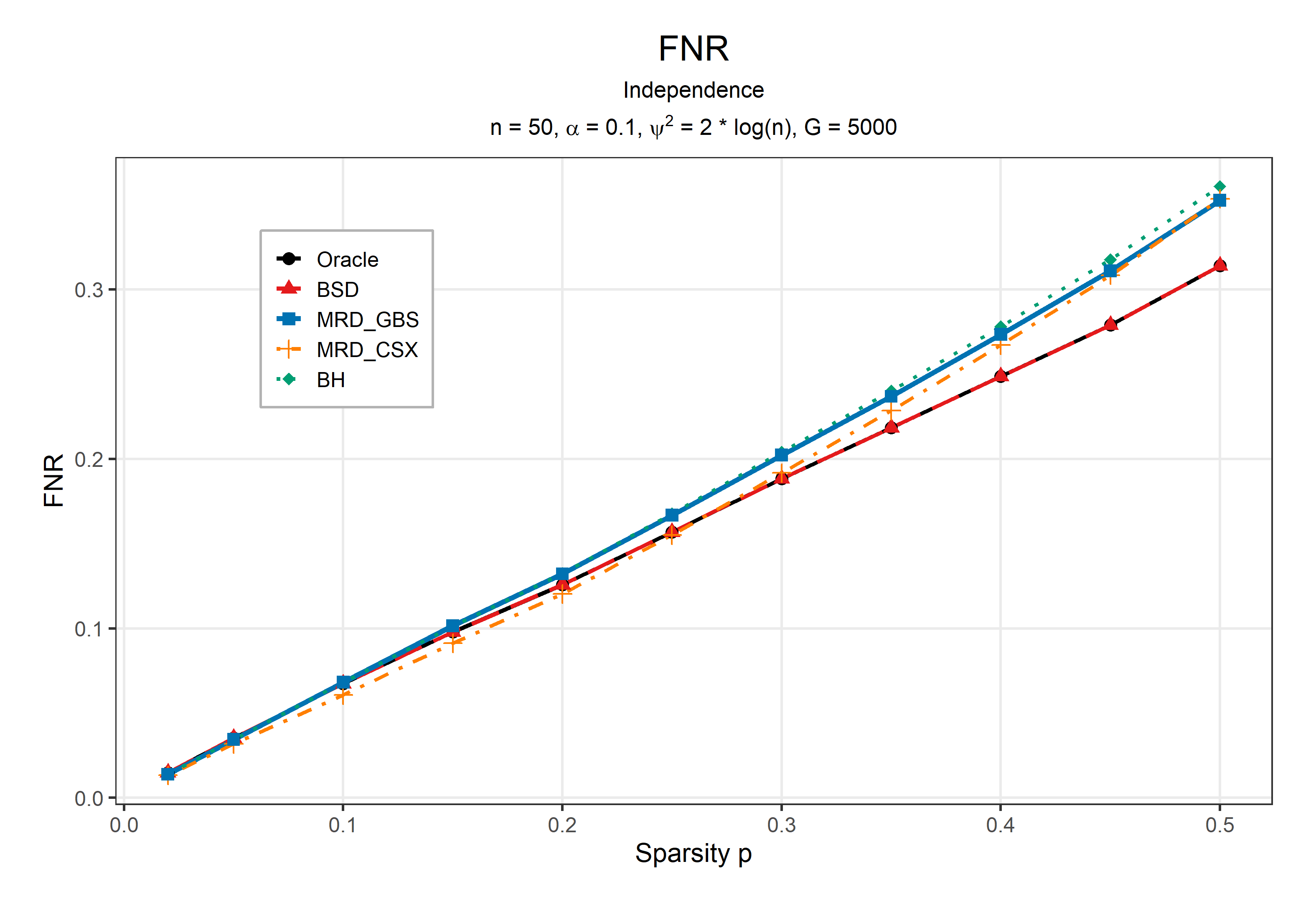}
		\caption{Independence}
	\end{subfigure}
	\hfill
	\begin{subfigure}{0.48\textwidth}
		\centering
		\includegraphics[width=\linewidth]{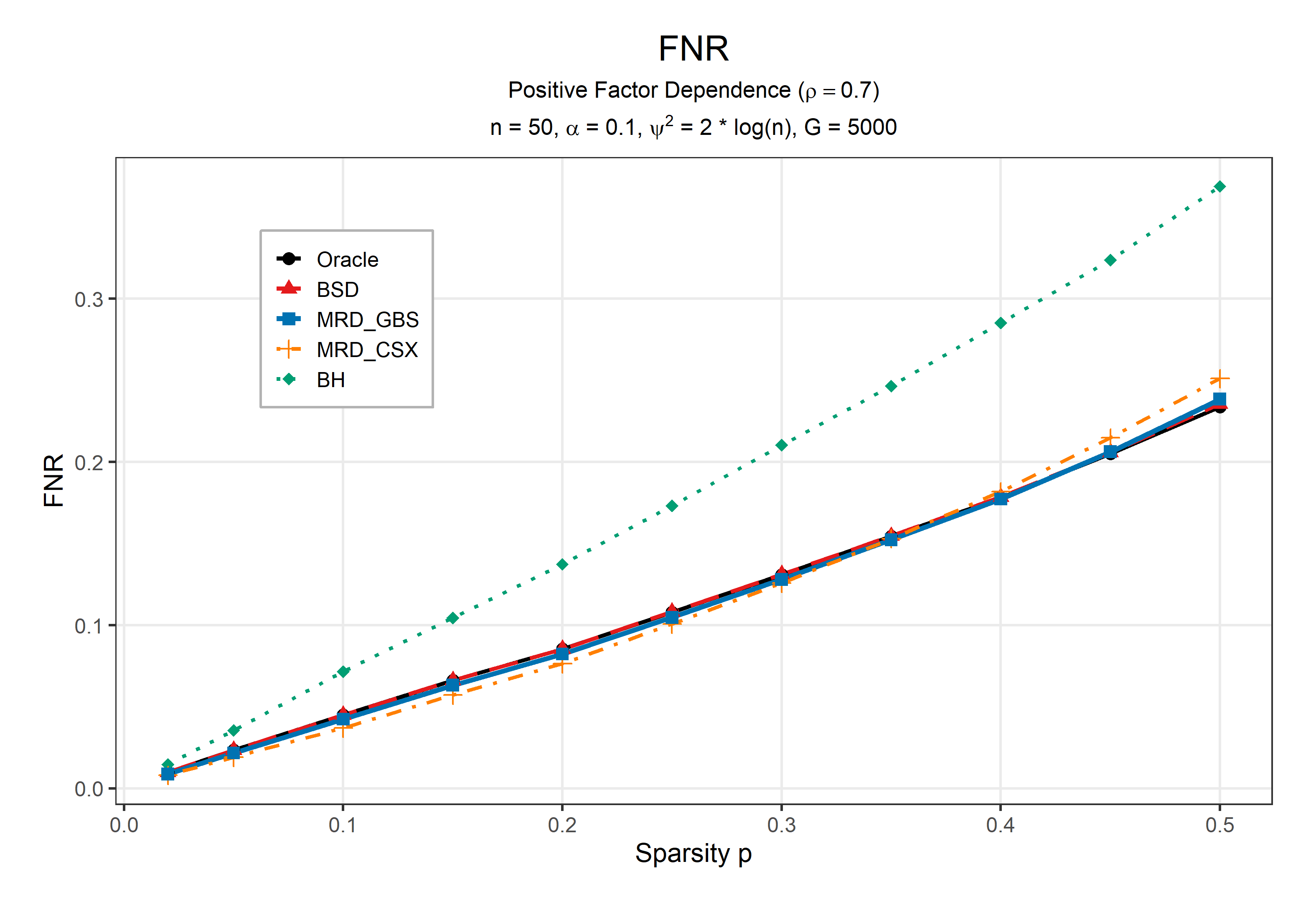}
		\caption{Positive Factor Dependence}
	\end{subfigure}
	
	\vspace{0.3cm}
	
	\begin{subfigure}{0.48\textwidth}
		\centering
		\includegraphics[width=\linewidth]{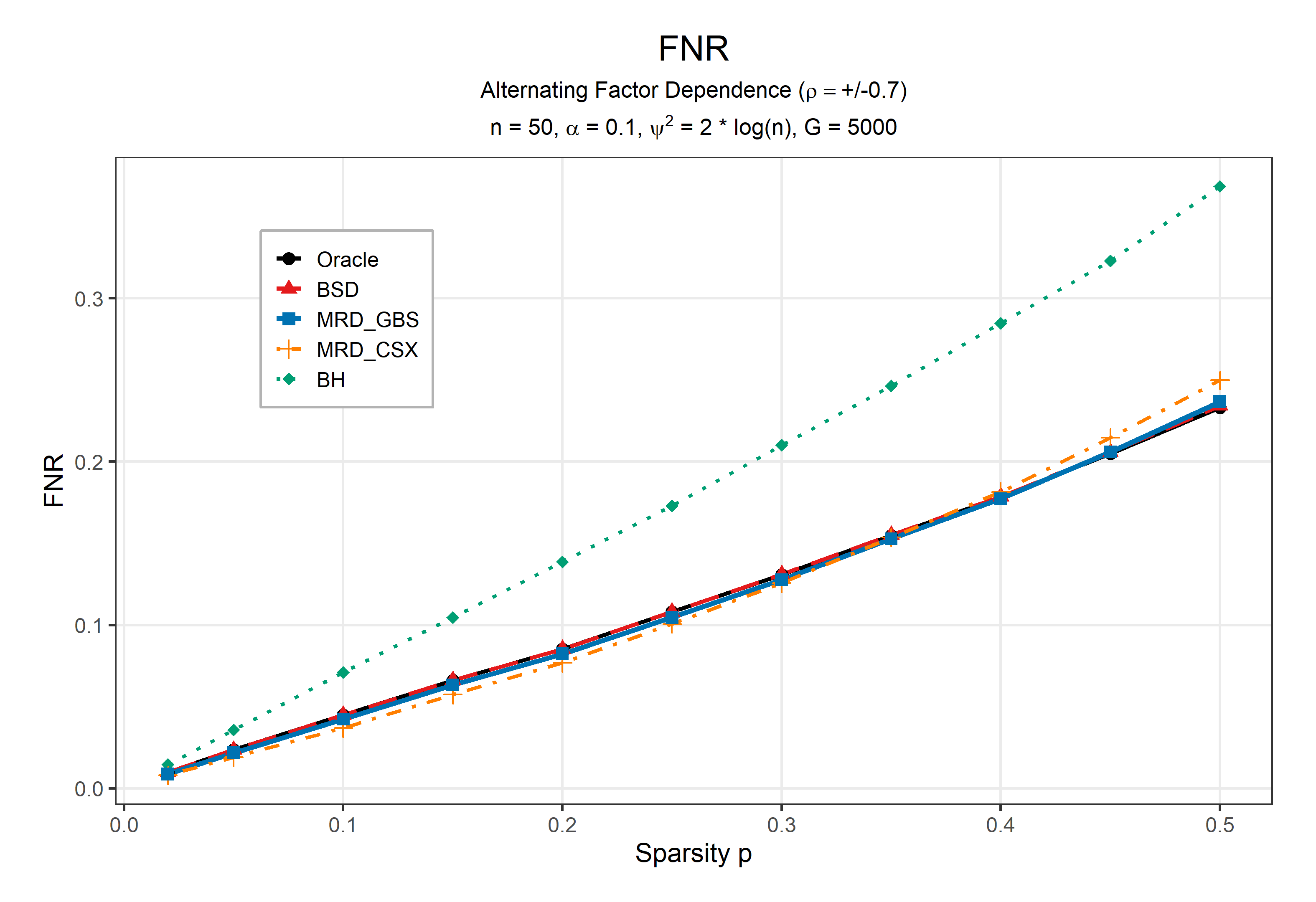}
		\caption{Alternating Factor Dependence}
	\end{subfigure}
	\hfill
	\begin{subfigure}{0.48\textwidth}
		\centering
		\includegraphics[width=\linewidth]{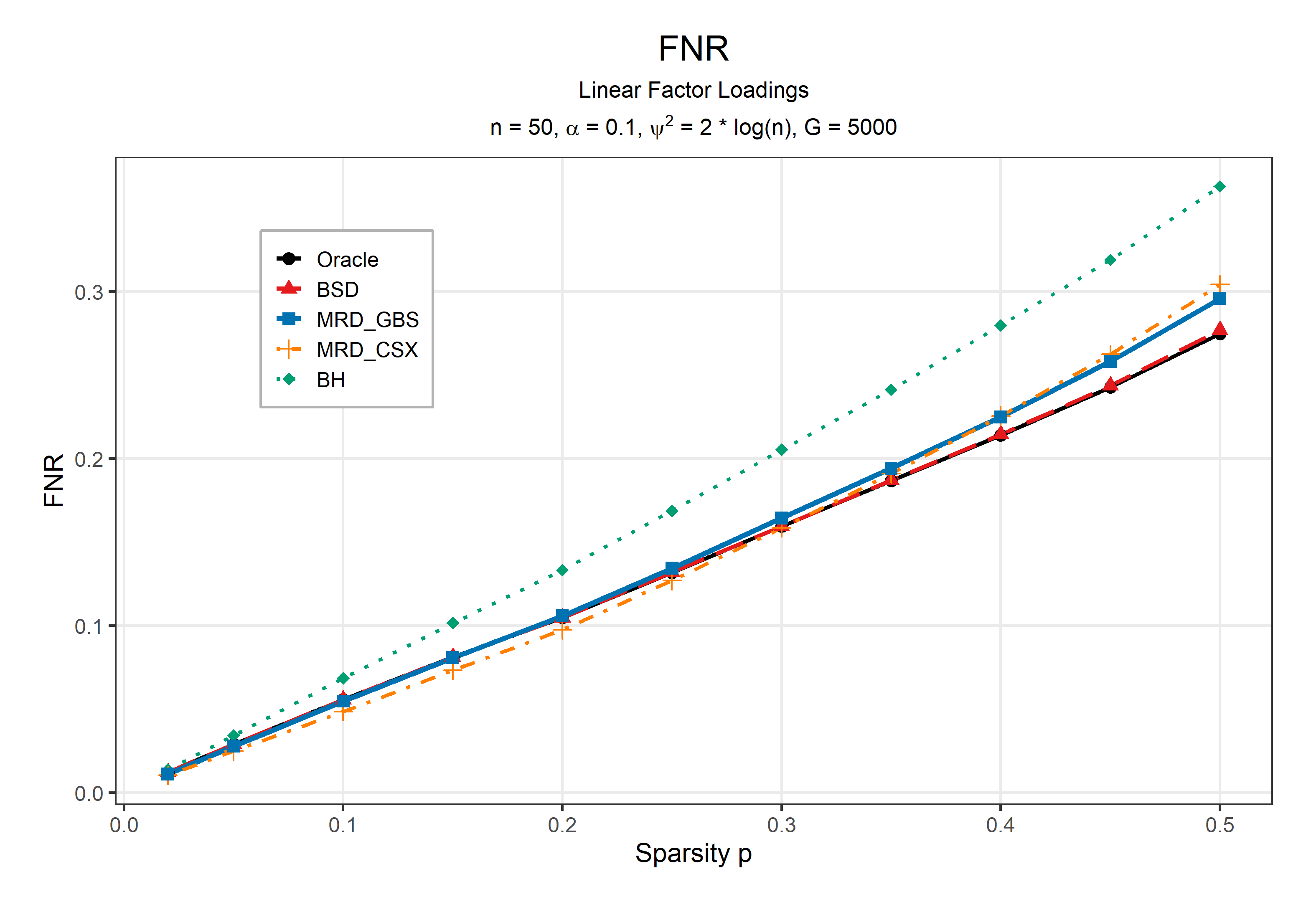}
		\caption{Linear Loadings}
	\end{subfigure}
	
	\vspace{0.3cm}
	
	\begin{subfigure}{0.48\textwidth}
		\centering
		\includegraphics[width=\linewidth]{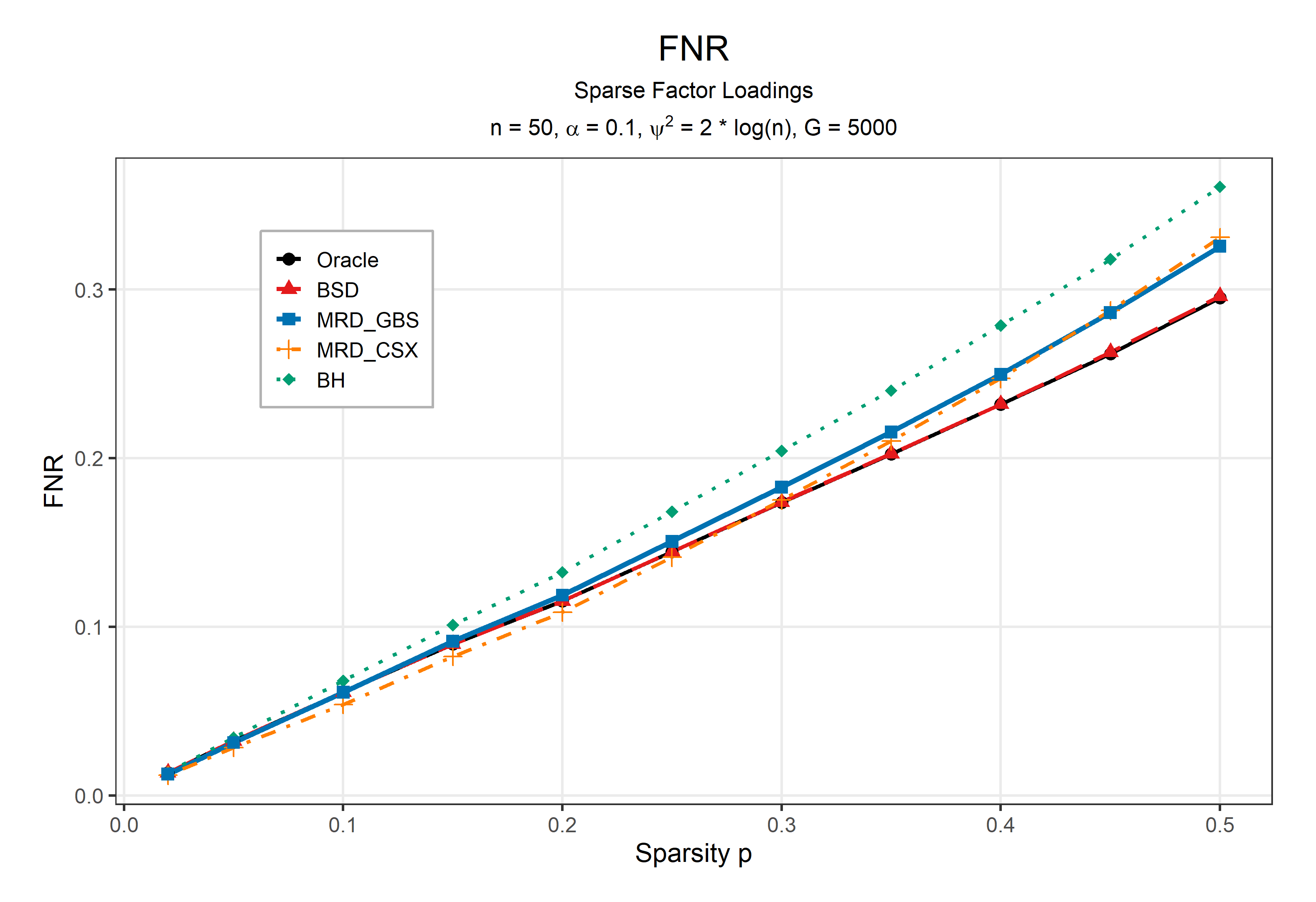}
		\caption{Sparse Factor Dependence}
	\end{subfigure}
	\hfill
	\begin{subfigure}{0.48\textwidth}
		\centering
		\includegraphics[width=\linewidth]{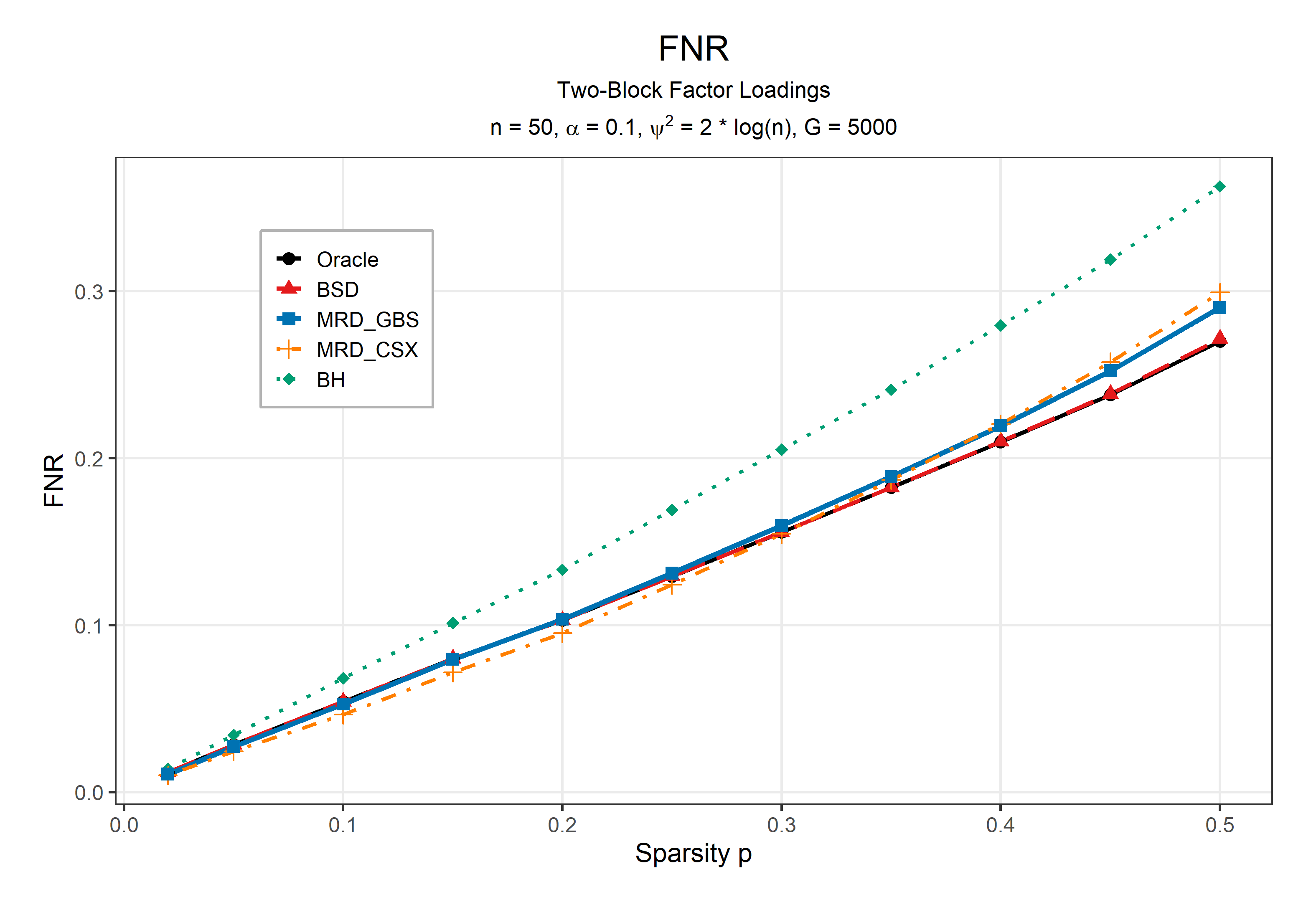}
		\caption{Two-Block Factor Dependence}
	\end{subfigure}
	
	\caption{Empirical false non-discovery rates under the six one-factor dependence structures when $n=50$, \(\alpha=0.1\), \(\psi^2=2\log n\), and \(G=5000\) Monte Carlo replications.}
	\label{fig:fnr-n50}
\end{figure}

\begin{figure}[p]
	\centering
	
	\begin{subfigure}{0.48\textwidth}
		\centering
		\includegraphics[width=\linewidth]{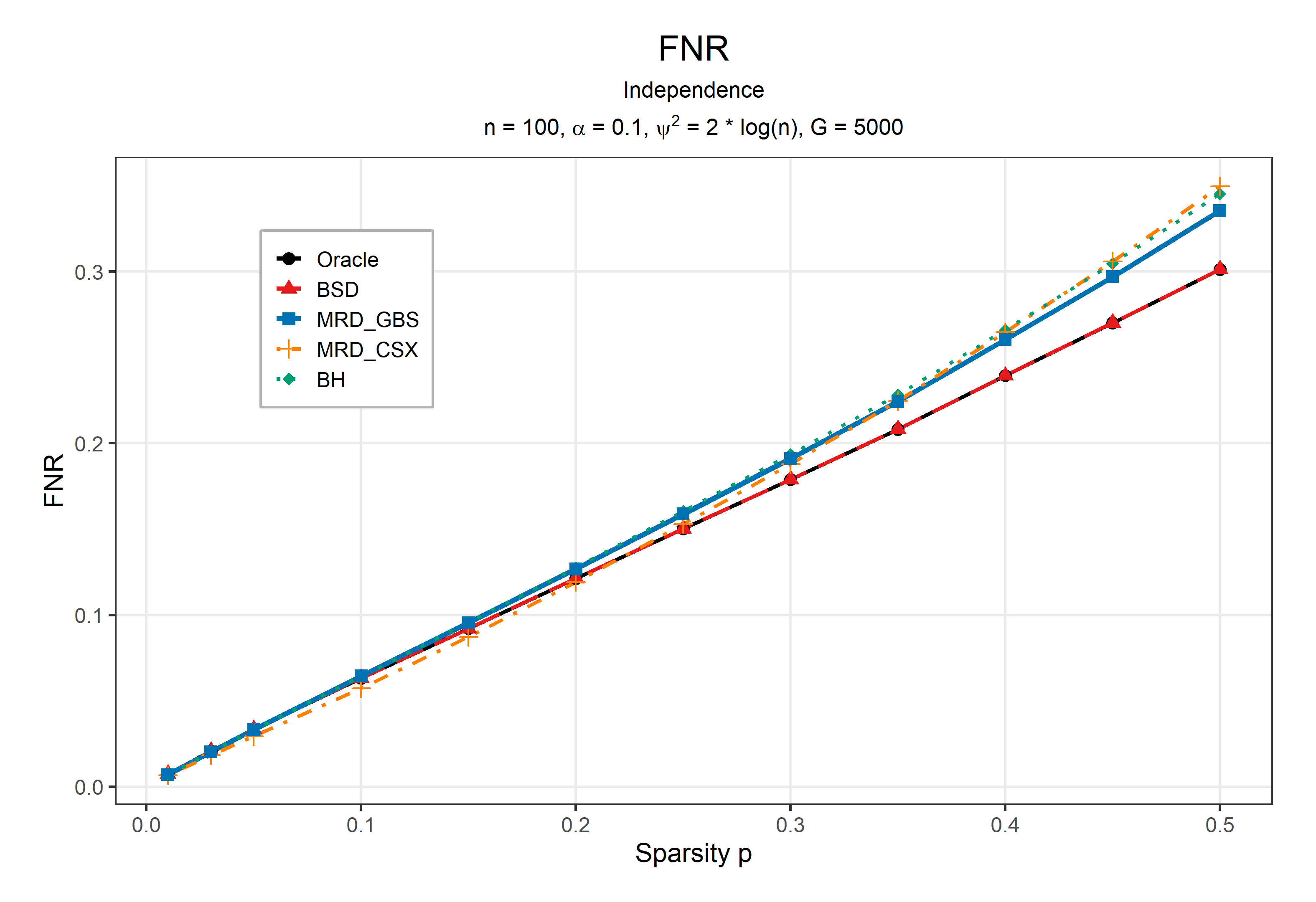}
		\caption{Independence}
	\end{subfigure}
	\hfill
	\begin{subfigure}{0.48\textwidth}
		\centering
		\includegraphics[width=\linewidth]{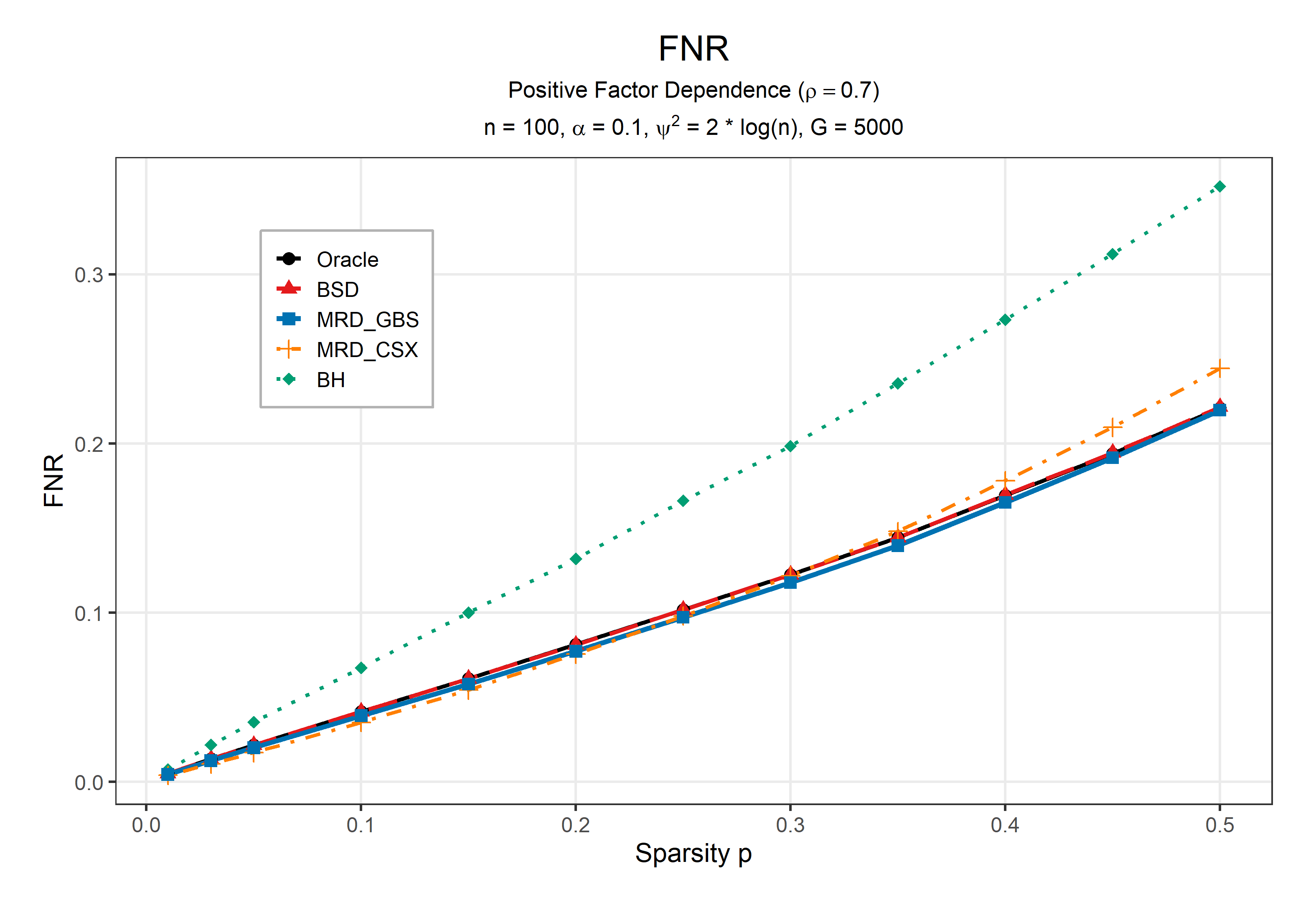}
		\caption{Positive Factor Dependence}
	\end{subfigure}
	
	\vspace{0.3cm}
	
	\begin{subfigure}{0.48\textwidth}
		\centering
		\includegraphics[width=\linewidth]{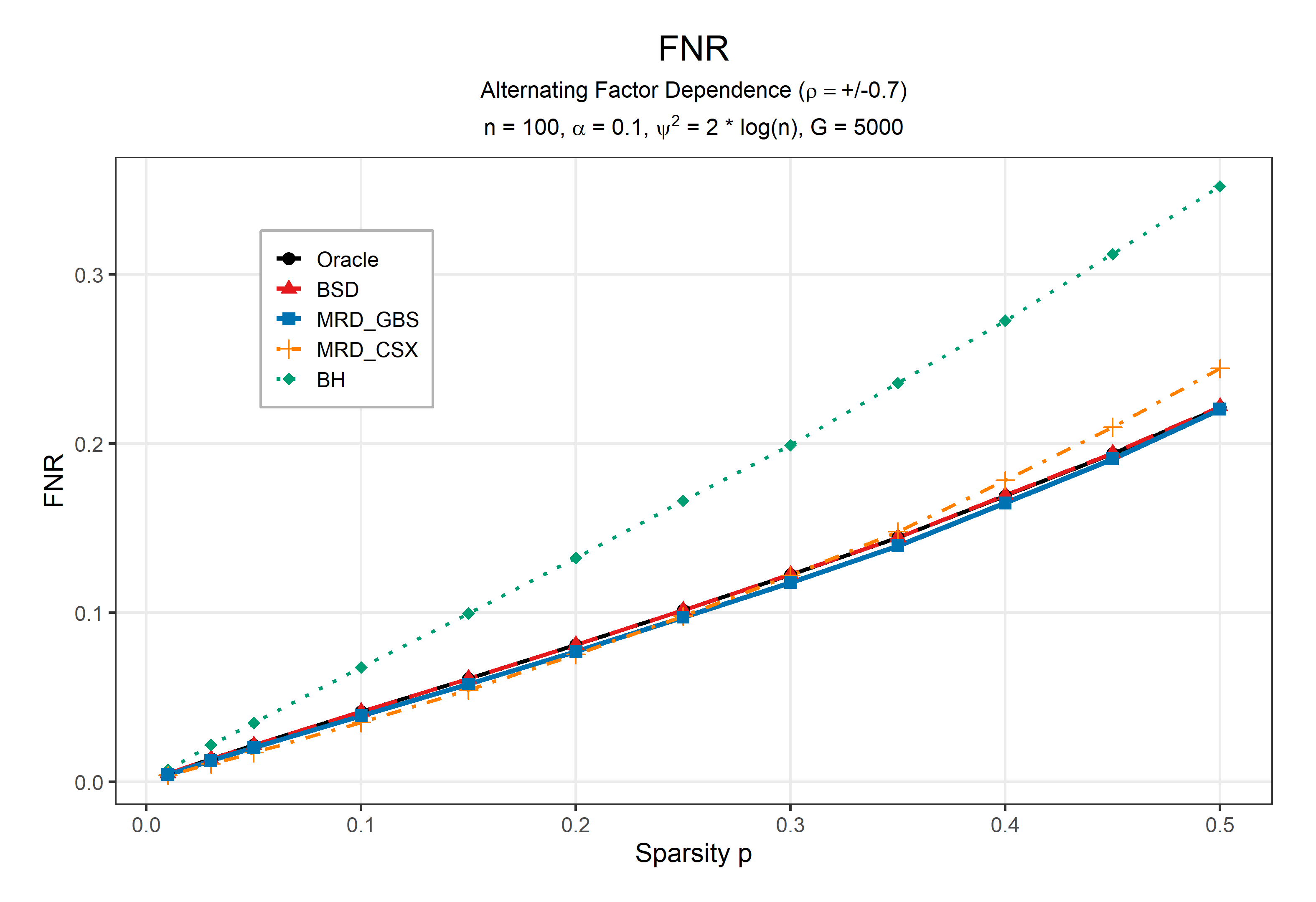}
		\caption{Alternating Factor Dependence}
	\end{subfigure}
	\hfill
	\begin{subfigure}{0.48\textwidth}
		\centering
		\includegraphics[width=\linewidth]{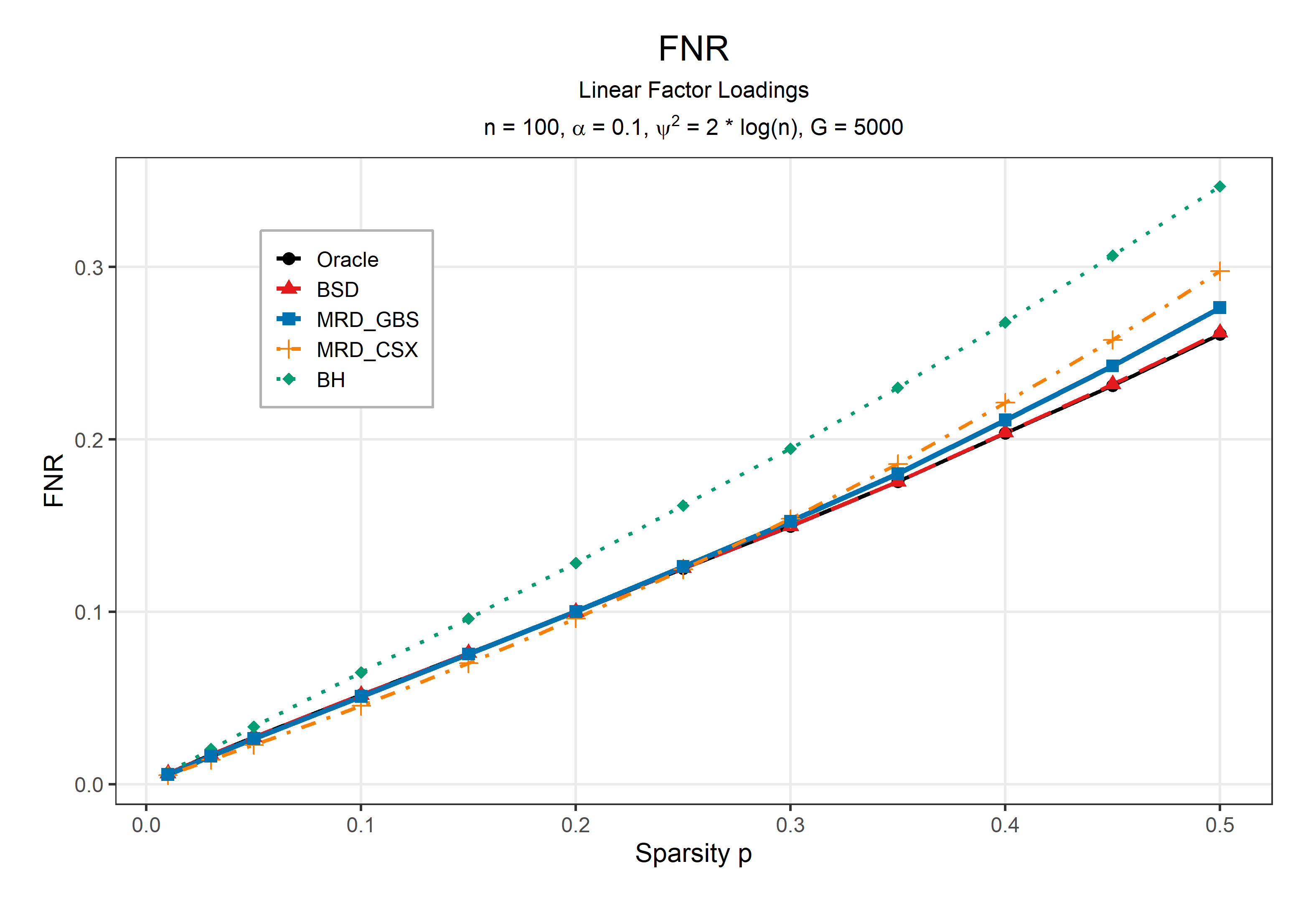}
		\caption{Linear Loadings}
	\end{subfigure}
	
	\vspace{0.3cm}
	
	\begin{subfigure}{0.48\textwidth}
		\centering
		\includegraphics[width=\linewidth]{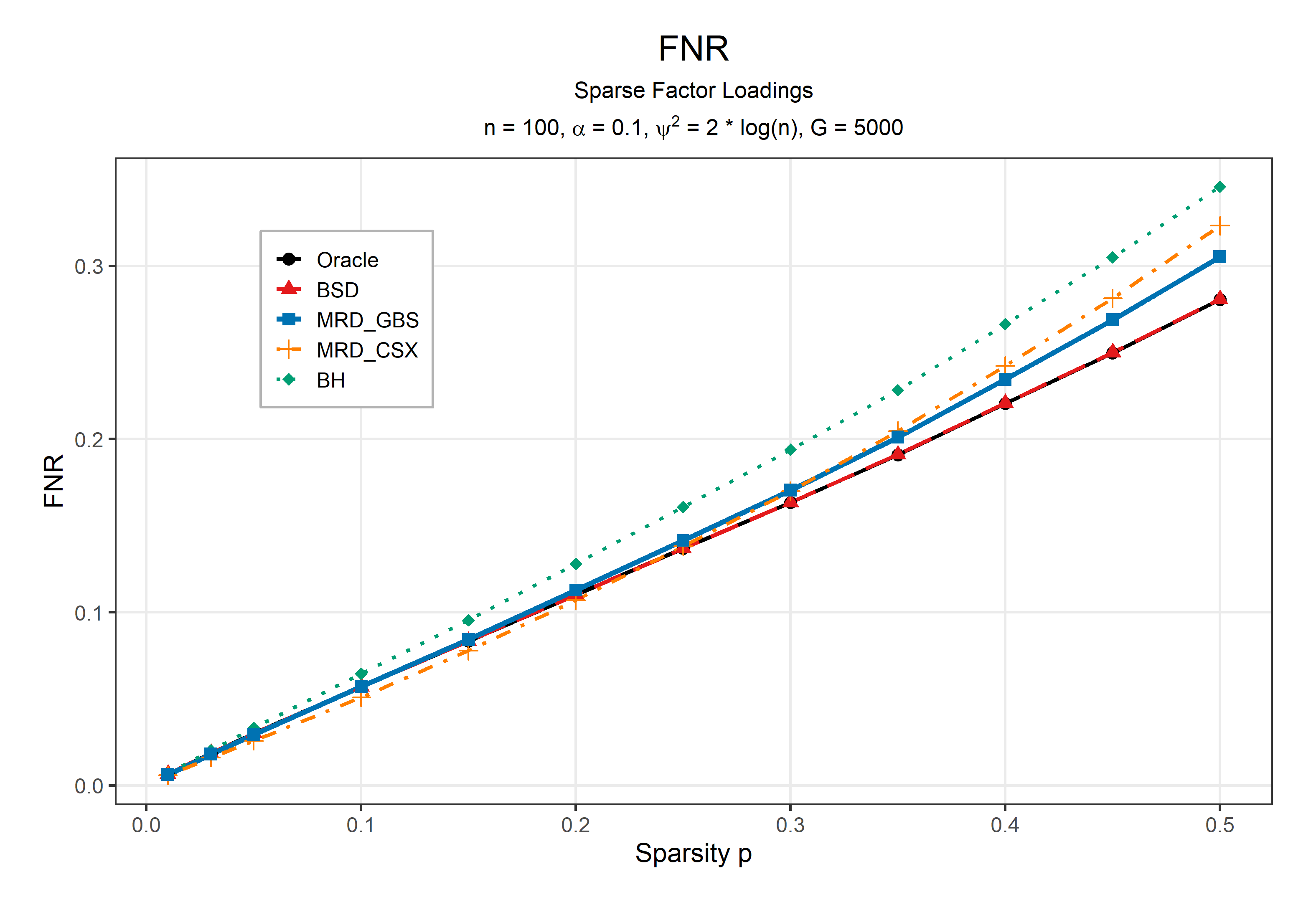}
		\caption{Sparse Factor Dependence}
	\end{subfigure}
	\hfill
	\begin{subfigure}{0.48\textwidth}
		\centering
		\includegraphics[width=\linewidth]{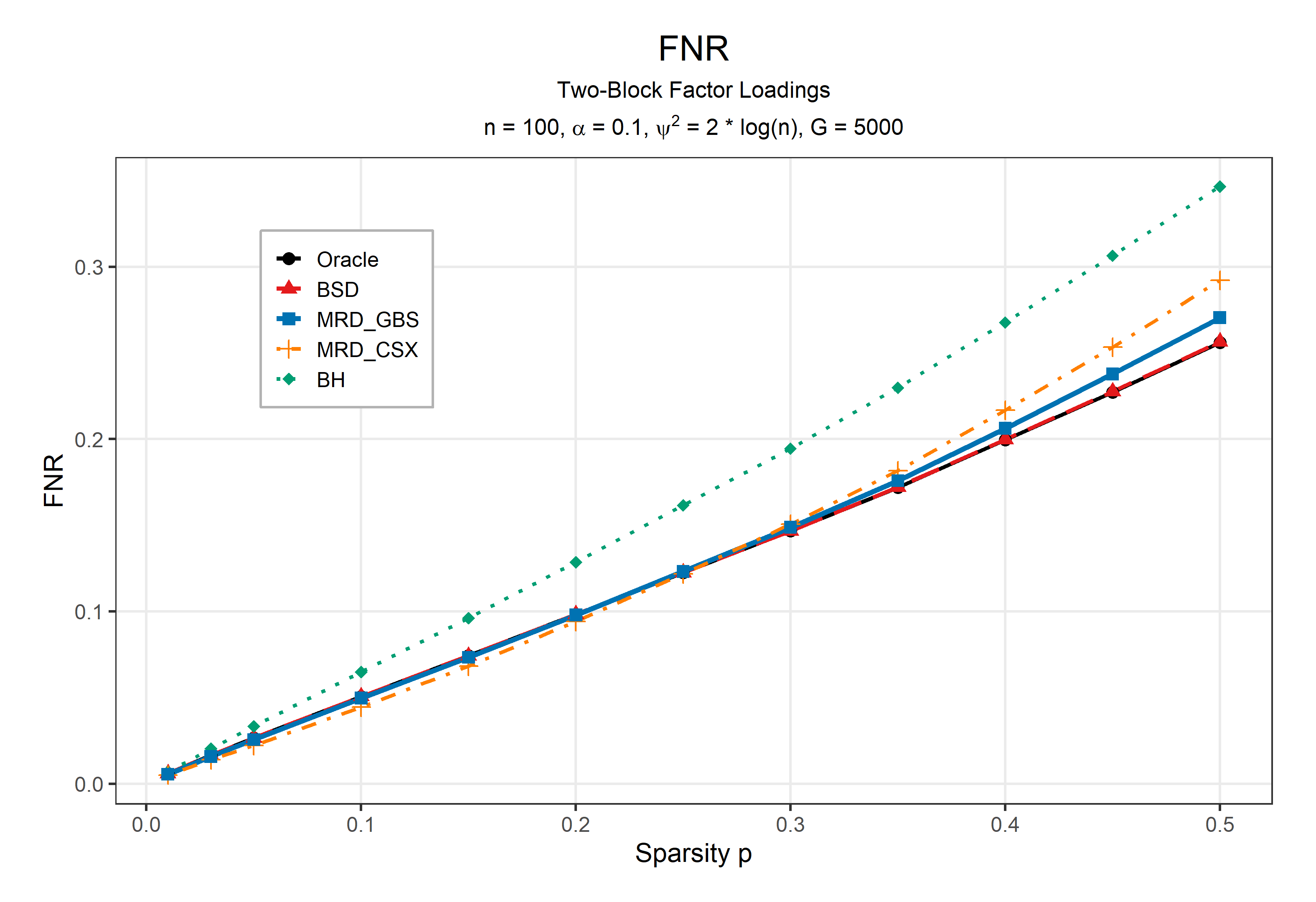}
		\caption{Two-Block Factor Dependence}
	\end{subfigure}
	
	\caption{Empirical false non-discovery rates under the six one-factor dependence structures when $n=100$, \(\alpha=0.1\), \(\psi^2=2\log n\), and \(G=5000\) Monte Carlo replications.}
	\label{fig:fnr-n100}
\end{figure}

A striking feature of the results is the close agreement between BSD and the Bayes Oracle across all sample sizes, sparsity levels, and dependence structures considered. The FNR curves of BSD are nearly indistinguishable from those of the Oracle throughout the entire range of configurations examined. This agreement complements the Bayes risk results reported earlier and provides additional evidence that the posterior model-pursuit mechanism underlying BSD successfully reproduces the signal-selection behavior of the Bayes Oracle even under substantial covariance dependence. In particular, BSD appears to recover nearly the same balance between false discoveries and missed discoveries that is induced by Bayes-risk minimization.

The behavior of MRD--GBS is also noteworthy. Across most configurations, the procedure attains FNR values that remain remarkably close to those of BSD and the Oracle despite being constructed from an entirely different frequentist framework. This finding is particularly significant because the GBS calibration was originally motivated by FDR control rather than Bayes-risk optimization. The results therefore suggest that the adaptive covariance-aware search mechanism employed by MRD--GBS is capable of recovering a large proportion of the signals while simultaneously maintaining stable control of false discoveries.

The performance of MRD--CSX is less competitive. Although the procedure occasionally attains FNR values comparable to those of BSD and MRD--GBS in certain moderately dense regimes, its FNRs are generally larger and are frequently accompanied by substantially elevated FDR values, as observed in the previous subsection. The BH procedure consistently exhibits the largest FNRs among the competing methods, reflecting its conservative behavior under the dependence structures considered here. These observations indicate that the lower Bayes risks achieved by BSD and MRD--GBS arise not merely from increased numbers of discoveries, but from a more favorable allocation of Type~I and Type~II errors.

Overall, the FNR results reinforce the conclusions of both the Bayes risk and FDR analyses. BSD consistently reproduces the signal-recovery performance of the Bayes Oracle across a broad collection of dependence structures, while MRD--GBS frequently attains comparable FNR values despite being derived from a fundamentally different methodology. Taken together, these findings provide further evidence that both procedures are highly effective at identifying sparse signals under dependence, with BSD exhibiting near-Oracle performance throughout the simulation settings considered.

\subsubsection{Power}

The empirical powers of the competing procedures are displayed in Figures \ref{fig:power-n20}--\ref{fig:power-n100}.
Since power measures the proportion of active signals that are correctly identified,
it provides a direct assessment of the signal-detection capability of the competing
methods. Together with the FDR and FNR analyses presented earlier, the power
results offer a complementary perspective on the sparse signal recovery behavior
of the procedures under consideration.

\begin{figure}[p]
	\centering
	
	\begin{subfigure}{0.48\textwidth}
		\centering
		\includegraphics[width=\linewidth]{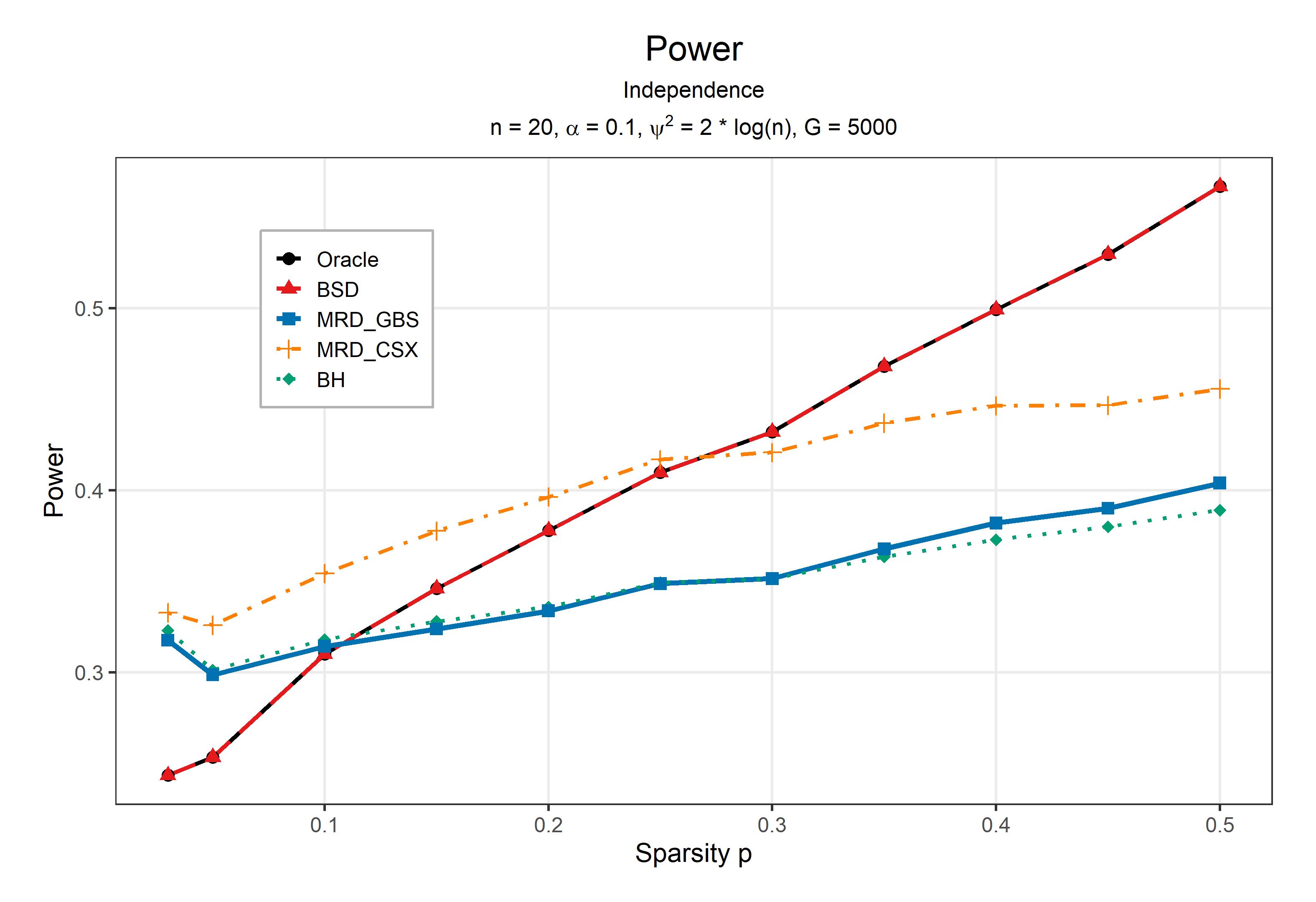}
		\caption{Independence}
	\end{subfigure}
	\hfill
	\begin{subfigure}{0.48\textwidth}
		\centering
		\includegraphics[width=\linewidth]{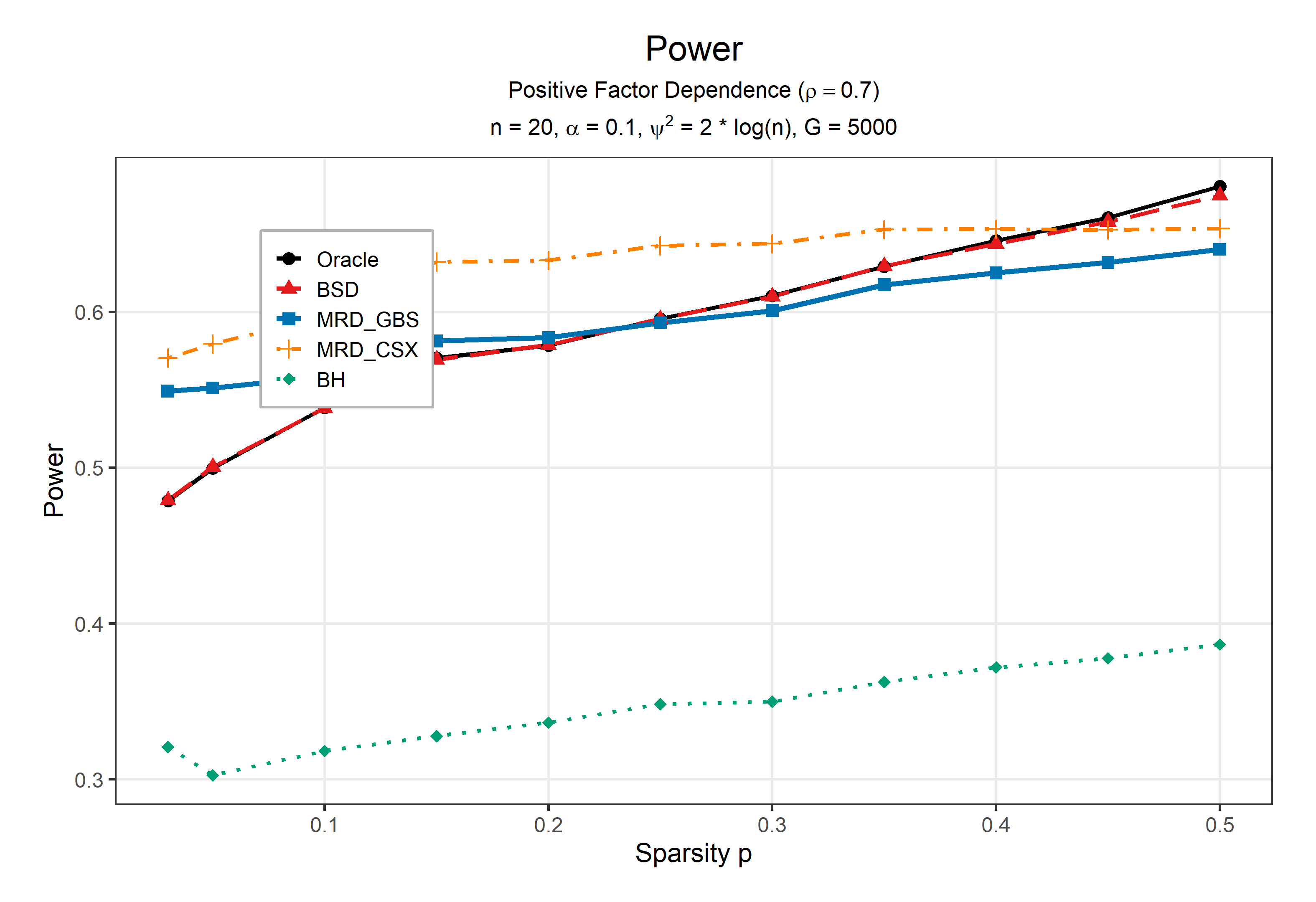}
		\caption{Positive Factor Dependence}
	\end{subfigure}
	
	\vspace{0.3cm}
	
	\begin{subfigure}{0.48\textwidth}
		\centering
		\includegraphics[width=\linewidth]{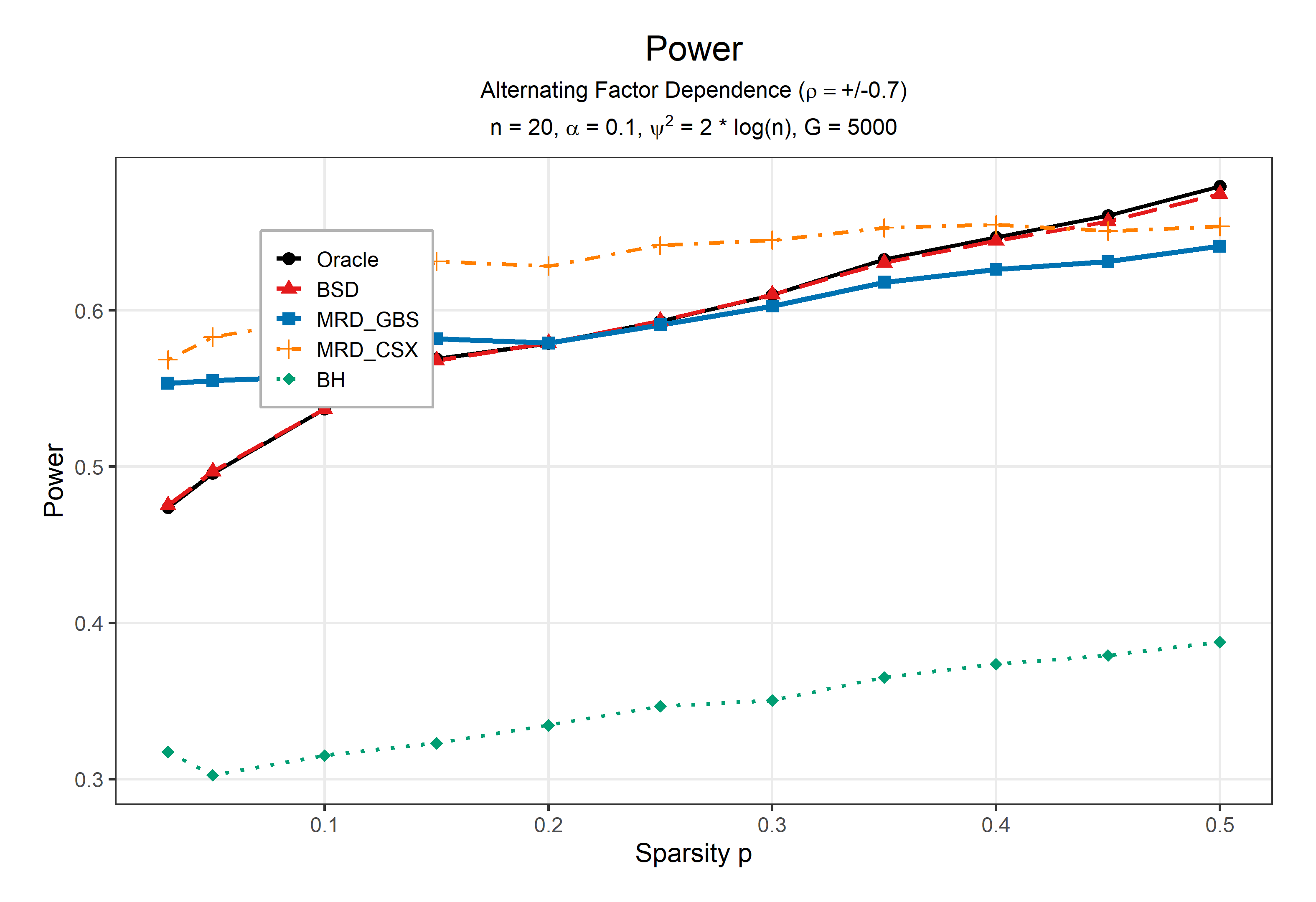}
		\caption{Alternating Factor Dependence}
	\end{subfigure}
	\hfill
	\begin{subfigure}{0.48\textwidth}
		\centering
		\includegraphics[width=\linewidth]{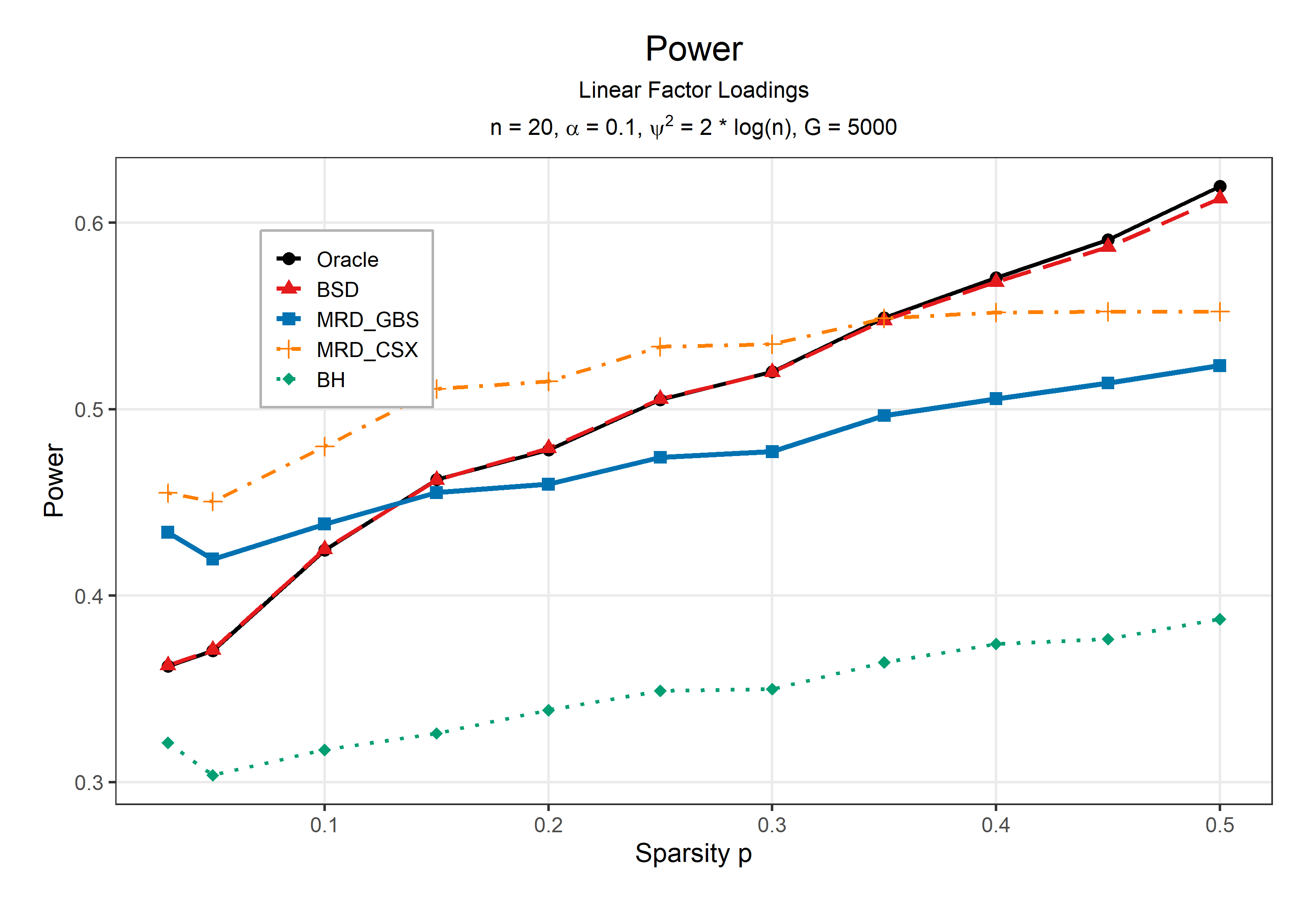}
		\caption{Linear Loadings}
	\end{subfigure}
	
	\vspace{0.3cm}
	
	\begin{subfigure}{0.48\textwidth}
		\centering
		\includegraphics[width=\linewidth]{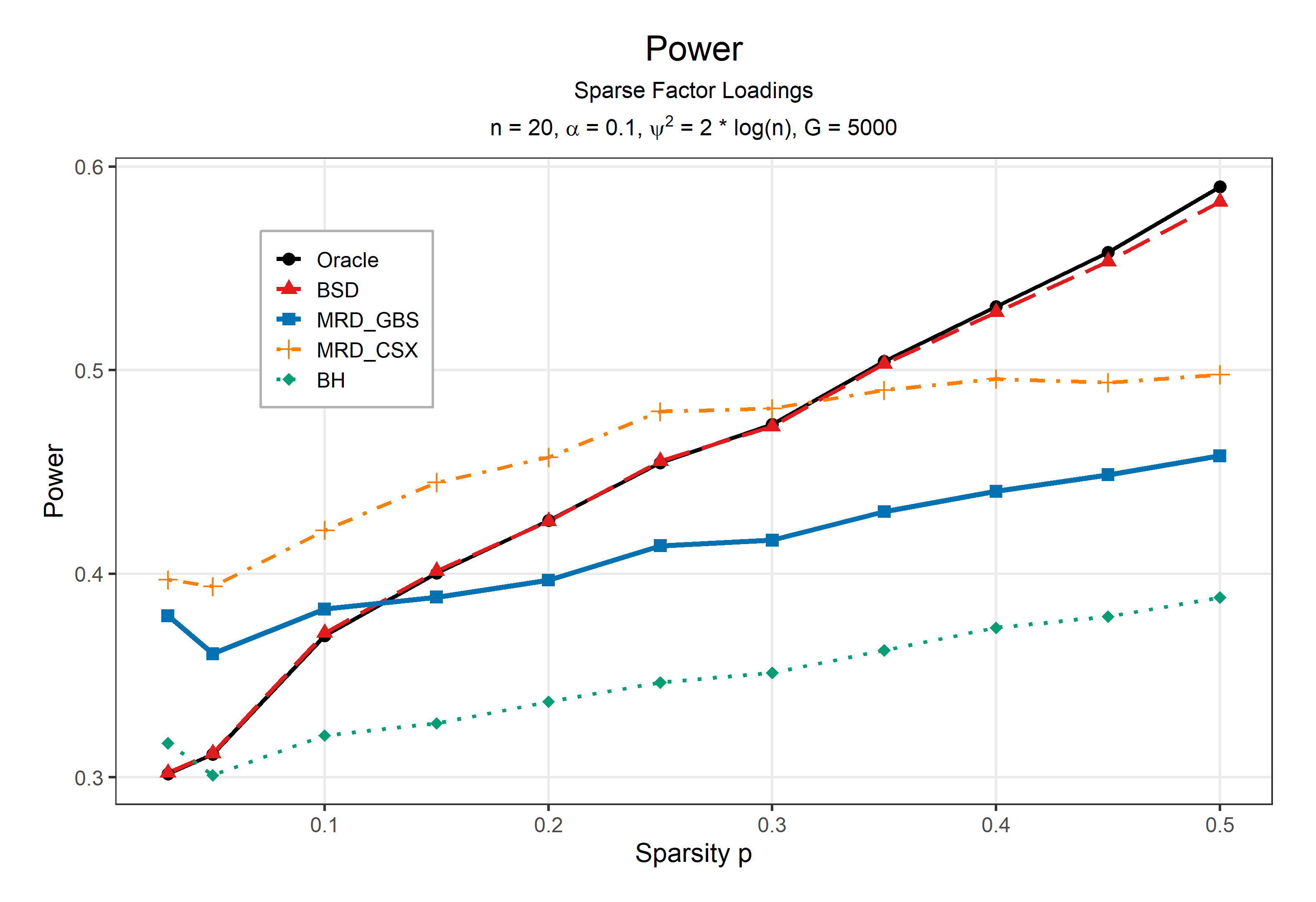}
		\caption{Sparse Factor Dependence}
	\end{subfigure}
	\hfill
	\begin{subfigure}{0.48\textwidth}
		\centering
		\includegraphics[width=\linewidth]{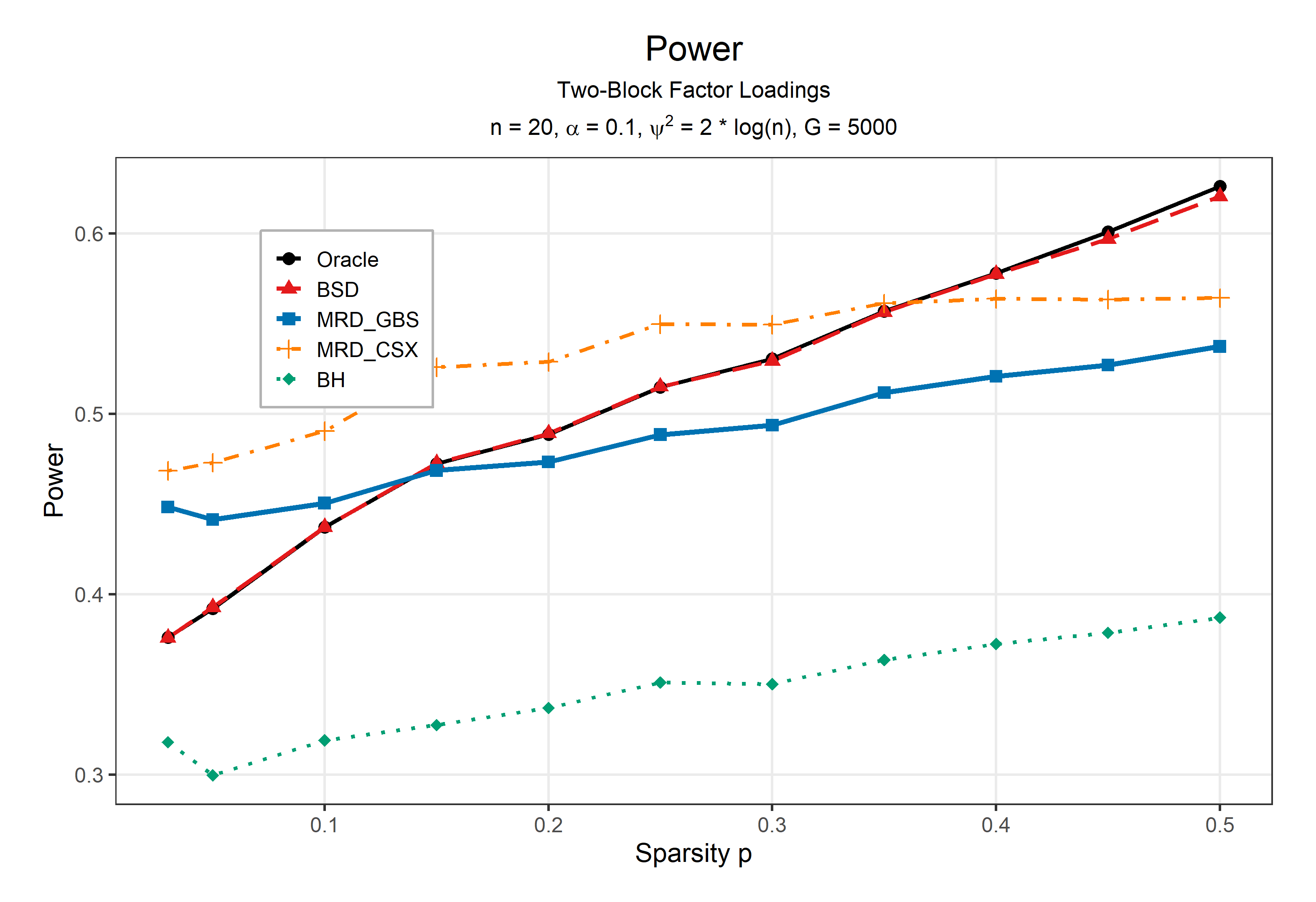}
		\caption{Two-Block Factor Dependence}
	\end{subfigure}
	
	\caption{Empirical powers under the six one-factor dependence structures when $n=20$, \(\alpha=0.1\), \(\psi^2=2\log n\), and \(G=5000\) Monte Carlo replications.}
	\label{fig:power-n20}
\end{figure}

\begin{figure}[p]
	\centering
	
	\begin{subfigure}{0.48\textwidth}
		\centering
		\includegraphics[width=\linewidth]{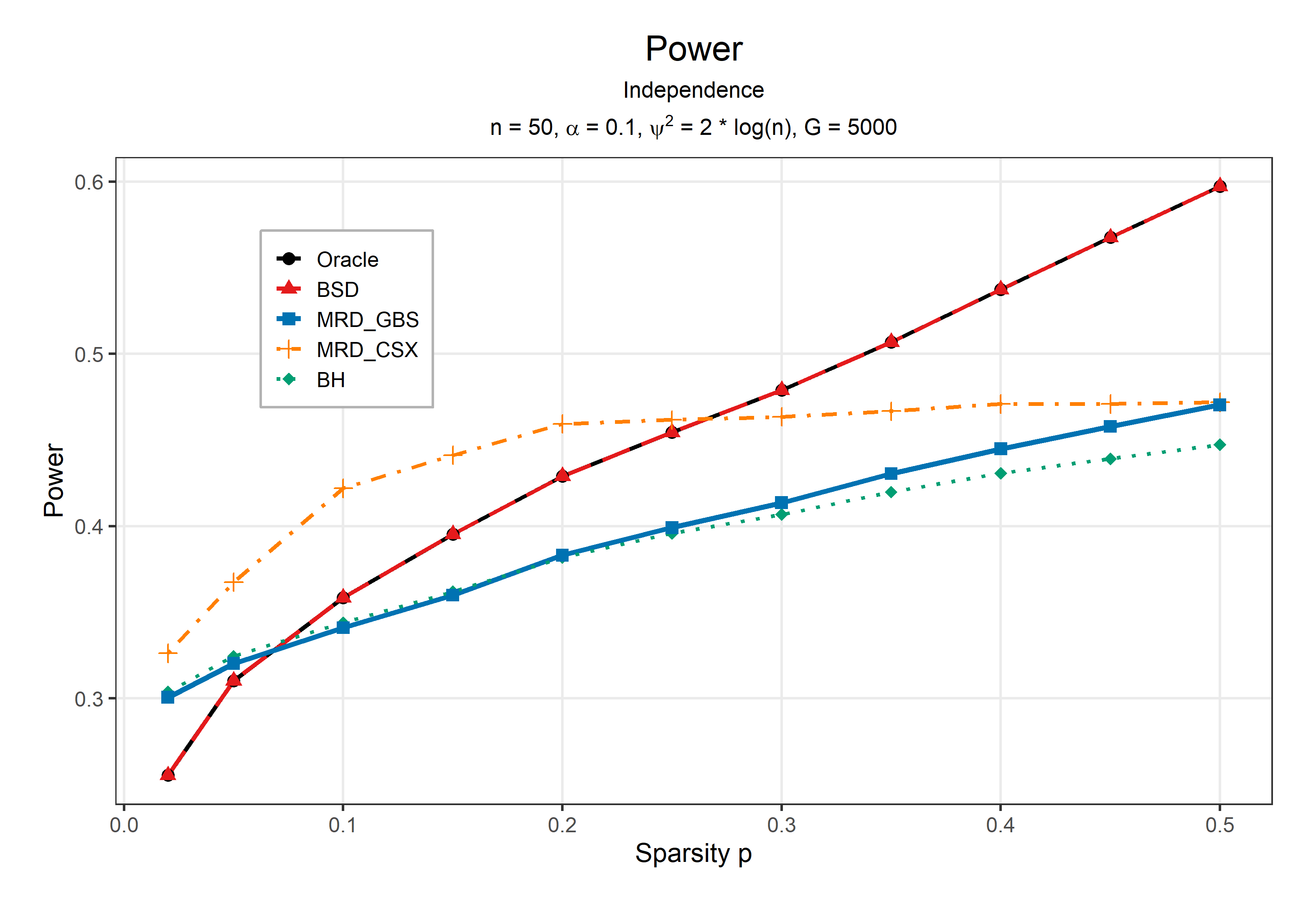}
		\caption{Independence}
	\end{subfigure}
	\hfill
	\begin{subfigure}{0.48\textwidth}
		\centering
		\includegraphics[width=\linewidth]{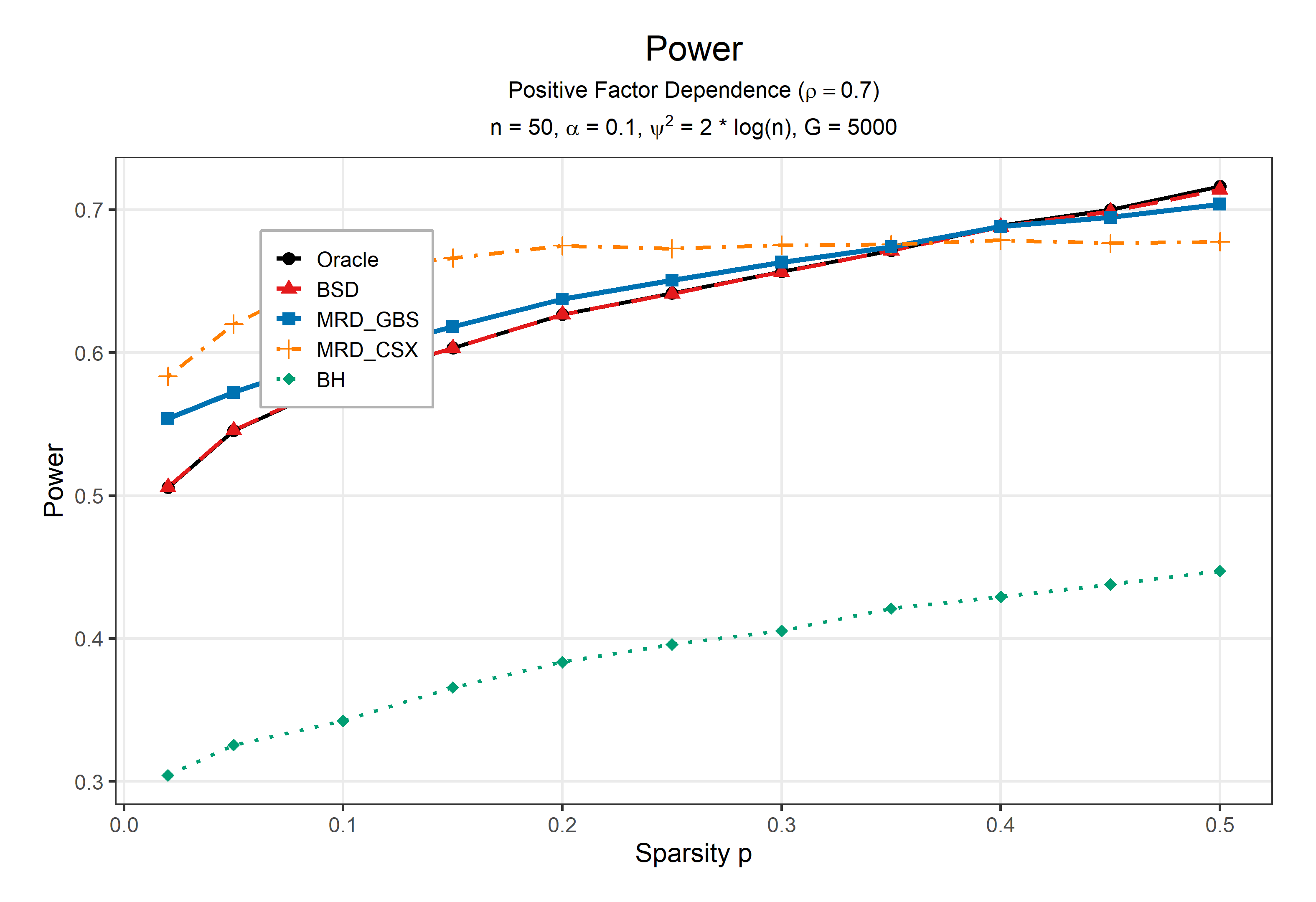}
		\caption{Positive Factor Dependence}
	\end{subfigure}
	
	\vspace{0.3cm}
	
	\begin{subfigure}{0.48\textwidth}
		\centering
		\includegraphics[width=\linewidth]{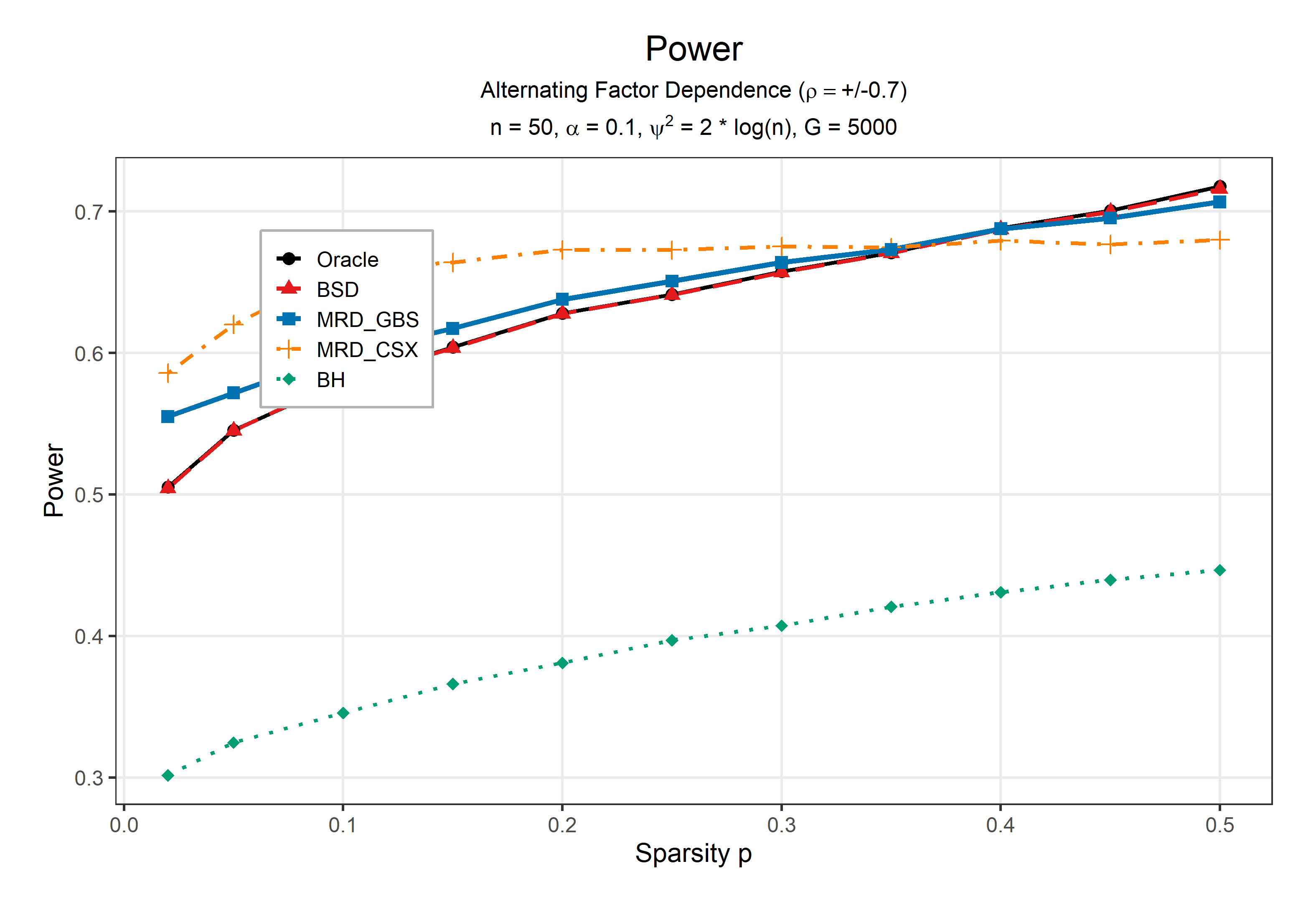}
		\caption{Alternating Factor Dependence}
	\end{subfigure}
	\hfill
	\begin{subfigure}{0.48\textwidth}
		\centering
		\includegraphics[width=\linewidth]{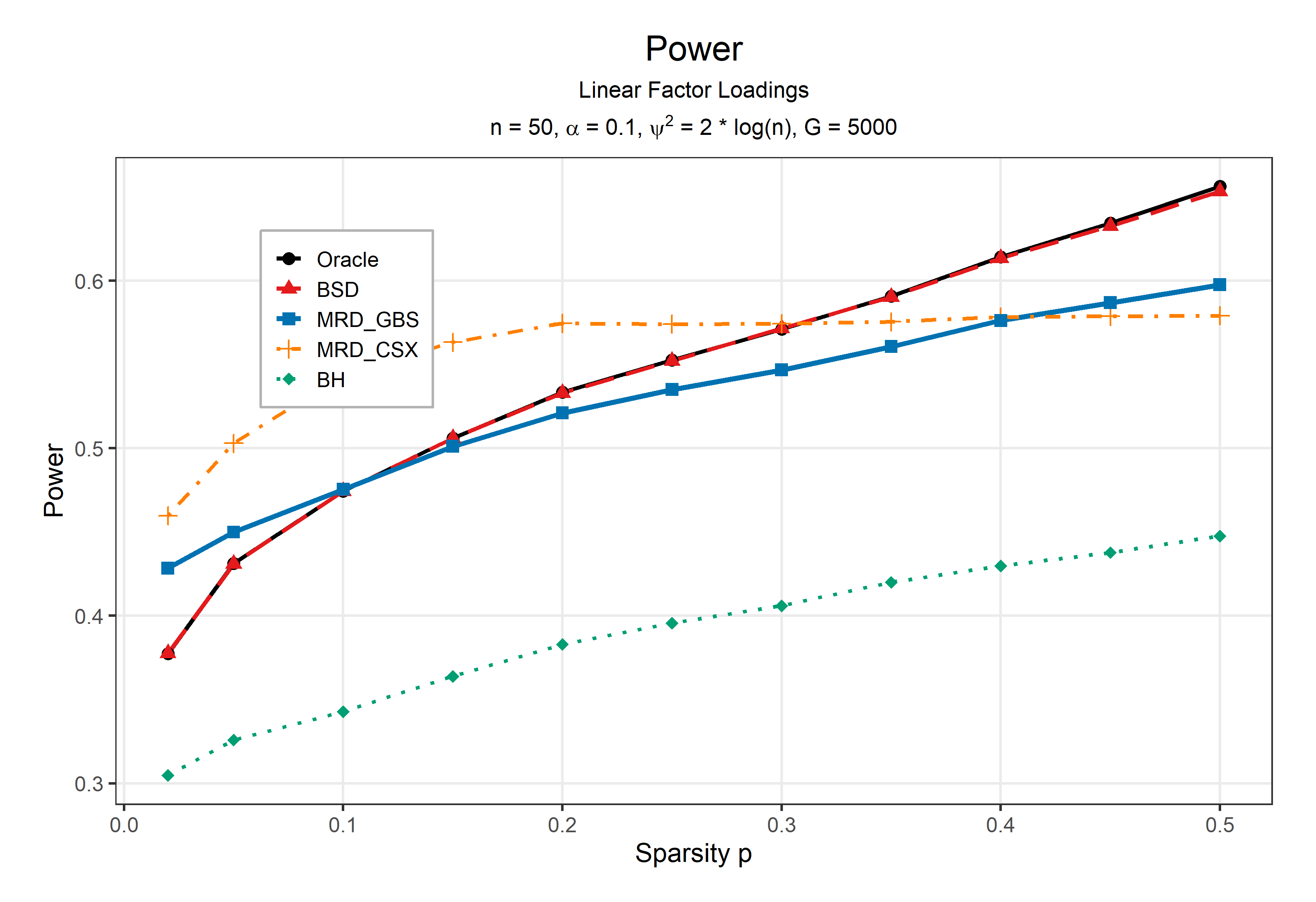}
		\caption{Linear Loadings}
	\end{subfigure}
	
	\vspace{0.3cm}
	
	\begin{subfigure}{0.48\textwidth}
		\centering
		\includegraphics[width=\linewidth]{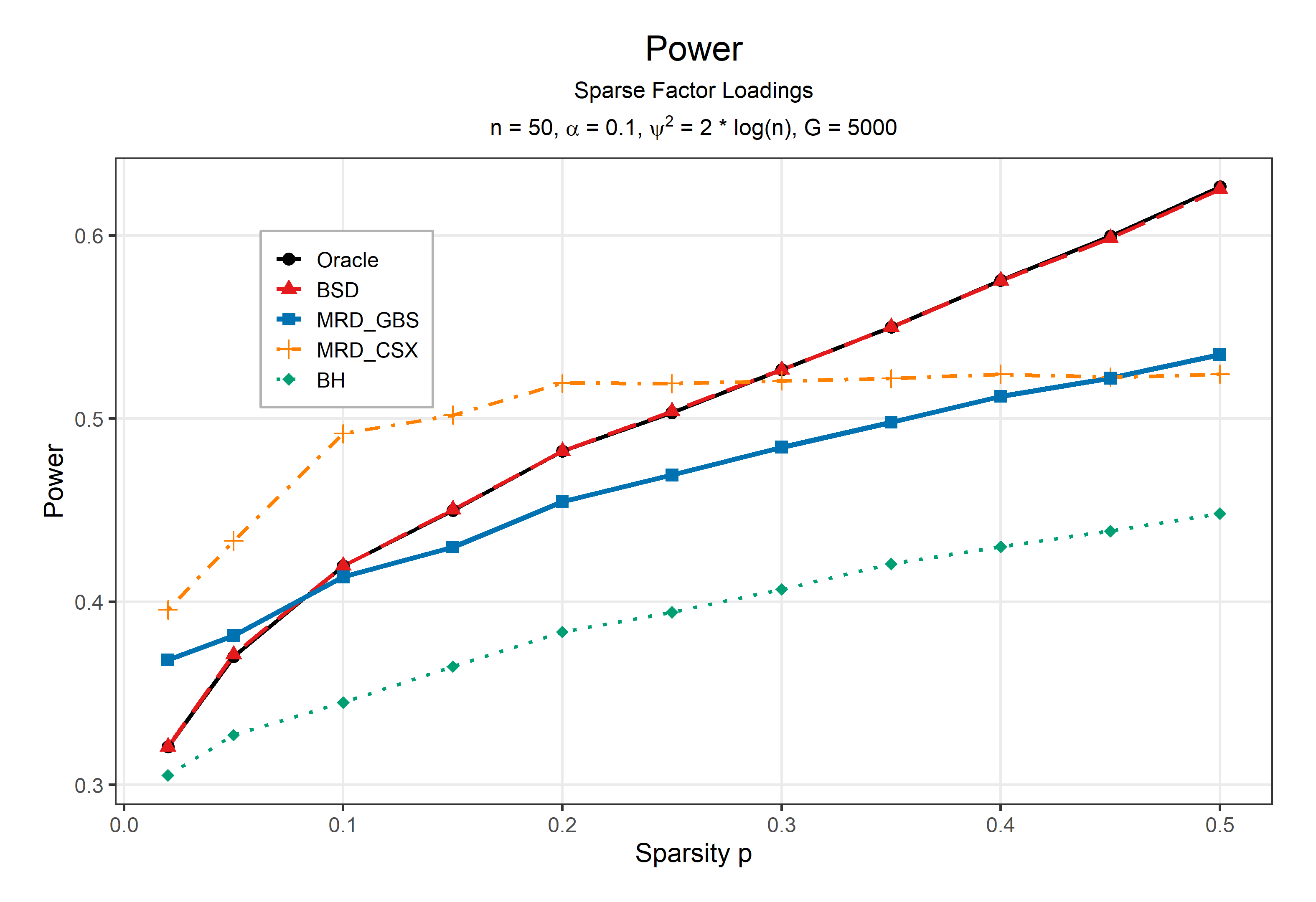}
		\caption{Sparse Factor Dependence}
	\end{subfigure}
	\hfill
	\begin{subfigure}{0.48\textwidth}
		\centering
		\includegraphics[width=\linewidth]{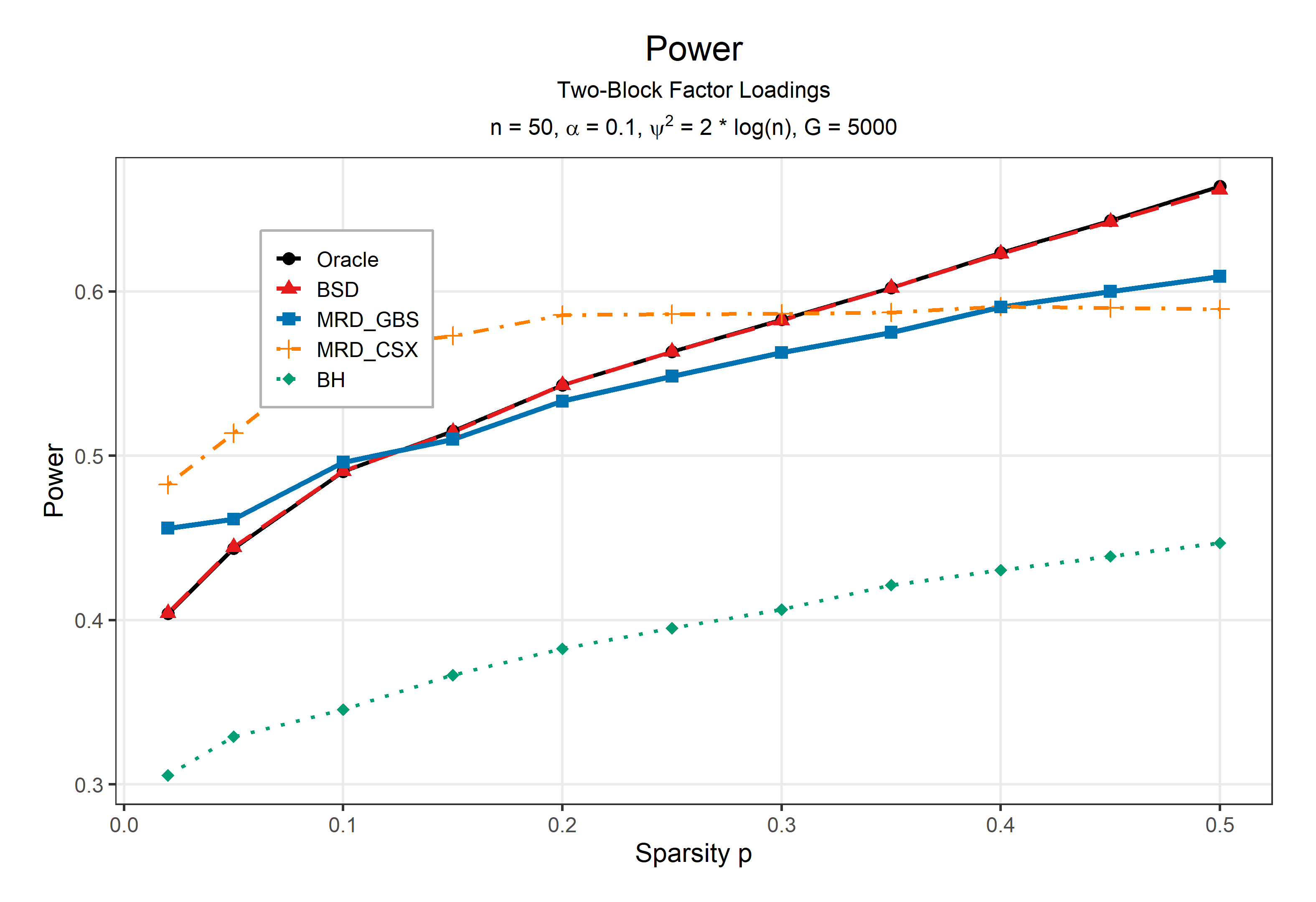}
		\caption{Two-Block Factor Dependence}
	\end{subfigure}
	
	\caption{Empirical powers under the six one-factor dependence structures when $n=50$, \(\alpha=0.1\), \(\psi^2=2\log n\), and \(G=5000\) Monte Carlo replications.}
	\label{fig:power-n50}
\end{figure}

\begin{figure}[p]
	\centering
	
	\begin{subfigure}{0.48\textwidth}
		\centering
		\includegraphics[width=\linewidth]{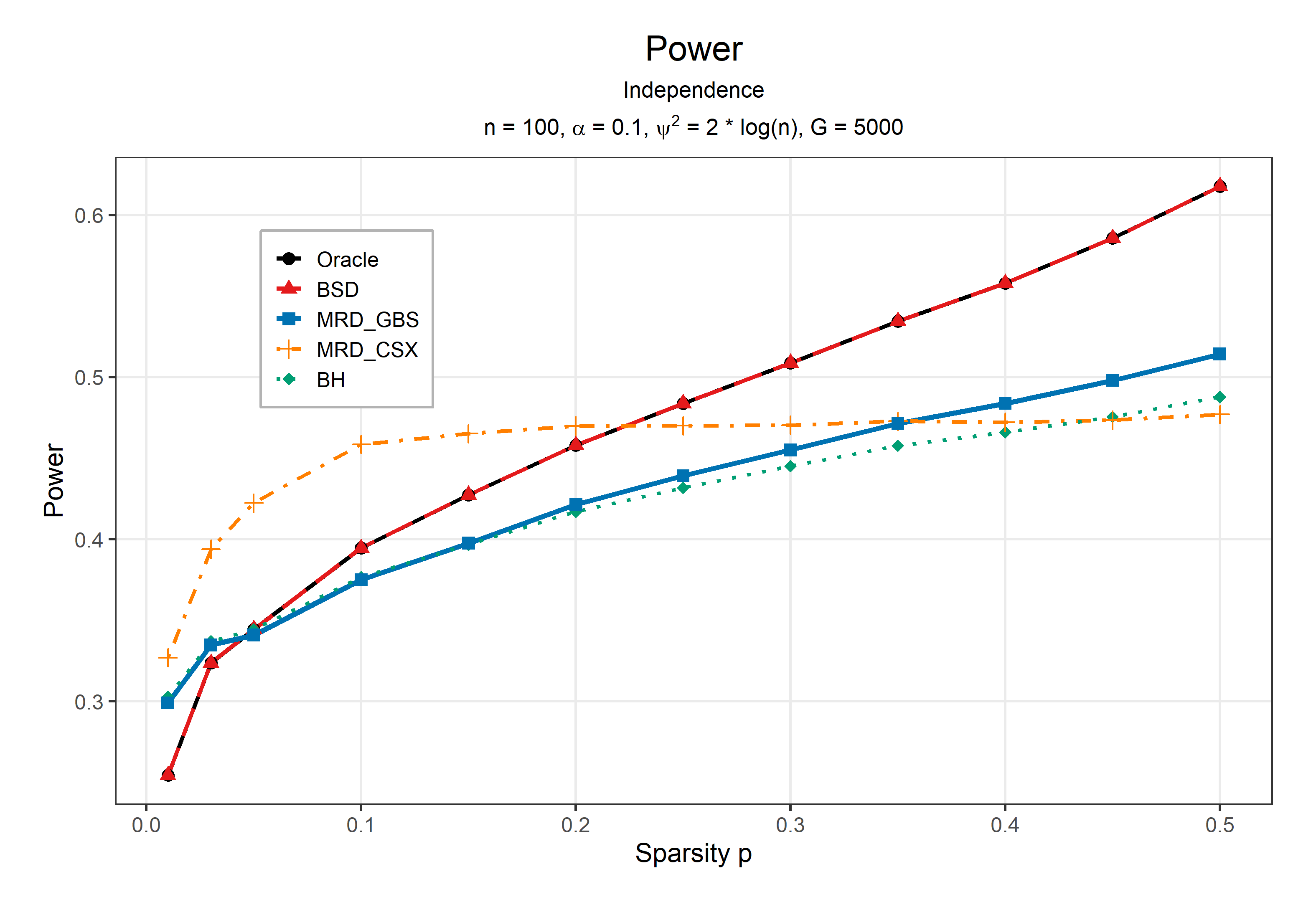}
		\caption{Independence}
	\end{subfigure}
	\hfill
	\begin{subfigure}{0.48\textwidth}
		\centering
		\includegraphics[width=\linewidth]{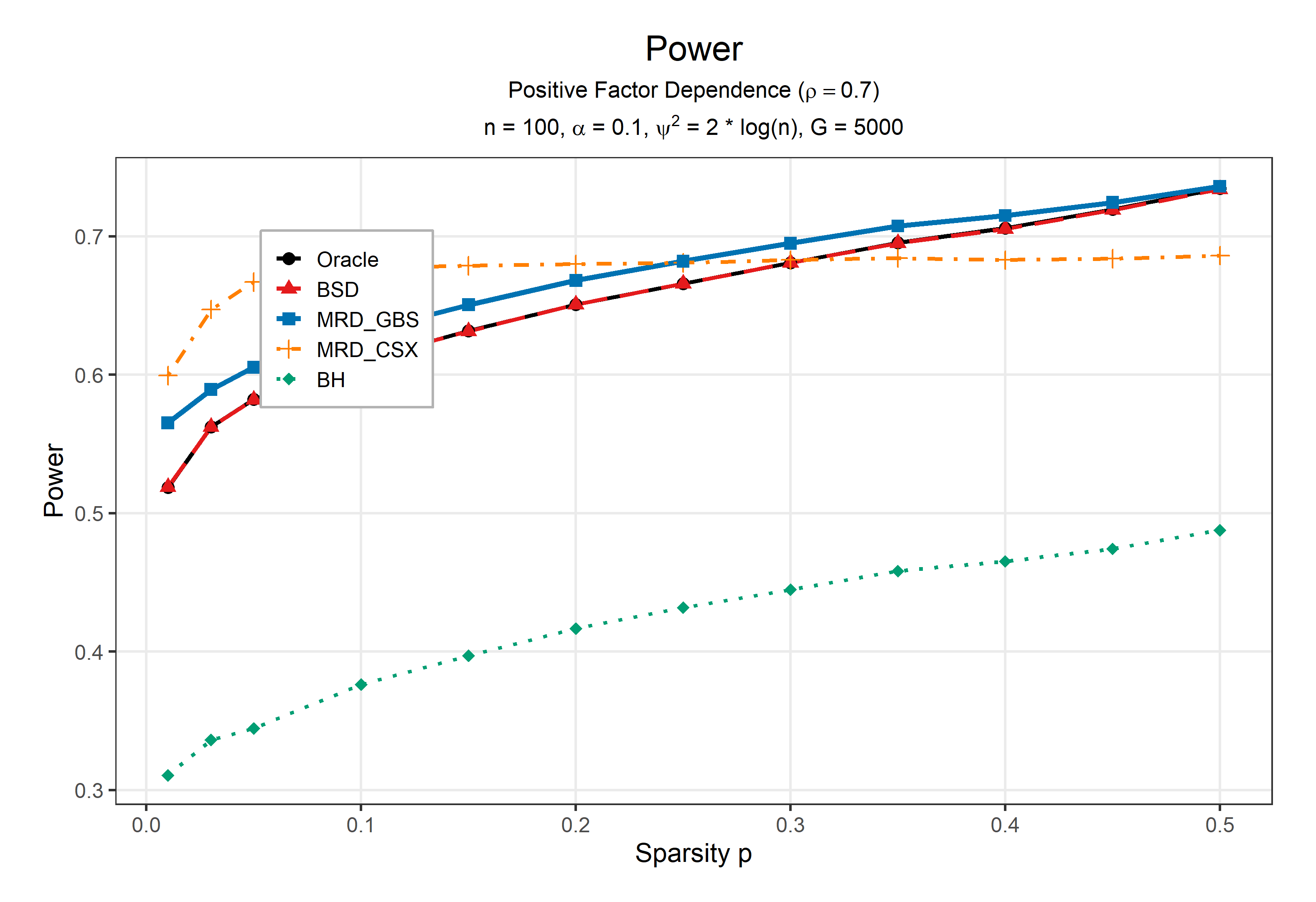}
		\caption{Positive Factor Dependence}
	\end{subfigure}
	
	\vspace{0.3cm}
	
	\begin{subfigure}{0.48\textwidth}
		\centering
		\includegraphics[width=\linewidth]{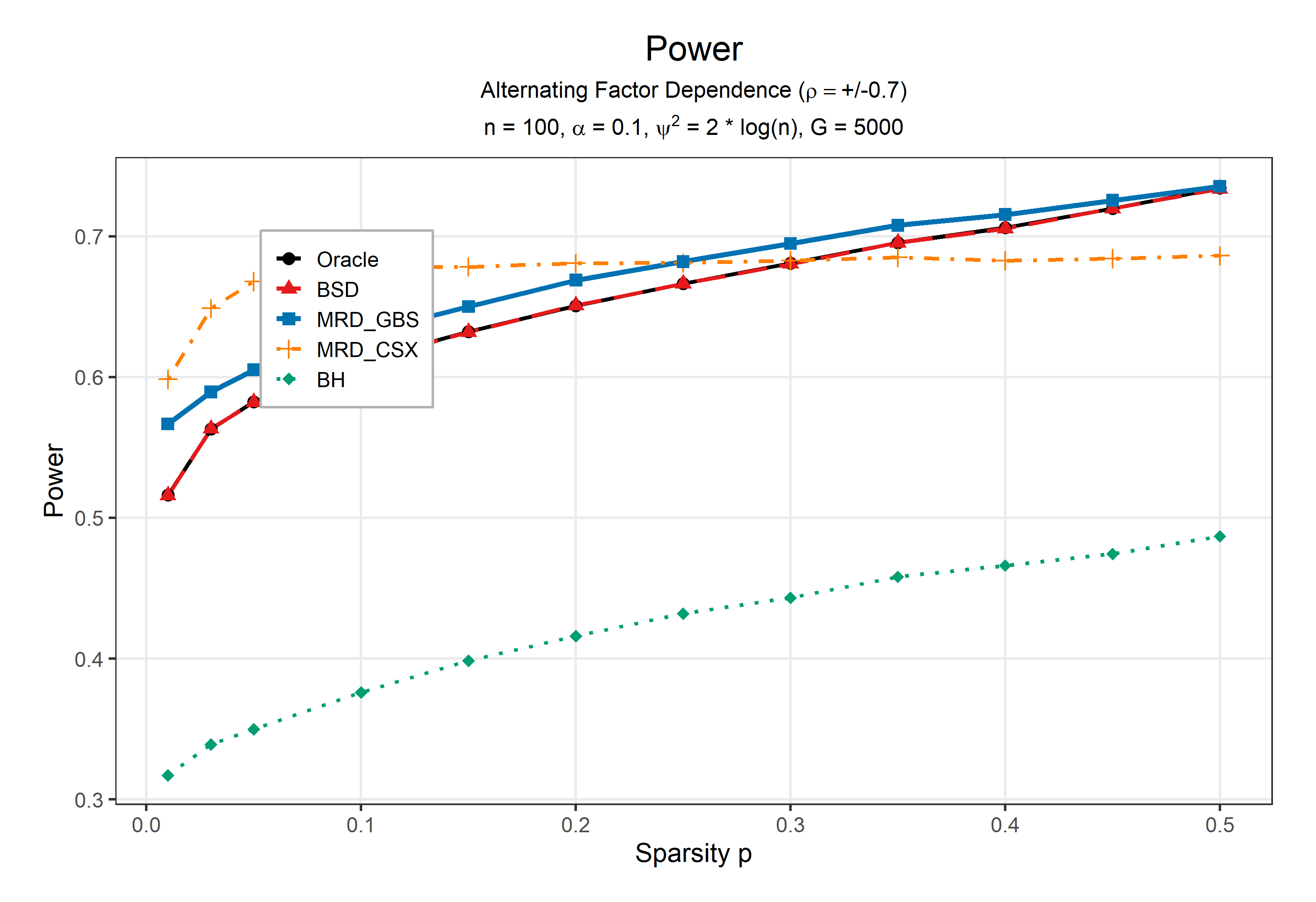}
		\caption{Alternating Factor Dependence}
	\end{subfigure}
	\hfill
	\begin{subfigure}{0.48\textwidth}
		\centering
		\includegraphics[width=\linewidth]{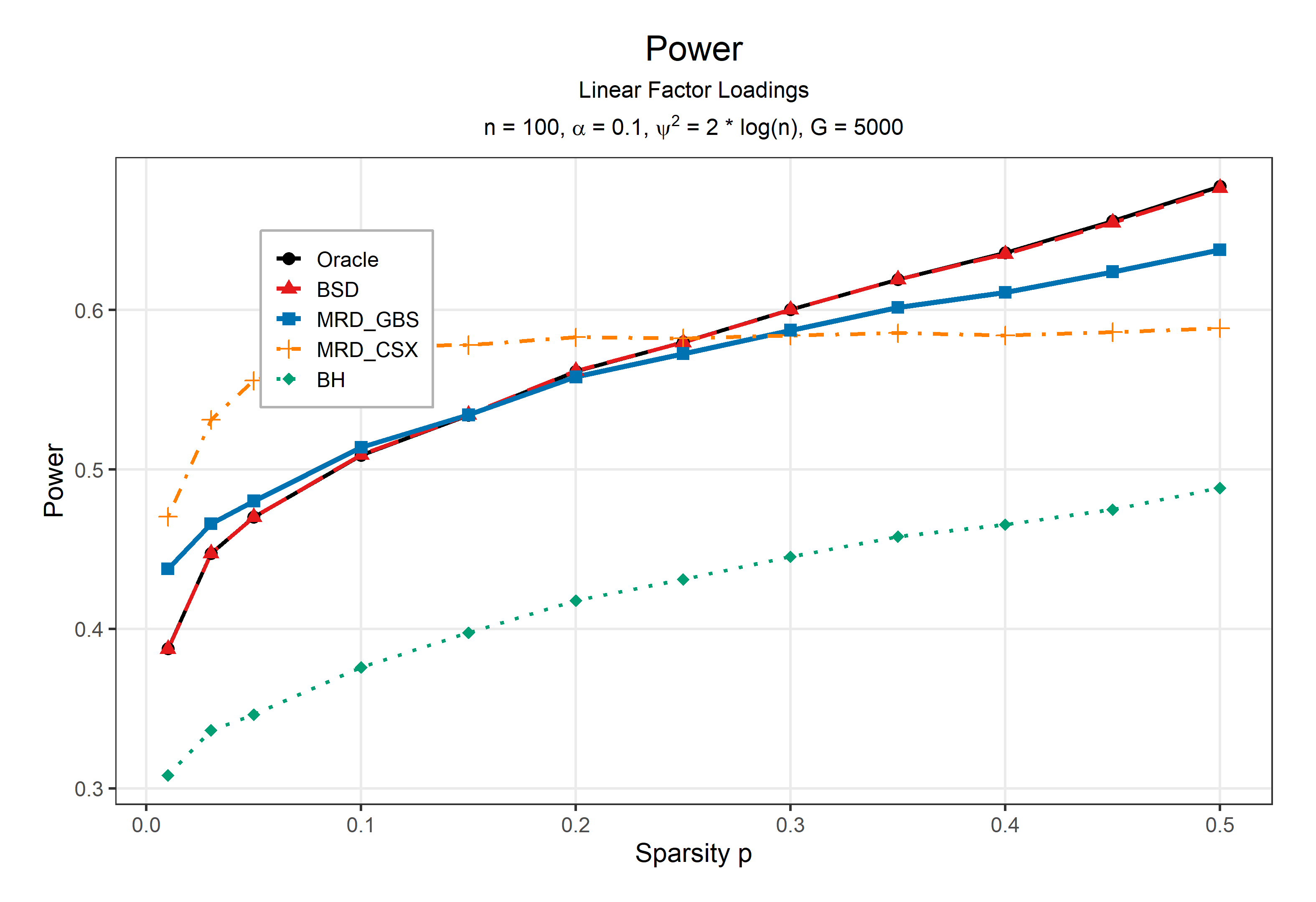}
		\caption{Linear Loadings}
	\end{subfigure}
	
	\vspace{0.3cm}
	
	\begin{subfigure}{0.48\textwidth}
		\centering
		\includegraphics[width=\linewidth]{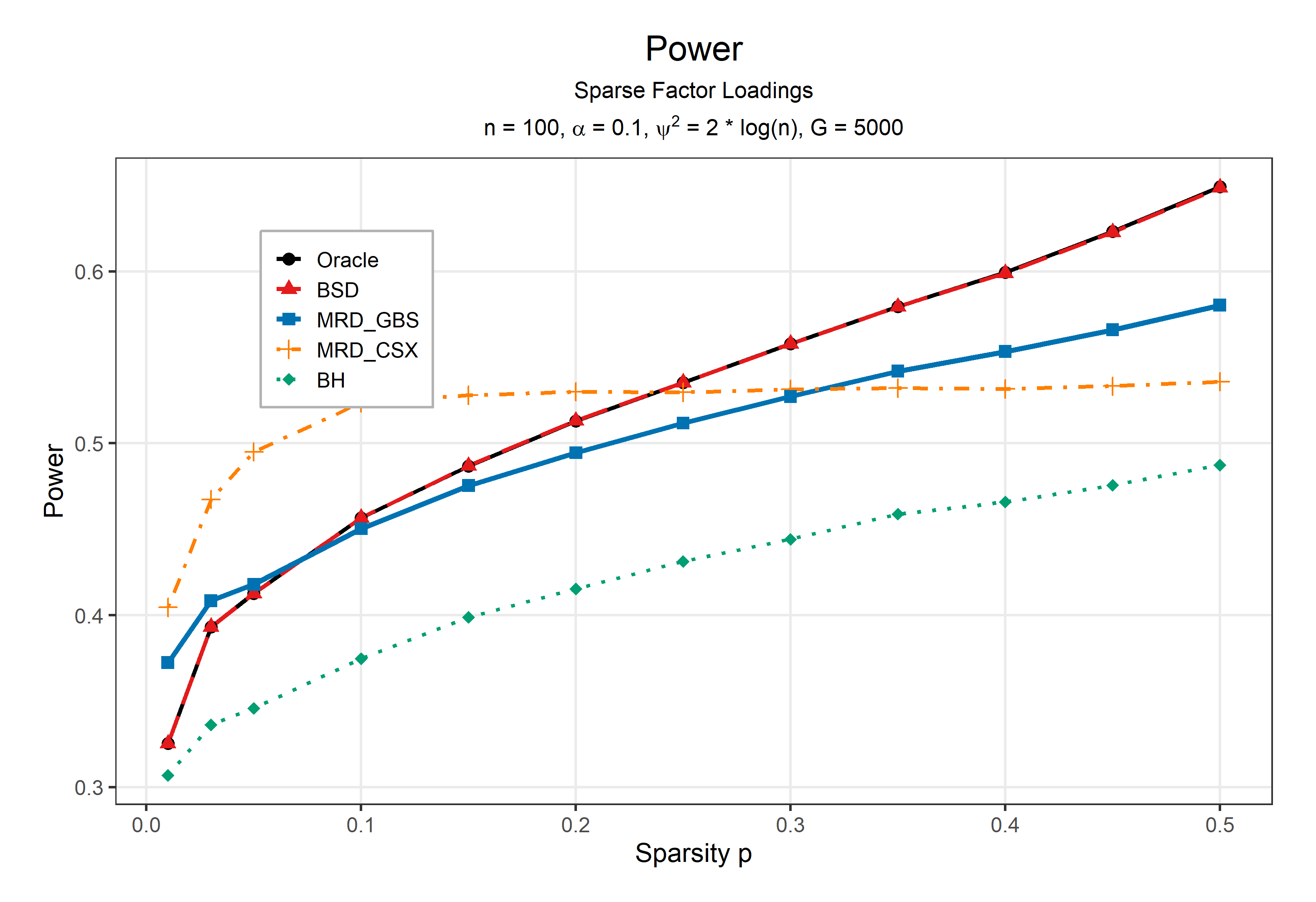}
		\caption{Sparse Factor Dependence}
	\end{subfigure}
	\hfill
	\begin{subfigure}{0.48\textwidth}
		\centering
		\includegraphics[width=\linewidth]{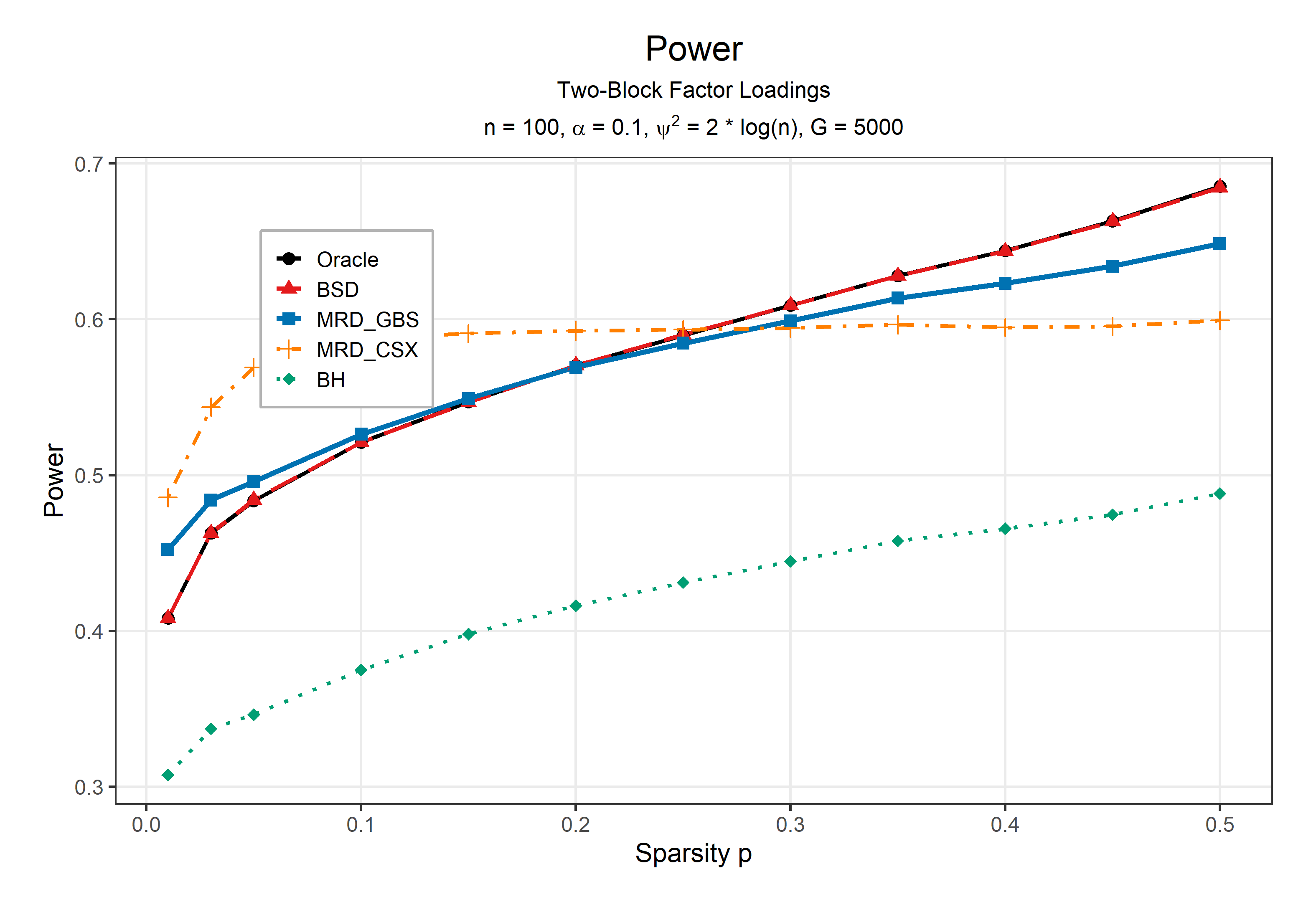}
		\caption{Two-Block Factor Dependence}
	\end{subfigure}
	
	\caption{Empirical powers under the six one-factor dependence structures when $n=100$, \(\alpha=0.1\), \(\psi^2=2\log n\), and \(G=5000\) Monte Carlo replications.}
	\label{fig:power-n100}
\end{figure}

A striking feature of the results is once again the remarkable agreement between
BSD and the Bayes Oracle. Across all dimensions, sparsity levels, and dependence
structures considered in this study, the power curves of BSD are nearly
indistinguishable from those of the Oracle. This close agreement mirrors the
corresponding FNR results and provides further evidence that the posterior
model-pursuit mechanism underlying BSD is able to recover essentially the same
set of active signals identified by the Bayes-optimal decision rule. The agreement
is particularly noteworthy because it persists across a broad spectrum of covariance
structures, including heterogeneous and clustered factor models.

The MRD--GBS procedure also exhibits highly competitive power performance.
Across most configurations, its power remains remarkably close to that of BSD and
the Bayes Oracle, often differing only slightly even in the more challenging sparse
dependence settings. Combined with the stable FDR behavior observed in the
previous subsection, these results suggest that MRD--GBS achieves an effective
balance between false discovery control and signal recovery, thereby explaining its
strong overall Bayes-risk performance.

The power behavior of MRD–CSX exhibits an interesting contrast with its FDR performance. Across several dependence structures, MRD–CSX frequently attains power values that are comparable to, and often exceed, those of BSD, the Bayes Oracle, and MRD–GBS. However, these power gains are accompanied by substantially larger false discovery rates, as observed in the previous subsection. Consequently, the increased signal-detection ability of MRD–CSX comes at the expense of a less favorable balance between Type I and Type II errors, which helps explain its comparatively larger Bayes risks despite its strong power performance.

Overall, the power results reinforce the conclusions drawn from the Bayes risk,
FDR, and FNR analyses. BSD consistently reproduces the signal-detection
performance of the Bayes Oracle across a broad collection of dependence
structures, while MRD--GBS frequently attains comparable power despite arising
from a fundamentally different inferential framework. These findings provide
further evidence that both procedures are highly effective at exploiting dependence
information for sparse signal recovery, with BSD exhibiting near-Oracle behavior
throughout the simulation settings considered.

\subsubsection{Average Number of Rejections}

We next examine the average number of rejections produced by the competing procedures. While the FDR, FNR, and power metrics considered in the previous subsections summarize different aspects of testing performance, the average number of rejections provides a more direct description of the overall aggressiveness of a multiple testing procedure. In particular, it reveals how frequently a procedure chooses to declare discoveries across varying levels of sparsity and dependence.

\begin{figure}[p]
	\centering
	
	\begin{subfigure}{0.48\textwidth}
		\centering
		\includegraphics[width=\linewidth]{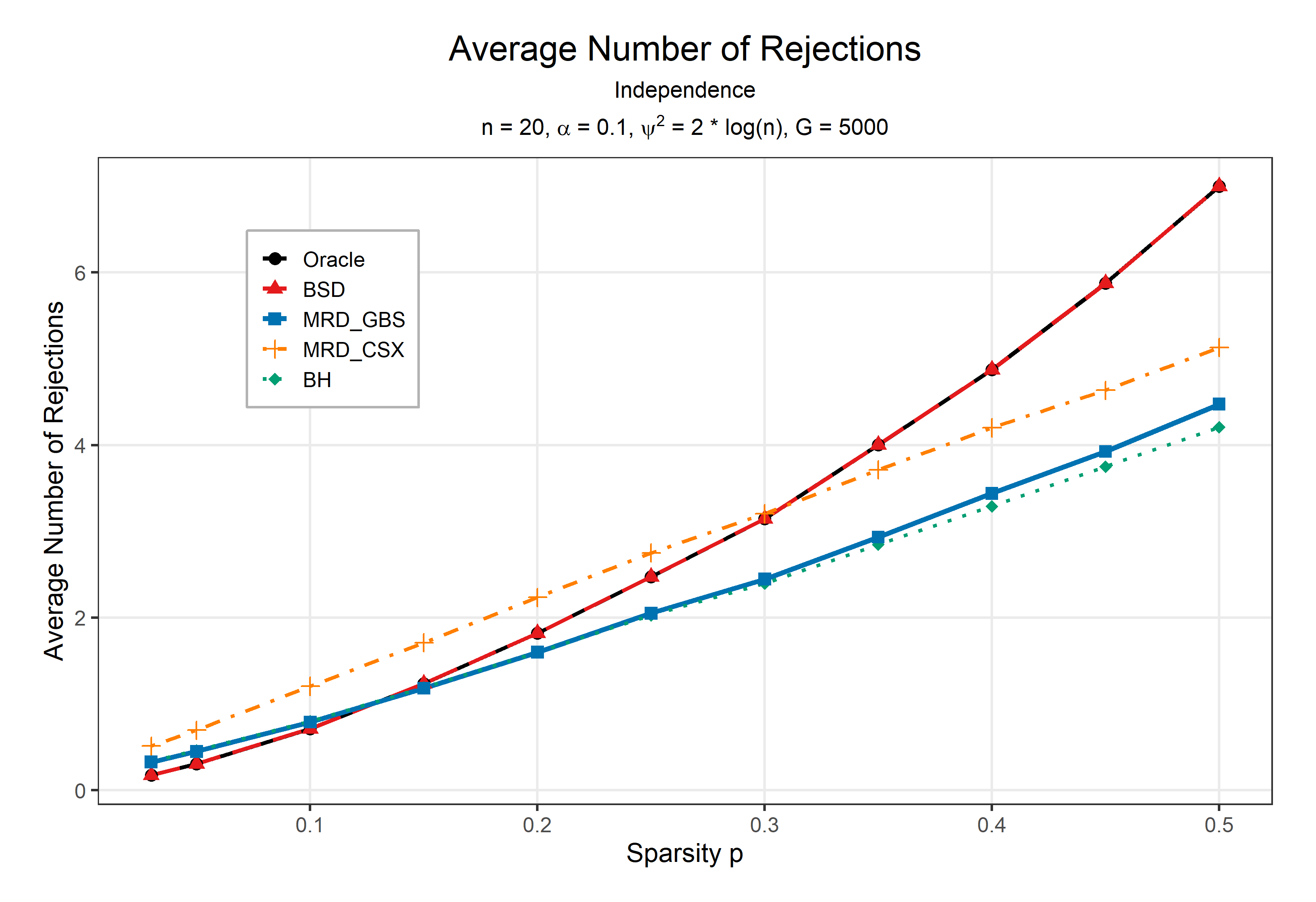}
		\caption{Independence}
	\end{subfigure}
	\hfill
	\begin{subfigure}{0.48\textwidth}
		\centering
		\includegraphics[width=\linewidth]{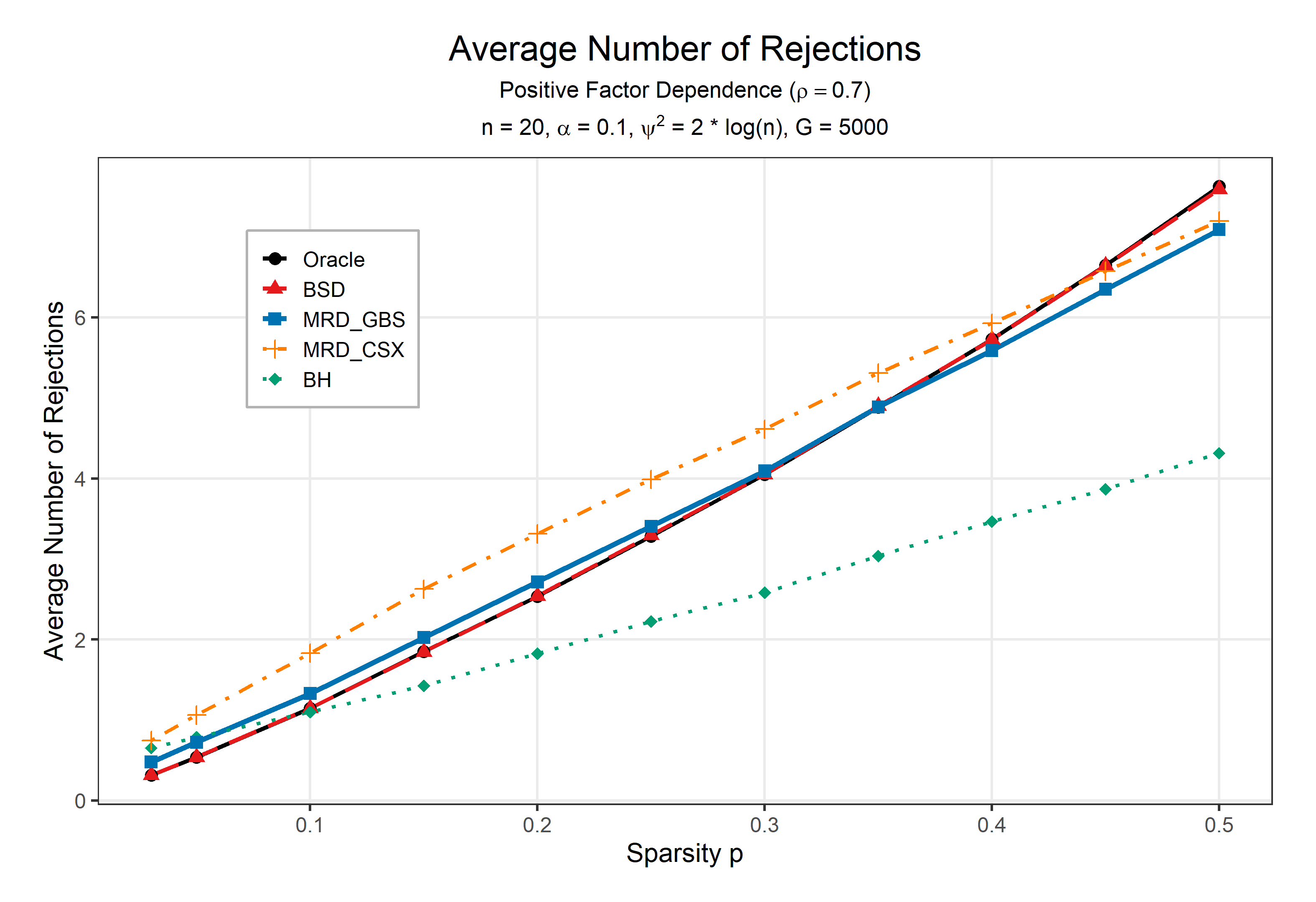}
		\caption{Positive Factor Dependence}
	\end{subfigure}
	
	\vspace{0.3cm}
	
	\begin{subfigure}{0.48\textwidth}
		\centering
		\includegraphics[width=\linewidth]{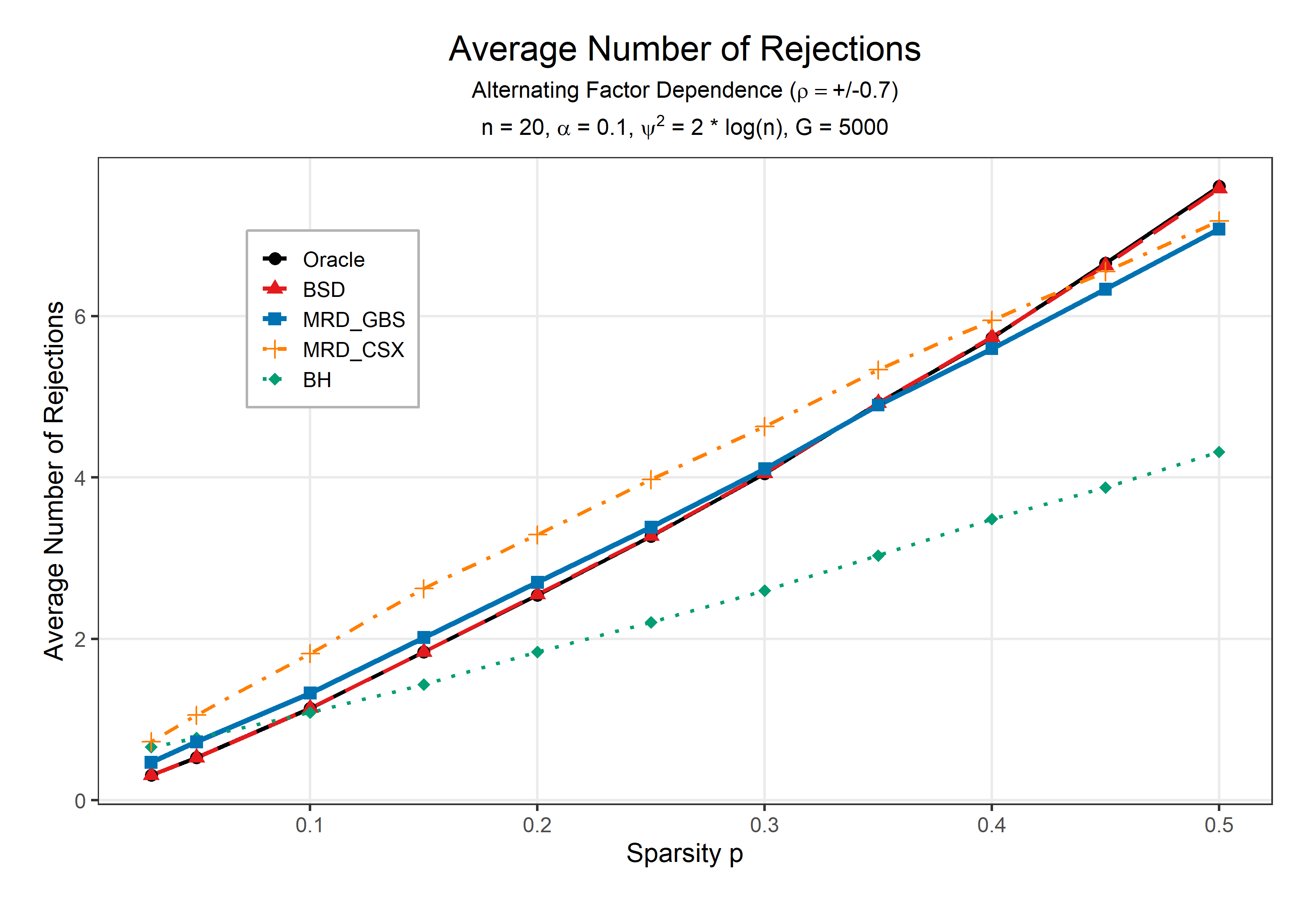}
		\caption{Alternating Factor Dependence}
	\end{subfigure}
	\hfill
	\begin{subfigure}{0.48\textwidth}
		\centering
		\includegraphics[width=\linewidth]{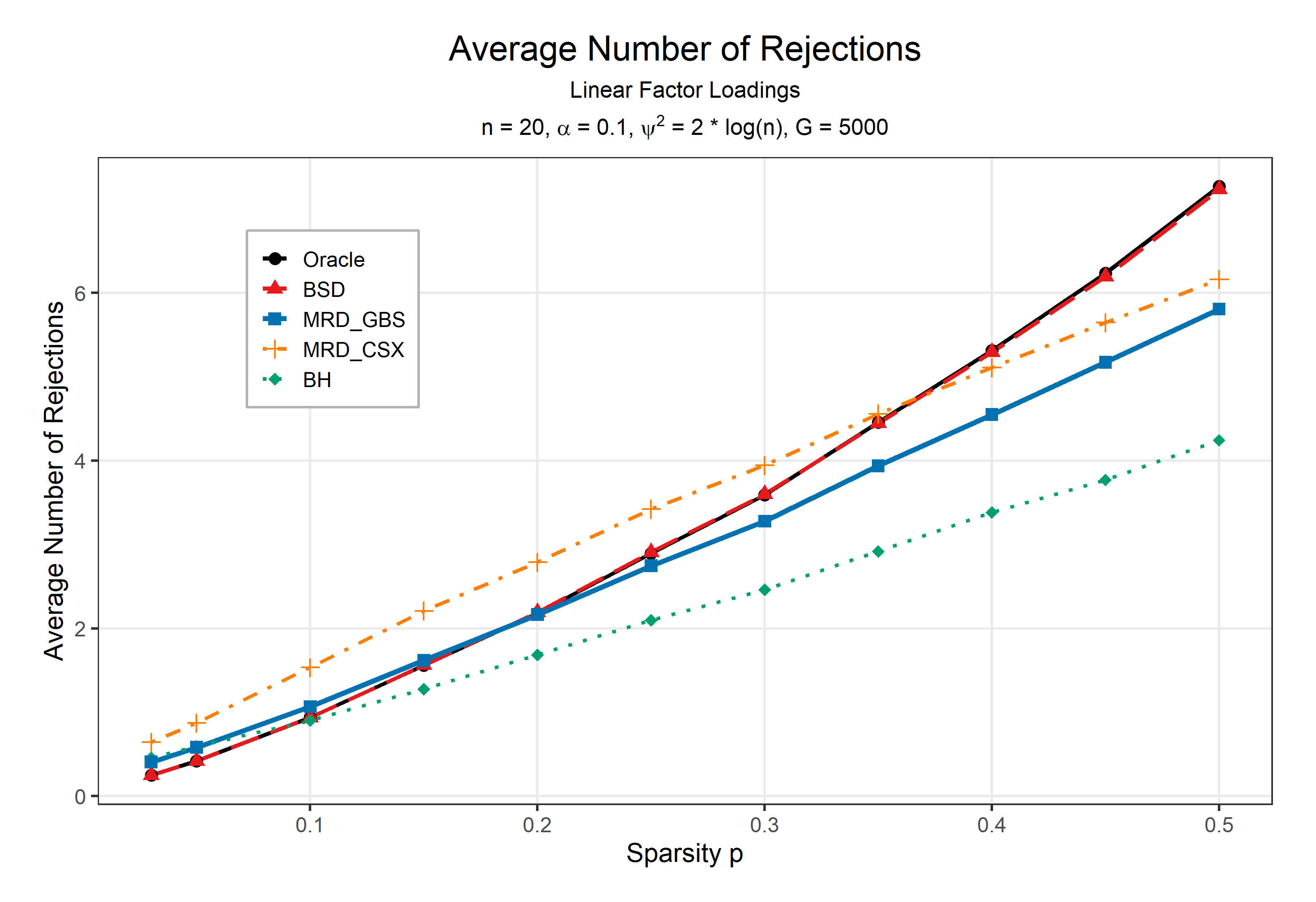}
		\caption{Linear Loadings}
	\end{subfigure}
	
	\vspace{0.3cm}
	
	\begin{subfigure}{0.48\textwidth}
		\centering
		\includegraphics[width=\linewidth]{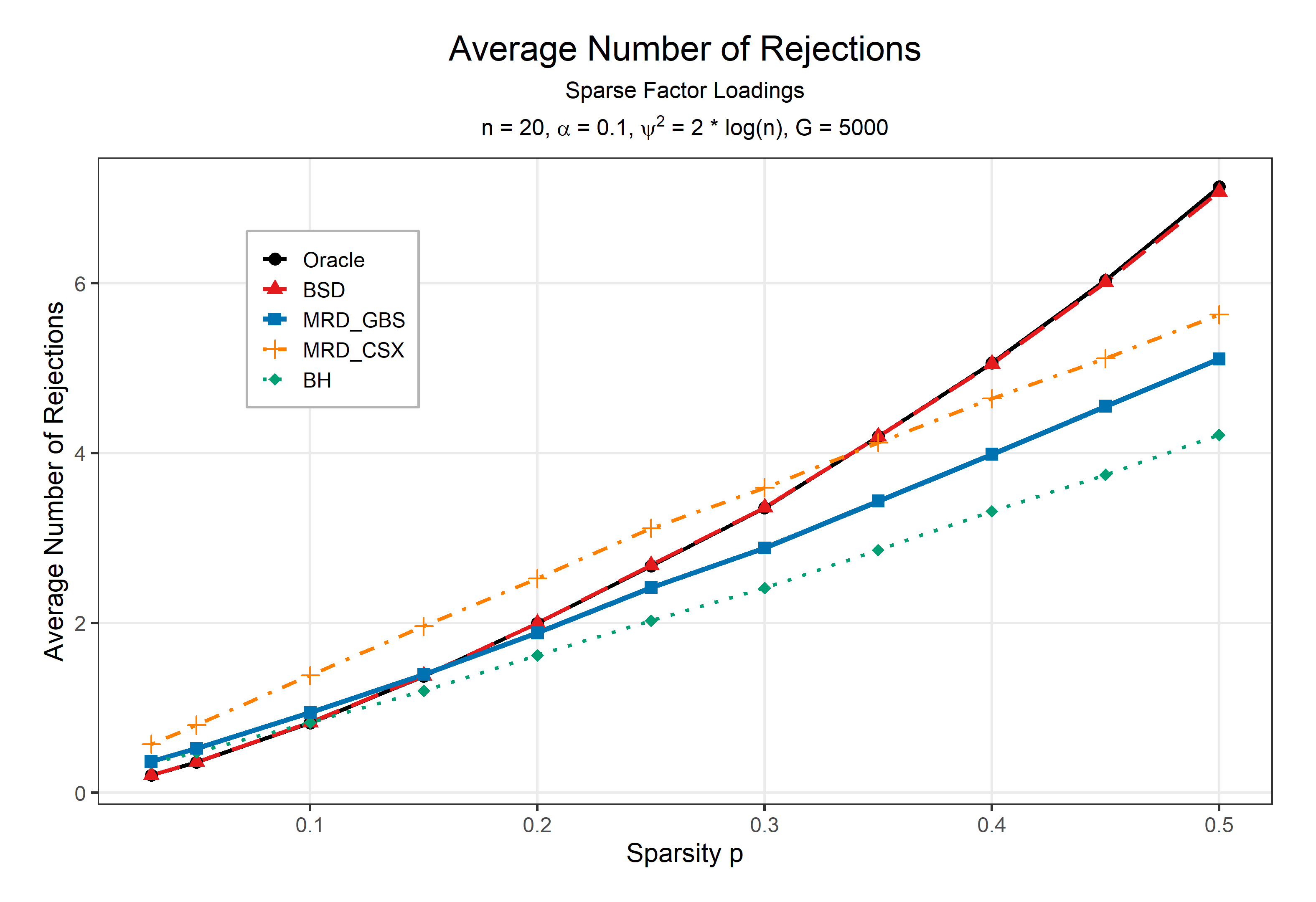}
		\caption{Sparse Factor Dependence}
	\end{subfigure}
	\hfill
	\begin{subfigure}{0.48\textwidth}
		\centering
		\includegraphics[width=\linewidth]{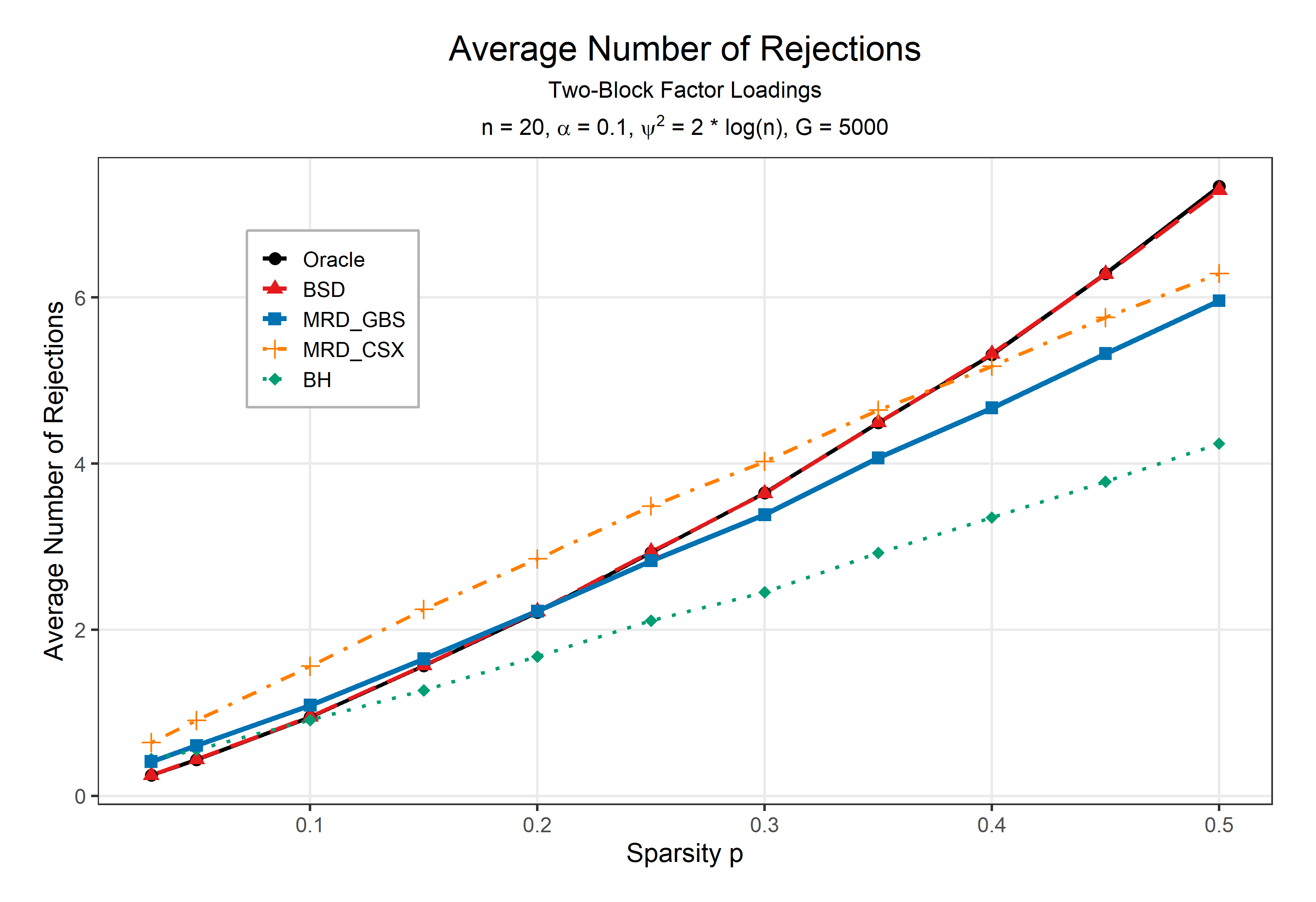}
		\caption{Two-Block Factor Dependence}
	\end{subfigure}
	
	\caption{Average number of rejections under the six one-factor dependence structures when $n=20$, \(\alpha=0.1\), \(\psi^2=2\log n\), and \(G=5000\) Monte Carlo replications.}
	\label{fig:anr-n20}
\end{figure}

\begin{figure}[p]
	\centering
	
	\begin{subfigure}{0.48\textwidth}
		\centering
		\includegraphics[width=\linewidth]{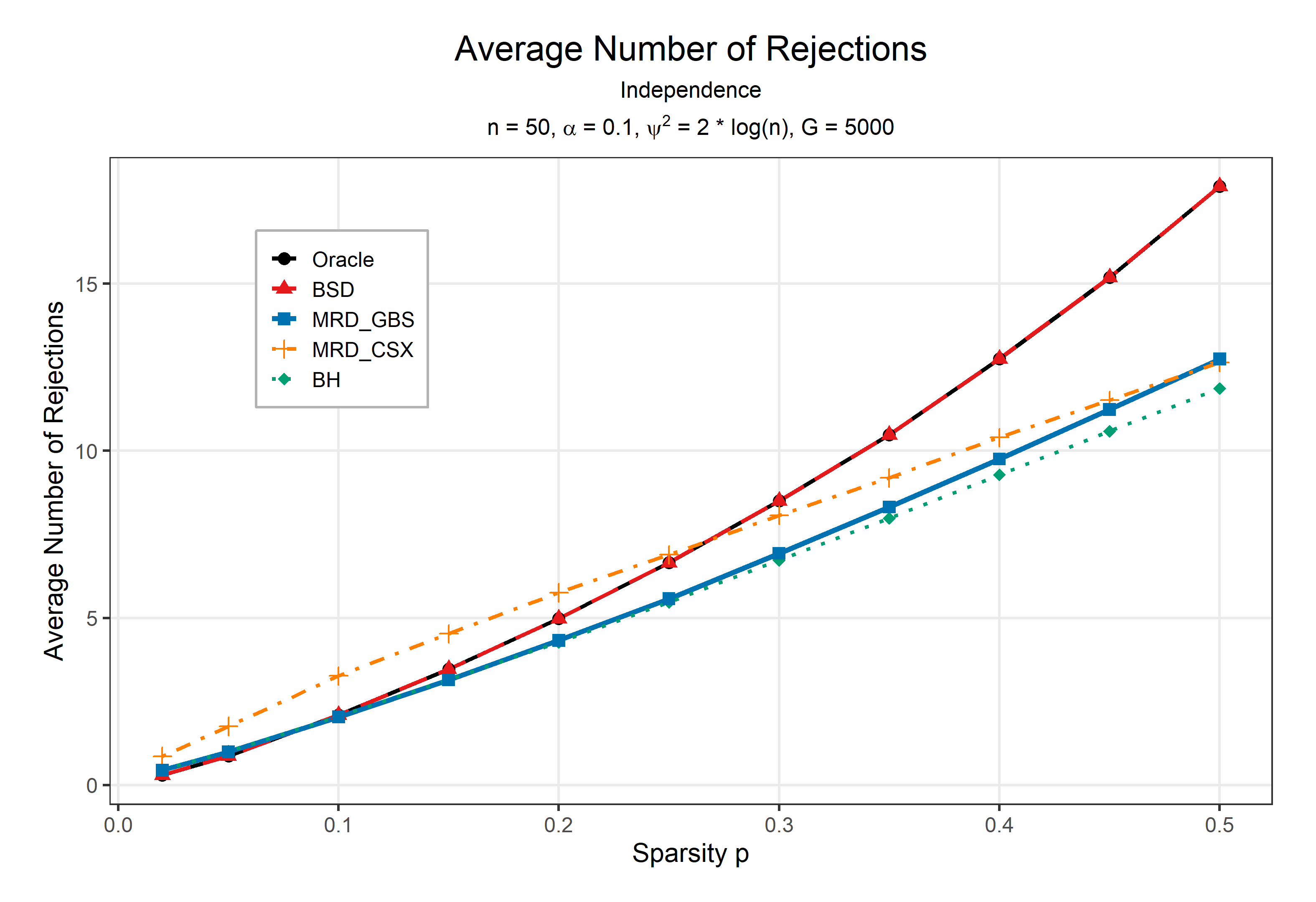}
		\caption{Independence}
	\end{subfigure}
	\hfill
	\begin{subfigure}{0.48\textwidth}
		\centering
		\includegraphics[width=\linewidth]{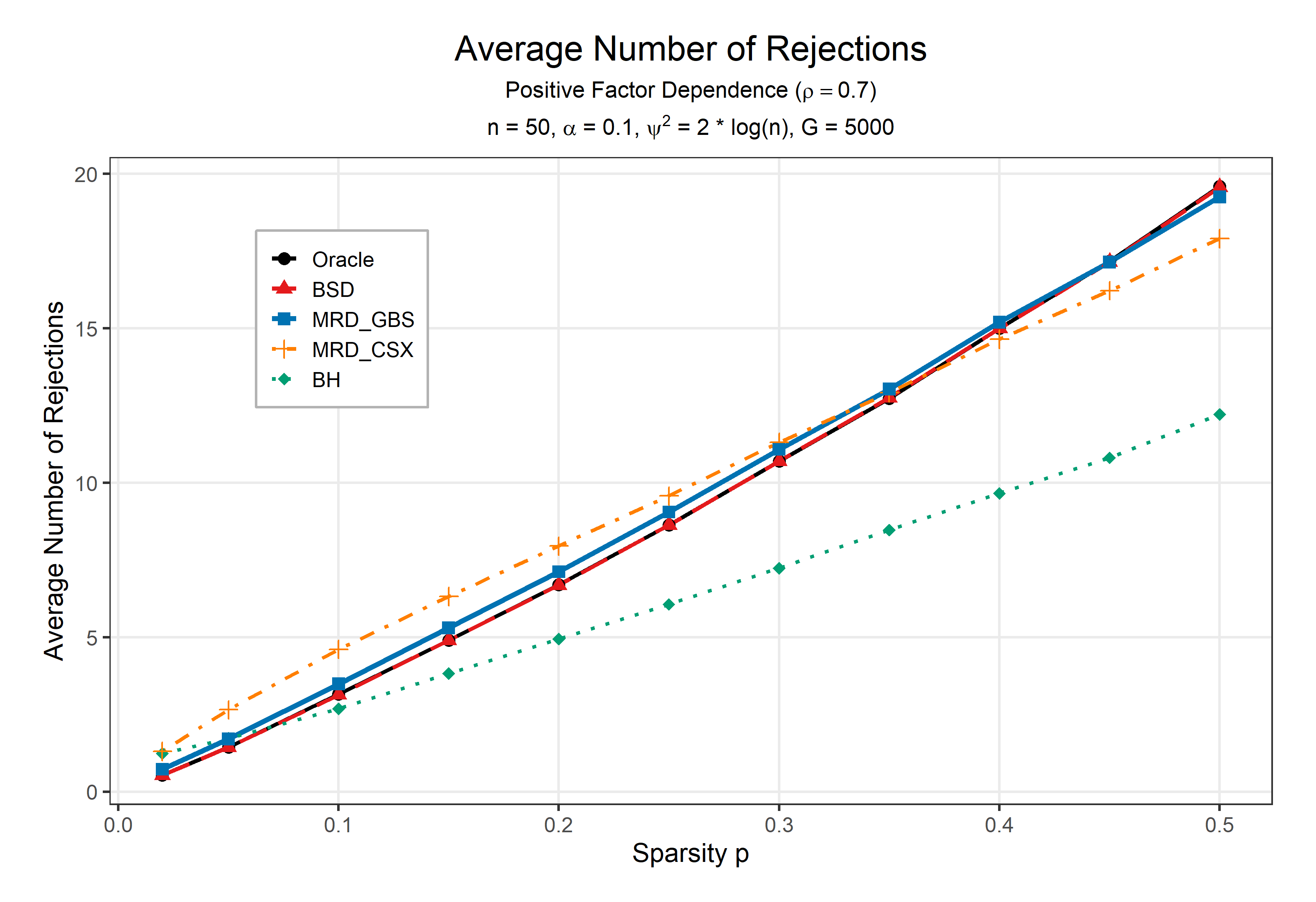}
		\caption{Positive Factor Dependence}
	\end{subfigure}
	
	\vspace{0.3cm}
	
	\begin{subfigure}{0.48\textwidth}
		\centering
		\includegraphics[width=\linewidth]{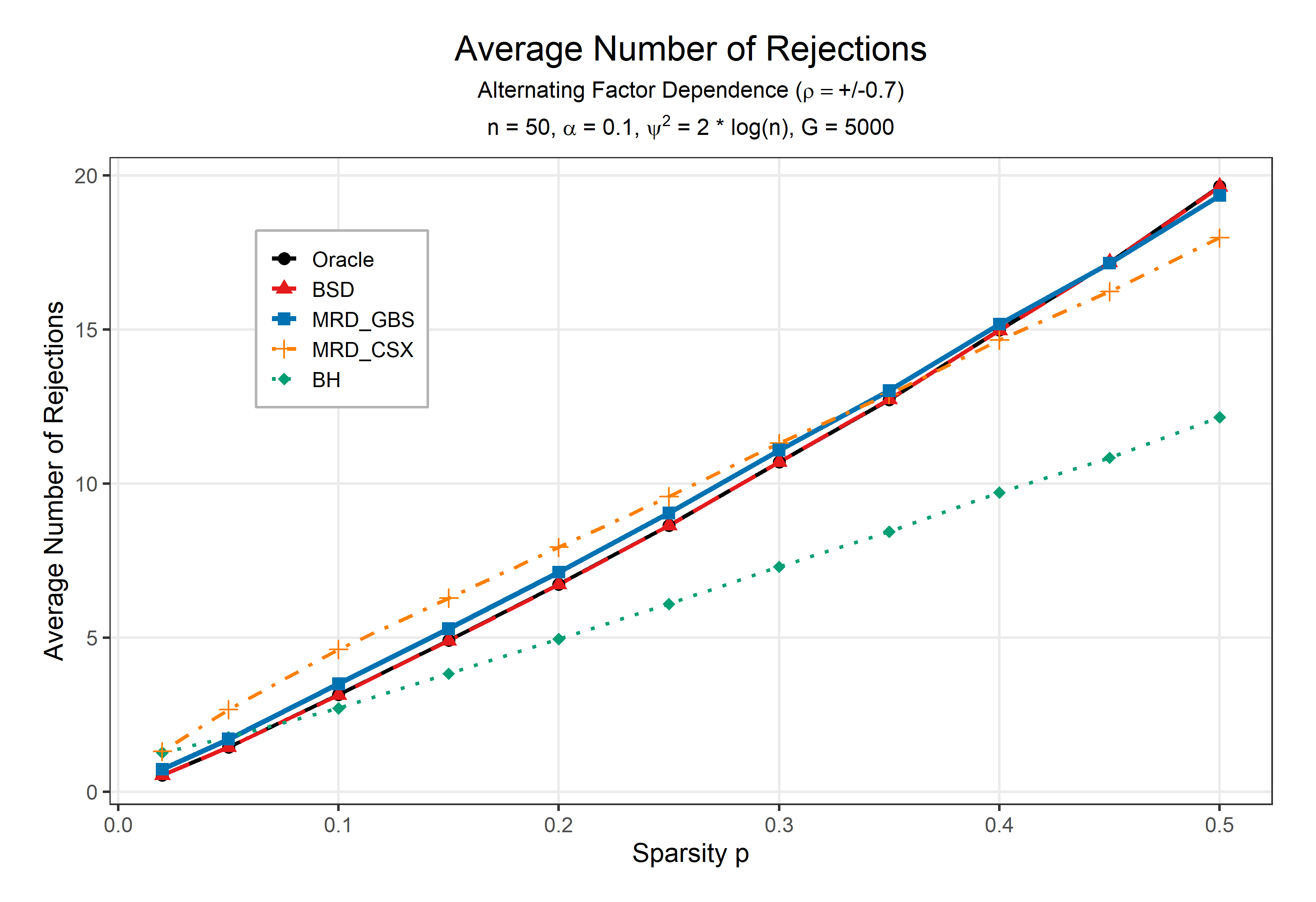}
		\caption{Alternating Factor Dependence}
	\end{subfigure}
	\hfill
	\begin{subfigure}{0.48\textwidth}
		\centering
		\includegraphics[width=\linewidth]{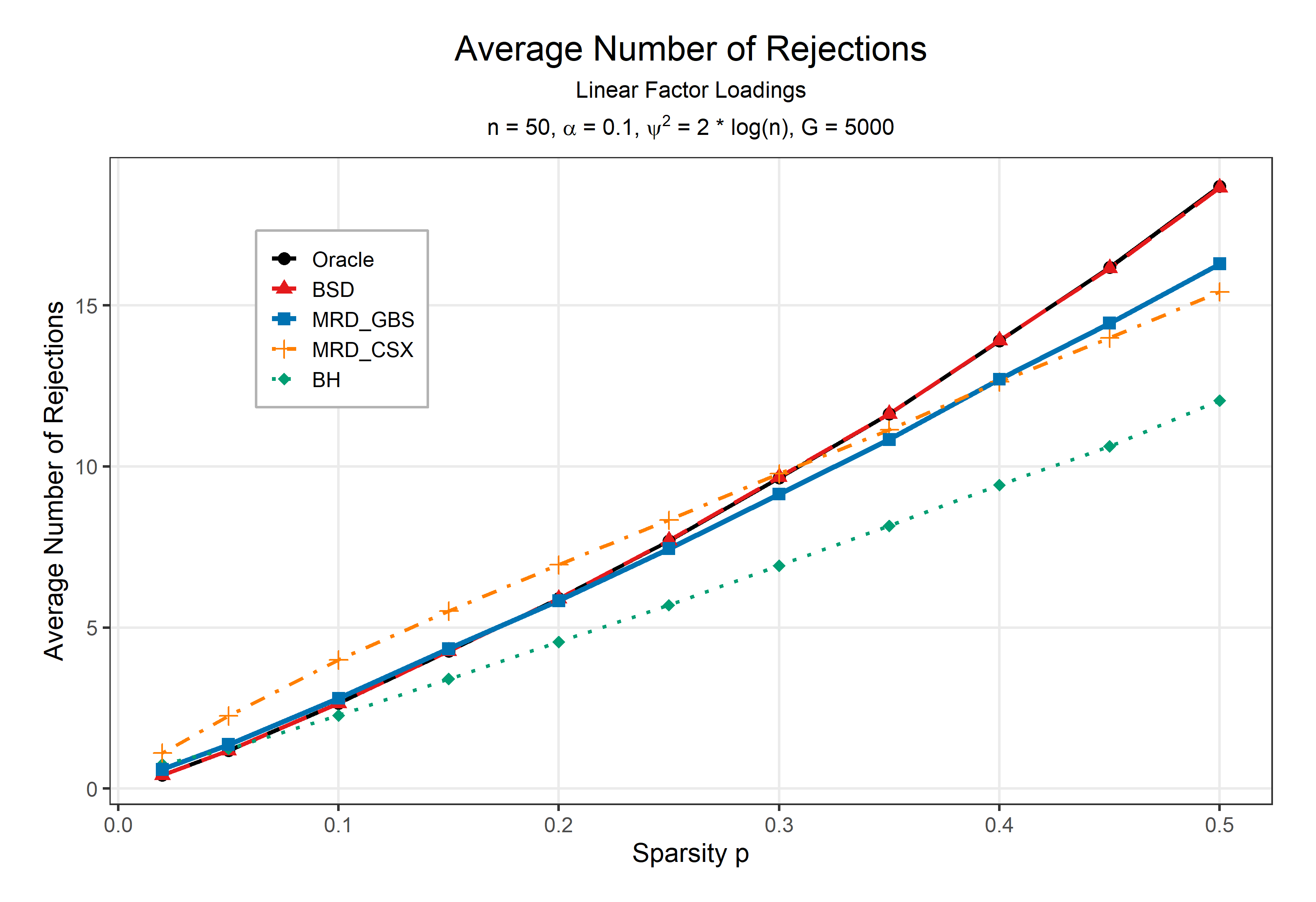}
		\caption{Linear Loadings}
	\end{subfigure}
	
	\vspace{0.3cm}
	
	\begin{subfigure}{0.48\textwidth}
		\centering
		\includegraphics[width=\linewidth]{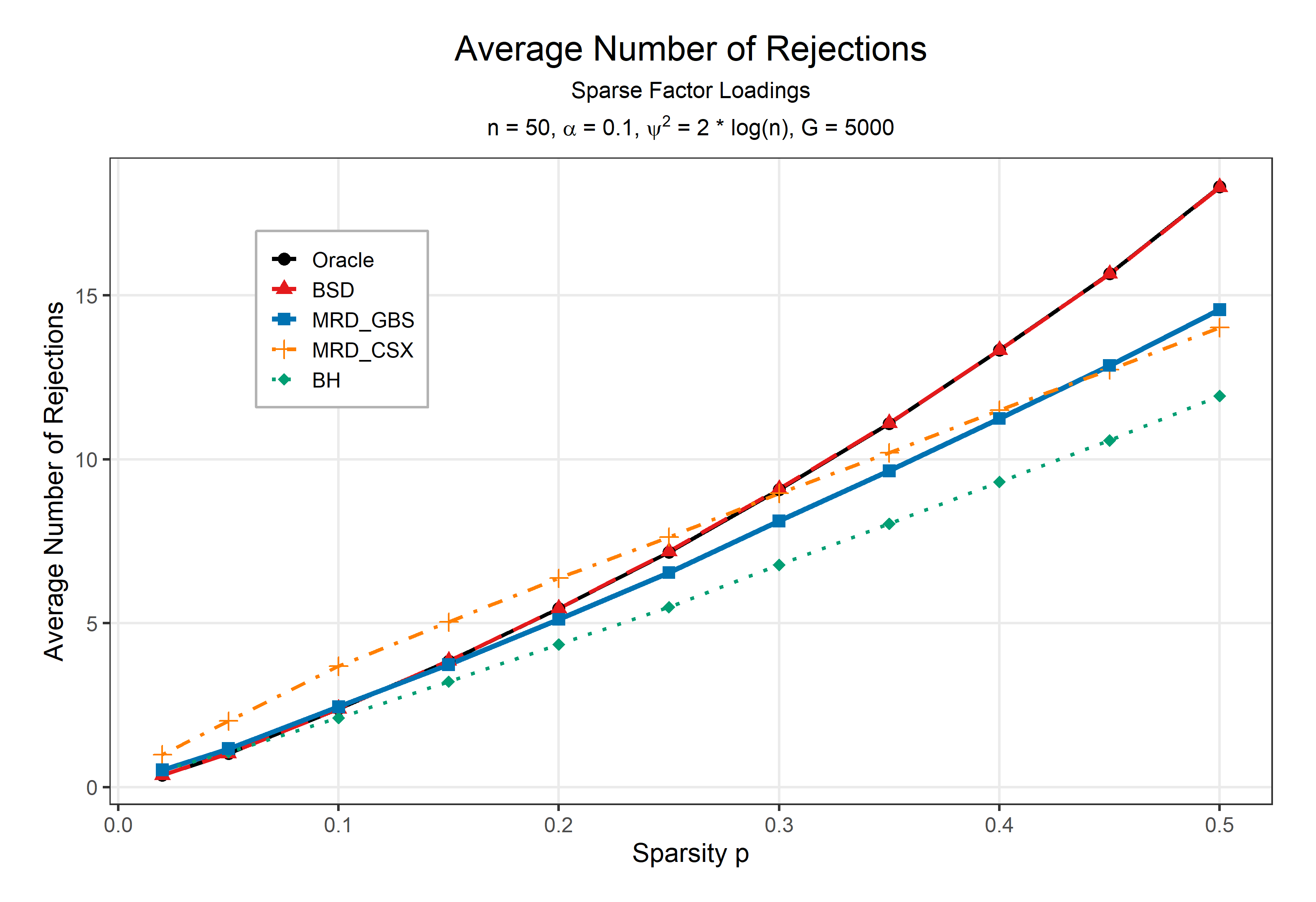}
		\caption{Sparse Factor Dependence}
	\end{subfigure}
	\hfill
	\begin{subfigure}{0.48\textwidth}
		\centering
		\includegraphics[width=\linewidth]{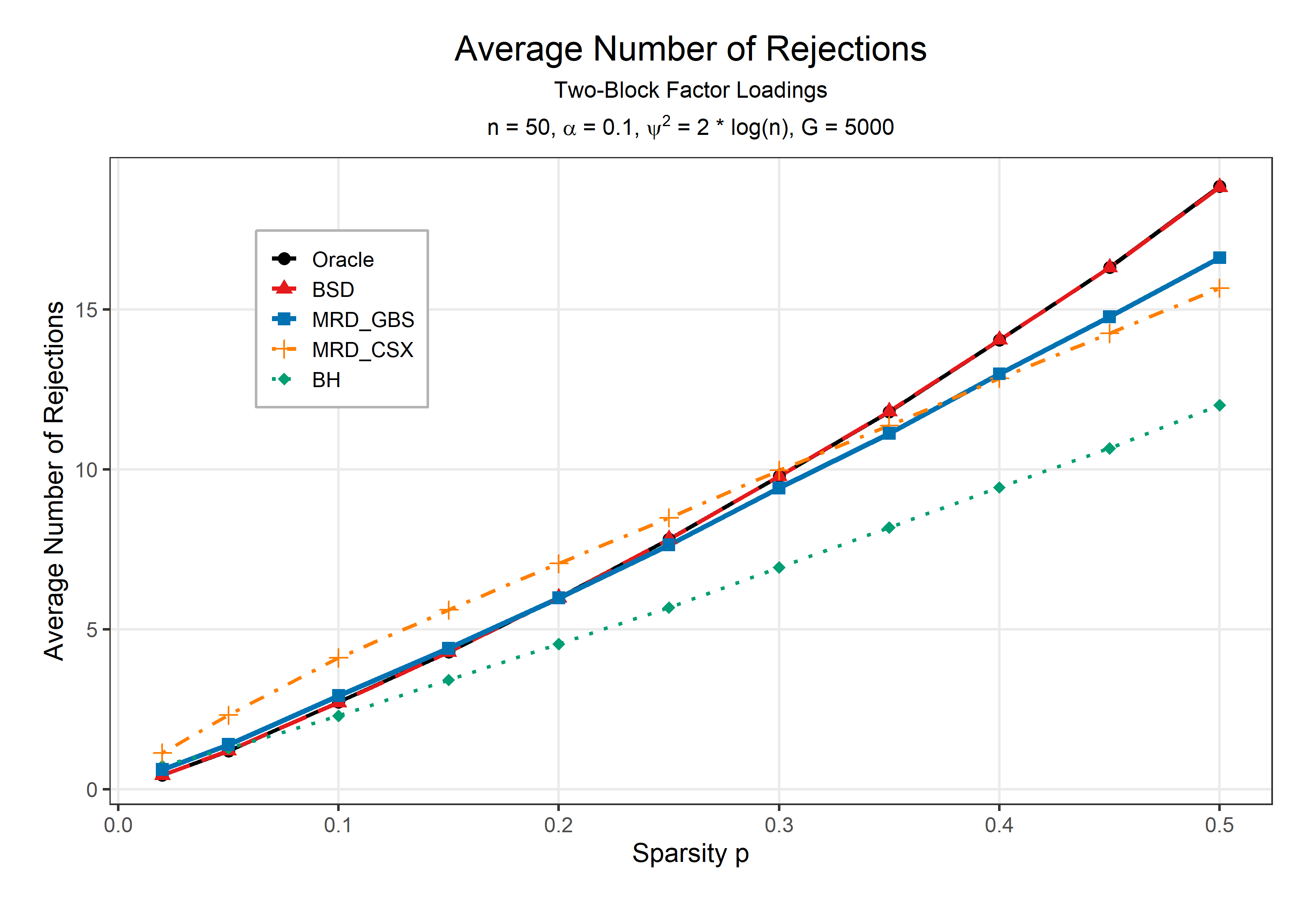}
		\caption{Two-Block Factor Dependence}
	\end{subfigure}
	
	\caption{Average number of rejections under the six one-factor dependence structures when $n=50$, \(\alpha=0.1\), \(\psi^2=2\log n\), and \(G=5000\) Monte Carlo replications.}
	\label{fig:anr-n50}
\end{figure}

\begin{figure}[p]
	\centering
	
	\begin{subfigure}{0.48\textwidth}
		\centering
		\includegraphics[width=\linewidth]{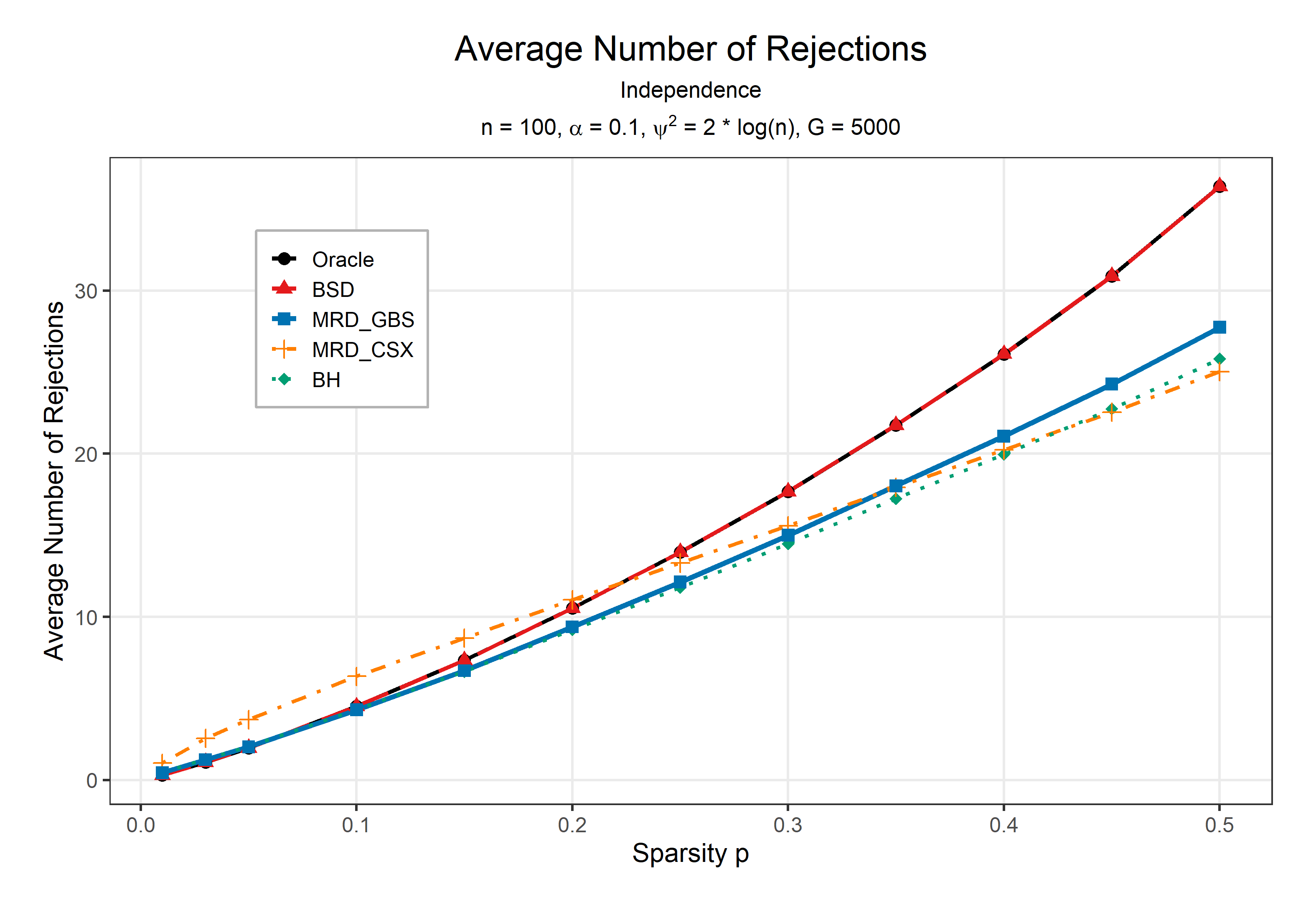}
		\caption{Independence}
	\end{subfigure}
	\hfill
	\begin{subfigure}{0.48\textwidth}
		\centering
		\includegraphics[width=\linewidth]{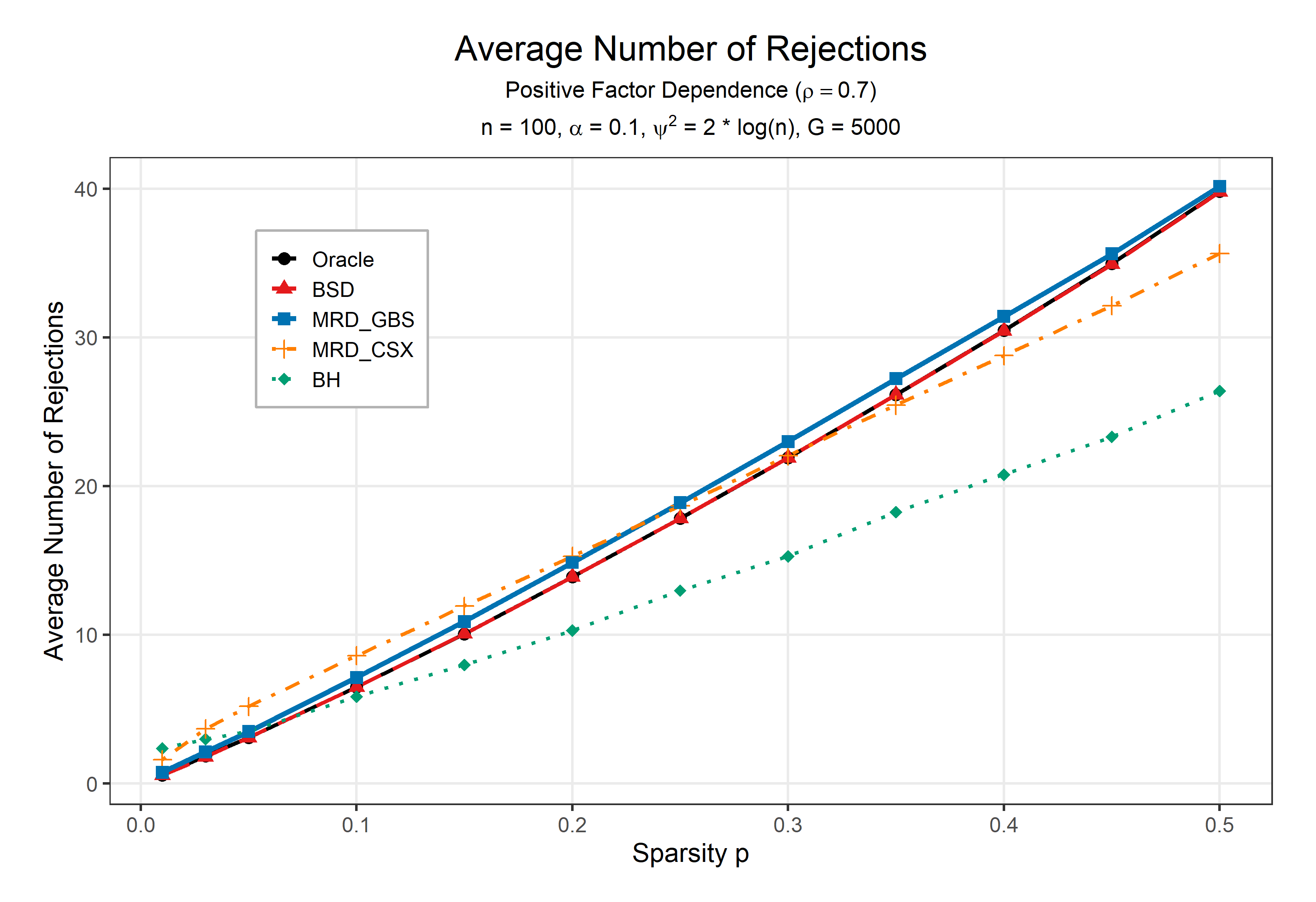}
		\caption{Positive Factor Dependence}
	\end{subfigure}
	
	\vspace{0.3cm}
	
	\begin{subfigure}{0.48\textwidth}
		\centering
		\includegraphics[width=\linewidth]{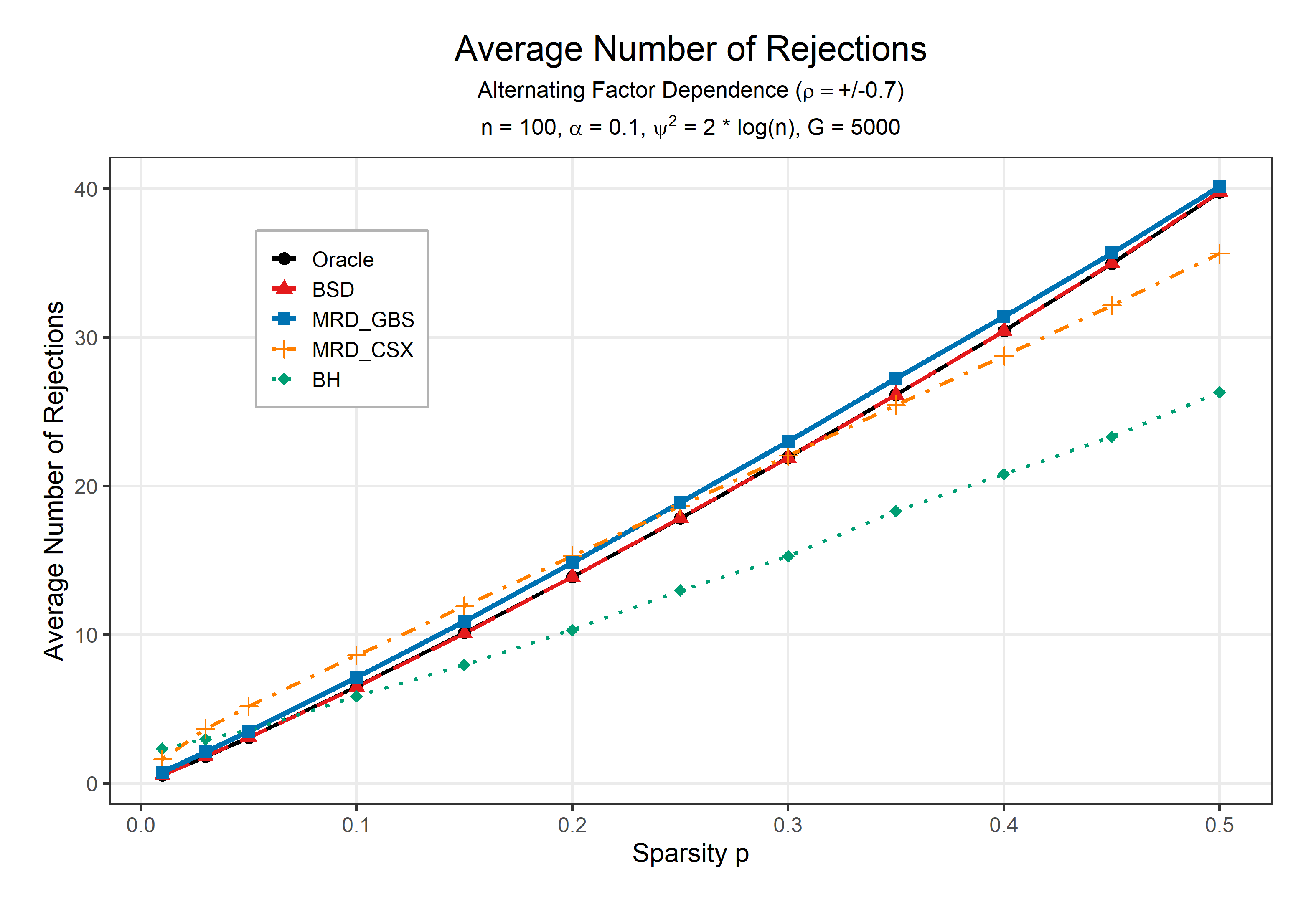}
		\caption{Alternating Factor Dependence}
	\end{subfigure}
	\hfill
	\begin{subfigure}{0.48\textwidth}
		\centering
		\includegraphics[width=\linewidth]{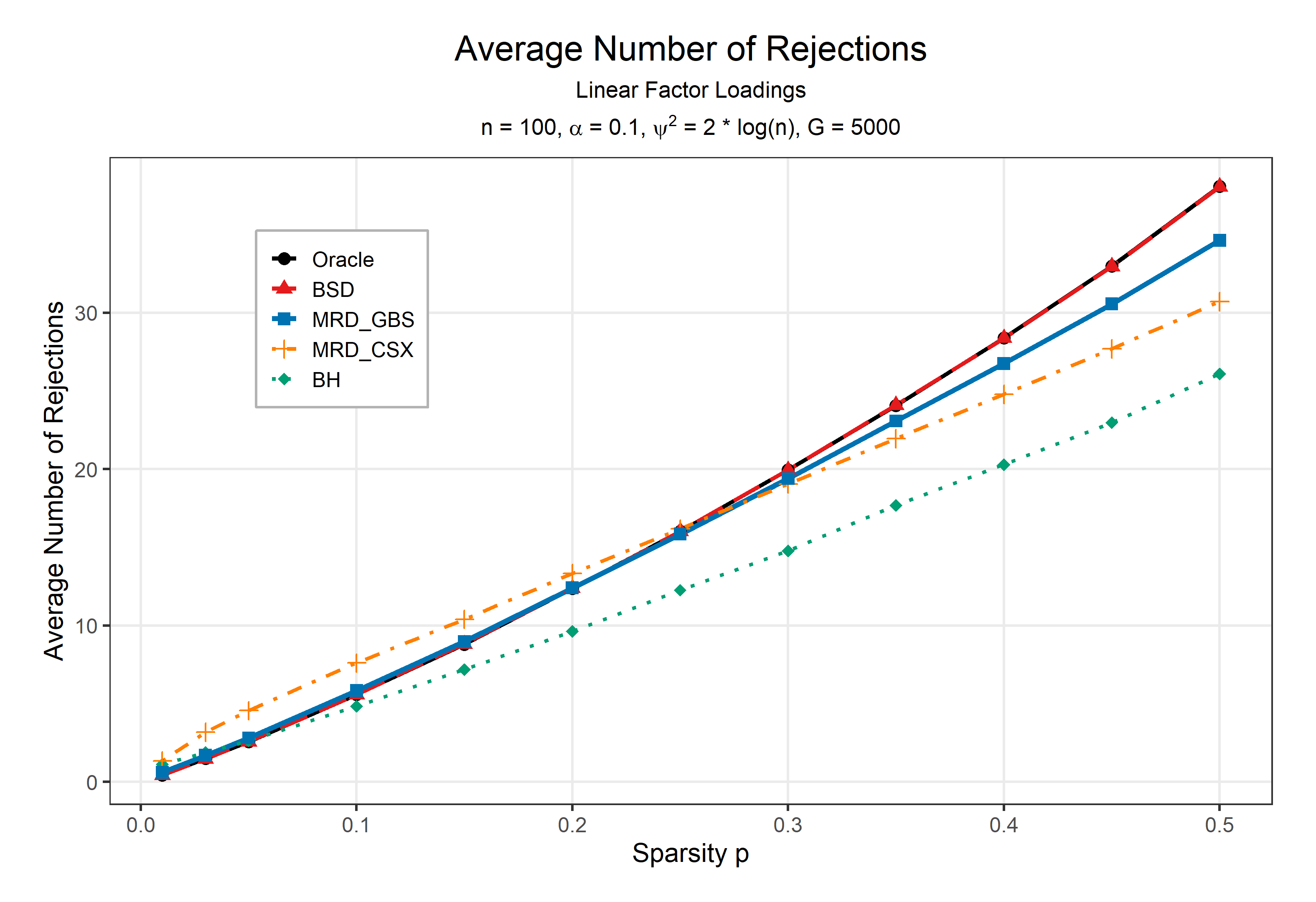}
		\caption{Linear Loadings}
	\end{subfigure}
	
	\vspace{0.3cm}
	
	\begin{subfigure}{0.48\textwidth}
		\centering
		\includegraphics[width=\linewidth]{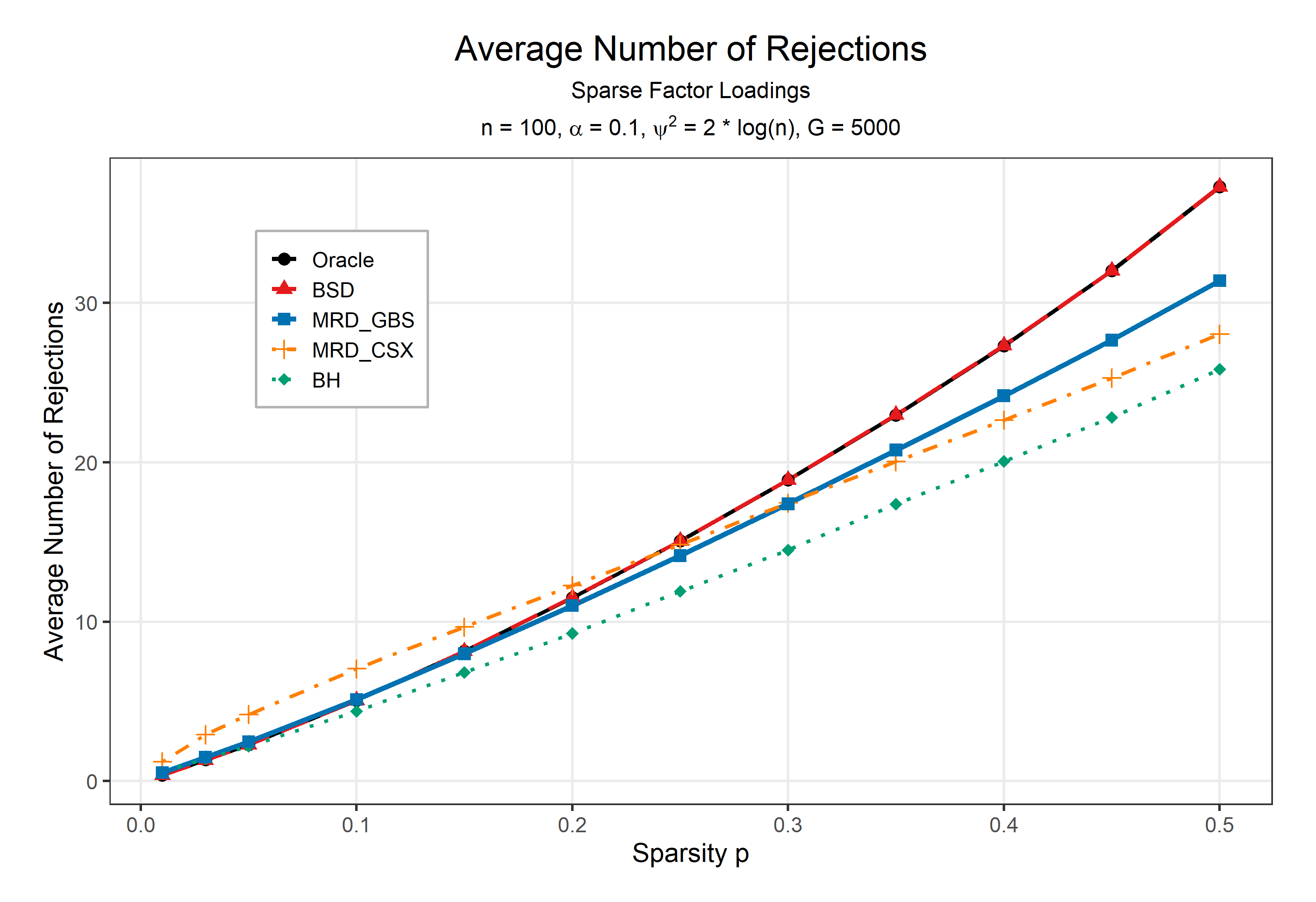}
		\caption{Sparse Factor Dependence}
	\end{subfigure}
	\hfill
	\begin{subfigure}{0.48\textwidth}
		\centering
		\includegraphics[width=\linewidth]{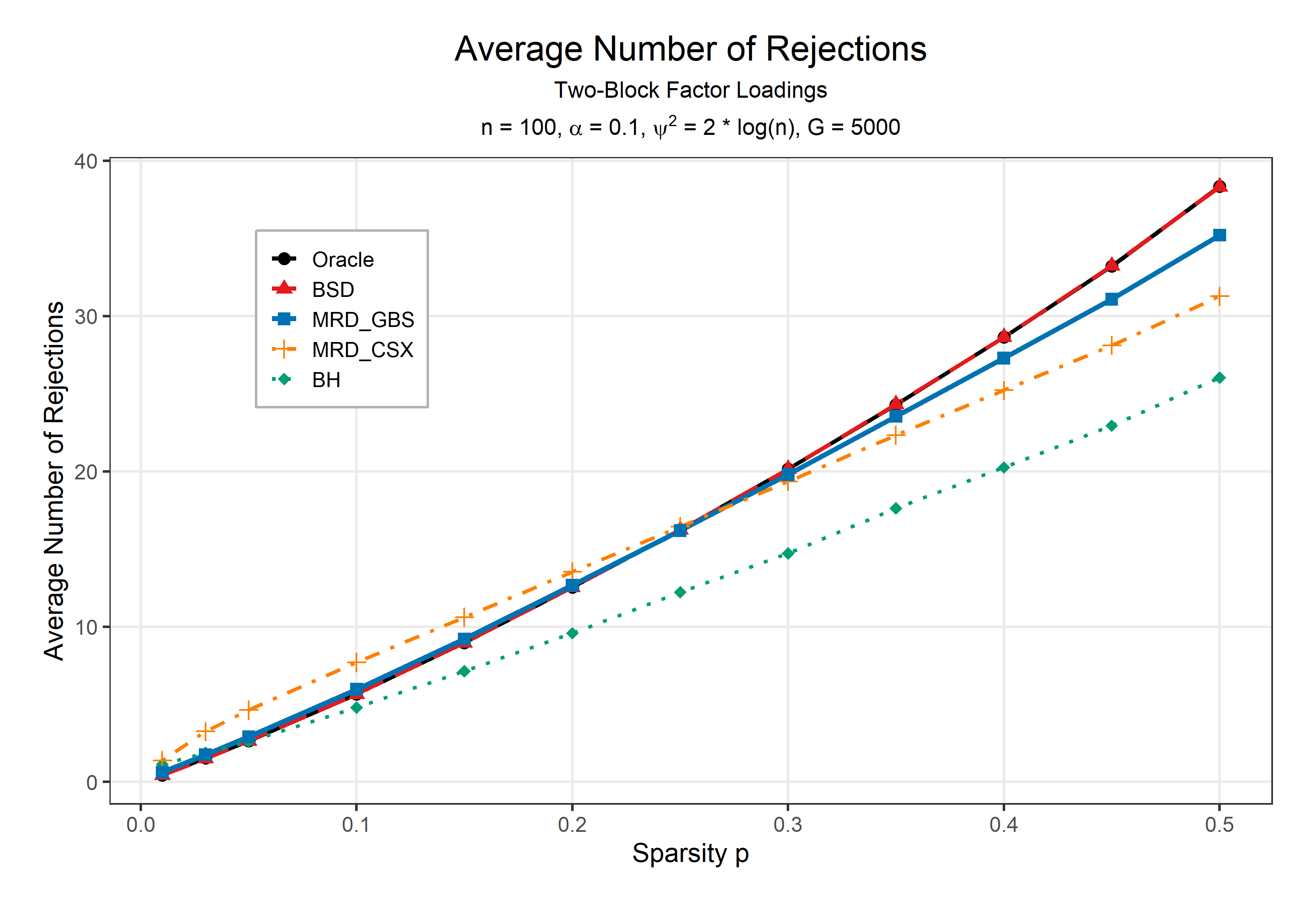}
		\caption{Two-Block Factor Dependence}
	\end{subfigure}
	
	\caption{Average number of rejections under the six one-factor dependence structures when $n=100$, \(\alpha=0.1\), \(\psi^2=2\log n\), and \(G=5000\) Monte Carlo replications.}
	\label{fig:anr-n100}
\end{figure}

Figures~\ref{fig:anr-n20}--\ref{fig:anr-n100} display the average numbers of rejections under the six one-factor dependence structures considered in this study. A striking feature of the results is the persistent agreement between BSD and the Bayes Oracle. Across all dimensions, sparsity levels, and dependence structures examined, the rejection curves of BSD are nearly indistinguishable from those of the Oracle. This finding complements the Bayes risk, FDR, FNR, and power results presented earlier and provides further evidence that BSD successfully reproduces the overall decision behavior of the Bayes-optimal rule under dependence.

The behavior of MRD--GBS is also noteworthy. Across most configurations, the procedure produces numbers of rejections that remain remarkably close to those of BSD and the Bayes Oracle despite arising from a fundamentally different inferential framework. In many settings, the rejection curves of MRD--GBS closely parallel those of BSD, suggesting that the covariance-aware search mechanism underlying the procedure successfully identifies a similar collection of signals while maintaining stable false discovery behavior.

\subsection*{Summary of Simulation Findings}

Taken together, the simulation results reveal a remarkably consistent pattern. Across all six one-factor dependence structures, a broad range of sparsity levels, and dimensions ranging from $n=20$ to $n=100$, the BSD procedure exhibits operating characteristics that are nearly indistinguishable from those of the Bayes Oracle. The close agreement is evident not only in terms of Bayes misclassification risk, but also across the individual components of testing performance, including FDR, FNR, power, and the average number of rejections. These findings suggest that the posterior model-pursuit mechanism underlying BSD is able to recover essentially the same signal-selection behavior as the Bayes-optimal decision rule while remaining computationally tractable under dependence.

An equally noteworthy finding is the strong performance of the MRD--GBS procedure. Despite arising from a fundamentally different inferential framework, MRD--GBS consistently produces operating characteristics that remain remarkably close to those of BSD and the Bayes Oracle across most configurations considered. Collectively, these observations suggest that effective utilization of covariance information may play a more fundamental role in sparse signal recovery under dependence than the Bayesian--frequentist distinction itself. They also indicate that dependence-aware sequential information accumulation, whether implemented through posterior-guided model pursuit or covariance-adaptive residual search, can lead to performance remarkably close to the Bayes-optimal frontier in sparse multiple testing problems.

\section{Discussion}
\label{sec:DISCUSSION}

The developments of this paper reveal that BSD possesses a considerably richer structure than what is apparent from its original posterior-odds formulation. While the procedure was initially motivated as a computationally feasible Bayesian alternative to the Bayes Oracle under dependence, the present analysis shows that several of its statistical, decision-theoretic, and computational properties are closely interconnected. In particular, the model-pursuit interpretation developed in Section~\ref{sec:BSD_METHOD}, together with the residual representation established in Section~\ref{sec:THEORETICAL_PROPERTIES}, provides a unified perspective from which the behavior, admissibility, and computational scalability of the procedure can be understood. Viewed in this way, BSD is not merely a Bayesian multiple testing algorithm, but a dependence-aware sequential search mechanism whose properties are shaped by the covariance-adaptive residual geometry underlying the procedure.

A central theme emerging from the present work is the interpretation of BSD as a posterior-guided model-pursuit strategy for sparse signal discovery under dependence. This perspective is conceptually related to the broader Bayesian model selection and stochastic search literature \citep{GM1993,GM1997,OS2009}, although BSD is developed in the distinct context of dependent multiple testing and sparse signal recovery. Rather than attempting to evaluate posterior probabilities over the entire configuration space simultaneously, BSD sequentially explores candidate sparse configurations by repeatedly extending the current active configuration using local posterior evidence. This perspective provides a natural explanation for both the statistical behavior and the computational feasibility of the procedure. In sparse settings, where substantial posterior mass is often concentrated on a relatively small collection of signal configurations, such a search strategy focuses computational effort on regions of the model space most likely to contain the underlying signal configuration. Consequently, BSD may be viewed as an adaptive mechanism for navigating the configuration space, balancing statistical evidence accumulation and computational tractability throughout the search process.

A second theme emerging from the present analysis is the central role of covariance-adaptive residual geometry. The representation theorem established in Section~\ref{sec:THEORETICAL_PROPERTIES} shows that the Bayesian posterior-odds statistics driving BSD are completely characterized by the corresponding MRD residual quantities. Consequently, the accumulation of posterior evidence throughout the BSD search process admits an equivalent geometric interpretation in terms of the evolution of covariance-adjusted residual systems. This connection reveals that Bayesian model pursuit and residual-based multiple testing are far more closely related than might initially appear, with the residual structure providing the geometric mechanism through which dependence information is incorporated into the sequential search process.

The present analysis also clarifies the precise nature of the relationship between BSD and MRD. Earlier interpretations suggested that BSD could be viewed as a monotone transformation of the corresponding MRD residual statistics. The results obtained here reveal a more nuanced picture. The transformation linking BSD and MRD depends explicitly on both the stage of the procedure and the accumulated active configuration through the adaptive residual variances. Consequently, BSD is more appropriately viewed as a member of a broader class of locally monotone residual-based procedures whose scoring rules evolve throughout the search process. This local monotonicity perspective provides a more refined description of the structural relationship between Bayesian posterior evidence and covariance-adaptive residual geometry.

The local monotone structure identified in Section~\ref{sec:THEORETICAL_PROPERTIES} also provides a natural explanation for the admissibility of BSD under arbitrary covariance dependence. Although BSD was originally derived from Bayesian posterior probability calculations, its admissibility does not depend directly on the Bayesian formulation itself. Rather, admissibility emerges as a consequence of the ordering structure induced by the covariance-adaptive residual system. Viewed from this perspective, the Bayesian and decision-theoretic aspects of BSD are closely intertwined: the posterior-guided search mechanism determines how evidence is accumulated, while the residual geometry underlying that mechanism governs the resulting admissibility properties. The present work therefore illustrates how Bayesian evidence accumulation and frequentist decision-theoretic optimality can coexist within a common dependence-aware testing framework.

The computational implications of these structural connections are equally noteworthy. The Bayes Oracle requires exploration of an exponentially large configuration space containing $2^n$ possible signal configurations, rendering exact implementation infeasible except for very small problems. BSD avoids this combinatorial explosion by replacing exhaustive posterior exploration with a sequential model-pursuit strategy requiring only a quadratic number of local posterior comparisons. Furthermore, through the BSD--MRD connection established in this paper, BSD inherits the computational simplifications available for the corresponding residual system. In particular, recent work by \citet{ghosh2026covariance} shows that the MRD residual statistics admit an alternative precision-matrix representation that substantially reduces the computational burden associated with residual evaluation. Thus, the computational scalability of BSD emerges not from algorithmic approximation, but from a sequence of structural simplifications rooted in the statistical formulation itself.

It is worth emphasizing that the present paper is not primarily concerned with empirical Bayes estimation of the hyperparameters governing the underlying sparse signal model. Rather, our objective is to understand the statistical, geometric, decision-theoretic, and computational properties of BSD and to investigate its relationship to the Bayes Oracle under covariance dependence. From this perspective, BSD serves both as a dependence-aware Bayesian model-pursuit procedure and as a useful benchmark for studying the attainable limits of sparse signal recovery under dependence. The development of robust data-driven strategies for estimating the underlying sparsity and signal-strength parameters, and the investigation of their impact on the resulting search dynamics, constitute important directions for future research.

It is also important to distinguish between the role of the Bayes Oracle and that of BSD. The Oracle serves primarily as an idealized benchmark for assessing the attainable limits of sparse signal recovery under dependence. Although Oracle risks can be evaluated under the controlled simulation settings considered in this paper, such calculations rely on complete knowledge of the underlying data-generating mechanism and are not generally available in practice. The value of BSD, therefore, lies not in reproducing an implementable Oracle procedure, but in providing a computationally tractable, dependence-aware search mechanism whose behavior can be studied relative to an ideal benchmark. Indeed, the primary purpose of the Oracle in the present work is to provide a benchmark for quantifying the extent to which a computationally tractable dependence-aware procedure can approach Bayes-optimal sparse signal recovery. From this perspective, the Oracle comparisons reported in this paper should be viewed as a tool for understanding the effectiveness of dependence-aware model pursuit rather than as a practical alternative to BSD itself.

One of the most striking empirical findings of this paper is the extent to which BSD approximates the Bayes Oracle under dependence. Across all one-factor structures considered, the discrepancies between BSD and the Oracle remain negligible relative to the overall Bayes risk, and the two procedures exhibit remarkably similar operating characteristics in terms of FDR, FNR, power, and support recovery. Although the present work does not provide a theoretical explanation for this phenomenon, the results suggest that the sequential posterior-guided search mechanism underlying BSD is able to capture a substantial portion of the information utilized by the Bayes-optimal decision rule. Understanding the theoretical origins of this near-oracle behavior remains an important direction for future research.

Equally noteworthy is the close empirical agreement among BSD, MRD--GBS, and the Bayes Oracle. Across a broad collection of one-factor dependence structures, dimensions, and sparsity regimes, BSD and MRD--GBS frequently exhibit remarkably similar operating characteristics and often achieve performance that is strikingly close to that of the Oracle itself. This observation is particularly intriguing because the two procedures originate from fundamentally different inferential paradigms. BSD arises from a Bayesian model-pursuit framework motivated by posterior evidence accumulation, whereas MRD--GBS was developed from a frequentist perspective and was not derived through Bayes-risk minimization considerations. To our knowledge, documented examples of computationally tractable frequentist multiple-testing procedures exhibiting such behavior under dependent sparse-signal models remain relatively limited. Taken together, these findings suggest that the effective exploitation of covariance information may play a more fundamental role in sparse signal recovery than the Bayesian--frequentist distinction itself. More broadly, they indicate that dependence-aware sequential information accumulation may constitute a key ingredient underlying successful large-scale inference under dependence and raise the possibility that covariance-adaptive residual-based procedures possess substantially stronger optimality properties than are presently understood.

Several interesting directions emerge from the present work. From a methodological perspective, it would be of interest to develop calibration strategies that explicitly account for covariance dependence within the model-pursuit framework. Such investigations may lead to more refined notions of multiplicity under dependence and a better understanding of how dependence influences the evolution of active configurations during sequential testing. From a theoretical standpoint, the model-pursuit interpretation developed here raises broader questions concerning information accumulation, evidence propagation, and testing complexity in high-dimensional dependent inference problems. A particularly interesting theoretical question raised by the present work concerns the near-oracle behavior of BSD. While the simulation results demonstrate a consistently close agreement between BSD and the Bayes Oracle across a broad collection of sparse dependence structures, a rigorous theoretical explanation for this phenomenon remains unavailable. Establishing conditions under which posterior-guided model-pursuit procedures can asymptotically approximate Bayes-optimal sparse signal recovery rules under dependence represents an important open problem. The striking empirical agreement between MRD--GBS and the Bayes Oracle raises the possibility that asymptotic Bayes-optimality under dependence may be attainable by suitably calibrated covariance-adaptive residual procedures. Establishing rigorous theoretical foundations for this phenomenon remains an important open problem for future research. 
The present work assumes that the covariance structure is known. While this assumption is standard in much of the dependent multiple-testing literature, extending Bayesian model-pursuit procedures to settings with unknown dependence remains an important challenge. A substantial literature now exists on high-dimensional covariance and precision-matrix estimation, including shrinkage estimators, thresholding methods, and sparse graphical modeling approaches \citep{LedoitWolf2004,BickelLevina2008,CaiLiu2011,FriedmanHastieTibshirani2008}. In such settings, estimation of the covariance structure is intrinsically intertwined with estimation of sparsity and signal-strength characteristics, and the impact of these uncertainties on posterior-guided sequential search is currently not well understood. It would also be of considerable interest to extend the present framework beyond the Gaussian setting and to investigate whether similar residual-geometric structures arise in more general classes of dependence models. We hope that the perspective developed in this paper provides a useful foundation for such future investigations.


\appendix  
\section*{Appendix}
\label{app}

\section{Proofs of the Main Theoretical Results}

%
%

%
%
%

This appendix contains proofs of the principal theoretical results established in the paper. In particular, we provide proofs of Theorems~\ref{THM:BSD_MRD_CONNECTION} and \ref{THM:BSD_Admissibility}.

\begin{flushleft}
	\textbf{Proof of Theorem~\ref{THM:BSD_MRD_CONNECTION}}
\end{flushleft}

\begin{proof}
	Fix a stage \(t\in\{1,\ldots,n\}\) and a conditioning history
	\((j_1,\ldots,j_{t-1})\). By definition,
	\begin{align}
		S_{tj}^{(j_1,\ldots,j_{t-1})}(\mathbf X)
		=
		\frac{
			\pi(\nu_j=1,\boldsymbol{\nu}^{(j_1,\ldots,j_{t-1},j)}=\mathbf 0)
			f(\mathbf X^{(j_1,\ldots,j_{t-1})}
			\mid
			\nu_j=1,
			\boldsymbol{\nu}^{(j_1,\ldots,j_{t-1},j)}=\mathbf 0)
		}{
			\pi(\nu_j=0,\boldsymbol{\nu}^{(j_1,\ldots,j_{t-1},j)}=\mathbf 0)
			f(\mathbf X^{(j_1,\ldots,j_{t-1})}
			\mid
			\nu_j=0,
			\boldsymbol{\nu}^{(j_1,\ldots,j_{t-1},j)}=\mathbf 0)
		}.
		\label{BSD_MRD_RELATION_PROOF_1}
	\end{align}
	
	Since the prior indicators are independent Bernoulli\((p)\),
	\[
	\frac{
		\pi(\nu_j=1,\boldsymbol{\nu}^{(j_1,\ldots,j_{t-1},j)}=\mathbf 0)
	}{
		\pi(\nu_j=0,\boldsymbol{\nu}^{(j_1,\ldots,j_{t-1},j)}=\mathbf 0)
	}
	=
	\frac{p}{1-p}.
	\]
	
	Therefore,
	\begin{align}
		S_{tj}^{(j_1,\ldots,j_{t-1})}(\mathbf X)
		=
		\frac{p}{1-p}
		\,
		\frac{
			f(\mathbf X^{(j_1,\ldots,j_{t-1})}
			\mid
			\nu_j=1,
			\boldsymbol{\nu}^{(j_1,\ldots,j_{t-1},j)}=\mathbf 0)
		}{
			f(\mathbf X^{(j_1,\ldots,j_{t-1})}
			\mid
			\nu_j=0,
			\boldsymbol{\nu}^{(j_1,\ldots,j_{t-1},j)}=\mathbf 0)
		}.
		\label{BSD_MRD_RELATION_PROOF_2}
	\end{align}
	
	Under both models appearing in
	\eqref{BSD_MRD_RELATION_PROOF_2},
	all coordinates other than the \(j\)th coordinate are null.
	Consequently, the two distributions differ only through the variance
	of the \(j\)th coordinate.
	
	Using the standard conditional decomposition of a multivariate normal distribution, we obtain
	\begin{align}
		&
		f(\mathbf X^{(j_1,\ldots,j_{t-1})}
		\mid
		\nu_j=0,
		\boldsymbol{\nu}^{(j_1,\ldots,j_{t-1},j)}=\mathbf 0)
		\nonumber\\
		&=
		N\!\left(
		u_{j\cdot(j_1,\ldots,j_{t-1})}(\mathbf X),
		\sigma_{j\cdot(j_1,\ldots,j_{t-1})}
		\right)(X_j)
		\nonumber\\
		&\qquad\times
		N\!\left(
		\mathbf 0,
		\Sigma_{(j_1,\ldots,j_{t-1},j)}
		\right)
		\left(
		\mathbf X^{(j_1,\ldots,j_{t-1},j)}
		\right),
		\label{BSD_MRD_RELATION_PROOF_3}
	\end{align}
	where
	\[
	u_{j\cdot(j_1,\ldots,j_{t-1})}(\mathbf X)
	=
	\bigl(
	\boldsymbol{\sigma}_{(j)}^{(j_1,\ldots,j_{t-1})}
	\bigr)^T
	\Sigma_{(j_1,\ldots,j_{t-1},j)}^{-1}
	\mathbf X^{(j_1,\ldots,j_{t-1},j)}.
	\]
	
	Similarly,
	\begin{align}
		&
		f(\mathbf X^{(j_1,\ldots,j_{t-1})}
		\mid
		\nu_j=1,
		\boldsymbol{\nu}^{(j_1,\ldots,j_{t-1},j)}=\mathbf 0)
		\nonumber\\
		&=
		N\!\left(
		u_{j\cdot(j_1,\ldots,j_{t-1})}(\mathbf X),
		\psi^2+\sigma_{j\cdot(j_1,\ldots,j_{t-1})}
		\right)(X_j)
		\nonumber\\
		&\qquad\times
		N\!\left(
		\mathbf 0,
		\Sigma_{(j_1,\ldots,j_{t-1},j)}
		\right)
		\left(
		\mathbf X^{(j_1,\ldots,j_{t-1},j)}
		\right).
		\label{BSD_MRD_RELATION_PROOF_4}
	\end{align}
	
	Substituting \eqref{BSD_MRD_RELATION_PROOF_3} and
	\eqref{BSD_MRD_RELATION_PROOF_4}
	into \eqref{BSD_MRD_RELATION_PROOF_2},
	the common density
	\[
	N\!\left(
	\mathbf 0,
	\Sigma_{(j_1,\ldots,j_{t-1},j)}
	\right)
	\left(
	\mathbf X^{(j_1,\ldots,j_{t-1},j)}
	\right)
	\]
	cancels from the numerator and denominator.
	A straightforward calculation using the normal density formula then yields
	\[
	S_{tj}^{(j_1,\ldots,j_{t-1})}(\mathbf X)
	=
	\frac{p}{1-p}
	\sqrt{
		\frac{
			\sigma_{j\cdot(j_1,\ldots,j_{t-1})}
		}{
			\psi^2+\sigma_{j\cdot(j_1,\ldots,j_{t-1})}
		}
	}
	\exp\!\left\{
	\frac{\psi^2}
	{2\{\psi^2+\sigma_{j\cdot(j_1,\ldots,j_{t-1})}\}}
	\left(
	U_{tj}^{(j_1,\ldots,j_{t-1})}(\mathbf X)
	\right)^2
	\right\},
	\]
	which coincides with the representation stated in the theorem.
	This completes the proof of Theorem~\ref{THM:BSD_MRD_CONNECTION}.
\end{proof}

\begin{flushleft}
	\textbf{Proof of Theorem~\ref{THM:BSD_Admissibility}}
\end{flushleft}

\begin{proof}
Fix a stage \(t\in\{1,\ldots,n\}\), a conditioning history
	\((j_1,\ldots,j_{t-1})\), and a candidate coordinate
	\(j\notin\{j_1,\ldots,j_{t-1}\}\). By Theorem~\ref{THM:BSD_MRD_CONNECTION},
	\[
	S_{tj}^{(j_1,\ldots,j_{t-1})}(\mathbf X)
	=
	\frac{p}{1-p}
	\sqrt{
		\frac{
			\sigma_{j\cdot(j_1,\ldots,j_{t-1})}
		}{
			\psi^2+\sigma_{j\cdot(j_1,\ldots,j_{t-1})}
		}
	}
	\exp\!\left\{
	\frac{\psi^2}
	{2\{\psi^2+\sigma_{j\cdot(j_1,\ldots,j_{t-1})}\}}
	\left(
	U_{tj}^{(j_1,\ldots,j_{t-1})}(\mathbf X)
	\right)^2
	\right\}.
	\]
	
	For the fixed stage, history, and coordinate under consideration, define
	\[
	A_{tj}^{(j_1,\ldots,j_{t-1})}
	=
	\frac{p}{1-p}
	\sqrt{
		\frac{
			\sigma_{j\cdot(j_1,\ldots,j_{t-1})}
		}{
			\psi^2+\sigma_{j\cdot(j_1,\ldots,j_{t-1})}
		}
	}
	\]
	and
	\[
	B_{tj}^{(j_1,\ldots,j_{t-1})}
	=
	\frac{\psi^2}
	{2\{\psi^2+\sigma_{j\cdot(j_1,\ldots,j_{t-1})}\}}.
	\]
	Since \(p\in(0,1)\), \(\psi^2>0\), and
	\(\sigma_{j\cdot(j_1,\ldots,j_{t-1})}>0\), we have
	\[
	A_{tj}^{(j_1,\ldots,j_{t-1})}>0,
	\qquad
	B_{tj}^{(j_1,\ldots,j_{t-1})}>0.
	\]
	
	Consequently,
	\[
	S_{tj}^{(j_1,\ldots,j_{t-1})}(\mathbf X)
	=
	A_{tj}^{(j_1,\ldots,j_{t-1})}
	\exp\!\left\{
	B_{tj}^{(j_1,\ldots,j_{t-1})}
	\left(
	U_{tj}^{(j_1,\ldots,j_{t-1})}(\mathbf X)
	\right)^2
	\right\}.
	\]
	Thus, for every fixed stage \(t\), conditioning history
	\((j_1,\ldots,j_{t-1})\), and candidate coordinate \(j\), the BSD statistic
	\(S_{tj}^{(j_1,\ldots,j_{t-1})}(\mathbf X)\) is a strictly increasing function of the magnitude of the corresponding MRD statistics, given by
	\[
	\left|U_{tj}^{(j_1,\ldots,j_{t-1})}(\mathbf X)\right|.
	\]
	
	It follows that the BSD procedure is a locally monotone residual-based
	step-down procedure in the sense of
	Lemma~\ref{lem:LEM_ADMISSIBILITY_MONOTONE_TRANSFORM}. Therefore, by
	Lemma~\ref{lem:LEM_ADMISSIBILITY_MONOTONE_TRANSFORM}, which invokes the
	general admissibility theorem of \citet{GC_ADMISSIBILITY_2026}, the BSD
	procedure is admissible with respect to the vector loss function
	\eqref{eq:VECTOR_LOSS}. This completes the proof of Theorem~\ref{THM:BSD_Admissibility}.
\end{proof}

\section{Additional Simulation Results Under One-factor Dependence}

This appendix provides a comprehensive numerical summary of the simulation experiments presented in the main simulation study. For each dependence structure, sample size, and sparsity level considered, we report the Bayes misclassification risk, relative Bayes misclassification risk, false discovery rate, false non-discovery rate, statistical power, and average number of rejections for all competing procedures. Results are presented for \(n=20, 50,\) and \(100\), sparsity levels ranging from very sparse to moderately dense regimes, nominal level \(\alpha=0.10\), signal strength \(\psi^2 = 2\log n\), and \(G=5000\) Monte Carlo replications. For ease of reference, the tables are organized by performance metric, with results grouped according to the underlying dependence structure and sparsity level. Together, these tables provide the detailed numerical evidence supporting the graphical summaries and principal conclusions presented in the main text.


	\subsection*{Bayes Misclassification Risks}
	
	Table~\ref{tab:appendix-risk-onefactor} reports the Bayes misclassification risks of the competing procedures across all one-factor dependence structures, sample sizes, and sparsity levels considered in the main simulation study. These values provide the complete numerical counterpart to the risk curves displayed in Figures~\ref{fig:onefactor-risk-n20}--\ref{fig:onefactor-risk-n100} and serve as the primary quantitative summary of overall decision performance under the additive misclassification loss.
	
	\small
	\setlength{\tabcolsep}{3.5pt}
	\renewcommand{\arraystretch}{1.08}
	\begin{longtable}{cc l rrrrr}
		\caption{Bayes misclassification risks of the competing procedures under the six one-factor dependence structures for $n=20,50,100$, $\alpha=0.10$, $\psi^2=2\log n$, and $G=5000$ Monte Carlo replications.}\label{tab:appendix-risk-onefactor}\\
		\toprule
		$n$ & $p$ & Dependence structure & Oracle & BSD & MRD--GBS & MRD--CSX & BH \\
		\midrule
		\endfirsthead
		\caption[]{Bayes misclassification risks of the competing procedures under the six one-factor dependence structures (continued).}\\
		\toprule
		$n$ & $p$ & Dependence structure & Oracle & BSD & MRD--GBS & MRD--CSX & BH \\
		\midrule
		\endhead
		\midrule
		\multicolumn{8}{r}{\emph{Continued on next page}}\\
		\endfoot
		\bottomrule
		\endlastfoot
		20 & 0.03 & Independence & 0.0240 & 0.0240 & 0.0270 & 0.0351 & 0.0273 \\
		& 0.05 &  & 0.0390 & 0.0390 & 0.0410 & 0.0498 & 0.0413 \\
		& 0.10 &  & 0.0743 & 0.0743 & 0.0754 & 0.0861 & 0.0755 \\
		& 0.15 &  & 0.1077 & 0.1077 & 0.1090 & 0.1168 & 0.1084 \\
		& 0.20 &  & 0.1407 & 0.1407 & 0.1438 & 0.1484 & 0.1432 \\
		& 0.25 &  & 0.1689 & 0.1689 & 0.1741 & 0.1740 & 0.1733 \\
		& 0.30 &  & 0.1985 & 0.1985 & 0.2070 & 0.2032 & 0.2061 \\
		& 0.35 &  & 0.2235 & 0.2235 & 0.2349 & 0.2273 & 0.2347 \\
		& 0.40 &  & 0.2447 & 0.2447 & 0.2608 & 0.2496 & 0.2620 \\
		& 0.45 &  & 0.2667 & 0.2667 & 0.2896 & 0.2762 & 0.2918 \\
		& 0.50 &  & 0.2826 & 0.2826 & 0.3126 & 0.2973 & 0.3166 \\
		\addlinespace[2pt]
		& 0.03 & Positive factor & 0.0169 & 0.0169 & 0.0210 & 0.0326 & 0.0431 \\
		& 0.05 &  & 0.0268 & 0.0268 & 0.0304 & 0.0436 & 0.0576 \\
		& 0.10 &  & 0.0505 & 0.0506 & 0.0536 & 0.0693 & 0.0903 \\
		& 0.15 &  & 0.0722 & 0.0723 & 0.0748 & 0.0882 & 0.1203 \\
		& 0.20 &  & 0.0962 & 0.0964 & 0.0993 & 0.1093 & 0.1540 \\
		& 0.25 &  & 0.1157 & 0.1165 & 0.1188 & 0.1250 & 0.1837 \\
		& 0.30 &  & 0.1366 & 0.1370 & 0.1401 & 0.1427 & 0.2162 \\
		& 0.35 &  & 0.1542 & 0.1546 & 0.1573 & 0.1572 & 0.2440 \\
		& 0.40 &  & 0.1699 & 0.1719 & 0.1741 & 0.1726 & 0.2719 \\
		& 0.45 &  & 0.1868 & 0.1895 & 0.1934 & 0.1897 & 0.2993 \\
		& 0.50 &  & 0.2013 & 0.2063 & 0.2102 & 0.2065 & 0.3244 \\
		\addlinespace[2pt]
		& 0.03 & Alternating factor & 0.0170 & 0.0170 & 0.0204 & 0.0318 & 0.0436 \\
		& 0.05 &  & 0.0267 & 0.0267 & 0.0303 & 0.0430 & 0.0570 \\
		& 0.10 &  & 0.0506 & 0.0508 & 0.0538 & 0.0690 & 0.0902 \\
		& 0.15 &  & 0.0728 & 0.0729 & 0.0746 & 0.0886 & 0.1223 \\
		& 0.20 &  & 0.0968 & 0.0972 & 0.1004 & 0.1101 & 0.1556 \\
		& 0.25 &  & 0.1164 & 0.1166 & 0.1189 & 0.1248 & 0.1837 \\
		& 0.30 &  & 0.1369 & 0.1370 & 0.1395 & 0.1427 & 0.2166 \\
		& 0.35 &  & 0.1539 & 0.1554 & 0.1577 & 0.1585 & 0.2429 \\
		& 0.40 &  & 0.1696 & 0.1719 & 0.1745 & 0.1727 & 0.2715 \\
		& 0.45 &  & 0.1874 & 0.1898 & 0.1936 & 0.1903 & 0.2996 \\
		& 0.50 &  & 0.2010 & 0.2051 & 0.2087 & 0.2050 & 0.3234 \\
		\addlinespace[2pt]
		& 0.03 & Linear loadings & 0.0205 & 0.0205 & 0.0240 & 0.0341 & 0.0333 \\
		& 0.05 &  & 0.0332 & 0.0332 & 0.0359 & 0.0466 & 0.0479 \\
		& 0.10 &  & 0.0624 & 0.0625 & 0.0642 & 0.0769 & 0.0804 \\
		& 0.15 &  & 0.0895 & 0.0896 & 0.0911 & 0.1022 & 0.1131 \\
		& 0.20 &  & 0.1186 & 0.1188 & 0.1210 & 0.1290 & 0.1465 \\
		& 0.25 &  & 0.1425 & 0.1430 & 0.1457 & 0.1495 & 0.1773 \\
		& 0.30 &  & 0.1681 & 0.1692 & 0.1730 & 0.1733 & 0.2100 \\
		& 0.35 &  & 0.1888 & 0.1899 & 0.1951 & 0.1924 & 0.2376 \\
		& 0.40 &  & 0.2091 & 0.2106 & 0.2173 & 0.2116 & 0.2662 \\
		& 0.45 &  & 0.2291 & 0.2313 & 0.2407 & 0.2334 & 0.2953 \\
		& 0.50 &  & 0.2439 & 0.2495 & 0.2619 & 0.2546 & 0.3199 \\
		\addlinespace[2pt]
		& 0.03 & Sparse factor & 0.0221 & 0.0221 & 0.0253 & 0.0342 & 0.0282 \\
		& 0.05 &  & 0.0361 & 0.0361 & 0.0386 & 0.0481 & 0.0425 \\
		& 0.10 &  & 0.0680 & 0.0681 & 0.0699 & 0.0817 & 0.0761 \\
		& 0.15 &  & 0.0988 & 0.0989 & 0.1004 & 0.1097 & 0.1095 \\
		& 0.20 &  & 0.1303 & 0.1307 & 0.1328 & 0.1389 & 0.1434 \\
		& 0.25 &  & 0.1562 & 0.1566 & 0.1596 & 0.1619 & 0.1743 \\
		& 0.30 &  & 0.1840 & 0.1850 & 0.1905 & 0.1879 & 0.2067 \\
		& 0.35 &  & 0.2082 & 0.2092 & 0.2164 & 0.2117 & 0.2353 \\
		& 0.40 &  & 0.2289 & 0.2310 & 0.2423 & 0.2338 & 0.2629 \\
		& 0.45 &  & 0.2496 & 0.2530 & 0.2690 & 0.2592 & 0.2924 \\
		& 0.50 &  & 0.2665 & 0.2716 & 0.2914 & 0.2816 & 0.3176 \\
		\addlinespace[2pt]
		& 0.03 & Two-block factor & 0.0199 & 0.0199 & 0.0236 & 0.0334 & 0.0328 \\
		& 0.05 &  & 0.0321 & 0.0321 & 0.0351 & 0.0462 & 0.0463 \\
		& 0.10 &  & 0.0607 & 0.0610 & 0.0635 & 0.0769 & 0.0806 \\
		& 0.15 &  & 0.0876 & 0.0878 & 0.0897 & 0.1001 & 0.1127 \\
		& 0.20 &  & 0.1162 & 0.1166 & 0.1191 & 0.1273 & 0.1465 \\
		& 0.25 &  & 0.1389 & 0.1395 & 0.1425 & 0.1459 & 0.1766 \\
		& 0.30 &  & 0.1645 & 0.1650 & 0.1687 & 0.1689 & 0.2092 \\
		& 0.35 &  & 0.1855 & 0.1867 & 0.1909 & 0.1885 & 0.2375 \\
		& 0.40 &  & 0.2040 & 0.2050 & 0.2119 & 0.2065 & 0.2658 \\
		& 0.45 &  & 0.2233 & 0.2271 & 0.2370 & 0.2295 & 0.2943 \\
		& 0.50 &  & 0.2413 & 0.2451 & 0.2553 & 0.2497 & 0.3206 \\
		\addlinespace[2pt]
		\midrule
		50 & 0.02 & Independence & 0.0155 & 0.0155 & 0.0167 & 0.0236 & 0.0169 \\
		& 0.05 &  & 0.0370 & 0.0370 & 0.0376 & 0.0471 & 0.0377 \\
		& 0.10 &  & 0.0704 & 0.0704 & 0.0711 & 0.0784 & 0.0710 \\
		& 0.15 &  & 0.1012 & 0.1012 & 0.1031 & 0.1065 & 0.1029 \\
		& 0.20 &  & 0.1276 & 0.1276 & 0.1304 & 0.1300 & 0.1300 \\
		& 0.25 &  & 0.1558 & 0.1558 & 0.1595 & 0.1564 & 0.1595 \\
		& 0.30 &  & 0.1827 & 0.1827 & 0.1880 & 0.1829 & 0.1885 \\
		& 0.35 &  & 0.2048 & 0.2048 & 0.2125 & 0.2067 & 0.2140 \\
		& 0.40 &  & 0.2253 & 0.2253 & 0.2369 & 0.2312 & 0.2396 \\
		& 0.45 &  & 0.2429 & 0.2429 & 0.2596 & 0.2558 & 0.2643 \\
		& 0.50 &  & 0.2607 & 0.2607 & 0.2817 & 0.2806 & 0.2883 \\
		\addlinespace[2pt]
		& 0.02 & Positive factor & 0.0104 & 0.0104 & 0.0122 & 0.0222 & 0.0323 \\
		& 0.05 &  & 0.0248 & 0.0248 & 0.0265 & 0.0398 & 0.0521 \\
		& 0.10 &  & 0.0468 & 0.0464 & 0.0482 & 0.0590 & 0.0840 \\
		& 0.15 &  & 0.0672 & 0.0671 & 0.0690 & 0.0761 & 0.1151 \\
		& 0.20 &  & 0.0838 & 0.0836 & 0.0859 & 0.0895 & 0.1429 \\
		& 0.25 &  & 0.1024 & 0.1025 & 0.1042 & 0.1054 & 0.1718 \\
		& 0.30 &  & 0.1199 & 0.1202 & 0.1216 & 0.1209 & 0.1997 \\
		& 0.35 &  & 0.1350 & 0.1352 & 0.1367 & 0.1352 & 0.2232 \\
		& 0.40 &  & 0.1490 & 0.1497 & 0.1517 & 0.1500 & 0.2483 \\
		& 0.45 &  & 0.1635 & 0.1643 & 0.1659 & 0.1655 & 0.2701 \\
		& 0.50 &  & 0.1754 & 0.1771 & 0.1790 & 0.1800 & 0.2954 \\
		\addlinespace[2pt]
		& 0.02 & Alternating factor & 0.0106 & 0.0105 & 0.0121 & 0.0220 & 0.0330 \\
		& 0.05 &  & 0.0250 & 0.0249 & 0.0265 & 0.0399 & 0.0528 \\
		& 0.10 &  & 0.0467 & 0.0463 & 0.0483 & 0.0589 & 0.0838 \\
		& 0.15 &  & 0.0671 & 0.0670 & 0.0688 & 0.0758 & 0.1154 \\
		& 0.20 &  & 0.0837 & 0.0837 & 0.0858 & 0.0897 & 0.1439 \\
		& 0.25 &  & 0.1024 & 0.1025 & 0.1038 & 0.1053 & 0.1716 \\
		& 0.30 &  & 0.1195 & 0.1198 & 0.1213 & 0.1207 & 0.1997 \\
		& 0.35 &  & 0.1352 & 0.1356 & 0.1371 & 0.1357 & 0.2227 \\
		& 0.40 &  & 0.1491 & 0.1496 & 0.1513 & 0.1495 & 0.2480 \\
		& 0.45 &  & 0.1632 & 0.1642 & 0.1655 & 0.1655 & 0.2691 \\
		& 0.50 &  & 0.1749 & 0.1764 & 0.1780 & 0.1791 & 0.2952 \\
		\addlinespace[2pt]
		& 0.02 & Linear loadings & 0.0130 & 0.0130 & 0.0143 & 0.0229 & 0.0226 \\
		& 0.05 &  & 0.0308 & 0.0308 & 0.0319 & 0.0432 & 0.0422 \\
		& 0.10 &  & 0.0580 & 0.0581 & 0.0591 & 0.0690 & 0.0754 \\
		& 0.15 &  & 0.0833 & 0.0834 & 0.0845 & 0.0903 & 0.1072 \\
		& 0.20 &  & 0.1046 & 0.1047 & 0.1061 & 0.1086 & 0.1350 \\
		& 0.25 &  & 0.1280 & 0.1283 & 0.1296 & 0.1295 & 0.1644 \\
		& 0.30 &  & 0.1503 & 0.1506 & 0.1526 & 0.1507 & 0.1929 \\
		& 0.35 &  & 0.1690 & 0.1695 & 0.1720 & 0.1696 & 0.2174 \\
		& 0.40 &  & 0.1868 & 0.1876 & 0.1909 & 0.1896 & 0.2431 \\
		& 0.45 &  & 0.2026 & 0.2038 & 0.2084 & 0.2085 & 0.2664 \\
		& 0.50 &  & 0.2173 & 0.2197 & 0.2258 & 0.2288 & 0.2916 \\
		\addlinespace[2pt]
		& 0.02 & Sparse factor & 0.0142 & 0.0142 & 0.0155 & 0.0233 & 0.0179 \\
		& 0.05 &  & 0.0338 & 0.0338 & 0.0348 & 0.0453 & 0.0385 \\
		& 0.10 &  & 0.0638 & 0.0639 & 0.0647 & 0.0733 & 0.0719 \\
		& 0.15 &  & 0.0923 & 0.0923 & 0.0938 & 0.0986 & 0.1032 \\
		& 0.20 &  & 0.1159 & 0.1160 & 0.1177 & 0.1189 & 0.1310 \\
		& 0.25 &  & 0.1420 & 0.1420 & 0.1442 & 0.1428 & 0.1607 \\
		& 0.30 &  & 0.1658 & 0.1662 & 0.1694 & 0.1666 & 0.1895 \\
		& 0.35 &  & 0.1867 & 0.1870 & 0.1918 & 0.1884 & 0.2145 \\
		& 0.40 &  & 0.2066 & 0.2067 & 0.2132 & 0.2107 & 0.2406 \\
		& 0.45 &  & 0.2235 & 0.2247 & 0.2347 & 0.2339 & 0.2647 \\
		& 0.50 &  & 0.2395 & 0.2406 & 0.2539 & 0.2559 & 0.2889 \\
		\addlinespace[2pt]
		& 0.02 & Two-block factor & 0.0127 & 0.0125 & 0.0138 & 0.0227 & 0.0222 \\
		& 0.05 &  & 0.0300 & 0.0299 & 0.0312 & 0.0432 & 0.0427 \\
		& 0.10 &  & 0.0564 & 0.0562 & 0.0573 & 0.0670 & 0.0756 \\
		& 0.15 &  & 0.0815 & 0.0817 & 0.0831 & 0.0891 & 0.1067 \\
		& 0.20 &  & 0.1026 & 0.1027 & 0.1041 & 0.1064 & 0.1351 \\
		& 0.25 &  & 0.1250 & 0.1251 & 0.1265 & 0.1267 & 0.1644 \\
		& 0.30 &  & 0.1459 & 0.1463 & 0.1484 & 0.1472 & 0.1927 \\
		& 0.35 &  & 0.1646 & 0.1648 & 0.1674 & 0.1659 & 0.2170 \\
		& 0.40 &  & 0.1824 & 0.1828 & 0.1856 & 0.1844 & 0.2427 \\
		& 0.45 &  & 0.1977 & 0.1983 & 0.2032 & 0.2040 & 0.2661 \\
		& 0.50 &  & 0.2128 & 0.2143 & 0.2207 & 0.2237 & 0.2915 \\
		\addlinespace[2pt]
		\midrule
		100 & 0.01 & Independence & 0.0079 & 0.0079 & 0.0085 & 0.0135 & 0.0086 \\
		& 0.03 &  & 0.0220 & 0.0220 & 0.0223 & 0.0311 & 0.0225 \\
		& 0.05 &  & 0.0353 & 0.0353 & 0.0357 & 0.0437 & 0.0356 \\
		& 0.10 &  & 0.0661 & 0.0661 & 0.0668 & 0.0713 & 0.0667 \\
		& 0.15 &  & 0.0949 & 0.0949 & 0.0961 & 0.0970 & 0.0960 \\
		& 0.20 &  & 0.1223 & 0.1223 & 0.1241 & 0.1225 & 0.1242 \\
		& 0.25 &  & 0.1478 & 0.1478 & 0.1505 & 0.1480 & 0.1512 \\
		& 0.30 &  & 0.1718 & 0.1718 & 0.1758 & 0.1736 & 0.1769 \\
		& 0.35 &  & 0.1934 & 0.1934 & 0.1991 & 0.1985 & 0.2012 \\
		& 0.40 &  & 0.2146 & 0.2146 & 0.2223 & 0.2244 & 0.2257 \\
		& 0.45 &  & 0.2319 & 0.2319 & 0.2431 & 0.2493 & 0.2487 \\
		& 0.50 &  & 0.2462 & 0.2462 & 0.2621 & 0.2736 & 0.2699 \\
		\addlinespace[2pt]
		& 0.01 & Positive factor & 0.0054 & 0.0052 & 0.0062 & 0.0138 & 0.0274 \\
		& 0.03 &  & 0.0148 & 0.0144 & 0.0154 & 0.0270 & 0.0394 \\
		& 0.05 &  & 0.0228 & 0.0228 & 0.0239 & 0.0347 & 0.0501 \\
		& 0.10 &  & 0.0428 & 0.0429 & 0.0443 & 0.0507 & 0.0824 \\
		& 0.15 &  & 0.0612 & 0.0612 & 0.0627 & 0.0656 & 0.1091 \\
		& 0.20 &  & 0.0789 & 0.0788 & 0.0804 & 0.0809 & 0.1354 \\
		& 0.25 &  & 0.0953 & 0.0953 & 0.0967 & 0.0959 & 0.1626 \\
		& 0.30 &  & 0.1109 & 0.1106 & 0.1122 & 0.1108 & 0.1848 \\
		& 0.35 &  & 0.1252 & 0.1253 & 0.1263 & 0.1256 & 0.2111 \\
		& 0.40 &  & 0.1399 & 0.1401 & 0.1410 & 0.1412 & 0.2347 \\
		& 0.45 &  & 0.1519 & 0.1521 & 0.1533 & 0.1557 & 0.2554 \\
		& 0.50 &  & 0.1632 & 0.1636 & 0.1641 & 0.1701 & 0.2757 \\
		\addlinespace[2pt]
		& 0.01 & Alternating factor & 0.0052 & 0.0052 & 0.0061 & 0.0139 & 0.0270 \\
		& 0.03 &  & 0.0144 & 0.0144 & 0.0154 & 0.0270 & 0.0394 \\
		& 0.05 &  & 0.0228 & 0.0228 & 0.0240 & 0.0345 & 0.0501 \\
		& 0.10 &  & 0.0429 & 0.0428 & 0.0443 & 0.0506 & 0.0826 \\
		& 0.15 &  & 0.0616 & 0.0613 & 0.0630 & 0.0658 & 0.1089 \\
		& 0.20 &  & 0.0788 & 0.0787 & 0.0802 & 0.0808 & 0.1358 \\
		& 0.25 &  & 0.0950 & 0.0950 & 0.0967 & 0.0957 & 0.1626 \\
		& 0.30 &  & 0.1111 & 0.1110 & 0.1123 & 0.1109 & 0.1858 \\
		& 0.35 &  & 0.1249 & 0.1250 & 0.1264 & 0.1252 & 0.2114 \\
		& 0.40 &  & 0.1395 & 0.1398 & 0.1405 & 0.1412 & 0.2343 \\
		& 0.45 &  & 0.1517 & 0.1522 & 0.1528 & 0.1558 & 0.2552 \\
		& 0.50 &  & 0.1631 & 0.1638 & 0.1647 & 0.1698 & 0.2755 \\
		\addlinespace[2pt]
		& 0.01 & Linear loadings & 0.0067 & 0.0066 & 0.0073 & 0.0137 & 0.0150 \\
		& 0.03 &  & 0.0183 & 0.0181 & 0.0187 & 0.0290 & 0.0284 \\
		& 0.05 &  & 0.0288 & 0.0288 & 0.0294 & 0.0392 & 0.0410 \\
		& 0.10 &  & 0.0541 & 0.0539 & 0.0546 & 0.0604 & 0.0723 \\
		& 0.15 &  & 0.0774 & 0.0775 & 0.0779 & 0.0803 & 0.1010 \\
		& 0.20 &  & 0.0991 & 0.0991 & 0.0998 & 0.1001 & 0.1282 \\
		& 0.25 &  & 0.1202 & 0.1203 & 0.1209 & 0.1204 & 0.1559 \\
		& 0.30 &  & 0.1394 & 0.1394 & 0.1405 & 0.1401 & 0.1797 \\
		& 0.35 &  & 0.1575 & 0.1576 & 0.1587 & 0.1598 & 0.2056 \\
		& 0.40 &  & 0.1751 & 0.1756 & 0.1774 & 0.1805 & 0.2296 \\
		& 0.45 &  & 0.1896 & 0.1903 & 0.1929 & 0.1993 & 0.2514 \\
		& 0.50 &  & 0.2030 & 0.2038 & 0.2072 & 0.2184 & 0.2719 \\
		\addlinespace[2pt]
		& 0.01 & Sparse factor & 0.0071 & 0.0071 & 0.0078 & 0.0137 & 0.0100 \\
		& 0.03 &  & 0.0200 & 0.0198 & 0.0203 & 0.0302 & 0.0237 \\
		& 0.05 &  & 0.0317 & 0.0317 & 0.0321 & 0.0413 & 0.0367 \\
		& 0.10 &  & 0.0594 & 0.0594 & 0.0601 & 0.0655 & 0.0679 \\
		& 0.15 &  & 0.0852 & 0.0852 & 0.0858 & 0.0879 & 0.0971 \\
		& 0.20 &  & 0.1100 & 0.1100 & 0.1112 & 0.1107 & 0.1255 \\
		& 0.25 &  & 0.1330 & 0.1330 & 0.1344 & 0.1333 & 0.1524 \\
		& 0.30 &  & 0.1544 & 0.1545 & 0.1566 & 0.1557 & 0.1774 \\
		& 0.35 &  & 0.1742 & 0.1744 & 0.1772 & 0.1782 & 0.2019 \\
		& 0.40 &  & 0.1936 & 0.1940 & 0.1976 & 0.2010 & 0.2268 \\
		& 0.45 &  & 0.2094 & 0.2097 & 0.2161 & 0.2229 & 0.2491 \\
		& 0.50 &  & 0.2237 & 0.2240 & 0.2326 & 0.2447 & 0.2706 \\
		\addlinespace[2pt]
		& 0.01 & Two-block factor & 0.0063 & 0.0063 & 0.0071 & 0.0138 & 0.0152 \\
		& 0.03 &  & 0.0177 & 0.0176 & 0.0182 & 0.0290 & 0.0283 \\
		& 0.05 &  & 0.0280 & 0.0280 & 0.0286 & 0.0386 & 0.0407 \\
		& 0.10 &  & 0.0525 & 0.0525 & 0.0535 & 0.0594 & 0.0720 \\
		& 0.15 &  & 0.0754 & 0.0753 & 0.0761 & 0.0786 & 0.1006 \\
		& 0.20 &  & 0.0973 & 0.0971 & 0.0979 & 0.0983 & 0.1283 \\
		& 0.25 &  & 0.1172 & 0.1172 & 0.1183 & 0.1176 & 0.1555 \\
		& 0.30 &  & 0.1363 & 0.1363 & 0.1375 & 0.1369 & 0.1795 \\
		& 0.35 &  & 0.1536 & 0.1539 & 0.1551 & 0.1559 & 0.2050 \\
		& 0.40 &  & 0.1710 & 0.1713 & 0.1732 & 0.1761 & 0.2290 \\
		& 0.45 &  & 0.1854 & 0.1859 & 0.1891 & 0.1952 & 0.2512 \\
		& 0.50 &  & 0.1980 & 0.1986 & 0.2023 & 0.2133 & 0.2718 \\
		        \addlinespace[2pt]
		\end{longtable}

\normalsize
\subsection*{Relative Bayes Misclassification Risks}

Table~\ref{tab:appendix-relative-risk-onefactor} reports the Bayes misclassification risks of the competing procedures relative to the corresponding Bayes Oracle risk across all one-factor dependence structures, sample sizes, and sparsity levels considered in the study. These relative risk values facilitate direct comparisons with the Bayes Oracle and provide a quantitative measure of the extent to which each procedure approximates the Bayes-optimal decision rule under dependence.

\small
	\setlength{\tabcolsep}{3.5pt}
	\renewcommand{\arraystretch}{1.08}
	\begin{longtable}{cc l rrrrr}
		\caption{Relative Bayes misclassification risks of the competing procedures, expressed as percentages of the corresponding Bayes Oracle risk, under the six one-factor dependence structures for $n=20,50,100$, $\alpha=0.10$, $\psi^2=2\log n$, and $G=5000$ Monte Carlo replications. The Oracle column is therefore equal to $100\%$ throughout. Values slightly below $100\%$ occasionally occur because all Bayes risks are estimated by Monte Carlo simulation and therefore reflect finite-simulation variability. Such values should not be interpreted as evidence that a procedure uniformly improves upon the Bayes Oracle.}\label{tab:appendix-relative-risk-onefactor}\\
		\toprule
		$n$ & $p$ & Dependence structure & Oracle & BSD & MRD--GBS & MRD--CSX & BH \\
		\midrule
		\endfirsthead
		\caption[]{Relative Bayes misclassification risks of the competing procedures, expressed as percentages of the corresponding Bayes Oracle risk (continued).}\\
		\toprule
		$n$ & $p$ & Dependence structure & Oracle & BSD & MRD--GBS & MRD--CSX & BH \\
		\midrule
		\endhead
		\midrule
		\multicolumn{8}{r}{\emph{Continued on next page}}\\
		\endfoot
		\bottomrule
		\endlastfoot
		20 & 0.03 & Independence & 100.00\% & 100.00\% & 112.14\% & 145.95\% & 113.43\% \\
		& 0.05 &  & 100.00\% & 100.00\% & 105.28\% & 127.66\% & 106.08\% \\
		& 0.10 &  & 100.00\% & 100.00\% & 101.49\% & 115.96\% & 101.70\% \\
		& 0.15 &  & 100.00\% & 100.00\% & 101.19\% & 108.41\% & 100.63\% \\
		& 0.20 &  & 100.00\% & 100.00\% & 102.17\% & 105.46\% & 101.72\% \\
		& 0.25 &  & 100.00\% & 100.00\% & 103.09\% & 103.02\% & 102.62\% \\
		& 0.30 &  & 100.00\% & 100.00\% & 104.24\% & 102.34\% & 103.78\% \\
		& 0.35 &  & 100.00\% & 100.00\% & 105.12\% & 101.73\% & 105.04\% \\
		& 0.40 &  & 100.00\% & 100.00\% & 106.59\% & 102.02\% & 107.06\% \\
		& 0.45 &  & 100.00\% & 100.00\% & 108.61\% & 103.57\% & 109.40\% \\
		& 0.50 &  & 100.00\% & 100.00\% & 110.60\% & 105.21\% & 112.03\% \\
		\addlinespace[2pt]
		& 0.03 & Positive factor & 100.00\% & 99.82\% & 123.88\% & 192.38\% & 254.67\% \\
		& 0.05 &  & 100.00\% & 99.89\% & 113.31\% & 162.68\% & 214.84\% \\
		& 0.10 &  & 100.00\% & 100.08\% & 106.02\% & 137.19\% & 178.69\% \\
		& 0.15 &  & 100.00\% & 100.10\% & 103.49\% & 122.03\% & 166.53\% \\
		& 0.20 &  & 100.00\% & 100.21\% & 103.21\% & 113.61\% & 160.07\% \\
		& 0.25 &  & 100.00\% & 100.63\% & 102.66\% & 108.01\% & 158.75\% \\
		& 0.30 &  & 100.00\% & 100.26\% & 102.53\% & 104.47\% & 158.30\% \\
		& 0.35 &  & 100.00\% & 100.27\% & 102.00\% & 101.96\% & 158.28\% \\
		& 0.40 &  & 100.00\% & 101.17\% & 102.45\% & 101.62\% & 160.01\% \\
		& 0.45 &  & 100.00\% & 101.46\% & 103.52\% & 101.57\% & 160.21\% \\
		& 0.50 &  & 100.00\% & 102.48\% & 104.44\% & 102.61\% & 161.14\% \\
		\addlinespace[2pt]
		& 0.03 & Alternating factor & 100.00\% & 99.94\% & 119.71\% & 187.06\% & 256.71\% \\
		& 0.05 &  & 100.00\% & 100.07\% & 113.48\% & 160.97\% & 213.52\% \\
		& 0.10 &  & 100.00\% & 100.28\% & 106.30\% & 136.33\% & 178.27\% \\
		& 0.15 &  & 100.00\% & 100.18\% & 102.51\% & 121.65\% & 168.05\% \\
		& 0.20 &  & 100.00\% & 100.40\% & 103.72\% & 113.73\% & 160.73\% \\
		& 0.25 &  & 100.00\% & 100.15\% & 102.19\% & 107.24\% & 157.79\% \\
		& 0.30 &  & 100.00\% & 100.10\% & 101.94\% & 104.28\% & 158.22\% \\
		& 0.35 &  & 100.00\% & 101.03\% & 102.52\% & 103.04\% & 157.84\% \\
		& 0.40 &  & 100.00\% & 101.37\% & 102.94\% & 101.88\% & 160.09\% \\
		& 0.45 &  & 100.00\% & 101.28\% & 103.31\% & 101.55\% & 159.83\% \\
		& 0.50 &  & 100.00\% & 102.03\% & 103.87\% & 102.03\% & 160.94\% \\
		\addlinespace[2pt]
		& 0.03 & Linear loadings & 100.00\% & 99.85\% & 117.01\% & 166.33\% & 162.38\% \\
		& 0.05 &  & 100.00\% & 100.12\% & 108.29\% & 140.33\% & 144.39\% \\
		& 0.10 &  & 100.00\% & 100.02\% & 102.85\% & 123.15\% & 128.79\% \\
		& 0.15 &  & 100.00\% & 100.10\% & 101.75\% & 114.18\% & 126.36\% \\
		& 0.20 &  & 100.00\% & 100.19\% & 101.99\% & 108.73\% & 123.53\% \\
		& 0.25 &  & 100.00\% & 100.35\% & 102.27\% & 104.91\% & 124.45\% \\
		& 0.30 &  & 100.00\% & 100.64\% & 102.88\% & 103.08\% & 124.92\% \\
		& 0.35 &  & 100.00\% & 100.57\% & 103.32\% & 101.88\% & 125.84\% \\
		& 0.40 &  & 100.00\% & 100.70\% & 103.92\% & 101.19\% & 127.33\% \\
		& 0.45 &  & 100.00\% & 100.97\% & 105.06\% & 101.91\% & 128.91\% \\
		& 0.50 &  & 100.00\% & 102.31\% & 107.38\% & 104.39\% & 131.16\% \\
		\addlinespace[2pt]
		& 0.03 & Sparse factor & 100.00\% & 99.95\% & 114.27\% & 154.56\% & 127.19\% \\
		& 0.05 &  & 100.00\% & 100.14\% & 106.98\% & 133.45\% & 117.74\% \\
		& 0.10 &  & 100.00\% & 100.12\% & 102.77\% & 120.19\% & 111.91\% \\
		& 0.15 &  & 100.00\% & 100.03\% & 101.55\% & 110.97\% & 110.77\% \\
		& 0.20 &  & 100.00\% & 100.25\% & 101.93\% & 106.58\% & 110.00\% \\
		& 0.25 &  & 100.00\% & 100.28\% & 102.22\% & 103.69\% & 111.58\% \\
		& 0.30 &  & 100.00\% & 100.54\% & 103.50\% & 102.10\% & 112.31\% \\
		& 0.35 &  & 100.00\% & 100.47\% & 103.96\% & 101.67\% & 113.01\% \\
		& 0.40 &  & 100.00\% & 100.94\% & 105.86\% & 102.16\% & 114.89\% \\
		& 0.45 &  & 100.00\% & 101.35\% & 107.77\% & 103.85\% & 117.17\% \\
		& 0.50 &  & 100.00\% & 101.92\% & 109.35\% & 105.66\% & 119.18\% \\
		\addlinespace[2pt]
		& 0.03 & Two-block factor & 100.00\% & 99.90\% & 118.22\% & 167.67\% & 164.51\% \\
		& 0.05 &  & 100.00\% & 100.12\% & 109.39\% & 144.18\% & 144.50\% \\
		& 0.10 &  & 100.00\% & 100.36\% & 104.49\% & 126.60\% & 132.76\% \\
		& 0.15 &  & 100.00\% & 100.33\% & 102.40\% & 114.30\% & 128.69\% \\
		& 0.20 &  & 100.00\% & 100.31\% & 102.43\% & 109.52\% & 126.03\% \\
		& 0.25 &  & 100.00\% & 100.43\% & 102.61\% & 105.05\% & 127.17\% \\
		& 0.30 &  & 100.00\% & 100.29\% & 102.58\% & 102.71\% & 127.22\% \\
		& 0.35 &  & 100.00\% & 100.63\% & 102.94\% & 101.65\% & 128.06\% \\
		& 0.40 &  & 100.00\% & 100.52\% & 103.86\% & 101.22\% & 130.29\% \\
		& 0.45 &  & 100.00\% & 101.71\% & 106.13\% & 102.80\% & 131.82\% \\
		& 0.50 &  & 100.00\% & 101.58\% & 105.80\% & 103.48\% & 132.85\% \\
		\midrule
		50 & 0.02 & Independence & 100.00\% & 100.00\% & 107.55\% & 151.78\% & 108.92\% \\
		& 0.05 &  & 100.00\% & 100.00\% & 101.70\% & 127.25\% & 101.77\% \\
		& 0.10 &  & 100.00\% & 100.00\% & 101.02\% & 111.38\% & 100.87\% \\
		& 0.15 &  & 100.00\% & 100.00\% & 101.91\% & 105.28\% & 101.68\% \\
		& 0.20 &  & 100.00\% & 100.00\% & 102.17\% & 101.89\% & 101.91\% \\
		& 0.25 &  & 100.00\% & 100.00\% & 102.36\% & 100.35\% & 102.37\% \\
		& 0.30 &  & 100.00\% & 100.00\% & 102.92\% & 100.15\% & 103.22\% \\
		& 0.35 &  & 100.00\% & 100.00\% & 103.76\% & 100.93\% & 104.48\% \\
		& 0.40 &  & 100.00\% & 100.00\% & 105.14\% & 102.60\% & 106.35\% \\
		& 0.45 &  & 100.00\% & 100.00\% & 106.88\% & 105.32\% & 108.81\% \\
		& 0.50 &  & 100.00\% & 100.00\% & 108.05\% & 107.61\% & 110.59\% \\
		\addlinespace[2pt]
		& 0.02 & Positive factor & 100.00\% & 99.77\% & 116.59\% & 212.22\% & 309.73\% \\
		& 0.05 &  & 100.00\% & 100.05\% & 106.76\% & 160.40\% & 210.04\% \\
		& 0.10 &  & 100.00\% & 99.22\% & 103.15\% & 126.13\% & 179.61\% \\
		& 0.15 &  & 100.00\% & 99.90\% & 102.74\% & 113.18\% & 171.26\% \\
		& 0.20 &  & 100.00\% & 99.82\% & 102.49\% & 106.76\% & 170.51\% \\
		& 0.25 &  & 100.00\% & 100.14\% & 101.80\% & 102.93\% & 167.83\% \\
		& 0.30 &  & 100.00\% & 100.23\% & 101.40\% & 100.80\% & 166.52\% \\
		& 0.35 &  & 100.00\% & 100.19\% & 101.30\% & 100.17\% & 165.41\% \\
		& 0.40 &  & 100.00\% & 100.46\% & 101.77\% & 100.62\% & 166.57\% \\
		& 0.45 &  & 100.00\% & 100.49\% & 101.50\% & 101.23\% & 165.22\% \\
		& 0.50 &  & 100.00\% & 100.96\% & 102.04\% & 102.60\% & 168.40\% \\
		\addlinespace[2pt]
		& 0.02 & Alternating factor & 100.00\% & 98.38\% & 113.91\% & 207.14\% & 310.30\% \\
		& 0.05 &  & 100.00\% & 99.49\% & 105.87\% & 159.72\% & 211.21\% \\
		& 0.10 &  & 100.00\% & 99.03\% & 103.30\% & 125.97\% & 179.35\% \\
		& 0.15 &  & 100.00\% & 99.84\% & 102.56\% & 112.92\% & 171.96\% \\
		& 0.20 &  & 100.00\% & 100.01\% & 102.54\% & 107.18\% & 171.93\% \\
		& 0.25 &  & 100.00\% & 100.07\% & 101.30\% & 102.80\% & 167.52\% \\
		& 0.30 &  & 100.00\% & 100.31\% & 101.50\% & 101.01\% & 167.13\% \\
		& 0.35 &  & 100.00\% & 100.28\% & 101.40\% & 100.37\% & 164.71\% \\
		& 0.40 &  & 100.00\% & 100.30\% & 101.47\% & 100.26\% & 166.30\% \\
		& 0.45 &  & 100.00\% & 100.60\% & 101.41\% & 101.38\% & 164.90\% \\
		& 0.50 &  & 100.00\% & 100.85\% & 101.77\% & 102.42\% & 168.82\% \\
		\addlinespace[2pt]
		& 0.02 & Linear loadings & 100.00\% & 100.06\% & 110.30\% & 176.16\% & 173.61\% \\
		& 0.05 &  & 100.00\% & 100.00\% & 103.80\% & 140.47\% & 137.27\% \\
		& 0.10 &  & 100.00\% & 100.08\% & 101.83\% & 118.89\% & 129.92\% \\
		& 0.15 &  & 100.00\% & 100.11\% & 101.42\% & 108.33\% & 128.66\% \\
		& 0.20 &  & 100.00\% & 100.12\% & 101.43\% & 103.83\% & 129.11\% \\
		& 0.25 &  & 100.00\% & 100.24\% & 101.23\% & 101.21\% & 128.46\% \\
		& 0.30 &  & 100.00\% & 100.22\% & 101.54\% & 100.29\% & 128.36\% \\
		& 0.35 &  & 100.00\% & 100.30\% & 101.73\% & 100.31\% & 128.61\% \\
		& 0.40 &  & 100.00\% & 100.46\% & 102.22\% & 101.52\% & 130.15\% \\
		& 0.45 &  & 100.00\% & 100.60\% & 102.87\% & 102.95\% & 131.53\% \\
		& 0.50 &  & 100.00\% & 101.10\% & 103.90\% & 105.31\% & 134.17\% \\
		\addlinespace[2pt]
		& 0.02 & Sparse factor & 100.00\% & 100.03\% & 109.04\% & 164.31\% & 126.03\% \\
		& 0.05 &  & 100.00\% & 99.89\% & 102.86\% & 133.96\% & 113.87\% \\
		& 0.10 &  & 100.00\% & 100.12\% & 101.36\% & 114.82\% & 112.69\% \\
		& 0.15 &  & 100.00\% & 99.99\% & 101.64\% & 106.88\% & 111.82\% \\
		& 0.20 &  & 100.00\% & 100.07\% & 101.57\% & 102.58\% & 112.97\% \\
		& 0.25 &  & 100.00\% & 99.99\% & 101.54\% & 100.60\% & 113.22\% \\
		& 0.30 &  & 100.00\% & 100.22\% & 102.18\% & 100.47\% & 114.30\% \\
		& 0.35 &  & 100.00\% & 100.18\% & 102.72\% & 100.91\% & 114.87\% \\
		& 0.40 &  & 100.00\% & 100.05\% & 103.22\% & 102.02\% & 116.47\% \\
		& 0.45 &  & 100.00\% & 100.52\% & 105.02\% & 104.67\% & 118.43\% \\
		& 0.50 &  & 100.00\% & 100.45\% & 105.99\% & 106.83\% & 120.61\% \\
		\addlinespace[2pt]
		& 0.02 & Two-block factor & 100.00\% & 98.33\% & 108.89\% & 179.17\% & 174.69\% \\
		& 0.05 &  & 100.00\% & 99.97\% & 104.17\% & 144.24\% & 142.57\% \\
		& 0.10 &  & 100.00\% & 99.65\% & 101.59\% & 118.74\% & 133.97\% \\
		& 0.15 &  & 100.00\% & 100.22\% & 101.96\% & 109.27\% & 130.93\% \\
		& 0.20 &  & 100.00\% & 100.09\% & 101.49\% & 103.68\% & 131.65\% \\
		& 0.25 &  & 100.00\% & 100.04\% & 101.21\% & 101.33\% & 131.51\% \\
		& 0.30 &  & 100.00\% & 100.22\% & 101.69\% & 100.84\% & 132.05\% \\
		& 0.35 &  & 100.00\% & 100.17\% & 101.70\% & 100.81\% & 131.87\% \\
		& 0.40 &  & 100.00\% & 100.23\% & 101.75\% & 101.13\% & 133.10\% \\
		& 0.45 &  & 100.00\% & 100.29\% & 102.81\% & 103.19\% & 134.63\% \\
		& 0.50 &  & 100.00\% & 100.74\% & 103.74\% & 105.14\% & 136.99\% \\
		\midrule
		100 & 0.01 & Independence & 100.00\% & 100.00\% & 107.73\% & 172.08\% & 109.26\% \\
		& 0.03 &  & 100.00\% & 100.00\% & 101.70\% & 141.45\% & 102.17\% \\
		& 0.05 &  & 100.00\% & 100.00\% & 101.14\% & 124.03\% & 101.02\% \\
		& 0.10 &  & 100.00\% & 100.00\% & 101.07\% & 107.83\% & 100.87\% \\
		& 0.15 &  & 100.00\% & 100.00\% & 101.23\% & 102.23\% & 101.15\% \\
		& 0.20 &  & 100.00\% & 100.00\% & 101.46\% & 100.16\% & 101.55\% \\
		& 0.25 &  & 100.00\% & 100.00\% & 101.81\% & 100.13\% & 102.30\% \\
		& 0.30 &  & 100.00\% & 100.00\% & 102.33\% & 101.10\% & 102.98\% \\
		& 0.35 &  & 100.00\% & 100.00\% & 102.93\% & 102.59\% & 103.99\% \\
		& 0.40 &  & 100.00\% & 100.00\% & 103.59\% & 104.57\% & 105.18\% \\
		& 0.45 &  & 100.00\% & 100.00\% & 104.85\% & 107.50\% & 107.25\% \\
		& 0.50 &  & 100.00\% & 100.00\% & 106.49\% & 111.13\% & 109.66\% \\
		\addlinespace[2pt]
		& 0.01 & Positive factor & 100.00\% & 96.23\% & 114.78\% & 256.94\% & 511.27\% \\
		& 0.03 &  & 100.00\% & 97.53\% & 104.07\% & 182.43\% & 266.59\% \\
		& 0.05 &  & 100.00\% & 99.98\% & 104.47\% & 151.84\% & 219.45\% \\
		& 0.10 &  & 100.00\% & 100.06\% & 103.37\% & 118.33\% & 192.35\% \\
		& 0.15 &  & 100.00\% & 100.00\% & 102.45\% & 107.21\% & 178.27\% \\
		& 0.20 &  & 100.00\% & 99.92\% & 101.87\% & 102.54\% & 171.58\% \\
		& 0.25 &  & 100.00\% & 99.99\% & 101.48\% & 100.66\% & 170.72\% \\
		& 0.30 &  & 100.00\% & 99.80\% & 101.19\% & 99.91\% & 166.74\% \\
		& 0.35 &  & 100.00\% & 100.05\% & 100.89\% & 100.30\% & 168.61\% \\
		& 0.40 &  & 100.00\% & 100.18\% & 100.80\% & 100.99\% & 167.79\% \\
		& 0.45 &  & 100.00\% & 100.11\% & 100.86\% & 102.50\% & 168.08\% \\
		& 0.50 &  & 100.00\% & 100.21\% & 100.54\% & 104.18\% & 168.92\% \\
		\addlinespace[2pt]
		& 0.01 & Alternating factor & 100.00\% & 100.12\% & 117.55\% & 268.40\% & 521.57\% \\
		& 0.03 &  & 100.00\% & 100.15\% & 106.86\% & 187.31\% & 273.94\% \\
		& 0.05 &  & 100.00\% & 100.08\% & 105.11\% & 151.43\% & 219.96\% \\
		& 0.10 &  & 100.00\% & 99.73\% & 103.36\% & 117.86\% & 192.63\% \\
		& 0.15 &  & 100.00\% & 99.46\% & 102.32\% & 106.83\% & 176.77\% \\
		& 0.20 &  & 100.00\% & 99.96\% & 101.84\% & 102.54\% & 172.42\% \\
		& 0.25 &  & 100.00\% & 100.06\% & 101.85\% & 100.77\% & 171.25\% \\
		& 0.30 &  & 100.00\% & 99.90\% & 101.13\% & 99.86\% & 167.23\% \\
		& 0.35 &  & 100.00\% & 100.12\% & 101.18\% & 100.27\% & 169.25\% \\
		& 0.40 &  & 100.00\% & 100.22\% & 100.76\% & 101.26\% & 167.98\% \\
		& 0.45 &  & 100.00\% & 100.30\% & 100.71\% & 102.67\% & 168.20\% \\
		& 0.50 &  & 100.00\% & 100.42\% & 100.96\% & 104.11\% & 168.91\% \\
		\addlinespace[2pt]
		& 0.01 & Linear loadings & 100.00\% & 97.24\% & 107.80\% & 202.61\% & 222.31\% \\
		& 0.03 &  & 100.00\% & 98.95\% & 102.40\% & 158.73\% & 155.18\% \\
		& 0.05 &  & 100.00\% & 99.98\% & 102.25\% & 136.14\% & 142.59\% \\
		& 0.10 &  & 100.00\% & 99.59\% & 100.87\% & 111.58\% & 133.68\% \\
		& 0.15 &  & 100.00\% & 100.09\% & 100.64\% & 103.75\% & 130.48\% \\
		& 0.20 &  & 100.00\% & 100.00\% & 100.74\% & 100.99\% & 129.37\% \\
		& 0.25 &  & 100.00\% & 100.04\% & 100.58\% & 100.14\% & 129.64\% \\
		& 0.30 &  & 100.00\% & 100.02\% & 100.80\% & 100.53\% & 128.94\% \\
		& 0.35 &  & 100.00\% & 100.08\% & 100.76\% & 101.43\% & 130.53\% \\
		& 0.40 &  & 100.00\% & 100.29\% & 101.30\% & 103.06\% & 131.07\% \\
		& 0.45 &  & 100.00\% & 100.40\% & 101.77\% & 105.14\% & 132.62\% \\
		& 0.50 &  & 100.00\% & 100.40\% & 102.11\% & 107.59\% & 133.97\% \\
		\addlinespace[2pt]
		& 0.01 & Sparse factor & 100.00\% & 100.08\% & 109.22\% & 192.08\% & 140.55\% \\
		& 0.03 &  & 100.00\% & 99.04\% & 101.85\% & 151.12\% & 118.76\% \\
		& 0.05 &  & 100.00\% & 99.97\% & 101.29\% & 130.26\% & 115.77\% \\
		& 0.10 &  & 100.00\% & 99.98\% & 101.22\% & 110.17\% & 114.23\% \\
		& 0.15 &  & 100.00\% & 99.98\% & 100.69\% & 103.20\% & 113.96\% \\
		& 0.20 &  & 100.00\% & 100.01\% & 101.12\% & 100.59\% & 114.04\% \\
		& 0.25 &  & 100.00\% & 100.05\% & 101.10\% & 100.23\% & 114.61\% \\
		& 0.30 &  & 100.00\% & 100.02\% & 101.40\% & 100.79\% & 114.87\% \\
		& 0.35 &  & 100.00\% & 100.08\% & 101.73\% & 102.25\% & 115.86\% \\
		& 0.40 &  & 100.00\% & 100.19\% & 102.08\% & 103.81\% & 117.13\% \\
		& 0.45 &  & 100.00\% & 100.17\% & 103.22\% & 106.46\% & 118.99\% \\
		& 0.50 &  & 100.00\% & 100.16\% & 103.97\% & 109.39\% & 120.97\% \\
		\addlinespace[2pt]
		& 0.01 & Two-block factor & 100.00\% & 100.13\% & 112.95\% & 220.57\% & 241.75\% \\
		& 0.03 &  & 100.00\% & 98.96\% & 102.69\% & 163.24\% & 159.21\% \\
		& 0.05 &  & 100.00\% & 99.94\% & 102.27\% & 137.76\% & 145.19\% \\
		& 0.10 &  & 100.00\% & 100.02\% & 101.96\% & 113.21\% & 137.23\% \\
		& 0.15 &  & 100.00\% & 99.82\% & 100.83\% & 104.18\% & 133.41\% \\
		& 0.20 &  & 100.00\% & 99.87\% & 100.62\% & 101.05\% & 131.91\% \\
		& 0.25 &  & 100.00\% & 99.94\% & 100.89\% & 100.32\% & 132.69\% \\
		& 0.30 &  & 100.00\% & 99.95\% & 100.83\% & 100.42\% & 131.66\% \\
		& 0.35 &  & 100.00\% & 100.19\% & 100.96\% & 101.50\% & 133.44\% \\
		& 0.40 &  & 100.00\% & 100.16\% & 101.26\% & 102.98\% & 133.89\% \\
		& 0.45 &  & 100.00\% & 100.25\% & 101.95\% & 105.25\% & 135.46\% \\
		& 0.50 &  & 100.00\% & 100.28\% & 102.14\% & 107.71\% & 137.23\% \\
	\addlinespace[2pt]
	\end{longtable}

\subsection*{False Discovery Rates}
Table~\ref{tab:appendix-fdr-onefactor} reports the empirical false discovery rates of the competing procedures across all one-factor dependence structures, sample sizes, and sparsity levels considered in the study. These results provide a detailed numerical assessment of Type~I error behavior under dependence and complement the graphical summaries presented in the main text. In particular, the table facilitates direct comparisons between the FDR-controlling procedures and the Bayes Oracle, thereby illustrating the extent to which the various methods differ in their allocation of false discoveries across sparsity regimes.
\small
	\setlength{\tabcolsep}{3.5pt}
	\renewcommand{\arraystretch}{1.08}
	\begin{longtable}{cc l rrrrr}
		\caption{False discovery rates of the competing procedures under the six one-factor dependence structures for $n=20,50,100$, $\alpha=0.10$, $\psi^2=2\log n$, and $G=5000$ Monte Carlo replications.}\label{tab:appendix-fdr-onefactor}\\
		\toprule
		$n$ & $p$ & Dependence structure & Oracle & BSD & MRD--GBS & MRD--CSX & BH \\
		\midrule
		\endfirsthead
		\caption[]{False discovery rates of the competing procedures under the six one-factor dependence structures (continued).}\\
		\toprule
		$n$ & $p$ & Dependence structure & Oracle & BSD & MRD--GBS & MRD--CSX & BH \\
		\midrule
		\endhead
		\midrule
		\multicolumn{8}{r}{\emph{Continued on next page}}\\
		\endfoot
		\bottomrule
		\endlastfoot
		20 & 0.03 & Independence & 0.0290 & 0.0290 & 0.0968 & 0.1366 & 0.1010 \\
		& 0.05 &  & 0.0437 & 0.0437 & 0.0937 & 0.1445 & 0.0986 \\
		& 0.10 &  & 0.0814 & 0.0814 & 0.0832 & 0.1673 & 0.0871 \\
		& 0.15 &  & 0.1200 & 0.1200 & 0.0796 & 0.1694 & 0.0821 \\
		& 0.20 &  & 0.1485 & 0.1485 & 0.0833 & 0.1717 & 0.0833 \\
		& 0.25 &  & 0.1564 & 0.1564 & 0.0755 & 0.1584 & 0.0753 \\
		& 0.30 &  & 0.1682 & 0.1682 & 0.0742 & 0.1529 & 0.0710 \\
		& 0.35 &  & 0.1783 & 0.1783 & 0.0720 & 0.1384 & 0.0682 \\
		& 0.40 &  & 0.1797 & 0.1797 & 0.0678 & 0.1238 & 0.0613 \\
		& 0.45 &  & 0.1821 & 0.1821 & 0.0610 & 0.1065 & 0.0550 \\
		& 0.50 &  & 0.1896 & 0.1896 & 0.0617 & 0.0973 & 0.0534 \\
		\addlinespace[2pt]
		& 0.03 & Positive factor & 0.0266 & 0.0266 & 0.1021 & 0.1677 & 0.0563 \\
		& 0.05 &  & 0.0362 & 0.0360 & 0.0966 & 0.1785 & 0.0567 \\
		& 0.10 &  & 0.0542 & 0.0546 & 0.0861 & 0.2043 & 0.0564 \\
		& 0.15 &  & 0.0714 & 0.0726 & 0.0871 & 0.2019 & 0.0521 \\
		& 0.20 &  & 0.0815 & 0.0832 & 0.0892 & 0.1851 & 0.0511 \\
		& 0.25 &  & 0.0834 & 0.0868 & 0.0820 & 0.1625 & 0.0524 \\
		& 0.30 &  & 0.0906 & 0.0913 & 0.0825 & 0.1414 & 0.0517 \\
		& 0.35 &  & 0.0973 & 0.0991 & 0.0840 & 0.1280 & 0.0493 \\
		& 0.40 &  & 0.0973 & 0.0997 & 0.0789 & 0.1111 & 0.0485 \\
		& 0.45 &  & 0.0991 & 0.1028 & 0.0769 & 0.0963 & 0.0427 \\
		& 0.50 &  & 0.1076 & 0.1120 & 0.0787 & 0.0895 & 0.0428 \\
		\addlinespace[2pt]
		& 0.03 & Alternating factor & 0.0257 & 0.0260 & 0.0937 & 0.1578 & 0.0553 \\
		& 0.05 &  & 0.0325 & 0.0330 & 0.0953 & 0.1774 & 0.0567 \\
		& 0.10 &  & 0.0545 & 0.0559 & 0.0871 & 0.2008 & 0.0564 \\
		& 0.15 &  & 0.0733 & 0.0740 & 0.0857 & 0.2033 & 0.0529 \\
		& 0.20 &  & 0.0838 & 0.0877 & 0.0905 & 0.1846 & 0.0550 \\
		& 0.25 &  & 0.0838 & 0.0853 & 0.0817 & 0.1619 & 0.0509 \\
		& 0.30 &  & 0.0907 & 0.0925 & 0.0834 & 0.1444 & 0.0508 \\
		& 0.35 &  & 0.0979 & 0.1022 & 0.0849 & 0.1318 & 0.0478 \\
		& 0.40 &  & 0.0975 & 0.1012 & 0.0814 & 0.1127 & 0.0501 \\
		& 0.45 &  & 0.1003 & 0.1021 & 0.0758 & 0.0951 & 0.0421 \\
		& 0.50 &  & 0.1067 & 0.1104 & 0.0767 & 0.0867 & 0.0415 \\
		\addlinespace[2pt]
		& 0.03 & Linear loadings & 0.0300 & 0.0299 & 0.1006 & 0.1545 & 0.0810 \\
		& 0.05 &  & 0.0409 & 0.0419 & 0.0948 & 0.1603 & 0.0769 \\
		& 0.10 &  & 0.0677 & 0.0685 & 0.0843 & 0.1866 & 0.0740 \\
		& 0.15 &  & 0.0935 & 0.0957 & 0.0852 & 0.1927 & 0.0684 \\
		& 0.20 &  & 0.1119 & 0.1147 & 0.0854 & 0.1804 & 0.0675 \\
		& 0.25 &  & 0.1170 & 0.1203 & 0.0791 & 0.1671 & 0.0637 \\
		& 0.30 &  & 0.1248 & 0.1286 & 0.0800 & 0.1522 & 0.0632 \\
		& 0.35 &  & 0.1338 & 0.1352 & 0.0808 & 0.1375 & 0.0583 \\
		& 0.40 &  & 0.1393 & 0.1404 & 0.0774 & 0.1215 & 0.0582 \\
		& 0.45 &  & 0.1406 & 0.1420 & 0.0727 & 0.1054 & 0.0500 \\
		& 0.50 &  & 0.1469 & 0.1524 & 0.0753 & 0.0984 & 0.0488 \\
		\addlinespace[2pt]
		& 0.03 & Sparse factor & 0.0255 & 0.0257 & 0.0977 & 0.1443 & 0.0930 \\
		& 0.05 &  & 0.0419 & 0.0430 & 0.0984 & 0.1555 & 0.0913 \\
		& 0.10 &  & 0.0720 & 0.0741 & 0.0880 & 0.1799 & 0.0832 \\
		& 0.15 &  & 0.1027 & 0.1044 & 0.0815 & 0.1821 & 0.0775 \\
		& 0.20 &  & 0.1307 & 0.1332 & 0.0859 & 0.1790 & 0.0777 \\
		& 0.25 &  & 0.1350 & 0.1372 & 0.0808 & 0.1673 & 0.0709 \\
		& 0.30 &  & 0.1474 & 0.1510 & 0.0826 & 0.1564 & 0.0666 \\
		& 0.35 &  & 0.1582 & 0.1605 & 0.0809 & 0.1438 & 0.0642 \\
		& 0.40 &  & 0.1591 & 0.1624 & 0.0797 & 0.1275 & 0.0595 \\
		& 0.45 &  & 0.1621 & 0.1667 & 0.0764 & 0.1127 & 0.0533 \\
		& 0.50 &  & 0.1725 & 0.1765 & 0.0741 & 0.1073 & 0.0515 \\
		\addlinespace[2pt]
		& 0.03 & Two-block factor & 0.0248 & 0.0246 & 0.1008 & 0.1530 & 0.0814 \\
		& 0.05 &  & 0.0387 & 0.0398 & 0.0968 & 0.1641 & 0.0745 \\
		& 0.10 &  & 0.0628 & 0.0648 & 0.0864 & 0.1912 & 0.0751 \\
		& 0.15 &  & 0.0860 & 0.0889 & 0.0841 & 0.1910 & 0.0688 \\
		& 0.20 &  & 0.1063 & 0.1103 & 0.0865 & 0.1850 & 0.0664 \\
		& 0.25 &  & 0.1099 & 0.1134 & 0.0807 & 0.1655 & 0.0645 \\
		& 0.30 &  & 0.1226 & 0.1236 & 0.0791 & 0.1500 & 0.0590 \\
		& 0.35 &  & 0.1302 & 0.1336 & 0.0822 & 0.1376 & 0.0599 \\
		& 0.40 &  & 0.1296 & 0.1318 & 0.0766 & 0.1172 & 0.0546 \\
		& 0.45 &  & 0.1338 & 0.1387 & 0.0767 & 0.1058 & 0.0503 \\
		& 0.50 &  & 0.1460 & 0.1490 & 0.0742 & 0.0982 & 0.0498 \\
		\addlinespace[2pt]
		\midrule
		50 & 0.02 & Independence & 0.0345 & 0.0345 & 0.0936 & 0.1753 & 0.0992 \\
		& 0.05 &  & 0.0774 & 0.0774 & 0.0882 & 0.2331 & 0.0923 \\
		& 0.10 &  & 0.1213 & 0.1213 & 0.0863 & 0.2472 & 0.0892 \\
		& 0.15 &  & 0.1399 & 0.1399 & 0.0867 & 0.2223 & 0.0891 \\
		& 0.20 &  & 0.1448 & 0.1448 & 0.0824 & 0.1891 & 0.0796 \\
		& 0.25 &  & 0.1496 & 0.1496 & 0.0775 & 0.1573 & 0.0741 \\
		& 0.30 &  & 0.1543 & 0.1543 & 0.0798 & 0.1338 & 0.0738 \\
		& 0.35 &  & 0.1533 & 0.1533 & 0.0735 & 0.1087 & 0.0650 \\
		& 0.40 &  & 0.1570 & 0.1570 & 0.0714 & 0.0926 & 0.0612 \\
		& 0.45 &  & 0.1602 & 0.1602 & 0.0689 & 0.0788 & 0.0577 \\
		& 0.50 &  & 0.1648 & 0.1648 & 0.0626 & 0.0645 & 0.0487 \\
		\addlinespace[2pt]
		& 0.02 & Positive factor & 0.0249 & 0.0245 & 0.0951 & 0.2268 & 0.0453 \\
		& 0.05 &  & 0.0502 & 0.0511 & 0.0942 & 0.2821 & 0.0467 \\
		& 0.10 &  & 0.0639 & 0.0637 & 0.0905 & 0.2470 & 0.0451 \\
		& 0.15 &  & 0.0727 & 0.0729 & 0.0944 & 0.1992 & 0.0534 \\
		& 0.20 &  & 0.0713 & 0.0714 & 0.0892 & 0.1553 & 0.0548 \\
		& 0.25 &  & 0.0751 & 0.0752 & 0.0877 & 0.1232 & 0.0538 \\
		& 0.30 &  & 0.0790 & 0.0799 & 0.0878 & 0.1029 & 0.0542 \\
		& 0.35 &  & 0.0773 & 0.0781 & 0.0824 & 0.0828 & 0.0481 \\
		& 0.40 &  & 0.0811 & 0.0822 & 0.0842 & 0.0721 & 0.0480 \\
		& 0.45 &  & 0.0836 & 0.0844 & 0.0807 & 0.0617 & 0.0420 \\
		& 0.50 &  & 0.0847 & 0.0861 & 0.0769 & 0.0514 & 0.0406 \\
		\addlinespace[2pt]
		& 0.02 & Alternating factor & 0.0256 & 0.0256 & 0.0960 & 0.2282 & 0.0461 \\
		& 0.05 &  & 0.0500 & 0.0513 & 0.0942 & 0.2835 & 0.0475 \\
		& 0.10 &  & 0.0647 & 0.0641 & 0.0932 & 0.2465 & 0.0463 \\
		& 0.15 &  & 0.0721 & 0.0719 & 0.0925 & 0.1964 & 0.0528 \\
		& 0.20 &  & 0.0728 & 0.0739 & 0.0902 & 0.1556 & 0.0568 \\
		& 0.25 &  & 0.0743 & 0.0749 & 0.0863 & 0.1228 & 0.0529 \\
		& 0.30 &  & 0.0775 & 0.0787 & 0.0870 & 0.1023 & 0.0551 \\
		& 0.35 &  & 0.0773 & 0.0783 & 0.0826 & 0.0831 & 0.0473 \\
		& 0.40 &  & 0.0808 & 0.0812 & 0.0832 & 0.0716 & 0.0488 \\
		& 0.45 &  & 0.0832 & 0.0849 & 0.0803 & 0.0619 & 0.0420 \\
		& 0.50 &  & 0.0850 & 0.0864 & 0.0775 & 0.0519 & 0.0404 \\
		\addlinespace[2pt]
		& 0.02 & Linear loadings & 0.0292 & 0.0293 & 0.0952 & 0.2053 & 0.0700 \\
		& 0.05 &  & 0.0668 & 0.0670 & 0.0933 & 0.2623 & 0.0681 \\
		& 0.10 &  & 0.0900 & 0.0898 & 0.0886 & 0.2563 & 0.0650 \\
		& 0.15 &  & 0.1039 & 0.1041 & 0.0931 & 0.2140 & 0.0694 \\
		& 0.20 &  & 0.1022 & 0.1031 & 0.0859 & 0.1724 & 0.0673 \\
		& 0.25 &  & 0.1066 & 0.1076 & 0.0838 & 0.1388 & 0.0652 \\
		& 0.30 &  & 0.1117 & 0.1133 & 0.0848 & 0.1168 & 0.0650 \\
		& 0.35 &  & 0.1106 & 0.1115 & 0.0785 & 0.0933 & 0.0571 \\
		& 0.40 &  & 0.1158 & 0.1176 & 0.0800 & 0.0819 & 0.0554 \\
		& 0.45 &  & 0.1186 & 0.1200 & 0.0764 & 0.0691 & 0.0490 \\
		& 0.50 &  & 0.1207 & 0.1231 & 0.0720 & 0.0585 & 0.0465 \\
		\addlinespace[2pt]
		& 0.02 & Sparse factor & 0.0320 & 0.0324 & 0.0938 & 0.1905 & 0.0892 \\
		& 0.05 &  & 0.0679 & 0.0681 & 0.0888 & 0.2491 & 0.0848 \\
		& 0.10 &  & 0.1041 & 0.1052 & 0.0883 & 0.2544 & 0.0811 \\
		& 0.15 &  & 0.1180 & 0.1188 & 0.0896 & 0.2194 & 0.0802 \\
		& 0.20 &  & 0.1217 & 0.1225 & 0.0843 & 0.1793 & 0.0768 \\
		& 0.25 &  & 0.1260 & 0.1266 & 0.0827 & 0.1477 & 0.0700 \\
		& 0.30 &  & 0.1299 & 0.1316 & 0.0840 & 0.1256 & 0.0711 \\
		& 0.35 &  & 0.1313 & 0.1321 & 0.0782 & 0.1023 & 0.0621 \\
		& 0.40 &  & 0.1367 & 0.1368 & 0.0753 & 0.0869 & 0.0597 \\
		& 0.45 &  & 0.1392 & 0.1410 & 0.0752 & 0.0762 & 0.0545 \\
		& 0.50 &  & 0.1433 & 0.1445 & 0.0692 & 0.0627 & 0.0486 \\
		\addlinespace[2pt]
		& 0.02 & Two-block factor & 0.0279 & 0.0270 & 0.0927 & 0.2067 & 0.0707 \\
		& 0.05 &  & 0.0613 & 0.0610 & 0.0923 & 0.2693 & 0.0718 \\
		& 0.10 &  & 0.0850 & 0.0848 & 0.0893 & 0.2521 & 0.0673 \\
		& 0.15 &  & 0.0959 & 0.0968 & 0.0912 & 0.2123 & 0.0711 \\
		& 0.20 &  & 0.1002 & 0.1012 & 0.0871 & 0.1690 & 0.0696 \\
		& 0.25 &  & 0.1024 & 0.1027 & 0.0846 & 0.1367 & 0.0642 \\
		& 0.30 &  & 0.1061 & 0.1069 & 0.0865 & 0.1150 & 0.0655 \\
		& 0.35 &  & 0.1064 & 0.1072 & 0.0790 & 0.0937 & 0.0576 \\
		& 0.40 &  & 0.1122 & 0.1130 & 0.0794 & 0.0793 & 0.0557 \\
		& 0.45 &  & 0.1142 & 0.1149 & 0.0769 & 0.0689 & 0.0491 \\
		& 0.50 &  & 0.1178 & 0.1193 & 0.0729 & 0.0571 & 0.0460 \\
		\addlinespace[2pt]
		\midrule
		100 & 0.01 & Independence & 0.0333 & 0.0333 & 0.0968 & 0.1993 & 0.1020 \\
		& 0.03 &  & 0.0790 & 0.0790 & 0.0922 & 0.3105 & 0.0962 \\
		& 0.05 &  & 0.1068 & 0.1068 & 0.0919 & 0.3131 & 0.0963 \\
		& 0.10 &  & 0.1251 & 0.1251 & 0.0902 & 0.2594 & 0.0914 \\
		& 0.15 &  & 0.1309 & 0.1309 & 0.0853 & 0.1967 & 0.0835 \\
		& 0.20 &  & 0.1320 & 0.1320 & 0.0840 & 0.1486 & 0.0793 \\
		& 0.25 &  & 0.1330 & 0.1330 & 0.0797 & 0.1147 & 0.0734 \\
		& 0.30 &  & 0.1383 & 0.1383 & 0.0796 & 0.0953 & 0.0709 \\
		& 0.35 &  & 0.1409 & 0.1409 & 0.0770 & 0.0784 & 0.0659 \\
		& 0.40 &  & 0.1447 & 0.1447 & 0.0728 & 0.0650 & 0.0595 \\
		& 0.45 &  & 0.1465 & 0.1465 & 0.0691 & 0.0543 & 0.0547 \\
		& 0.50 &  & 0.1495 & 0.1495 & 0.0661 & 0.0455 & 0.0507 \\
		\addlinespace[2pt]
		& 0.01 & Positive factor & 0.0251 & 0.0248 & 0.0978 & 0.2709 & 0.0447 \\
		& 0.03 &  & 0.0496 & 0.0499 & 0.0958 & 0.3594 & 0.0456 \\
		& 0.05 &  & 0.0556 & 0.0555 & 0.0938 & 0.3227 & 0.0452 \\
		& 0.10 &  & 0.0607 & 0.0614 & 0.0942 & 0.2119 & 0.0538 \\
		& 0.15 &  & 0.0631 & 0.0629 & 0.0918 & 0.1489 & 0.0531 \\
		& 0.20 &  & 0.0649 & 0.0648 & 0.0902 & 0.1102 & 0.0514 \\
		& 0.25 &  & 0.0653 & 0.0654 & 0.0877 & 0.0861 & 0.0518 \\
		& 0.30 &  & 0.0683 & 0.0684 & 0.0877 & 0.0710 & 0.0482 \\
		& 0.35 &  & 0.0707 & 0.0710 & 0.0861 & 0.0595 & 0.0515 \\
		& 0.40 &  & 0.0724 & 0.0727 & 0.0839 & 0.0496 & 0.0492 \\
		& 0.45 &  & 0.0729 & 0.0732 & 0.0803 & 0.0417 & 0.0426 \\
		& 0.50 &  & 0.0758 & 0.0760 & 0.0780 & 0.0356 & 0.0419 \\
		\addlinespace[2pt]
		& 0.01 & Alternating factor & 0.0243 & 0.0246 & 0.0965 & 0.2719 & 0.0424 \\
		& 0.03 &  & 0.0490 & 0.0495 & 0.0943 & 0.3591 & 0.0450 \\
		& 0.05 &  & 0.0543 & 0.0545 & 0.0948 & 0.3224 & 0.0450 \\
		& 0.10 &  & 0.0614 & 0.0615 & 0.0948 & 0.2116 & 0.0541 \\
		& 0.15 &  & 0.0652 & 0.0642 & 0.0942 & 0.1494 & 0.0521 \\
		& 0.20 &  & 0.0642 & 0.0643 & 0.0898 & 0.1103 & 0.0519 \\
		& 0.25 &  & 0.0645 & 0.0650 & 0.0877 & 0.0853 & 0.0520 \\
		& 0.30 &  & 0.0687 & 0.0690 & 0.0880 & 0.0713 & 0.0494 \\
		& 0.35 &  & 0.0699 & 0.0704 & 0.0864 & 0.0587 & 0.0516 \\
		& 0.40 &  & 0.0716 & 0.0720 & 0.0830 & 0.0492 & 0.0490 \\
		& 0.45 &  & 0.0731 & 0.0740 & 0.0802 & 0.0420 & 0.0427 \\
		& 0.50 &  & 0.0751 & 0.0761 & 0.0786 & 0.0351 & 0.0410 \\
		\addlinespace[2pt]
		& 0.01 & Linear loadings & 0.0298 & 0.0301 & 0.0966 & 0.2377 & 0.0701 \\
		& 0.03 &  & 0.0654 & 0.0653 & 0.0940 & 0.3425 & 0.0697 \\
		& 0.05 &  & 0.0771 & 0.0769 & 0.0936 & 0.3281 & 0.0706 \\
		& 0.10 &  & 0.0873 & 0.0870 & 0.0922 & 0.2349 & 0.0744 \\
		& 0.15 &  & 0.0918 & 0.0928 & 0.0891 & 0.1679 & 0.0703 \\
		& 0.20 &  & 0.0923 & 0.0927 & 0.0876 & 0.1250 & 0.0664 \\
		& 0.25 &  & 0.0942 & 0.0941 & 0.0851 & 0.0978 & 0.0650 \\
		& 0.30 &  & 0.0978 & 0.0980 & 0.0840 & 0.0810 & 0.0594 \\
		& 0.35 &  & 0.1007 & 0.1013 & 0.0819 & 0.0672 & 0.0608 \\
		& 0.40 &  & 0.1035 & 0.1042 & 0.0795 & 0.0563 & 0.0560 \\
		& 0.45 &  & 0.1046 & 0.1054 & 0.0757 & 0.0468 & 0.0491 \\
		& 0.50 &  & 0.1084 & 0.1093 & 0.0729 & 0.0394 & 0.0465 \\
		\addlinespace[2pt]
		& 0.01 & Sparse factor & 0.0310 & 0.0315 & 0.0954 & 0.2220 & 0.0927 \\
		& 0.03 &  & 0.0715 & 0.0706 & 0.0925 & 0.3324 & 0.0873 \\
		& 0.05 &  & 0.0915 & 0.0917 & 0.0936 & 0.3248 & 0.0854 \\
		& 0.10 &  & 0.1020 & 0.1020 & 0.0906 & 0.2473 & 0.0837 \\
		& 0.15 &  & 0.1072 & 0.1070 & 0.0886 & 0.1817 & 0.0791 \\
		& 0.20 &  & 0.1097 & 0.1097 & 0.0871 & 0.1351 & 0.0753 \\
		& 0.25 &  & 0.1104 & 0.1110 & 0.0830 & 0.1046 & 0.0710 \\
		& 0.30 &  & 0.1156 & 0.1158 & 0.0824 & 0.0867 & 0.0653 \\
		& 0.35 &  & 0.1181 & 0.1186 & 0.0802 & 0.0722 & 0.0646 \\
		& 0.40 &  & 0.1218 & 0.1227 & 0.0766 & 0.0599 & 0.0591 \\
		& 0.45 &  & 0.1239 & 0.1244 & 0.0733 & 0.0502 & 0.0530 \\
		& 0.50 &  & 0.1279 & 0.1283 & 0.0693 & 0.0427 & 0.0491 \\
		\addlinespace[2pt]
		& 0.01 & Two-block factor & 0.0282 & 0.0286 & 0.0971 & 0.2461 & 0.0769 \\
		& 0.03 &  & 0.0597 & 0.0593 & 0.0928 & 0.3467 & 0.0713 \\
		& 0.05 &  & 0.0748 & 0.0747 & 0.0934 & 0.3280 & 0.0689 \\
		& 0.10 &  & 0.0819 & 0.0822 & 0.0927 & 0.2326 & 0.0729 \\
		& 0.15 &  & 0.0873 & 0.0876 & 0.0896 & 0.1661 & 0.0683 \\
		& 0.20 &  & 0.0896 & 0.0897 & 0.0881 & 0.1242 & 0.0662 \\
		& 0.25 &  & 0.0895 & 0.0898 & 0.0856 & 0.0961 & 0.0639 \\
		& 0.30 &  & 0.0944 & 0.0946 & 0.0848 & 0.0791 & 0.0591 \\
		& 0.35 &  & 0.0966 & 0.0974 & 0.0826 & 0.0661 & 0.0601 \\
		& 0.40 &  & 0.0998 & 0.1002 & 0.0803 & 0.0554 & 0.0552 \\
		& 0.45 &  & 0.1013 & 0.1022 & 0.0767 & 0.0463 & 0.0493 \\
		& 0.50 &  & 0.1047 & 0.1053 & 0.0730 & 0.0396 & 0.0466 \\
		\addlinespace[2pt]
	\end{longtable}

\normalsize
\subsection*{False Non-Discovery Rate}

Table~\ref{tab:appendix-fnr-onefactor} reports the empirical false non-discovery rates of the competing procedures across all one-factor dependence structures, sample sizes, and sparsity levels considered in the study. These results provide a detailed numerical assessment of missed-signal behavior under dependence and complement the graphical summaries presented in the main text. Since the false non-discovery rate quantifies the proportion of true signals among the accepted hypotheses, this table provides additional insight into the signal recovery properties of the competing procedures and helps explain the differences in their overall Bayes misclassification risks.

\small
	\setlength{\tabcolsep}{3.5pt}
	\renewcommand{\arraystretch}{1.08}
	\begin{longtable}{cc l rrrrr}
		\caption{False non-discovery rates of the competing procedures under the six one-factor dependence structures for $n=20,50,100$, $\alpha=0.10$, $\psi^2=2\log n$, and $G=5000$ Monte Carlo replications.}\label{tab:appendix-fnr-onefactor}\\
		\toprule
		$n$ & $p$ & Dependence structure & Oracle & BSD & MRD--GBS & MRD--CSX & BH \\
		\midrule
		\endfirsthead
		\caption[]{False non-discovery rates of the competing procedures under the six one-factor dependence structures (continued).}\\
		\toprule
		$n$ & $p$ & Dependence structure & Oracle & BSD & MRD--GBS & MRD--CSX & BH \\
		\midrule
		\endhead
		\midrule
		\multicolumn{8}{r}{\emph{Continued on next page}}\\
		\endfoot
		\bottomrule
		\endlastfoot
		20 & 0.03 & Independence & 0.0227 & 0.0227 & 0.0204 & 0.0200 & 0.0203 \\
		& 0.05 &  & 0.0371 & 0.0371 & 0.0347 & 0.0331 & 0.0345 \\
		& 0.10 &  & 0.0711 & 0.0711 & 0.0698 & 0.0657 & 0.0697 \\
		& 0.15 &  & 0.1040 & 0.1040 & 0.1054 & 0.0978 & 0.1050 \\
		& 0.20 &  & 0.1376 & 0.1376 & 0.1430 & 0.1325 & 0.1428 \\
		& 0.25 &  & 0.1691 & 0.1691 & 0.1789 & 0.1656 & 0.1792 \\
		& 0.30 &  & 0.2034 & 0.2034 & 0.2189 & 0.2040 & 0.2194 \\
		& 0.35 &  & 0.2342 & 0.2342 & 0.2558 & 0.2398 & 0.2577 \\
		& 0.40 &  & 0.2655 & 0.2655 & 0.2936 & 0.2774 & 0.2970 \\
		& 0.45 &  & 0.3013 & 0.3013 & 0.3371 & 0.3222 & 0.3412 \\
		& 0.50 &  & 0.3327 & 0.3327 & 0.3763 & 0.3629 & 0.3824 \\
		\addlinespace[2pt]
		& 0.03 & Positive factor & 0.0157 & 0.0157 & 0.0137 & 0.0129 & 0.0212 \\
		& 0.05 &  & 0.0253 & 0.0253 & 0.0225 & 0.0209 & 0.0362 \\
		& 0.10 &  & 0.0487 & 0.0488 & 0.0458 & 0.0416 & 0.0726 \\
		& 0.15 &  & 0.0712 & 0.0712 & 0.0680 & 0.0607 & 0.1089 \\
		& 0.20 &  & 0.0973 & 0.0973 & 0.0943 & 0.0860 & 0.1474 \\
		& 0.25 &  & 0.1214 & 0.1217 & 0.1197 & 0.1102 & 0.1848 \\
		& 0.30 &  & 0.1479 & 0.1482 & 0.1480 & 0.1387 & 0.2253 \\
		& 0.35 &  & 0.1730 & 0.1731 & 0.1739 & 0.1656 & 0.2651 \\
		& 0.40 &  & 0.1995 & 0.2014 & 0.2041 & 0.1974 & 0.3047 \\
		& 0.45 &  & 0.2302 & 0.2329 & 0.2395 & 0.2341 & 0.3486 \\
		& 0.50 &  & 0.2595 & 0.2651 & 0.2748 & 0.2719 & 0.3908 \\
		\addlinespace[2pt]
		& 0.03 & Alternating factor & 0.0158 & 0.0158 & 0.0135 & 0.0130 & 0.0214 \\
		& 0.05 &  & 0.0254 & 0.0254 & 0.0224 & 0.0207 & 0.0360 \\
		& 0.10 &  & 0.0489 & 0.0489 & 0.0459 & 0.0418 & 0.0721 \\
		& 0.15 &  & 0.0716 & 0.0718 & 0.0681 & 0.0610 & 0.1092 \\
		& 0.20 &  & 0.0975 & 0.0976 & 0.0953 & 0.0871 & 0.1472 \\
		& 0.25 &  & 0.1220 & 0.1221 & 0.1201 & 0.1103 & 0.1846 \\
		& 0.30 &  & 0.1480 & 0.1480 & 0.1473 & 0.1383 & 0.2257 \\
		& 0.35 &  & 0.1723 & 0.1733 & 0.1741 & 0.1661 & 0.2638 \\
		& 0.40 &  & 0.1992 & 0.2012 & 0.2043 & 0.1972 & 0.3032 \\
		& 0.45 &  & 0.2308 & 0.2337 & 0.2406 & 0.2352 & 0.3489 \\
		& 0.50 &  & 0.2594 & 0.2642 & 0.2739 & 0.2710 & 0.3903 \\
		\addlinespace[2pt]
		& 0.03 & Linear loadings & 0.0191 & 0.0191 & 0.0170 & 0.0163 & 0.0205 \\
		& 0.05 &  & 0.0314 & 0.0314 & 0.0289 & 0.0272 & 0.0348 \\
		& 0.10 &  & 0.0600 & 0.0599 & 0.0576 & 0.0531 & 0.0701 \\
		& 0.15 &  & 0.0872 & 0.0872 & 0.0864 & 0.0788 & 0.1059 \\
		& 0.20 &  & 0.1178 & 0.1177 & 0.1189 & 0.1098 & 0.1432 \\
		& 0.25 &  & 0.1458 & 0.1457 & 0.1498 & 0.1381 & 0.1806 \\
		& 0.30 &  & 0.1769 & 0.1775 & 0.1848 & 0.1726 & 0.2207 \\
		& 0.35 &  & 0.2044 & 0.2055 & 0.2164 & 0.2042 & 0.2589 \\
		& 0.40 &  & 0.2347 & 0.2363 & 0.2515 & 0.2397 & 0.2987 \\
		& 0.45 &  & 0.2691 & 0.2715 & 0.2914 & 0.2813 & 0.3435 \\
		& 0.50 &  & 0.2995 & 0.3053 & 0.3305 & 0.3230 & 0.3851 \\
		\addlinespace[2pt]
		& 0.03 & Sparse factor & 0.0210 & 0.0210 & 0.0186 & 0.0180 & 0.0205 \\
		& 0.05 &  & 0.0343 & 0.0343 & 0.0317 & 0.0298 & 0.0347 \\
		& 0.10 &  & 0.0654 & 0.0653 & 0.0634 & 0.0593 & 0.0694 \\
		& 0.15 &  & 0.0963 & 0.0962 & 0.0964 & 0.0884 & 0.1054 \\
		& 0.20 &  & 0.1283 & 0.1285 & 0.1314 & 0.1212 & 0.1425 \\
		& 0.25 &  & 0.1581 & 0.1582 & 0.1642 & 0.1517 & 0.1797 \\
		& 0.30 &  & 0.1911 & 0.1916 & 0.2023 & 0.1882 & 0.2194 \\
		& 0.35 &  & 0.2213 & 0.2221 & 0.2377 & 0.2240 & 0.2580 \\
		& 0.40 &  & 0.2523 & 0.2541 & 0.2756 & 0.2616 & 0.2972 \\
		& 0.45 &  & 0.2871 & 0.2900 & 0.3180 & 0.3060 & 0.3415 \\
		& 0.50 &  & 0.3189 & 0.3240 & 0.3582 & 0.3482 & 0.3832 \\
		\addlinespace[2pt]
		& 0.03 & Two-block factor & 0.0188 & 0.0188 & 0.0166 & 0.0159 & 0.0205 \\
		& 0.05 &  & 0.0304 & 0.0304 & 0.0278 & 0.0260 & 0.0348 \\
		& 0.10 &  & 0.0587 & 0.0588 & 0.0567 & 0.0524 & 0.0699 \\
		& 0.15 &  & 0.0859 & 0.0860 & 0.0849 & 0.0767 & 0.1058 \\
		& 0.20 &  & 0.1159 & 0.1159 & 0.1166 & 0.1075 & 0.1433 \\
		& 0.25 &  & 0.1429 & 0.1430 & 0.1462 & 0.1344 & 0.1799 \\
		& 0.30 &  & 0.1735 & 0.1740 & 0.1803 & 0.1681 & 0.2205 \\
		& 0.35 &  & 0.2016 & 0.2025 & 0.2117 & 0.2003 & 0.2588 \\
		& 0.40 &  & 0.2310 & 0.2318 & 0.2457 & 0.2349 & 0.2985 \\
		& 0.45 &  & 0.2639 & 0.2676 & 0.2869 & 0.2772 & 0.3429 \\
		& 0.50 &  & 0.2966 & 0.3009 & 0.3242 & 0.3182 & 0.3851 \\
		\addlinespace[2pt]
		\midrule
		50 & 0.02 & Independence & 0.0148 & 0.0148 & 0.0139 & 0.0133 & 0.0139 \\
		& 0.05 &  & 0.0353 & 0.0353 & 0.0345 & 0.0318 & 0.0343 \\
		& 0.10 &  & 0.0675 & 0.0675 & 0.0682 & 0.0607 & 0.0680 \\
		& 0.15 &  & 0.0980 & 0.0980 & 0.1015 & 0.0913 & 0.1014 \\
		& 0.20 &  & 0.1257 & 0.1257 & 0.1322 & 0.1203 & 0.1325 \\
		& 0.25 &  & 0.1568 & 0.1568 & 0.1668 & 0.1551 & 0.1677 \\
		& 0.30 &  & 0.1884 & 0.1884 & 0.2022 & 0.1917 & 0.2041 \\
		& 0.35 &  & 0.2185 & 0.2185 & 0.2369 & 0.2285 & 0.2402 \\
		& 0.40 &  & 0.2486 & 0.2486 & 0.2735 & 0.2671 & 0.2781 \\
		& 0.45 &  & 0.2790 & 0.2790 & 0.3110 & 0.3084 & 0.3176 \\
		& 0.50 &  & 0.3141 & 0.3141 & 0.3525 & 0.3535 & 0.3608 \\
		\addlinespace[2pt]
		& 0.02 & Positive factor & 0.0099 & 0.0099 & 0.0088 & 0.0081 & 0.0147 \\
		& 0.05 &  & 0.0236 & 0.0236 & 0.0218 & 0.0192 & 0.0356 \\
		& 0.10 &  & 0.0450 & 0.0450 & 0.0424 & 0.0371 & 0.0715 \\
		& 0.15 &  & 0.0663 & 0.0663 & 0.0633 & 0.0573 & 0.1045 \\
		& 0.20 &  & 0.0854 & 0.0855 & 0.0824 & 0.0765 & 0.1374 \\
		& 0.25 &  & 0.1081 & 0.1082 & 0.1048 & 0.1008 & 0.1733 \\
		& 0.30 &  & 0.1310 & 0.1312 & 0.1281 & 0.1259 & 0.2104 \\
		& 0.35 &  & 0.1546 & 0.1547 & 0.1524 & 0.1531 & 0.2465 \\
		& 0.40 &  & 0.1781 & 0.1786 & 0.1773 & 0.1818 & 0.2853 \\
		& 0.45 &  & 0.2052 & 0.2060 & 0.2063 & 0.2148 & 0.3238 \\
		& 0.50 &  & 0.2338 & 0.2355 & 0.2385 & 0.2511 & 0.3689 \\
		\addlinespace[2pt]
		& 0.02 & Alternating factor & 0.0099 & 0.0099 & 0.0088 & 0.0080 & 0.0146 \\
		& 0.05 &  & 0.0236 & 0.0236 & 0.0218 & 0.0192 & 0.0358 \\
		& 0.10 &  & 0.0450 & 0.0449 & 0.0423 & 0.0369 & 0.0711 \\
		& 0.15 &  & 0.0662 & 0.0663 & 0.0633 & 0.0575 & 0.1047 \\
		& 0.20 &  & 0.0852 & 0.0852 & 0.0823 & 0.0769 & 0.1386 \\
		& 0.25 &  & 0.1080 & 0.1081 & 0.1046 & 0.1007 & 0.1730 \\
		& 0.30 &  & 0.1308 & 0.1310 & 0.1278 & 0.1257 & 0.2101 \\
		& 0.35 &  & 0.1549 & 0.1551 & 0.1528 & 0.1537 & 0.2464 \\
		& 0.40 &  & 0.1783 & 0.1787 & 0.1773 & 0.1815 & 0.2847 \\
		& 0.45 &  & 0.2050 & 0.2057 & 0.2060 & 0.2148 & 0.3229 \\
		& 0.50 &  & 0.2329 & 0.2343 & 0.2367 & 0.2499 & 0.3685 \\
		\addlinespace[2pt]
		& 0.02 & Linear loadings & 0.0124 & 0.0124 & 0.0113 & 0.0105 & 0.0139 \\
		& 0.05 &  & 0.0292 & 0.0292 & 0.0279 & 0.0250 & 0.0345 \\
		& 0.10 &  & 0.0558 & 0.0559 & 0.0549 & 0.0487 & 0.0685 \\
		& 0.15 &  & 0.0812 & 0.0813 & 0.0810 & 0.0733 & 0.1017 \\
		& 0.20 &  & 0.1049 & 0.1049 & 0.1059 & 0.0975 & 0.1334 \\
		& 0.25 &  & 0.1319 & 0.1320 & 0.1345 & 0.1271 & 0.1689 \\
		& 0.30 &  & 0.1595 & 0.1596 & 0.1645 & 0.1587 & 0.2054 \\
		& 0.35 &  & 0.1867 & 0.1872 & 0.1942 & 0.1911 & 0.2414 \\
		& 0.40 &  & 0.2140 & 0.2145 & 0.2251 & 0.2256 & 0.2799 \\
		& 0.45 &  & 0.2428 & 0.2439 & 0.2582 & 0.2625 & 0.3191 \\
		& 0.50 &  & 0.2749 & 0.2772 & 0.2960 & 0.3044 & 0.3631 \\
		\addlinespace[2pt]
		& 0.02 & Sparse factor & 0.0136 & 0.0136 & 0.0126 & 0.0119 & 0.0138 \\
		& 0.05 &  & 0.0323 & 0.0322 & 0.0313 & 0.0285 & 0.0343 \\
		& 0.10 &  & 0.0614 & 0.0614 & 0.0611 & 0.0540 & 0.0680 \\
		& 0.15 &  & 0.0899 & 0.0898 & 0.0914 & 0.0824 & 0.1011 \\
		& 0.20 &  & 0.1152 & 0.1152 & 0.1188 & 0.1087 & 0.1324 \\
		& 0.25 &  & 0.1446 & 0.1445 & 0.1506 & 0.1412 & 0.1682 \\
		& 0.30 &  & 0.1737 & 0.1739 & 0.1828 & 0.1752 & 0.2043 \\
		& 0.35 &  & 0.2024 & 0.2026 & 0.2155 & 0.2102 & 0.2401 \\
		& 0.40 &  & 0.2319 & 0.2320 & 0.2497 & 0.2473 & 0.2787 \\
		& 0.45 &  & 0.2619 & 0.2629 & 0.2863 & 0.2875 & 0.3180 \\
		& 0.50 &  & 0.2951 & 0.2960 & 0.3256 & 0.3309 & 0.3609 \\
		\addlinespace[2pt]
		& 0.02 & Two-block factor & 0.0119 & 0.0119 & 0.0108 & 0.0101 & 0.0139 \\
		& 0.05 &  & 0.0286 & 0.0285 & 0.0273 & 0.0244 & 0.0343 \\
		& 0.10 &  & 0.0543 & 0.0543 & 0.0528 & 0.0465 & 0.0683 \\
		& 0.15 &  & 0.0800 & 0.0801 & 0.0796 & 0.0718 & 0.1013 \\
		& 0.20 &  & 0.1030 & 0.1029 & 0.1035 & 0.0952 & 0.1333 \\
		& 0.25 &  & 0.1292 & 0.1292 & 0.1312 & 0.1242 & 0.1690 \\
		& 0.30 &  & 0.1557 & 0.1559 & 0.1597 & 0.1548 & 0.2053 \\
		& 0.35 &  & 0.1825 & 0.1826 & 0.1891 & 0.1870 & 0.2411 \\
		& 0.40 &  & 0.2098 & 0.2101 & 0.2194 & 0.2205 & 0.2797 \\
		& 0.45 &  & 0.2381 & 0.2386 & 0.2524 & 0.2576 & 0.3189 \\
		& 0.50 &  & 0.2701 & 0.2716 & 0.2903 & 0.2994 & 0.3629 \\
		\addlinespace[2pt]
		\midrule
		100 & 0.01 & Independence & 0.0075 & 0.0075 & 0.0070 & 0.0067 & 0.0070 \\
		& 0.03 &  & 0.0209 & 0.0209 & 0.0204 & 0.0184 & 0.0203 \\
		& 0.05 &  & 0.0335 & 0.0335 & 0.0333 & 0.0294 & 0.0332 \\
		& 0.10 &  & 0.0632 & 0.0632 & 0.0644 & 0.0573 & 0.0644 \\
		& 0.15 &  & 0.0922 & 0.0922 & 0.0955 & 0.0872 & 0.0957 \\
		& 0.20 &  & 0.1212 & 0.1212 & 0.1268 & 0.1192 & 0.1277 \\
		& 0.25 &  & 0.1502 & 0.1502 & 0.1588 & 0.1530 & 0.1605 \\
		& 0.30 &  & 0.1789 & 0.1789 & 0.1910 & 0.1879 & 0.1936 \\
		& 0.35 &  & 0.2081 & 0.2081 & 0.2244 & 0.2247 & 0.2284 \\
		& 0.40 &  & 0.2394 & 0.2394 & 0.2604 & 0.2648 & 0.2661 \\
		& 0.45 &  & 0.2701 & 0.2701 & 0.2970 & 0.3060 & 0.3048 \\
		& 0.50 &  & 0.3014 & 0.3014 & 0.3355 & 0.3497 & 0.3453 \\
		\addlinespace[2pt]
		& 0.01 & Positive factor & 0.0049 & 0.0049 & 0.0044 & 0.0039 & 0.0073 \\
		& 0.03 &  & 0.0136 & 0.0136 & 0.0125 & 0.0107 & 0.0218 \\
		& 0.05 &  & 0.0217 & 0.0217 & 0.0202 & 0.0173 & 0.0353 \\
		& 0.10 &  & 0.0415 & 0.0415 & 0.0391 & 0.0352 & 0.0675 \\
		& 0.15 &  & 0.0611 & 0.0611 & 0.0578 & 0.0543 & 0.1001 \\
		& 0.20 &  & 0.0812 & 0.0811 & 0.0772 & 0.0755 & 0.1320 \\
		& 0.25 &  & 0.1018 & 0.1018 & 0.0973 & 0.0981 & 0.1663 \\
		& 0.30 &  & 0.1224 & 0.1224 & 0.1178 & 0.1218 & 0.1985 \\
		& 0.35 &  & 0.1445 & 0.1446 & 0.1397 & 0.1481 & 0.2356 \\
		& 0.40 &  & 0.1693 & 0.1697 & 0.1652 & 0.1781 & 0.2733 \\
		& 0.45 &  & 0.1941 & 0.1945 & 0.1916 & 0.2097 & 0.3121 \\
		& 0.50 &  & 0.2211 & 0.2216 & 0.2198 & 0.2446 & 0.3521 \\
		\addlinespace[2pt]
		& 0.01 & Alternating factor & 0.0049 & 0.0049 & 0.0043 & 0.0039 & 0.0072 \\
		& 0.03 &  & 0.0136 & 0.0136 & 0.0125 & 0.0107 & 0.0218 \\
		& 0.05 &  & 0.0217 & 0.0217 & 0.0202 & 0.0172 & 0.0348 \\
		& 0.10 &  & 0.0414 & 0.0414 & 0.0391 & 0.0351 & 0.0677 \\
		& 0.15 &  & 0.0609 & 0.0610 & 0.0579 & 0.0544 & 0.0997 \\
		& 0.20 &  & 0.0812 & 0.0811 & 0.0771 & 0.0754 & 0.1323 \\
		& 0.25 &  & 0.1016 & 0.1016 & 0.0973 & 0.0980 & 0.1662 \\
		& 0.30 &  & 0.1225 & 0.1226 & 0.1179 & 0.1219 & 0.1991 \\
		& 0.35 &  & 0.1444 & 0.1444 & 0.1397 & 0.1479 & 0.2357 \\
		& 0.40 &  & 0.1692 & 0.1695 & 0.1649 & 0.1782 & 0.2727 \\
		& 0.45 &  & 0.1941 & 0.1943 & 0.1910 & 0.2096 & 0.3121 \\
		& 0.50 &  & 0.2213 & 0.2219 & 0.2203 & 0.2444 & 0.3521 \\
		\addlinespace[2pt]
		& 0.01 & Linear loadings & 0.0062 & 0.0062 & 0.0056 & 0.0052 & 0.0070 \\
		& 0.03 &  & 0.0171 & 0.0171 & 0.0163 & 0.0142 & 0.0204 \\
		& 0.05 &  & 0.0273 & 0.0273 & 0.0264 & 0.0228 & 0.0333 \\
		& 0.10 &  & 0.0518 & 0.0518 & 0.0509 & 0.0454 & 0.0649 \\
		& 0.15 &  & 0.0761 & 0.0761 & 0.0755 & 0.0701 & 0.0962 \\
		& 0.20 &  & 0.1001 & 0.1000 & 0.1002 & 0.0962 & 0.1283 \\
		& 0.25 &  & 0.1252 & 0.1253 & 0.1262 & 0.1247 & 0.1617 \\
		& 0.30 &  & 0.1497 & 0.1497 & 0.1524 & 0.1539 & 0.1946 \\
		& 0.35 &  & 0.1755 & 0.1755 & 0.1801 & 0.1857 & 0.2300 \\
		& 0.40 &  & 0.2036 & 0.2040 & 0.2113 & 0.2213 & 0.2679 \\
		& 0.45 &  & 0.2314 & 0.2321 & 0.2426 & 0.2577 & 0.3067 \\
		& 0.50 &  & 0.2612 & 0.2619 & 0.2764 & 0.2977 & 0.3468 \\
		\addlinespace[2pt]
		& 0.01 & Sparse factor & 0.0068 & 0.0068 & 0.0063 & 0.0058 & 0.0070 \\
		& 0.03 &  & 0.0188 & 0.0188 & 0.0181 & 0.0161 & 0.0203 \\
		& 0.05 &  & 0.0301 & 0.0301 & 0.0295 & 0.0258 & 0.0331 \\
		& 0.10 &  & 0.0571 & 0.0571 & 0.0572 & 0.0509 & 0.0646 \\
		& 0.15 &  & 0.0833 & 0.0833 & 0.0842 & 0.0778 & 0.0955 \\
		& 0.20 &  & 0.1101 & 0.1101 & 0.1128 & 0.1071 & 0.1281 \\
		& 0.25 &  & 0.1370 & 0.1370 & 0.1415 & 0.1382 & 0.1608 \\
		& 0.30 &  & 0.1634 & 0.1634 & 0.1705 & 0.1700 & 0.1939 \\
		& 0.35 &  & 0.1910 & 0.1911 & 0.2011 & 0.2047 & 0.2283 \\
		& 0.40 &  & 0.2207 & 0.2209 & 0.2345 & 0.2423 & 0.2665 \\
		& 0.45 &  & 0.2497 & 0.2500 & 0.2689 & 0.2813 & 0.3050 \\
		& 0.50 &  & 0.2806 & 0.2809 & 0.3053 & 0.3234 & 0.3458 \\
		\addlinespace[2pt]
		& 0.01 & Two-block factor & 0.0059 & 0.0059 & 0.0055 & 0.0050 & 0.0070 \\
		& 0.03 &  & 0.0167 & 0.0167 & 0.0158 & 0.0138 & 0.0203 \\
		& 0.05 &  & 0.0266 & 0.0265 & 0.0256 & 0.0222 & 0.0333 \\
		& 0.10 &  & 0.0507 & 0.0506 & 0.0497 & 0.0444 & 0.0650 \\
		& 0.15 &  & 0.0742 & 0.0742 & 0.0733 & 0.0682 & 0.0961 \\
		& 0.20 &  & 0.0982 & 0.0982 & 0.0978 & 0.0942 & 0.1285 \\
		& 0.25 &  & 0.1226 & 0.1225 & 0.1231 & 0.1218 & 0.1616 \\
		& 0.30 &  & 0.1467 & 0.1467 & 0.1488 & 0.1506 & 0.1945 \\
		& 0.35 &  & 0.1720 & 0.1721 & 0.1759 & 0.1818 & 0.2298 \\
		& 0.40 &  & 0.1997 & 0.1999 & 0.2063 & 0.2168 & 0.2677 \\
		& 0.45 &  & 0.2273 & 0.2276 & 0.2378 & 0.2534 & 0.3065 \\
		& 0.50 &  & 0.2561 & 0.2566 & 0.2705 & 0.2922 & 0.3466 \\
		\addlinespace[2pt]
	\end{longtable}
\normalsize

\subsection*{Powers}

Table~\ref{tab:appendix-power-onefactor} reports the empirical powers of the competing procedures across all one-factor dependence structures, sample sizes, and sparsity levels considered in the study. These results provide a detailed numerical assessment of signal detection performance under dependence and complement the graphical summaries presented in the main text. Since power measures the proportion of truly non-null hypotheses that are correctly identified, the table offers a direct quantitative summary of the signal recovery capabilities of the competing procedures and helps explain the observed differences in false non-discovery rates and overall decision performance.
\small
	\setlength{\tabcolsep}{3.5pt}
	\renewcommand{\arraystretch}{1.08}
	\begin{longtable}{cc l rrrrr}
		\caption{Empirical powers of the competing procedures under the six one-factor dependence structures for $n=20,50,100$, $\alpha=0.10$, $\psi^2=2\log n$, and $G=5000$ Monte Carlo replications.}\label{tab:appendix-power-onefactor}\\
		\toprule
		$n$ & $p$ & Dependence structure & Oracle & BSD & MRD--GBS & MRD--CSX & BH \\
		\midrule
		\endfirsthead
		\caption[]{Empirical powers of the competing procedures under the six one-factor dependence structures (continued).}\\
		\toprule
		$n$ & $p$ & Dependence structure & Oracle & BSD & MRD--GBS & MRD--CSX & BH \\
		\midrule
		\endhead
		\midrule
		\multicolumn{8}{r}{\emph{Continued on next page}}\\
		\endfoot
		\bottomrule
		\endlastfoot
		20 & 0.03 & Independence & 0.2435 & 0.2435 & 0.3177 & 0.3328 & 0.3231 \\
		& 0.05 &  & 0.2536 & 0.2536 & 0.2987 & 0.3259 & 0.3014 \\
		& 0.10 &  & 0.3100 & 0.3100 & 0.3143 & 0.3544 & 0.3180 \\
		& 0.15 &  & 0.3461 & 0.3461 & 0.3239 & 0.3777 & 0.3280 \\
		& 0.20 &  & 0.3780 & 0.3780 & 0.3338 & 0.3964 & 0.3359 \\
		& 0.25 &  & 0.4099 & 0.4099 & 0.3488 & 0.4168 & 0.3494 \\
		& 0.30 &  & 0.4321 & 0.4321 & 0.3515 & 0.4207 & 0.3513 \\
		& 0.35 &  & 0.4680 & 0.4680 & 0.3678 & 0.4368 & 0.3635 \\
		& 0.40 &  & 0.4993 & 0.4993 & 0.3820 & 0.4464 & 0.3728 \\
		& 0.45 &  & 0.5296 & 0.5296 & 0.3901 & 0.4466 & 0.3800 \\
		& 0.50 &  & 0.5668 & 0.5668 & 0.4039 & 0.4556 & 0.3891 \\
		\addlinespace[2pt]
		& 0.03 & Positive factor & 0.4788 & 0.4790 & 0.5491 & 0.5705 & 0.3209 \\
		& 0.05 &  & 0.4998 & 0.5005 & 0.5511 & 0.5794 & 0.3027 \\
		& 0.10 &  & 0.5386 & 0.5383 & 0.5585 & 0.5968 & 0.3184 \\
		& 0.15 &  & 0.5705 & 0.5694 & 0.5813 & 0.6321 & 0.3279 \\
		& 0.20 &  & 0.5785 & 0.5788 & 0.5834 & 0.6330 & 0.3365 \\
		& 0.25 &  & 0.5956 & 0.5955 & 0.5929 & 0.6425 & 0.3483 \\
		& 0.30 &  & 0.6105 & 0.6098 & 0.6008 & 0.6439 & 0.3500 \\
		& 0.35 &  & 0.6290 & 0.6293 & 0.6174 & 0.6529 & 0.3625 \\
		& 0.40 &  & 0.6457 & 0.6434 & 0.6250 & 0.6531 & 0.3720 \\
		& 0.45 &  & 0.6604 & 0.6578 & 0.6317 & 0.6526 & 0.3778 \\
		& 0.50 &  & 0.6805 & 0.6745 & 0.6401 & 0.6534 & 0.3866 \\
		\addlinespace[2pt]
		& 0.03 & Alternating factor & 0.4740 & 0.4754 & 0.5534 & 0.5685 & 0.3176 \\
		& 0.05 &  & 0.4959 & 0.4969 & 0.5551 & 0.5829 & 0.3025 \\
		& 0.10 &  & 0.5369 & 0.5370 & 0.5578 & 0.5952 & 0.3152 \\
		& 0.15 &  & 0.5690 & 0.5680 & 0.5818 & 0.6314 & 0.3231 \\
		& 0.20 &  & 0.5790 & 0.5793 & 0.5791 & 0.6285 & 0.3346 \\
		& 0.25 &  & 0.5932 & 0.5934 & 0.5908 & 0.6416 & 0.3468 \\
		& 0.30 &  & 0.6100 & 0.6102 & 0.6027 & 0.6450 & 0.3504 \\
		& 0.35 &  & 0.6326 & 0.6306 & 0.6180 & 0.6529 & 0.3651 \\
		& 0.40 &  & 0.6468 & 0.6446 & 0.6262 & 0.6548 & 0.3736 \\
		& 0.45 &  & 0.6606 & 0.6569 & 0.6312 & 0.6509 & 0.3792 \\
		& 0.50 &  & 0.6794 & 0.6746 & 0.6411 & 0.6536 & 0.3878 \\
		\addlinespace[2pt]
		& 0.03 & Linear loadings & 0.3622 & 0.3627 & 0.4340 & 0.4551 & 0.3212 \\
		& 0.05 &  & 0.3705 & 0.3710 & 0.4196 & 0.4506 & 0.3037 \\
		& 0.10 &  & 0.4245 & 0.4249 & 0.4385 & 0.4801 & 0.3174 \\
		& 0.15 &  & 0.4623 & 0.4622 & 0.4553 & 0.5109 & 0.3262 \\
		& 0.20 &  & 0.4783 & 0.4791 & 0.4598 & 0.5150 & 0.3387 \\
		& 0.25 &  & 0.5052 & 0.5056 & 0.4742 & 0.5336 & 0.3491 \\
		& 0.30 &  & 0.5202 & 0.5200 & 0.4773 & 0.5349 & 0.3499 \\
		& 0.35 &  & 0.5490 & 0.5477 & 0.4966 & 0.5488 & 0.3643 \\
		& 0.40 &  & 0.5706 & 0.5683 & 0.5057 & 0.5519 & 0.3742 \\
		& 0.45 &  & 0.5910 & 0.5871 & 0.5141 & 0.5524 & 0.3768 \\
		& 0.50 &  & 0.6196 & 0.6131 & 0.5234 & 0.5524 & 0.3875 \\
		\addlinespace[2pt]
		& 0.03 & Sparse factor & 0.3017 & 0.3023 & 0.3794 & 0.3970 & 0.3169 \\
		& 0.05 &  & 0.3113 & 0.3119 & 0.3607 & 0.3937 & 0.3010 \\
		& 0.10 &  & 0.3697 & 0.3709 & 0.3827 & 0.4214 & 0.3206 \\
		& 0.15 &  & 0.4003 & 0.4013 & 0.3885 & 0.4449 & 0.3265 \\
		& 0.20 &  & 0.4262 & 0.4259 & 0.3969 & 0.4572 & 0.3372 \\
		& 0.25 &  & 0.4546 & 0.4554 & 0.4138 & 0.4797 & 0.3466 \\
		& 0.30 &  & 0.4734 & 0.4724 & 0.4165 & 0.4812 & 0.3514 \\
		& 0.35 &  & 0.5043 & 0.5032 & 0.4305 & 0.4902 & 0.3624 \\
		& 0.40 &  & 0.5313 & 0.5285 & 0.4405 & 0.4957 & 0.3734 \\
		& 0.45 &  & 0.5580 & 0.5534 & 0.4487 & 0.4939 & 0.3790 \\
		& 0.50 &  & 0.5901 & 0.5827 & 0.4579 & 0.4979 & 0.3884 \\
		\addlinespace[2pt]
		& 0.03 & Two-block factor & 0.3760 & 0.3759 & 0.4483 & 0.4685 & 0.3180 \\
		& 0.05 &  & 0.3920 & 0.3928 & 0.4414 & 0.4729 & 0.2996 \\
		& 0.10 &  & 0.4372 & 0.4373 & 0.4505 & 0.4905 & 0.3190 \\
		& 0.15 &  & 0.4724 & 0.4725 & 0.4686 & 0.5259 & 0.3275 \\
		& 0.20 &  & 0.4885 & 0.4892 & 0.4733 & 0.5288 & 0.3369 \\
		& 0.25 &  & 0.5149 & 0.5151 & 0.4885 & 0.5497 & 0.3511 \\
		& 0.30 &  & 0.5306 & 0.5293 & 0.4937 & 0.5495 & 0.3502 \\
		& 0.35 &  & 0.5569 & 0.5564 & 0.5117 & 0.5612 & 0.3636 \\
		& 0.40 &  & 0.5780 & 0.5774 & 0.5208 & 0.5638 & 0.3724 \\
		& 0.45 &  & 0.6008 & 0.5969 & 0.5270 & 0.5634 & 0.3786 \\
		& 0.50 &  & 0.6261 & 0.6206 & 0.5374 & 0.5642 & 0.3870 \\
		\midrule
		50 & 0.02 & Independence & 0.2553 & 0.2553 & 0.3006 & 0.3263 & 0.3038 \\
		& 0.05 &  & 0.3101 & 0.3101 & 0.3203 & 0.3675 & 0.3245 \\
		& 0.10 &  & 0.3585 & 0.3585 & 0.3410 & 0.4220 & 0.3439 \\
		& 0.15 &  & 0.3953 & 0.3953 & 0.3599 & 0.4412 & 0.3619 \\
		& 0.20 &  & 0.4291 & 0.4291 & 0.3831 & 0.4594 & 0.3817 \\
		& 0.25 &  & 0.4545 & 0.4545 & 0.3991 & 0.4617 & 0.3958 \\
		& 0.30 &  & 0.4791 & 0.4791 & 0.4136 & 0.4635 & 0.4068 \\
		& 0.35 &  & 0.5069 & 0.5069 & 0.4305 & 0.4668 & 0.4198 \\
		& 0.40 &  & 0.5374 & 0.5374 & 0.4449 & 0.4710 & 0.4307 \\
		& 0.45 &  & 0.5676 & 0.5676 & 0.4579 & 0.4710 & 0.4392 \\
		& 0.50 &  & 0.5973 & 0.5973 & 0.4704 & 0.4720 & 0.4473 \\
		\addlinespace[2pt]
		& 0.02 & Positive factor & 0.5056 & 0.5061 & 0.5539 & 0.5834 & 0.3043 \\
		& 0.05 &  & 0.5455 & 0.5456 & 0.5722 & 0.6200 & 0.3256 \\
		& 0.10 &  & 0.5803 & 0.5806 & 0.5979 & 0.6577 & 0.3425 \\
		& 0.15 &  & 0.6033 & 0.6031 & 0.6181 & 0.6662 & 0.3658 \\
		& 0.20 &  & 0.6268 & 0.6265 & 0.6374 & 0.6748 & 0.3835 \\
		& 0.25 &  & 0.6416 & 0.6412 & 0.6506 & 0.6726 & 0.3959 \\
		& 0.30 &  & 0.6568 & 0.6565 & 0.6633 & 0.6751 & 0.4054 \\
		& 0.35 &  & 0.6713 & 0.6713 & 0.6741 & 0.6757 & 0.4209 \\
		& 0.40 &  & 0.6889 & 0.6881 & 0.6882 & 0.6787 & 0.4291 \\
		& 0.45 &  & 0.7001 & 0.6988 & 0.6946 & 0.6765 & 0.4378 \\
		& 0.50 &  & 0.7163 & 0.7141 & 0.7038 & 0.6776 & 0.4473 \\
		\addlinespace[2pt]
		& 0.02 & Alternating factor & 0.5053 & 0.5045 & 0.5551 & 0.5859 & 0.3016 \\
		& 0.05 &  & 0.5455 & 0.5452 & 0.5717 & 0.6202 & 0.3247 \\
		& 0.10 &  & 0.5803 & 0.5813 & 0.5991 & 0.6588 & 0.3457 \\
		& 0.15 &  & 0.6040 & 0.6037 & 0.6174 & 0.6641 & 0.3661 \\
		& 0.20 &  & 0.6280 & 0.6278 & 0.6379 & 0.6729 & 0.3810 \\
		& 0.25 &  & 0.6414 & 0.6412 & 0.6508 & 0.6728 & 0.3971 \\
		& 0.30 &  & 0.6575 & 0.6569 & 0.6640 & 0.6754 & 0.4074 \\
		& 0.35 &  & 0.6708 & 0.6706 & 0.6731 & 0.6746 & 0.4207 \\
		& 0.40 &  & 0.6883 & 0.6874 & 0.6878 & 0.6793 & 0.4309 \\
		& 0.45 &  & 0.7006 & 0.6998 & 0.6953 & 0.6767 & 0.4397 \\
		& 0.50 &  & 0.7177 & 0.7160 & 0.7067 & 0.6798 & 0.4466 \\
		\addlinespace[2pt]
		& 0.02 & Linear loadings & 0.3774 & 0.3777 & 0.4284 & 0.4598 & 0.3048 \\
		& 0.05 &  & 0.4313 & 0.4310 & 0.4500 & 0.5030 & 0.3259 \\
		& 0.10 &  & 0.4745 & 0.4744 & 0.4755 & 0.5438 & 0.3427 \\
		& 0.15 &  & 0.5061 & 0.5061 & 0.5012 & 0.5633 & 0.3639 \\
		& 0.20 &  & 0.5335 & 0.5331 & 0.5211 & 0.5747 & 0.3830 \\
		& 0.25 &  & 0.5526 & 0.5523 & 0.5351 & 0.5742 & 0.3956 \\
		& 0.30 &  & 0.5712 & 0.5716 & 0.5468 & 0.5744 & 0.4060 \\
		& 0.35 &  & 0.5909 & 0.5903 & 0.5608 & 0.5755 & 0.4200 \\
		& 0.40 &  & 0.6142 & 0.6135 & 0.5763 & 0.5785 & 0.4299 \\
		& 0.45 &  & 0.6346 & 0.6328 & 0.5870 & 0.5789 & 0.4378 \\
		& 0.50 &  & 0.6564 & 0.6535 & 0.5977 & 0.5792 & 0.4477 \\
		\addlinespace[2pt]
		& 0.02 & Sparse factor & 0.3207 & 0.3207 & 0.3682 & 0.3956 & 0.3052 \\
		& 0.05 &  & 0.3700 & 0.3711 & 0.3815 & 0.4332 & 0.3272 \\
		& 0.10 &  & 0.4196 & 0.4196 & 0.4135 & 0.4917 & 0.3448 \\
		& 0.15 &  & 0.4500 & 0.4504 & 0.4298 & 0.5019 & 0.3645 \\
		& 0.20 &  & 0.4823 & 0.4822 & 0.4547 & 0.5193 & 0.3835 \\
		& 0.25 &  & 0.5032 & 0.5039 & 0.4692 & 0.5191 & 0.3943 \\
		& 0.30 &  & 0.5266 & 0.5266 & 0.4844 & 0.5206 & 0.4067 \\
		& 0.35 &  & 0.5501 & 0.5499 & 0.4979 & 0.5218 & 0.4206 \\
		& 0.40 &  & 0.5755 & 0.5753 & 0.5121 & 0.5242 & 0.4299 \\
		& 0.45 &  & 0.5997 & 0.5984 & 0.5220 & 0.5227 & 0.4386 \\
		& 0.50 &  & 0.6266 & 0.6253 & 0.5349 & 0.5243 & 0.4481 \\
		\addlinespace[2pt]
		& 0.02 & Two-block factor & 0.4040 & 0.4042 & 0.4559 & 0.4824 & 0.3055 \\
		& 0.05 &  & 0.4437 & 0.4444 & 0.4614 & 0.5138 & 0.3291 \\
		& 0.10 &  & 0.4904 & 0.4907 & 0.4961 & 0.5656 & 0.3455 \\
		& 0.15 &  & 0.5150 & 0.5146 & 0.5100 & 0.5731 & 0.3665 \\
		& 0.20 &  & 0.5429 & 0.5429 & 0.5333 & 0.5858 & 0.3826 \\
		& 0.25 &  & 0.5633 & 0.5633 & 0.5485 & 0.5863 & 0.3950 \\
		& 0.30 &  & 0.5829 & 0.5823 & 0.5628 & 0.5865 & 0.4065 \\
		& 0.35 &  & 0.6022 & 0.6022 & 0.5752 & 0.5873 & 0.4214 \\
		& 0.40 &  & 0.6237 & 0.6232 & 0.5906 & 0.5909 & 0.4305 \\
		& 0.45 &  & 0.6432 & 0.6424 & 0.6001 & 0.5900 & 0.4388 \\
		& 0.50 &  & 0.6640 & 0.6621 & 0.6092 & 0.5893 & 0.4470 \\
		\midrule
		100 & 0.01 & Independence & 0.2544 & 0.2544 & 0.2990 & 0.3269 & 0.3027 \\
		& 0.03 &  & 0.3237 & 0.3237 & 0.3348 & 0.3939 & 0.3370 \\
		& 0.05 &  & 0.3445 & 0.3445 & 0.3408 & 0.4224 & 0.3445 \\
		& 0.10 &  & 0.3945 & 0.3945 & 0.3751 & 0.4586 & 0.3765 \\
		& 0.15 &  & 0.4273 & 0.4273 & 0.3976 & 0.4652 & 0.3966 \\
		& 0.20 &  & 0.4579 & 0.4579 & 0.4213 & 0.4698 & 0.4171 \\
		& 0.25 &  & 0.4837 & 0.4837 & 0.4392 & 0.4700 & 0.4317 \\
		& 0.30 &  & 0.5087 & 0.5087 & 0.4551 & 0.4705 & 0.4451 \\
		& 0.35 &  & 0.5344 & 0.5344 & 0.4714 & 0.4729 & 0.4578 \\
		& 0.40 &  & 0.5580 & 0.5580 & 0.4838 & 0.4722 & 0.4660 \\
		& 0.45 &  & 0.5857 & 0.5857 & 0.4980 & 0.4735 & 0.4755 \\
		& 0.50 &  & 0.6177 & 0.6177 & 0.5142 & 0.4770 & 0.4878 \\
		\addlinespace[2pt]
		& 0.01 & Positive factor & 0.5188 & 0.5190 & 0.5654 & 0.5995 & 0.3107 \\
		& 0.03 &  & 0.5625 & 0.5625 & 0.5893 & 0.6472 & 0.3362 \\
		& 0.05 &  & 0.5822 & 0.5821 & 0.6055 & 0.6670 & 0.3445 \\
		& 0.10 &  & 0.6105 & 0.6107 & 0.6307 & 0.6770 & 0.3763 \\
		& 0.15 &  & 0.6316 & 0.6316 & 0.6506 & 0.6788 & 0.3972 \\
		& 0.20 &  & 0.6507 & 0.6508 & 0.6682 & 0.6800 & 0.4167 \\
		& 0.25 &  & 0.6657 & 0.6658 & 0.6820 & 0.6809 & 0.4318 \\
		& 0.30 &  & 0.6810 & 0.6810 & 0.6951 & 0.6830 & 0.4448 \\
		& 0.35 &  & 0.6953 & 0.6951 & 0.7075 & 0.6844 & 0.4582 \\
		& 0.40 &  & 0.7059 & 0.7054 & 0.7150 & 0.6830 & 0.4652 \\
		& 0.45 &  & 0.7197 & 0.7191 & 0.7245 & 0.6840 & 0.4743 \\
		& 0.50 &  & 0.7347 & 0.7341 & 0.7362 & 0.6862 & 0.4877 \\
		\addlinespace[2pt]
		& 0.01 & Alternating factor & 0.5164 & 0.5161 & 0.5667 & 0.5985 & 0.3172 \\
		& 0.03 &  & 0.5633 & 0.5633 & 0.5895 & 0.6488 & 0.3391 \\
		& 0.05 &  & 0.5825 & 0.5823 & 0.6050 & 0.6679 & 0.3498 \\
		& 0.10 &  & 0.6116 & 0.6115 & 0.6307 & 0.6779 & 0.3760 \\
		& 0.15 &  & 0.6323 & 0.6320 & 0.6501 & 0.6784 & 0.3984 \\
		& 0.20 &  & 0.6507 & 0.6509 & 0.6688 & 0.6808 & 0.4160 \\
		& 0.25 &  & 0.6663 & 0.6664 & 0.6820 & 0.6811 & 0.4320 \\
		& 0.30 &  & 0.6808 & 0.6805 & 0.6948 & 0.6827 & 0.4433 \\
		& 0.35 &  & 0.6955 & 0.6954 & 0.7078 & 0.6849 & 0.4581 \\
		& 0.40 &  & 0.7062 & 0.7056 & 0.7154 & 0.6827 & 0.4662 \\
		& 0.45 &  & 0.7199 & 0.7196 & 0.7255 & 0.6841 & 0.4744 \\
		& 0.50 &  & 0.7342 & 0.7336 & 0.7355 & 0.6862 & 0.4869 \\
		\addlinespace[2pt]
		& 0.01 & Linear loadings & 0.3877 & 0.3875 & 0.4377 & 0.4704 & 0.3082 \\
		& 0.03 &  & 0.4473 & 0.4475 & 0.4660 & 0.5312 & 0.3364 \\
		& 0.05 &  & 0.4700 & 0.4702 & 0.4802 & 0.5557 & 0.3462 \\
		& 0.10 &  & 0.5088 & 0.5091 & 0.5139 & 0.5773 & 0.3760 \\
		& 0.15 &  & 0.5341 & 0.5343 & 0.5342 & 0.5780 & 0.3976 \\
		& 0.20 &  & 0.5616 & 0.5619 & 0.5580 & 0.5830 & 0.4177 \\
		& 0.25 &  & 0.5799 & 0.5796 & 0.5726 & 0.5822 & 0.4310 \\
		& 0.30 &  & 0.6002 & 0.6002 & 0.5873 & 0.5841 & 0.4453 \\
		& 0.35 &  & 0.6191 & 0.6192 & 0.6016 & 0.5855 & 0.4579 \\
		& 0.40 &  & 0.6358 & 0.6350 & 0.6110 & 0.5840 & 0.4654 \\
		& 0.45 &  & 0.6558 & 0.6547 & 0.6240 & 0.5861 & 0.4748 \\
		& 0.50 &  & 0.6776 & 0.6766 & 0.6377 & 0.5886 & 0.4884 \\
		\addlinespace[2pt]
		& 0.01 & Sparse factor & 0.3254 & 0.3254 & 0.3724 & 0.4046 & 0.3069 \\
		& 0.03 &  & 0.3933 & 0.3933 & 0.4084 & 0.4674 & 0.3363 \\
		& 0.05 &  & 0.4126 & 0.4128 & 0.4179 & 0.4951 & 0.3459 \\
		& 0.10 &  & 0.4566 & 0.4567 & 0.4503 & 0.5236 & 0.3747 \\
		& 0.15 &  & 0.4868 & 0.4869 & 0.4754 & 0.5280 & 0.3989 \\
		& 0.20 &  & 0.5130 & 0.5131 & 0.4945 & 0.5300 & 0.4154 \\
		& 0.25 &  & 0.5352 & 0.5353 & 0.5117 & 0.5297 & 0.4313 \\
		& 0.30 &  & 0.5579 & 0.5579 & 0.5272 & 0.5314 & 0.4441 \\
		& 0.35 &  & 0.5794 & 0.5794 & 0.5420 & 0.5321 & 0.4589 \\
		& 0.40 &  & 0.5993 & 0.5990 & 0.5534 & 0.5317 & 0.4659 \\
		& 0.45 &  & 0.6233 & 0.6228 & 0.5659 & 0.5333 & 0.4756 \\
		& 0.50 &  & 0.6492 & 0.6488 & 0.5803 & 0.5359 & 0.4874 \\
		\addlinespace[2pt]
		& 0.01 & Two-block factor & 0.4082 & 0.4081 & 0.4523 & 0.4857 & 0.3075 \\
		& 0.03 &  & 0.4630 & 0.4628 & 0.4840 & 0.5437 & 0.3373 \\
		& 0.05 &  & 0.4835 & 0.4840 & 0.4960 & 0.5692 & 0.3465 \\
		& 0.10 &  & 0.5210 & 0.5211 & 0.5263 & 0.5873 & 0.3750 \\
		& 0.15 &  & 0.5470 & 0.5469 & 0.5491 & 0.5910 & 0.3980 \\
		& 0.20 &  & 0.5704 & 0.5704 & 0.5692 & 0.5926 & 0.4163 \\
		& 0.25 &  & 0.5897 & 0.5901 & 0.5847 & 0.5933 & 0.4311 \\
		& 0.30 &  & 0.6090 & 0.6090 & 0.5991 & 0.5945 & 0.4448 \\
		& 0.35 &  & 0.6281 & 0.6279 & 0.6136 & 0.5966 & 0.4578 \\
		& 0.40 &  & 0.6440 & 0.6437 & 0.6230 & 0.5948 & 0.4657 \\
		& 0.45 &  & 0.6631 & 0.6628 & 0.6341 & 0.5954 & 0.4747 \\
		& 0.50 &  & 0.6852 & 0.6845 & 0.6486 & 0.5993 & 0.4882 \\
	\end{longtable}
\normalsize

\subsection*{Average Number of Rejections}

Table~\ref{tab:appendix-anr-onefactor} reports the average numbers of rejections produced by the competing procedures across all one-factor dependence structures, sample sizes, and sparsity levels considered in the study. These results provide a detailed numerical summary of the overall rejection behavior of the various methods and complement the graphical summaries presented in the main text. Since the average number of rejections reflects the overall level of signal discovery, this table provides additional insight into the operating characteristics of the competing procedures and complements the corresponding FDR, FNR, and power summaries presented above.
\small
	\setlength{\tabcolsep}{3.5pt}
	\renewcommand{\arraystretch}{1.08}
	\begin{longtable}{cc l rrrrr}
		\caption{Average numbers of rejections of the competing procedures under the six one-factor dependence structures for $n=20,50,100$, $\alpha=0.10$, $\psi^2=2\log n$, and $G=5000$ Monte Carlo replications.}\label{tab:appendix-anr-onefactor}\\
		\toprule
		$n$ & $p$ & Dependence structure & Oracle & BSD & MRD--GBS & MRD--CSX & BH \\
		\midrule
		\endfirsthead
		\caption[]{Average numbers of rejections of the competing procedures under the six one-factor dependence structures (continued).}\\
		\toprule
		$n$ & $p$ & Dependence structure & Oracle & BSD & MRD--GBS & MRD--CSX & BH \\
		\midrule
		\endhead
		\midrule
		\multicolumn{8}{r}{\emph{Continued on next page}}\\
		\endfoot
		\bottomrule
		\endlastfoot
		20 & 0.03 & Independence & 0.1738 & 0.1738 & 0.3278 & 0.5128 & 0.3404 \\
		& 0.05 &  & 0.3050 & 0.3050 & 0.4498 & 0.6962 & 0.4620 \\
		& 0.10 &  & 0.7124 & 0.7124 & 0.7894 & 1.2054 & 0.8008 \\
		& 0.15 &  & 1.2368 & 1.2368 & 1.1820 & 1.7116 & 1.1904 \\
		& 0.20 &  & 1.8204 & 1.8204 & 1.5992 & 2.2354 & 1.5988 \\
		& 0.25 &  & 2.4750 & 2.4750 & 2.0508 & 2.7518 & 2.0276 \\
		& 0.30 &  & 3.1474 & 3.1474 & 2.4466 & 3.2046 & 2.3960 \\
		& 0.35 &  & 4.0022 & 4.0022 & 2.9338 & 3.7136 & 2.8474 \\
		& 0.40 &  & 4.8748 & 4.8748 & 3.4404 & 4.2012 & 3.2890 \\
		& 0.45 &  & 5.8732 & 5.8732 & 3.9256 & 4.6340 & 3.7502 \\
		& 0.50 &  & 6.9972 & 6.9972 & 4.4742 & 5.1316 & 4.2082 \\
		\addlinespace[2pt]
		& 0.03 & Positive factor & 0.3132 & 0.3138 & 0.4784 & 0.7466 & 0.6518 \\
		& 0.05 &  & 0.5382 & 0.5384 & 0.7224 & 1.0612 & 0.7870 \\
		& 0.10 &  & 1.1458 & 1.1442 & 1.3282 & 1.8320 & 1.0982 \\
		& 0.15 &  & 1.8480 & 1.8466 & 2.0228 & 2.6264 & 1.4258 \\
		& 0.20 &  & 2.5350 & 2.5386 & 2.7156 & 3.3124 & 1.8266 \\
		& 0.25 &  & 3.2832 & 3.2954 & 3.4060 & 3.9894 & 2.2250 \\
		& 0.30 &  & 4.0476 & 4.0474 & 4.0926 & 4.6126 & 2.5828 \\
		& 0.35 &  & 4.8868 & 4.8998 & 4.8886 & 5.3116 & 3.0372 \\
		& 0.40 &  & 5.7296 & 5.7270 & 5.5888 & 5.9278 & 3.4648 \\
		& 0.45 &  & 6.6460 & 6.6444 & 6.3486 & 6.5762 & 3.8660 \\
		& 0.50 &  & 7.6280 & 7.5936 & 7.0904 & 7.1990 & 4.3142 \\
		\addlinespace[2pt]
		& 0.03 & Alternating factor & 0.3108 & 0.3118 & 0.4722 & 0.7264 & 0.6620 \\
		& 0.05 &  & 0.5310 & 0.5322 & 0.7258 & 1.0558 & 0.7744 \\
		& 0.10 &  & 1.1396 & 1.1432 & 1.3286 & 1.8202 & 1.0912 \\
		& 0.15 &  & 1.8410 & 1.8384 & 2.0164 & 2.6238 & 1.4350 \\
		& 0.20 &  & 2.5406 & 2.5508 & 2.7014 & 3.2920 & 1.8372 \\
		& 0.25 &  & 3.2720 & 3.2764 & 3.3854 & 3.9762 & 2.2044 \\
		& 0.30 &  & 4.0508 & 4.0548 & 4.1068 & 4.6336 & 2.5974 \\
		& 0.35 &  & 4.9208 & 4.9228 & 4.8970 & 5.3354 & 3.0328 \\
		& 0.40 &  & 5.7324 & 5.7394 & 5.5980 & 5.9466 & 3.4870 \\
		& 0.45 &  & 6.6564 & 6.6284 & 6.3344 & 6.5520 & 3.8768 \\
		& 0.50 &  & 7.6082 & 7.5810 & 7.0790 & 7.1804 & 4.3176 \\
		\addlinespace[2pt]
		& 0.03 & Linear loadings & 0.2476 & 0.2478 & 0.4054 & 0.6414 & 0.4576 \\
		& 0.05 &  & 0.4182 & 0.4202 & 0.5804 & 0.8706 & 0.5944 \\
		& 0.10 &  & 0.9350 & 0.9368 & 1.0666 & 1.5322 & 0.9014 \\
		& 0.15 &  & 1.5582 & 1.5616 & 1.6168 & 2.2080 & 1.2744 \\
		& 0.20 &  & 2.1822 & 2.1930 & 2.1642 & 2.7880 & 1.6844 \\
		& 0.25 &  & 2.8918 & 2.9082 & 2.7458 & 3.4220 & 2.0996 \\
		& 0.30 &  & 3.5918 & 3.6012 & 3.2776 & 3.9430 & 2.4622 \\
		& 0.35 &  & 4.4560 & 4.4492 & 3.9388 & 4.5598 & 2.9176 \\
		& 0.40 &  & 5.3134 & 5.2952 & 4.5494 & 5.1104 & 3.3844 \\
		& 0.45 &  & 6.2374 & 6.1980 & 5.1730 & 5.6468 & 3.7710 \\
		& 0.50 &  & 7.2720 & 7.2382 & 5.8080 & 6.1636 & 4.2440 \\
		\addlinespace[2pt]
		& 0.03 & Sparse factor & 0.2052 & 0.2058 & 0.3688 & 0.5736 & 0.3504 \\
		& 0.05 &  & 0.3598 & 0.3620 & 0.5214 & 0.7972 & 0.4794 \\
		& 0.10 &  & 0.8230 & 0.8278 & 0.9446 & 1.3804 & 0.8274 \\
		& 0.15 &  & 1.3742 & 1.3796 & 1.3920 & 1.9618 & 1.2032 \\
		& 0.20 &  & 1.9962 & 1.9984 & 1.8832 & 2.5238 & 1.6168 \\
		& 0.25 &  & 2.6712 & 2.6812 & 2.4170 & 3.1162 & 2.0254 \\
		& 0.30 &  & 3.3568 & 3.3614 & 2.8828 & 3.5934 & 2.4102 \\
		& 0.35 &  & 4.1960 & 4.1972 & 3.4348 & 4.1240 & 2.8590 \\
		& 0.40 &  & 5.0622 & 5.0512 & 3.9866 & 4.6424 & 3.3156 \\
		& 0.45 &  & 6.0350 & 6.0134 & 4.5522 & 5.1162 & 3.7462 \\
		& 0.50 &  & 7.1378 & 7.0792 & 5.1096 & 5.6332 & 4.2138 \\
		\addlinespace[2pt]
		& 0.03 & Two-block factor & 0.2496 & 0.2488 & 0.4134 & 0.6436 & 0.4434 \\
		& 0.05 &  & 0.4376 & 0.4392 & 0.6074 & 0.9094 & 0.5582 \\
		& 0.10 &  & 0.9506 & 0.9530 & 1.0920 & 1.5618 & 0.9122 \\
		& 0.15 &  & 1.5686 & 1.5732 & 1.6490 & 2.2462 & 1.2722 \\
		& 0.20 &  & 2.2112 & 2.2212 & 2.2248 & 2.8548 & 1.6786 \\
		& 0.25 &  & 2.9280 & 2.9392 & 2.8306 & 3.4900 & 2.1090 \\
		& 0.30 &  & 3.6498 & 3.6416 & 3.3854 & 4.0278 & 2.4500 \\
		& 0.35 &  & 4.4922 & 4.4954 & 4.0696 & 4.6454 & 2.9258 \\
		& 0.40 &  & 5.3148 & 5.3200 & 4.6712 & 5.1738 & 3.3522 \\
		& 0.45 &  & 6.2896 & 6.2826 & 5.3242 & 5.7594 & 3.7832 \\
		& 0.50 &  & 7.3378 & 7.2922 & 5.9604 & 6.2892 & 4.2410 \\
		\addlinespace[2pt]
		\midrule
		50 & 0.02 & Independence & 0.2950 & 0.2950 & 0.4480 & 0.8596 & 0.4638 \\
		& 0.05 &  & 0.8750 & 0.8750 & 0.9932 & 1.7636 & 1.0170 \\
		& 0.10 &  & 2.0964 & 2.0964 & 2.0422 & 3.2670 & 2.0590 \\
		& 0.15 &  & 3.4732 & 3.4732 & 3.1478 & 4.5308 & 3.1544 \\
		& 0.20 &  & 4.9826 & 4.9826 & 4.3312 & 5.7562 & 4.2758 \\
		& 0.25 &  & 6.6522 & 6.6522 & 5.5760 & 6.8920 & 5.4670 \\
		& 0.30 &  & 8.5056 & 8.5056 & 6.9322 & 8.0726 & 6.7308 \\
		& 0.35 &  & 10.4848 & 10.4848 & 8.3184 & 9.1932 & 7.9822 \\
		& 0.40 &  & 12.7590 & 12.7590 & 9.7570 & 10.4022 & 9.2848 \\
		& 0.45 &  & 15.1890 & 15.1890 & 11.2382 & 11.5122 & 10.5870 \\
		& 0.50 &  & 17.9120 & 17.9120 & 12.7442 & 12.6474 & 11.8666 \\
		\addlinespace[2pt]
		& 0.02 & Positive factor & 0.5350 & 0.5350 & 0.7308 & 1.3108 & 1.2394 \\
		& 0.05 &  & 1.4440 & 1.4450 & 1.7110 & 2.6594 & 1.7358 \\
		& 0.10 &  & 3.1582 & 3.1420 & 3.4902 & 4.6144 & 2.6860 \\
		& 0.15 &  & 4.9104 & 4.9022 & 5.3078 & 6.3246 & 3.8288 \\
		& 0.20 &  & 6.7096 & 6.6954 & 7.1238 & 7.9570 & 4.9492 \\
		& 0.25 &  & 8.6396 & 8.6366 & 9.0572 & 9.5774 & 6.0690 \\
		& 0.30 &  & 10.7058 & 10.7094 & 11.0830 & 11.3120 & 7.2406 \\
		& 0.35 &  & 12.7388 & 12.7536 & 13.0348 & 12.9172 & 8.4754 \\
		& 0.40 &  & 15.0020 & 15.0010 & 15.2030 & 14.6528 & 9.6592 \\
		& 0.45 &  & 17.1742 & 17.1550 & 17.1480 & 16.2254 & 10.8114 \\
		& 0.50 &  & 19.5986 & 19.5696 & 19.2556 & 17.9088 & 12.2166 \\
		\addlinespace[2pt]
		& 0.02 & Alternating factor & 0.5462 & 0.5356 & 0.7306 & 1.3098 & 1.2758 \\
		& 0.05 &  & 1.4516 & 1.4452 & 1.7090 & 2.6670 & 1.7756 \\
		& 0.10 &  & 3.1532 & 3.1386 & 3.5016 & 4.6198 & 2.7056 \\
		& 0.15 &  & 4.9120 & 4.9010 & 5.2918 & 6.2868 & 3.8346 \\
		& 0.20 &  & 6.7268 & 6.7252 & 7.1264 & 7.9330 & 4.9534 \\
		& 0.25 &  & 8.6476 & 8.6420 & 9.0438 & 9.5794 & 6.0916 \\
		& 0.30 &  & 10.6966 & 10.6998 & 11.0842 & 11.3072 & 7.3004 \\
		& 0.35 &  & 12.7244 & 12.7380 & 13.0136 & 12.8954 & 8.4380 \\
		& 0.40 &  & 14.9842 & 14.9698 & 15.1740 & 14.6542 & 9.7102 \\
		& 0.45 &  & 17.1744 & 17.1846 & 17.1528 & 16.2302 & 10.8396 \\
		& 0.50 &  & 19.6444 & 19.6270 & 19.3462 & 17.9746 & 12.1584 \\
		\addlinespace[2pt]
		& 0.02 & Linear loadings & 0.4128 & 0.4136 & 0.5934 & 1.1084 & 0.7514 \\
		& 0.05 &  & 1.1794 & 1.1798 & 1.3706 & 2.2564 & 1.2392 \\
		& 0.10 &  & 2.6504 & 2.6512 & 2.8016 & 3.9956 & 2.2656 \\
		& 0.15 &  & 4.2640 & 4.2632 & 4.3428 & 5.5160 & 3.4036 \\
		& 0.20 &  & 5.8968 & 5.8964 & 5.8354 & 6.9482 & 4.5482 \\
		& 0.25 &  & 7.6964 & 7.7028 & 7.4494 & 8.3372 & 5.6966 \\
		& 0.30 &  & 9.6476 & 9.6716 & 9.1448 & 9.7732 & 6.9218 \\
		& 0.35 &  & 11.6386 & 11.6374 & 10.8382 & 11.1350 & 8.1564 \\
		& 0.40 &  & 13.9060 & 13.9144 & 12.7080 & 12.6264 & 9.4258 \\
		& 0.45 &  & 16.1864 & 16.1614 & 14.4532 & 13.9928 & 10.6282 \\
		& 0.50 &  & 18.6992 & 18.6664 & 16.2944 & 15.4258 & 12.0466 \\
		\addlinespace[2pt]
		& 0.02 & Sparse factor & 0.3578 & 0.3580 & 0.5212 & 0.9884 & 0.5186 \\
		& 0.05 &  & 1.0218 & 1.0232 & 1.1686 & 2.0152 & 1.0628 \\
		& 0.10 &  & 2.3830 & 2.3872 & 2.4500 & 3.6888 & 2.1124 \\
		& 0.15 &  & 3.8462 & 3.8560 & 3.7372 & 5.0346 & 3.2150 \\
		& 0.20 &  & 5.4470 & 5.4516 & 5.1168 & 6.3714 & 4.3488 \\
		& 0.25 &  & 7.1724 & 7.1866 & 6.5430 & 7.6280 & 5.4878 \\
		& 0.30 &  & 9.0822 & 9.0968 & 8.1174 & 8.9608 & 6.7790 \\
		& 0.35 &  & 11.0882 & 11.0982 & 9.6456 & 10.1974 & 8.0326 \\
		& 0.40 &  & 13.3340 & 13.3290 & 11.2418 & 11.4982 & 9.3086 \\
		& 0.45 &  & 15.6602 & 15.6572 & 12.8620 & 12.7356 & 10.5760 \\
		& 0.50 &  & 18.3134 & 18.3042 & 14.5616 & 14.0226 & 11.9338 \\
		\addlinespace[2pt]
		& 0.02 & Two-block factor & 0.4464 & 0.4358 & 0.6168 & 1.1372 & 0.7328 \\
		& 0.05 &  & 1.2036 & 1.2048 & 1.3976 & 2.3158 & 1.2812 \\
		& 0.10 &  & 2.7264 & 2.7178 & 2.9252 & 4.1102 & 2.2954 \\
		& 0.15 &  & 4.2964 & 4.2974 & 4.4070 & 5.6050 & 3.4188 \\
		& 0.20 &  & 5.9836 & 5.9908 & 5.9810 & 7.0550 & 4.5384 \\
		& 0.25 &  & 7.8146 & 7.8186 & 7.6406 & 8.4876 & 5.6818 \\
		& 0.30 &  & 9.7846 & 9.7846 & 9.4200 & 9.9704 & 6.9354 \\
		& 0.35 &  & 11.7980 & 11.8112 & 11.1248 & 11.3646 & 8.1792 \\
		& 0.40 &  & 14.0488 & 14.0506 & 12.9866 & 12.8496 & 9.4346 \\
		& 0.45 &  & 16.3164 & 16.3102 & 14.7720 & 14.2542 & 10.6594 \\
		& 0.50 &  & 18.8488 & 18.8264 & 16.6178 & 15.6696 & 12.0048 \\
		\addlinespace[2pt]
		\midrule
		100 & 0.01 & Independence & 0.2934 & 0.2934 & 0.4514 & 1.0364 & 0.4706 \\
		& 0.03 &  & 1.1056 & 1.1056 & 1.2506 & 2.5640 & 1.2728 \\
		& 0.05 &  & 1.9646 & 1.9646 & 2.0384 & 3.7078 & 2.0676 \\
		& 0.10 &  & 4.5160 & 4.5160 & 4.3044 & 6.3624 & 4.3086 \\
		& 0.15 &  & 7.3460 & 7.3460 & 6.7026 & 8.7088 & 6.6446 \\
		& 0.20 &  & 10.5440 & 10.5440 & 9.3864 & 11.0488 & 9.2106 \\
		& 0.25 &  & 13.9728 & 13.9728 & 12.1336 & 13.3090 & 11.8040 \\
		& 0.30 &  & 17.6852 & 17.6852 & 15.0066 & 15.5954 & 14.4932 \\
		& 0.35 &  & 21.7452 & 21.7452 & 18.0282 & 17.9440 & 17.2426 \\
		& 0.40 &  & 26.1136 & 26.1136 & 21.0762 & 20.2304 & 19.9566 \\
		& 0.45 &  & 30.8952 & 30.8952 & 24.2652 & 22.5530 & 22.7564 \\
		& 0.50 &  & 36.3894 & 36.3894 & 27.7470 & 25.0374 & 25.8282 \\
		\addlinespace[2pt]
		& 0.01 & Positive factor & 0.5716 & 0.5522 & 0.7528 & 1.6124 & 2.3648 \\
		& 0.03 &  & 1.8474 & 1.8104 & 2.1272 & 3.6776 & 2.9876 \\
		& 0.05 &  & 3.0888 & 3.0872 & 3.4936 & 5.1892 & 3.5192 \\
		& 0.10 &  & 6.4822 & 6.4890 & 7.1252 & 8.6090 & 5.8458 \\
		& 0.15 &  & 10.0692 & 10.0684 & 10.8856 & 11.9270 & 7.9734 \\
		& 0.20 &  & 13.9142 & 13.9150 & 14.8570 & 15.2874 & 10.2900 \\
		& 0.25 &  & 17.8380 & 17.8406 & 18.8886 & 18.6720 & 12.9738 \\
		& 0.30 &  & 21.9202 & 21.9020 & 22.9942 & 22.0394 & 15.2666 \\
		& 0.35 &  & 26.1596 & 26.1532 & 27.2238 & 25.4504 & 18.2614 \\
		& 0.40 &  & 30.4640 & 30.4438 & 31.4094 & 28.7772 & 20.7786 \\
		& 0.45 &  & 34.9582 & 34.9264 & 35.6268 & 32.1438 & 23.3220 \\
		& 0.50 &  & 39.8322 & 39.8072 & 40.1634 & 35.6550 & 26.3982 \\
		\addlinespace[2pt]
		& 0.01 & Alternating factor & 0.5496 & 0.5498 & 0.7498 & 1.6230 & 2.3350 \\
		& 0.03 &  & 1.8128 & 1.8138 & 2.1228 & 3.6800 & 2.9952 \\
		& 0.05 &  & 3.0868 & 3.0874 & 3.4972 & 5.1880 & 3.5696 \\
		& 0.10 &  & 6.5128 & 6.4962 & 7.1280 & 8.6102 & 5.8736 \\
		& 0.15 &  & 10.1346 & 10.0900 & 10.9072 & 11.9318 & 7.9758 \\
		& 0.20 &  & 13.9010 & 13.9074 & 14.8640 & 15.3132 & 10.3238 \\
		& 0.25 &  & 17.8406 & 17.8518 & 18.8926 & 18.6640 & 12.9806 \\
		& 0.30 &  & 21.9338 & 21.8998 & 22.9950 & 22.0376 & 15.2804 \\
		& 0.35 &  & 26.1364 & 26.1492 & 27.2434 & 25.4430 & 18.3060 \\
		& 0.40 &  & 30.4518 & 30.4372 & 31.3996 & 28.7554 & 20.8104 \\
		& 0.45 &  & 34.9616 & 34.9842 & 35.6700 & 32.1618 & 23.3112 \\
		& 0.50 &  & 39.7758 & 39.7838 & 40.1560 & 35.6418 & 26.3118 \\
		\addlinespace[2pt]
		& 0.01 & Linear loadings & 0.4420 & 0.4230 & 0.6094 & 1.3420 & 1.1196 \\
		& 0.03 &  & 1.5002 & 1.4810 & 1.7004 & 3.1872 & 1.8796 \\
		& 0.05 &  & 2.5646 & 2.5632 & 2.8110 & 4.5620 & 2.6312 \\
		& 0.10 &  & 5.5908 & 5.5764 & 5.8324 & 7.6068 & 4.8450 \\
		& 0.15 &  & 8.7924 & 8.8028 & 8.9632 & 10.3998 & 7.1812 \\
		& 0.20 &  & 12.3688 & 12.3796 & 12.4034 & 13.3294 & 9.6212 \\
		& 0.25 &  & 16.0332 & 16.0286 & 15.8488 & 16.1768 & 12.2604 \\
		& 0.30 &  & 19.9360 & 19.9368 & 19.3874 & 19.0528 & 14.7770 \\
		& 0.35 &  & 24.0656 & 24.0836 & 23.0726 & 21.9532 & 17.6892 \\
		& 0.40 &  & 28.3838 & 28.3710 & 26.7470 & 24.7846 & 20.2888 \\
		& 0.45 &  & 32.9704 & 32.9498 & 30.5602 & 27.6900 & 22.9702 \\
		& 0.50 &  & 38.0752 & 38.0570 & 34.6252 & 30.7074 & 26.0800 \\
		\addlinespace[2pt]
		& 0.01 & Sparse factor & 0.3658 & 0.3664 & 0.5338 & 1.2136 & 0.6164 \\
		& 0.03 &  & 1.3336 & 1.3140 & 1.5006 & 2.9186 & 1.4110 \\
		& 0.05 &  & 2.2990 & 2.2992 & 2.4694 & 4.1866 & 2.1910 \\
		& 0.10 &  & 5.0748 & 5.0768 & 5.1220 & 7.0516 & 4.3940 \\
		& 0.15 &  & 8.1464 & 8.1468 & 7.9900 & 9.6658 & 6.8240 \\
		& 0.20 &  & 11.5236 & 11.5258 & 11.0172 & 12.2686 & 9.2684 \\
		& 0.25 &  & 15.0714 & 15.0814 & 14.1536 & 14.8294 & 11.9092 \\
		& 0.30 &  & 18.9028 & 18.9060 & 17.3948 & 17.4494 & 14.4854 \\
		& 0.35 &  & 22.9680 & 22.9786 & 20.7642 & 20.0578 & 17.3854 \\
		& 0.40 &  & 27.3202 & 27.3296 & 24.1726 & 22.6556 & 20.0604 \\
		& 0.45 &  & 32.0264 & 32.0216 & 27.6726 & 25.2890 & 22.8240 \\
		& 0.50 &  & 37.2918 & 37.2850 & 31.3936 & 28.0442 & 25.8400 \\
		\addlinespace[2pt]
		& 0.01 & Two-block factor & 0.4470 & 0.4474 & 0.6234 & 1.3930 & 1.1326 \\
		& 0.03 &  & 1.5334 & 1.5130 & 1.7616 & 3.2544 & 1.8774 \\
		& 0.05 &  & 2.6352 & 2.6368 & 2.9012 & 4.6324 & 2.5910 \\
		& 0.10 &  & 5.6564 & 5.6592 & 5.9628 & 7.7040 & 4.7916 \\
		& 0.15 &  & 8.9692 & 8.9542 & 9.2034 & 10.6098 & 7.1438 \\
		& 0.20 &  & 12.5448 & 12.5320 & 12.6638 & 13.5422 & 9.5802 \\
		& 0.25 &  & 16.2234 & 16.2374 & 16.1896 & 16.4518 & 12.2240 \\
		& 0.30 &  & 20.1592 & 20.1558 & 19.7922 & 19.3574 & 14.7308 \\
		& 0.35 &  & 24.3038 & 24.3152 & 23.5446 & 22.3368 & 17.6290 \\
		& 0.40 &  & 28.6452 & 28.6440 & 27.2972 & 25.2240 & 20.2622 \\
		& 0.45 &  & 33.2172 & 33.2366 & 31.0892 & 28.1190 & 22.9490 \\
		& 0.50 &  & 38.3472 & 38.3314 & 35.2140 & 31.2716 & 26.0442 \\
		\addlinespace[2pt]
	\end{longtable}
\normalsize


\section{Additional Simulation Study Under General Covariance Dependence}

The experiments reported in Section~\ref{sec:SIMULATION} focused on one-factor dependence structures because exact Bayes Oracle risks can be computed efficiently in that setting. We now move beyond the one-factor framework and investigate the behavior of BSD and competing procedures under a broader collection of covariance structures. Throughout this appendix, we fix $n=200$, $\alpha=0.10$, $\psi^2=2\log n$, and $G=5000$ Monte Carlo replications. Since exact Oracle evaluation becomes computationally infeasible under these general covariance models, our comparisons focus on the implementable procedures BSD, MRD--GBS, MRD--CSX, and BH.

\subsection*{Simulation Design}

To assess the robustness of the competing procedures under a broad range of dependence settings, we consider six covariance structures spanning both global and localized forms of dependence.

\paragraph{Equicorrelation Model.}
The covariance matrix is given by
\[
\Sigma_{ij}
=
\rho +(1-\rho)I(i=j),
\qquad \rho =0.7.
\]

\paragraph{Toeplitz Model.}
The covariance matrix is specified by
\[
\Sigma_{ij}
=
\rho^{|i-j|},
\qquad \rho =0.9.
\]

\paragraph{Fractional Gaussian Noise Model.}
The covariance matrix is generated from a fractional Gaussian noise process with Hurst parameter
\[
H=0.9.
\]
This model exhibits long-range dependence whose strength increases with \(H\).

\paragraph{Heterogeneous Block Model.}
The coordinates are partitioned into blocks of varying sizes. Within each block, observations follow an equicorrelated covariance structure, while observations belonging to different blocks are independent. The within-block correlations vary across blocks.

\paragraph{Factor Model.}
Observations are generated according to
\[
X_i
=
\theta_i+\lambda_iF+\varepsilon_i,
\qquad
i=1,\ldots,n,
\]
where \(F\sim N(0,1)\) is a common factor,
\(\varepsilon_i\stackrel{\mathrm{ind}}{\sim}N(0,1)\),
and \(\lambda_i\) are independently generated from a
\(\mathrm{Uniform}(0.5,1)\) distribution.

\paragraph{Sparse Precision Model.}
The precision matrix
\[
\boldsymbol{\Omega}
=
\boldsymbol{\Sigma}^{-1}
\]
is generated to be sparse, with each coordinate directly connected to only a small number of neighboring coordinates. The covariance matrix is obtained by inverting \(\boldsymbol{\Omega}\).

Collectively, these covariance structures range from highly global forms of dependence to substantially more localized conditional dependence patterns. This diversity allows us to investigate the robustness of BSD and competing procedures across a broad spectrum of covariance geometries extending well beyond the one-factor models considered in Section~\ref{sec:SIMULATION}.

\subsection*{Bayes Misclassification Risks}

Figure~\ref{fig:general-cov-risk} displays the Bayes misclassification risks of BSD, MRD--GBS, MRD--CSX, and BH under the six covariance structures considered.

\begin{figure}[!htbp]
	\centering
	
	\begin{subfigure}[b]{0.48\textwidth}
		\centering
		\includegraphics[width=\textwidth]{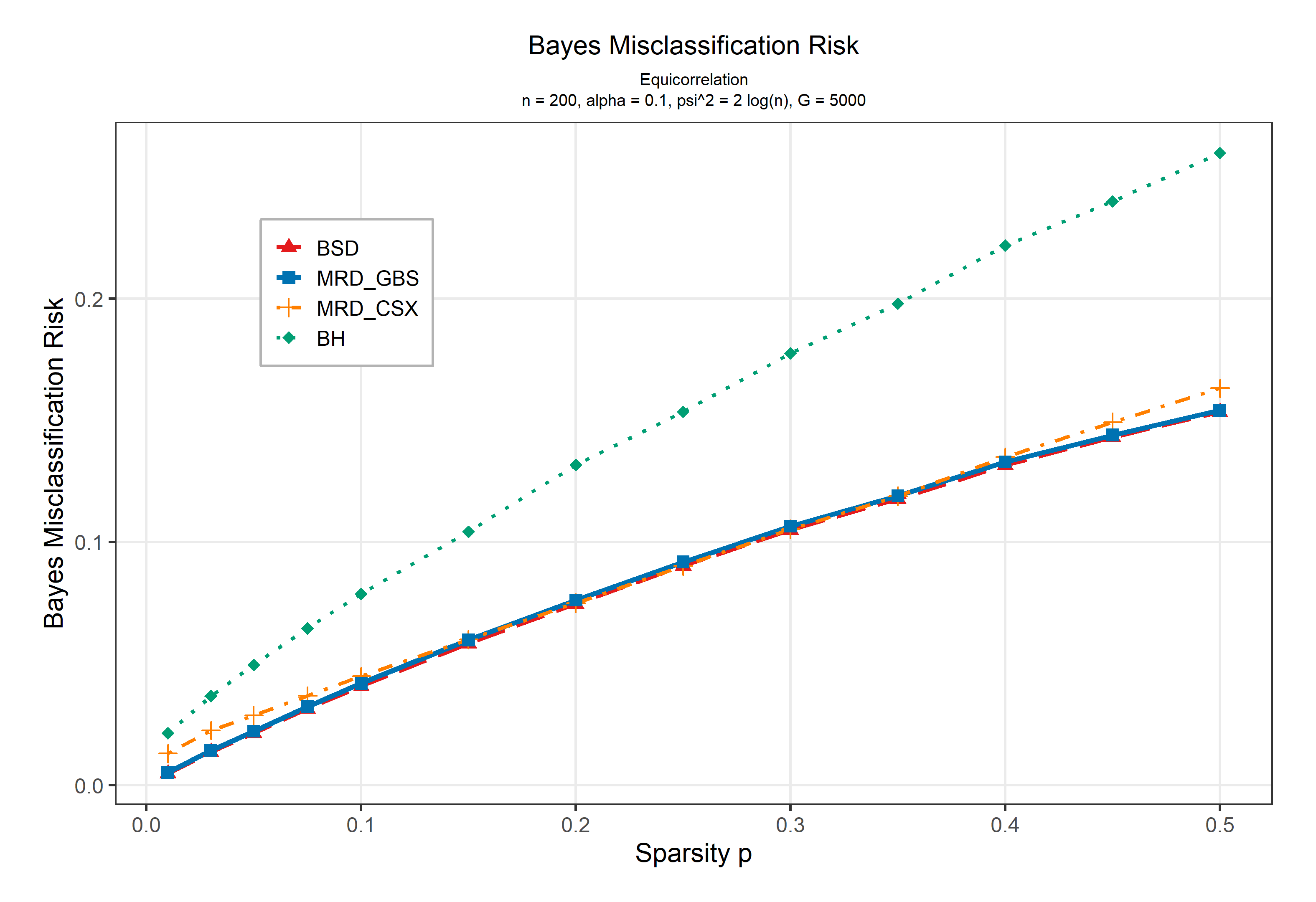}
		\caption{Equicorrelation}
	\end{subfigure}
	\hfill
	\begin{subfigure}[b]{0.48\textwidth}
		\centering
		\includegraphics[width=\textwidth]{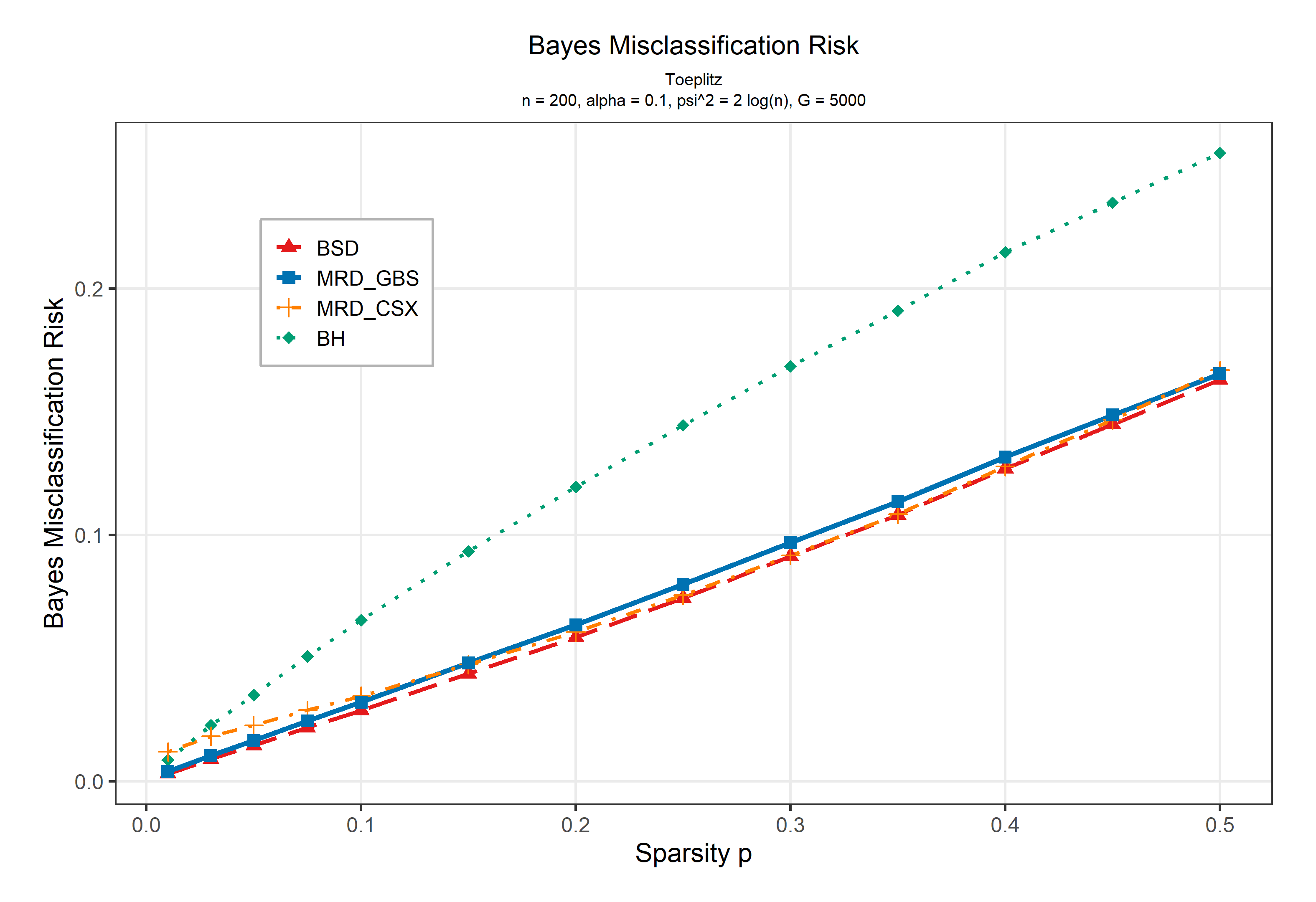}
		\caption{Toeplitz}
	\end{subfigure}
	
	\vspace{0.25cm}
	
	\begin{subfigure}[b]{0.48\textwidth}
		\centering
		\includegraphics[width=\textwidth]{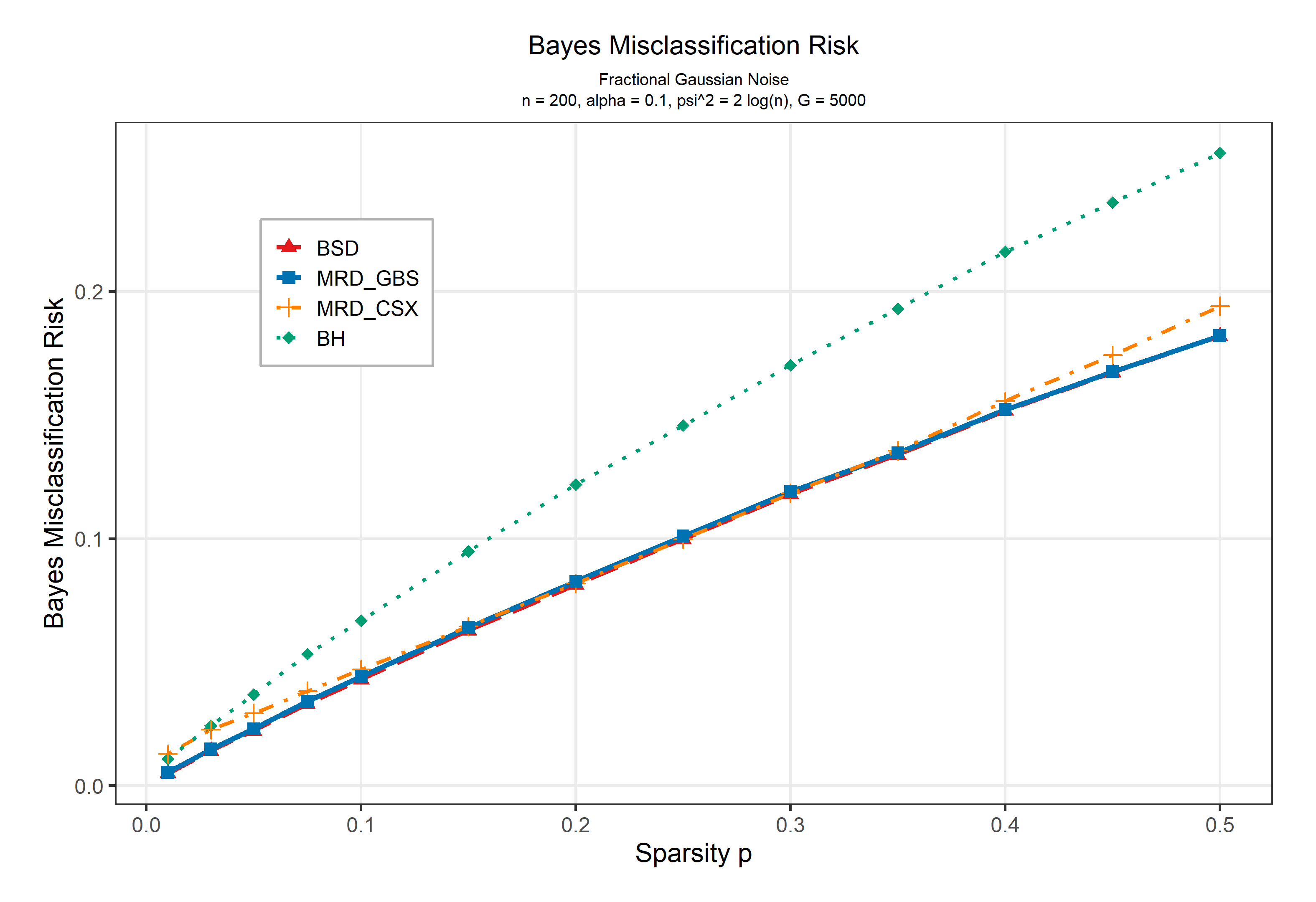}
		\caption{Fractional Gaussian Noise}
	\end{subfigure}
	\hfill
	\begin{subfigure}[b]{0.48\textwidth}
		\centering
		\includegraphics[width=\textwidth]{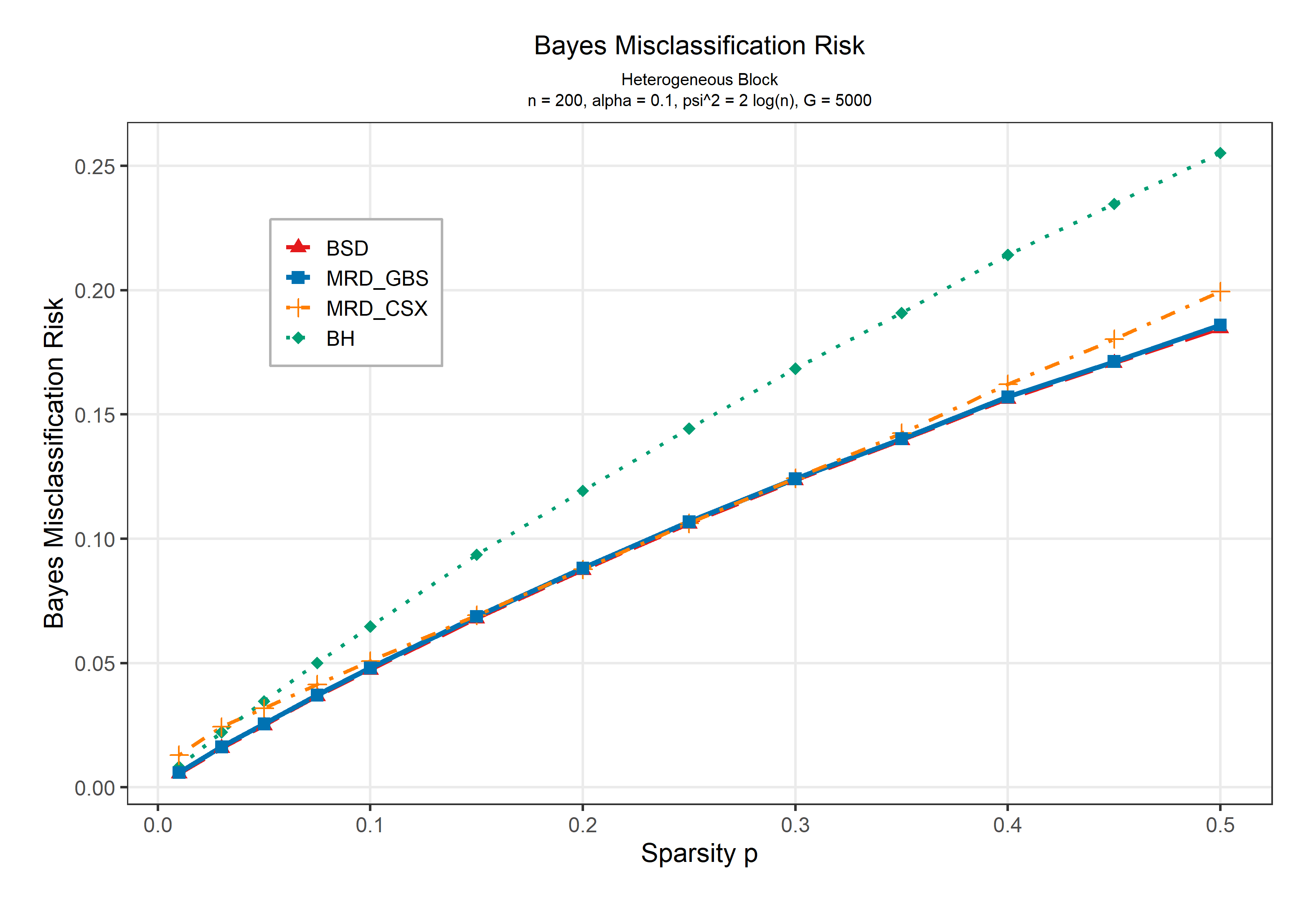}
		\caption{Heterogeneous Block}
	\end{subfigure}
	
	\vspace{0.25cm}
	
	\begin{subfigure}[b]{0.48\textwidth}
		\centering
		\includegraphics[width=\textwidth]{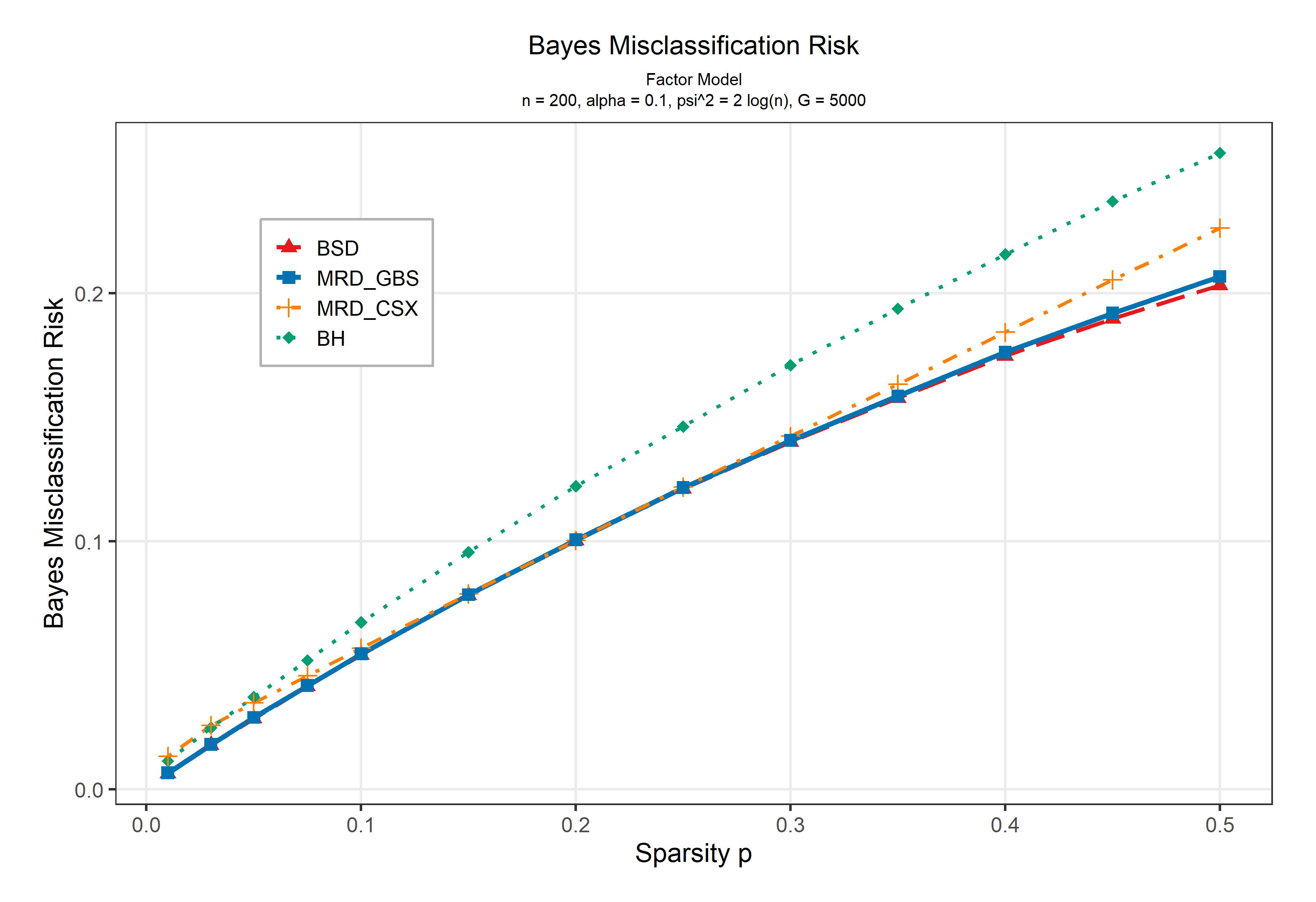}
		\caption{Factor Model}
	\end{subfigure}
	\hfill
	\begin{subfigure}[b]{0.48\textwidth}
		\centering
		\includegraphics[width=\textwidth]{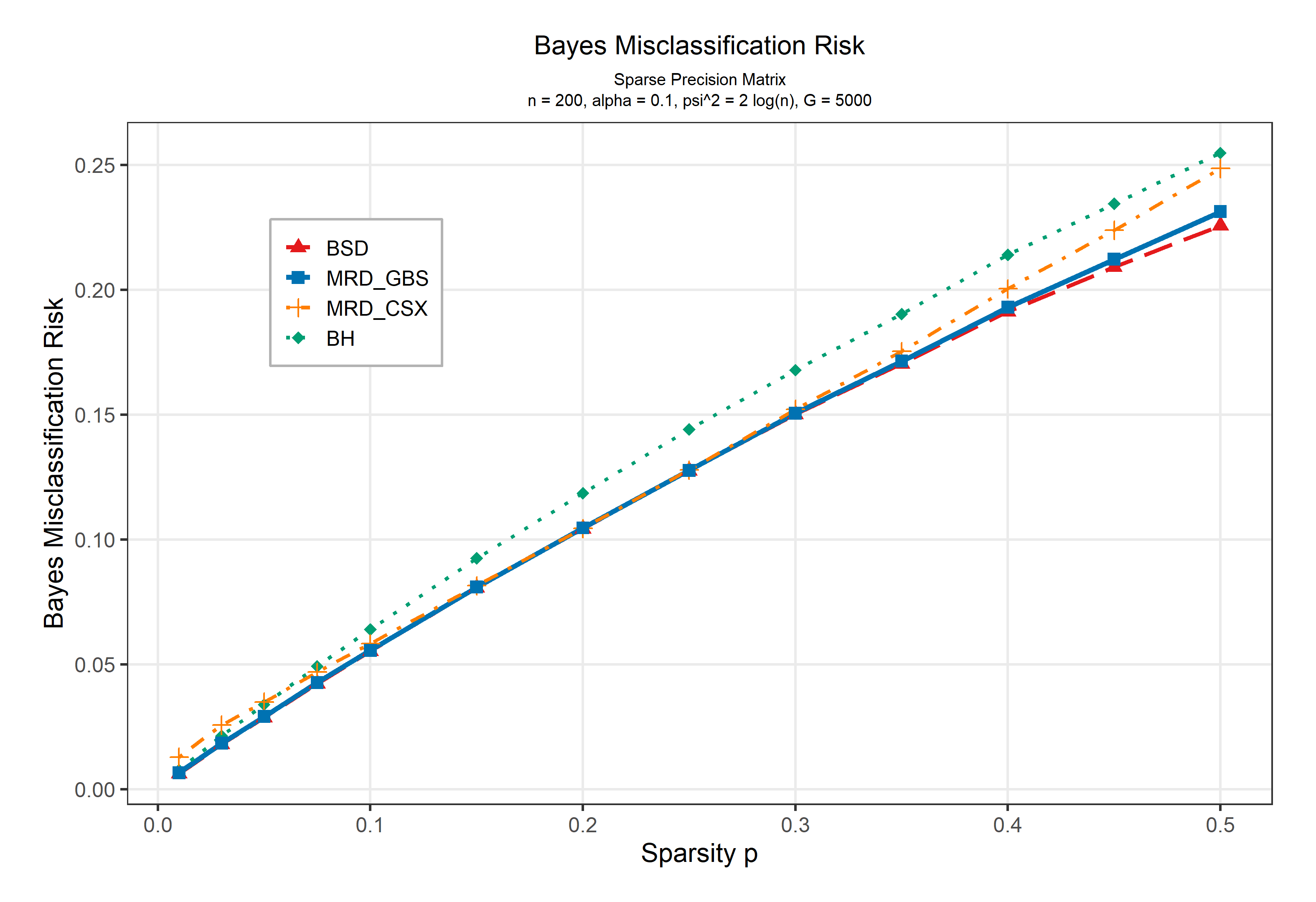}
		\caption{Sparse Precision Matrix}
	\end{subfigure}
	
	\caption{
		Bayes misclassification risks of BSD, MRD--GBS, MRD--CSX, and BH under six general covariance structures with $n=200$, $\alpha=0.10$, $\psi^2=2\log n$, and $G=5000$ Monte Carlo replications. Exact Oracle risks are unavailable in these settings because evaluating posterior inclusion probabilities under general covariance dependence is computationally infeasible.
	}
	\label{fig:general-cov-risk}
\end{figure}

Table~\ref{tab:appendix-risk-general-cov} below reports the Bayes misclassification risks of the competing procedures under the six general covariance structures considered above. The experiments are conducted with $n=200$, $\alpha=0.10$, $\psi^2=2\log n$, and $G=5000$ Monte Carlo replications.

\small
\setlength{\tabcolsep}{4.2pt}
\renewcommand{\arraystretch}{1.08}

\normalsize

\small
\setlength{\tabcolsep}{4.2pt}
\renewcommand{\arraystretch}{1.08}
\begin{longtable}{c l rrrr}
	\caption{Bayes misclassification risks of the competing procedures under six general covariance structures for $n=200$, $\alpha=0.10$, $\psi^2=2\log n$, and $G=5000$ Monte Carlo replications.}
	\label{tab:appendix-risk-general-cov}\\
	\toprule
	$p$ & Dependence structure & BSD & MRD--GBS & MRD--CSX & BH \\
	\midrule
	\endfirsthead
	\caption[]{Bayes misclassification risks of the competing procedures under six general covariance structures (continued).}\\
	\toprule
	$p$ & Dependence structure & BSD & MRD--GBS & MRD--CSX & BH \\
	\midrule
	\endhead
	\midrule
	\multicolumn{6}{r}{\emph{Continued on next page}}\\
	\endfoot
	\bottomrule
	\endlastfoot
	0.01 & Equicorrelation & 0.0047 & 0.0053 & 0.0131 & 0.0215 \\
	0.03 &  & 0.0136 & 0.0143 & 0.0226 & 0.0367 \\
	0.05 &  & 0.0213 & 0.0222 & 0.0287 & 0.0495 \\
	0.07 &  & 0.0314 & 0.0324 & 0.0367 & 0.0645 \\
	0.10 &  & 0.0407 & 0.0418 & 0.0448 & 0.0786 \\
	0.15 &  & 0.0582 & 0.0597 & 0.0600 & 0.1042 \\
	0.20 &  & 0.0745 & 0.0761 & 0.0749 & 0.1317 \\
	0.25 &  & 0.0902 & 0.0918 & 0.0902 & 0.1535 \\
	0.30 &  & 0.1049 & 0.1065 & 0.1053 & 0.1776 \\
	0.35 &  & 0.1177 & 0.1190 & 0.1189 & 0.1980 \\
	0.40 &  & 0.1316 & 0.1329 & 0.1349 & 0.2219 \\
	0.45 &  & 0.1431 & 0.1439 & 0.1492 & 0.2401 \\
	0.50 &  & 0.1533 & 0.1541 & 0.1632 & 0.2600 \\
	\addlinespace[2pt] 
	0.01 & Toeplitz & 0.0031 & 0.0040 & 0.0120 & 0.0087 \\
	0.03 &  & 0.0091 & 0.0104 & 0.0183 & 0.0228 \\
	0.05 &  & 0.0145 & 0.0165 & 0.0227 & 0.0350 \\
	0.07 &  & 0.0219 & 0.0245 & 0.0289 & 0.0507 \\
	0.10 &  & 0.0289 & 0.0322 & 0.0346 & 0.0655 \\
	0.15 &  & 0.0436 & 0.0481 & 0.0474 & 0.0934 \\
	0.20 &  & 0.0585 & 0.0635 & 0.0606 & 0.1194 \\
	0.25 &  & 0.0744 & 0.0799 & 0.0755 & 0.1446 \\
	0.30 &  & 0.0913 & 0.0970 & 0.0917 & 0.1685 \\
	0.35 &  & 0.1082 & 0.1135 & 0.1085 & 0.1911 \\
	0.40 &  & 0.1269 & 0.1316 & 0.1277 & 0.2148 \\
	0.45 &  & 0.1449 & 0.1487 & 0.1468 & 0.2349 \\
	0.50 &  & 0.1631 & 0.1655 & 0.1669 & 0.2551 \\
	\addlinespace[2pt]
	0.01 & Fractional Gaussian Noise & 0.0049 & 0.0054 & 0.0129 & 0.0108 \\
	0.03 &  & 0.0141 & 0.0148 & 0.0227 & 0.0243 \\
	0.05 &  & 0.0222 & 0.0230 & 0.0293 & 0.0369 \\
	0.07 &  & 0.0331 & 0.0341 & 0.0383 & 0.0533 \\
	0.10 &  & 0.0430 & 0.0442 & 0.0471 & 0.0668 \\
	0.15 &  & 0.0627 & 0.0640 & 0.0644 & 0.0948 \\
	0.20 &  & 0.0813 & 0.0827 & 0.0818 & 0.1219 \\
	0.25 &  & 0.0998 & 0.1010 & 0.0997 & 0.1458 \\
	0.30 &  & 0.1179 & 0.1190 & 0.1183 & 0.1702 \\
	0.35 &  & 0.1339 & 0.1346 & 0.1355 & 0.1930 \\
	0.40 &  & 0.1517 & 0.1522 & 0.1556 & 0.2160 \\
	0.45 &  & 0.1672 & 0.1675 & 0.1742 & 0.2360 \\
	0.50 &  & 0.1819 & 0.1821 & 0.1941 & 0.2560 \\
	\addlinespace[2pt]
	0.01 & Heterogeneous block & 0.0055 & 0.0060 & 0.0130 & 0.0084 \\
	0.03 &  & 0.0158 & 0.0163 & 0.0244 & 0.0222 \\
	0.05 &  & 0.0249 & 0.0256 & 0.0318 & 0.0347 \\
	0.07 &  & 0.0367 & 0.0372 & 0.0414 & 0.0501 \\
	0.10 &  & 0.0473 & 0.0481 & 0.0508 & 0.0648 \\
	0.15 &  & 0.0680 & 0.0688 & 0.0693 & 0.0936 \\
	0.20 &  & 0.0875 & 0.0882 & 0.0877 & 0.1193 \\
	0.25 &  & 0.1061 & 0.1068 & 0.1062 & 0.1444 \\
	0.30 &  & 0.1235 & 0.1242 & 0.1244 & 0.1685 \\
	0.35 &  & 0.1396 & 0.1402 & 0.1426 & 0.1908 \\
	0.40 &  & 0.1563 & 0.1571 & 0.1622 & 0.2142 \\
	0.45 &  & 0.1707 & 0.1713 & 0.1804 & 0.2347 \\
	0.50 &  & 0.1847 & 0.1860 & 0.1995 & 0.2552 \\
	\addlinespace[2pt]
	0.01 & Factor model & 0.0064 & 0.0067 & 0.0133 & 0.0115 \\
	0.03 &  & 0.0179 & 0.0182 & 0.0259 & 0.0249 \\
	0.05 &  & 0.0286 & 0.0290 & 0.0350 & 0.0373 \\
	0.07 &  & 0.0416 & 0.0419 & 0.0458 & 0.0521 \\
	0.10 &  & 0.0542 & 0.0546 & 0.0570 & 0.0674 \\
	0.15 &  & 0.0783 & 0.0785 & 0.0789 & 0.0956 \\
	0.20 &  & 0.1003 & 0.1006 & 0.1004 & 0.1223 \\
	0.25 &  & 0.1212 & 0.1216 & 0.1219 & 0.1463 \\
	0.30 &  & 0.1401 & 0.1407 & 0.1424 & 0.1710 \\
	0.35 &  & 0.1578 & 0.1585 & 0.1633 & 0.1938 \\
	0.40 &  & 0.1747 & 0.1762 & 0.1843 & 0.2157 \\
	0.45 &  & 0.1897 & 0.1920 & 0.2054 & 0.2370 \\
	0.50 &  & 0.2032 & 0.2067 & 0.2263 & 0.2566 \\
	\addlinespace[2pt]
	0.01 & Sparse precision matrix & 0.0062 & 0.0066 & 0.0129 & 0.0076 \\
	0.03 &  & 0.0181 & 0.0184 & 0.0257 & 0.0213 \\
	0.05 &  & 0.0287 & 0.0291 & 0.0349 & 0.0339 \\
	0.07 &  & 0.0423 & 0.0427 & 0.0469 & 0.0493 \\
	0.10 &  & 0.0552 & 0.0556 & 0.0582 & 0.0640 \\
	0.15 &  & 0.0806 & 0.0809 & 0.0815 & 0.0925 \\
	0.20 &  & 0.1043 & 0.1046 & 0.1044 & 0.1186 \\
	0.25 &  & 0.1276 & 0.1277 & 0.1278 & 0.1441 \\
	0.30 &  & 0.1502 & 0.1506 & 0.1521 & 0.1678 \\
	0.35 &  & 0.1704 & 0.1714 & 0.1753 & 0.1903 \\
	0.40 &  & 0.1912 & 0.1930 & 0.2004 & 0.2141 \\
	0.45 &  & 0.2092 & 0.2122 & 0.2239 & 0.2344 \\
	0.50 &  & 0.2257 & 0.2313 & 0.2486 & 0.2548 \\
\end{longtable}

Several conclusions emerge from Figure~\ref{fig:general-cov-risk} and Table~\ref{tab:appendix-risk-general-cov}. First, BSD consistently attains the smallest Bayes misclassification risks across all six covariance structures considered. Second, MRD--GBS remains remarkably close to BSD throughout the entire sparsity range despite arising from a fundamentally different inferential framework. Third, the gap between BSD and BH is substantial across all dependence structures, illustrating the benefits of incorporating covariance information into the testing procedure. Finally, the close agreement between BSD and MRD--GBS persists even under covariance geometries substantially more complex than the one-factor models considered in Section~\ref{sec:SIMULATION}, including Toeplitz, heterogeneous block, and sparse precision matrix structures.

Taken together, these experiments suggest that the favorable behavior of BSD extends well beyond the one-factor settings for which Oracle benchmarking is available. Equally noteworthy is the persistent proximity between BSD and MRD--GBS across all dependence structures considered. Combined with the Oracle comparisons of Section~\ref{sec:SIMULATION}, these findings provide further empirical evidence that suitably calibrated covariance-adaptive residual procedures can closely mimic the decision-theoretic performance of substantially more computationally intensive Bayesian model-selection approaches.

%
%
%


\bibliography{BSD_Reference.bib}

\end{document}